\documentclass[11pt,hidelinks]{article}
\usepackage{amsmath,amsfonts}
\usepackage{enumerate}
\usepackage{natbib}
\usepackage{url} 
\usepackage{soul} 
\usepackage{amsthm}
\usepackage{bbm}
\usepackage[pdftex]{graphicx}
\usepackage{indentfirst}
\usepackage{verbatim}
\usepackage{setspace}
\usepackage{mathtools}
\usepackage[capposition=top]{floatrow}
\usepackage{booktabs}
\usepackage{enumitem}
\usepackage[flushleft]{threeparttable}
\usepackage{comment}
\usepackage{color}
\usepackage{bm}
\usepackage{subcaption}
\usepackage{changepage}
\usepackage{hyperref}
\usepackage[toc,title,page]{appendix}
\usepackage[dvipsnames]{xcolor}
\usepackage{ longtable, xpatch,  dcolumn,booktabs,adjustbox}

\newcommand{\LamY}{\Lambda_y}
\newcommand{\LamX}{\Lambda_x}
\newcommand{\Ny}{N_y}
\newcommand{\Nx}{N_x}
\newcommand{\eX}{e_x}
\newcommand{\eY}{e_y}
\newcommand{\WY}{W^Y}
\newcommand{\QY}{Q^Y}

\newcommand{\Talpha}{T}

\newcommand{\Var}{\mathrm{Var}}

\newcommand{\T}{\top}

\newcommand{\F}{\mathrm{f}}
\newcommand{\R}{\mathbb{R}}
\newcommand{\E}{\mathbb{E}}
\newcommand{\Lam}{\Lambda}

\newcommand{\Ncal}{\mathcal{N}}
\newcommand{\iid}{\mathrm{i.i.d.}}

\def\+#1{\mathbb{#1}}
\def\*#1{\mathbf{#1}}

\usepackage{mathrsfs}

\def\z(#1){\mathbf{z}_{(#1)}}

\def\c(#1){\mathbf{#1}_{c}}
\def\n(#1){{#1}_{c,(i)}}
\def\s(#1){{#1}_{c,i}}
\def\onenorm(#1){\lVert#1\rVert_1}

\newcommand{\Lp}{\left(}
\newcommand{\Rp}{\right)}
\newcommand{\Ls}{\left[}
\newcommand{\Rs}{\right]}

\newcommand{\PP}{\mathbb{P}}
\usepackage{amsmath}
\usepackage{chngcntr}

\makeatletter
\newcommand{\Biggg}{\bBigg@{3.5}}
\def\Bigggl{\mathopen\Biggg}

\def\Bigggr{\mathclose\Biggg}

\newcommand{\norm}[1]{\left\lVert#1\right\rVert}

\definecolor{Blue}{rgb}{0,0,1}
\definecolor{Grey}{rgb}{.5,.5,.5}

\usepackage{multirow, rotating}
\usepackage{pst-node}
\usepackage{inconsolata}
\usepackage[T1]{fontenc}

\newtheorem{assumption}{Assumption}
\newtheorem{assumpS}{Assumption}
\newtheorem{assumpG}{Assumption}

\newtheorem{proposition}{Proposition}

\newtheorem{example}{Example}

\newtheorem{theorem}{Theorem}
\newtheorem{corollary}{Corollary}
\usepackage{natbib}

\usepackage[margin=10pt,labelfont=bf]{caption}

\newcommand\tcaptab[1]{\captionsetup{position=top, font=normalsize, labelfont=bf, textfont=normalfont, justification=centering, margin=0mm, aboveskip=1mm, belowskip=0mm, labelsep=colon, singlelinecheck=false}\caption{#1}}

\newcommand\tcapfig[1]{\captionsetup{position=top, font=normalsize, labelfont=bf, textfont=normalfont, justification=centering, margin=0mm, aboveskip=2mm, belowskip=0mm, labelsep=colon, singlelinecheck=false}\caption{#1}}

\newcolumntype{L}[1]{>{\raggedright\let\newline\\\arraybackslash\hspace{0pt}}m{#1}}
\newcolumntype{C}[1]{>{\centering\let\newline\\\arraybackslash\hspace{0pt}}m{#1}}
\newcolumntype{R}[1]{>{\raggedleft\let\newline\\\arraybackslash\hspace{0pt}}m{#1}}

\usepackage{setspace}
\usepackage{enumitem}
\setlist{nolistsep}
\setlist{noitemsep} 

\usepackage[margin=1in]{geometry}

\newcommand*{\thisdraft}{This draft: August 25, 2023}
\newcommand*{\firstdraft}{First draft: December 14, 2022} 

\begin{document}
 	\title{Target PCA: Transfer Learning Large Dimensional Panel Data\thanks{\scriptsize We thank Jose Blanchet, Kay Giesecke, Serena Ng, Neil Shephard and seminar and conference participants at Stanford, University of Rochester, Oxford, the NBER-NSF Time-Series Conference, California Econometrics Conference, SoFiE Conference, and NASMES for helpful comments.}}

\date{\thisdraft \\ \firstdraft}
	
		\author{Junting Duan\thanks{ \scriptsize Stanford University, Department of Management Science \& Engineering, Email: \texttt{duanjt@stanford.edu}.}
                \and
		Markus Pelger\thanks{\scriptsize Stanford University, Department of Management Science \& Engineering, Email: \texttt{mpelger@stanford.edu}.}
               \and
              Ruoxuan Xiong\thanks{ \scriptsize Emory University, Department of Quantitative Theory and Methods, Email: \texttt{ruoxuan.xiong@emory.edu}.}
	}
	
	\onehalfspacing

	\begin{titlepage}
		\maketitle
		\thispagestyle{empty}
		

		\begin{abstract}

This paper develops a novel method to estimate a latent factor model for a large target panel with missing observations by optimally using the information from auxiliary panel data sets. We refer to our estimator as target-PCA. Transfer learning from auxiliary panel data allows us to deal with a large fraction of missing observations and weak signals in the target panel. We show that our estimator is more efficient and can consistently estimate weak factors, which are not identifiable with conventional methods. We provide the asymptotic inferential theory for target-PCA under very general assumptions on the approximate factor model and missing patterns. In an empirical study of imputing data in a mixed-frequency macroeconomic panel, we demonstrate that target-PCA significantly outperforms all benchmark methods.
			
			\vspace{1cm}
			

		\end{abstract}
\end{titlepage}

	\begin{onehalfspacing}

		\section{Introduction}

Panel data with a large number of units and time periods are widely available in macroeconomics, finance, and many other areas of social sciences. In many cases, these panels can be well described by an approximate factor structure, that is, a small number of common factors explain a large portion of the co-movements. A common approach is to estimate latent factors with statistical methods from the panel of interest, which we refer to as target data. In the era of big data, there often exist auxiliary panels, that contain relevant information and share some common factors with the target panel. Combining the information in multiple panels can increase the efficiency of the estimated factor model for the target. Even more importantly, using auxiliary panels can identify the factors that only affect a small subset of units or are not detectable due to the missing observations in the target. The idea of using auxiliary panel data to estimate a model that is applied to a target data set is conceptually similar to transfer learning, which has been successfully used for machine learning tasks. 

There is a broad class of practical problems which our setup is relevant. A particularly important case is mixed-frequency data, which is empirically studied in this paper. For example, some macroeconomic time series are only available at a quarterly or lower frequency, for example, GDP, while other time series are available at a higher frequency, for example, stock returns. As long as some factors in stock returns are also correlated with the macroeconomic movements, they can be used to obtain higher frequency factors and imputed values for the lower frequency target panel. Another important case is casual inference on panel data with non-random treatment, where missing data in the target panel correspond to the unobserved counterfactual outcomes. The auxiliary data could be panels for another set of units or the same set of units but with different outcome variables, which are correlated with the target time series. By leveraging auxiliary data, the latent factor model can be more precisely estimated, thus improving the precision of missing data imputation in the target panel.

These examples illustrate the benefits, but also fundamental challenges, of using auxiliary data.
First, the cross-sectional dimensionality and the signal-to-noise ratio of the panels can be very different. Second, the target panel might include information not contained in the auxiliary data. Therefore, naive methods, such as simply concatenating target and auxiliary data to one large panel or using them separately, can be sub-optimal or infeasible.

We propose to estimate a latent factor model for the target data by optimally combining information from auxiliary panels with the target panel. We refer to this method as target-PCA (T-PCA). This method is broadly applicable but easy to implement: It applies principal component analysis (PCA) to a weighted average of the {second-moment matrices} of the target and auxiliary panels. {Our method can be interpreted as applying PCA to an auxiliary panel with a reward for factors that are useful for the target panel.} We show the consistency and provide the inferential theory for target-PCA under general assumptions on the latent factor model and missing observations. The asymptotic distribution is essential for two reasons: First, it provides guidance on selecting efficient weights in target-PCA; second, it provides confidence intervals for missing data imputation from the estimated latent factor model.

We show two important effects of the relative weights between target and auxiliary data in target-PCA. The first effect is the consistency effect for factor identification. This matters when neither the target nor the auxiliary data alone are sufficient to estimate all the factors that we care about. Factors cannot be consistently estimated from the target panel when their signal is too weak or when the partially observed data is insufficient, for example, when certain times are unobserved for the full panel. Suppose that these weak factors are strong in the auxiliary data, but the auxiliary data might not contain the other factors for the target. In this case, target-PCA can consistently estimate all factors, if we select the target weight for the auxiliary data at the right rate to account for the different dimensions of the panels.

The second effect is the efficiency effect in the estimation of latent factors and loadings. This effect arises when target and auxiliary data are observed with different noise levels. If the weighting in target-PCA properly accounts for the noise ratio between target and auxiliary data, then we can improve the efficiency of the estimated latent factor model. Hence, after selecting the target weight in the right order to ensure consistency, we can improve the efficiency by selecting the optimal scale of the weights. 

These two important effects show that the optimal selection of weights in target-PCA is a challenging problem that can depend on the relative factor strength, observation pattern, noise level, and dimensionality between target and auxiliary data. To address this problem, we develop the inferential theory for the estimated factors, loadings, and imputed values for a general weighting scheme in target-PCA. The inferential theory is then used as guidance for selecting the weights in target-PCA. The naive cases of concatenating the target and auxiliary data, or using only one of the panels, are special cases of our general method. We show that these special cases are generally less efficient and even lack identification in the worst case.

Our work contributes to four distinct areas. First, we contribute to the literature on large dimensional factor models by proposing a new setup where the latent factors can be jointly estimated from multiple panel data. Second, we provide a new solution to the problem of weak factors, which cannot be consistently estimated with conventional PCA estimators and require additional signals. We show how to leverage the information in supplementary panels to overcome the weak signal problem. Third, our paper complements the recent work on imputing missing data in large panels. We show that leveraging auxiliary data allows us to impute missing observations with higher precision and makes it possible to impute values that could otherwise not be imputed using only the target panel.
Last but not least, we contribute to causal inference and can estimate heterogeneous and time-dependent treatment effects for general interventions. 

Our asymptotic results are developed under the framework of an approximate latent factor structure for both target and auxiliary data, both with large cross-section and time-series dimensions. 
Our new setup of using multiple panels generalizes the existing factor modeling literature, which only uses one panel so far.
When the data is fully observed, \cite{bai2002determining} show that the factor model can be estimated with PCA applied to the covariance matrix of the data. \cite{bai2003inferential} and \cite{fan2013large} derive the consistency and asymptotic normality of the estimated factors, loadings and common components. Extensions of latent factor models with fully observed data include among others adding observable factors in \cite{bai2009panel}, sparse and interpretable latent factors in \cite{pelgerxiongsparse2020}, conditional loadings in \cite{fan2016projected}, \cite{pelger2018state}, \cite{chen2021semiparametric} and \cite{chen2022unified}, time-varying, locally estimated loadings in \cite{su2017time} and high-frequency estimation in \cite{pelger2019large}. 
The idea of increasing the efficiency by weighting panels differently is related to the GLS type weighting for PCA estimators suggested among others in \cite{breitung2011gls} and \cite{choi2012efficient}. \cite{boivin2006more} study empirically the benefits for PCA estimators when down-weighting or dropping uninformative information.  

Our work complements the recent work on estimating the latent factor model from one large panel with missing observations and developing entry-wise inferential theory. 
Our results extend the framework of \cite{xiong2019large} to the new setup of multiple panels.
Other closely related work includes the recent papers by \cite{jin2020factor}, \cite{bai2019matrix}, and \cite{cahan2021factor}. These papers differ in the algorithms to impute the missing observations, the generality of the missing patterns, and the proportion of required observed entries relative to the missing entries. \cite{bai2019matrix}, and \cite{cahan2021factor} leverage a block structure in missing data. \cite{jin2020factor} focuses on the case of missing-at-random and analyzes the EM estimator considered in \cite{stock1998}. We show that by leveraging auxiliary data, we not only increase the precision for imputing missing values, but can also accommodate more general missing patterns. The problem of missing data imputation has been actively studied in the matrix completion literature since \cite{candes2009exact}. The data imputation algorithms are largely based on rank-regularized methods with worst-case guarantees. Until recently, the entry-wise inferential theory is provided by the celebrated work of \cite{chen2019inference} under the assumption of i.i.d. sampling. In contrast, we build on the large dimensional factor modeling literature, allowing us to make progress on the inferential theory under general observation patterns.

Our proposal of using auxiliary data brings the idea of transfer learning to the estimation of latent factors. The concept of transfer learning is to apply a model estimated on auxiliary data to target data.  Transfer learning has been successfully used for machine learning tasks as surveyed by \cite{pan2009survey}. Related to our work, \cite{huang2021scaled} propose to forecast one target time series by scaling each of the auxiliary predictors by its predictive power on the target. In contrast, our target-PCA allows for a large cross-sectional dimension of the target with missing observations.    

Our work provides a complementary and novel solution for weak factor estimation. \cite{onatski2012asymptotics} has shown that conventional PCA estimators cannot consistently estimate factors that are too weak. \cite{onatski2010}, \cite{ahn2013}, \cite{bai2021approximate}, \cite{onatski2022} and others have studied the properties of weak factors and test the number of factors, with a focus on one panel.
\cite{lettau2020estimating} have demonstrated that including additional information from other moments of the data can overcome the problem for certain types of weak factors. Similarly, our novel use of additional information from the auxiliary data can upweight the weak signals, allowing for the identification and efficient estimation of weak factors. 

Our work is also complementary to the literature on mixed-frequency data imputation. A specific application of our general framework is to impute low-frequency observations in a target panel using auxiliary panels of higher frequency. One approach to dealing with data sampled at different frequencies has emerged in work by \cite{ghysels2002}, and \cite{andreou2010} using Mi(xed) Da(ta) S(ampling) (MIDAS) regressions, and its many extensions have been surveyed among others in \cite{ghysels2020}. The MIDAS regression relates the low-frequency time series that we wish to predict to observables at high and low frequencies. Our framework shares some similarities as it relates latent factors of higher frequency to a large panel of lower-frequency observations. Another alternative is to use state space models to deal with mixed frequency data - usually estimated using the Kalman filter (see \cite{bai2013} for a comparison with MIDAS). State space models are parameter-driven and impose different assumptions on the data generating process.\footnote{\cite{stock2016} discusses the issues with state space estimation of factor models with missing data.} Another complementary idea is based on interpolation arguments as for example in \cite{chow1971}. The recent work of \cite{ng2022} emphasizes the importance of residual correlation in the imputation of mixed frequency data. They propose a dynamic matrix completion approach by combining the state space setup and \cite{chow1971} type bridge regressions with latent factor models estimated from a large partially observed panel.

In simulation and empirical studies, we show the superior performance of our target-PCA method relative to benchmarks under a variety of settings.
Our comprehensive simulations compare target-PCA to the natural benchmarks of applying PCA to separate panels or simple concatenated panels. Target-PCA substantially outperforms the alternatives in- and out-of-sample under different observation patterns. Our empirical analysis shows the good performance of target-PCA for imputing missing values in popular macroeconomic panels. We demonstrate the potential of target-PCA for nowcasting macroeconomic panels by imputing unbalanced low-frequency panels with higher-frequency auxiliary data.

The rest of the paper is organized as follows. Section \ref{sec:model} introduces the model and target-PCA. Section \ref{sec: A Simplified Factor Model} illustrates the two important effects of weighting auxiliary data. Section \ref{subsec: assumptions} formalizes the assumptions on the observation pattern and approximate factor model. Section \ref{sec: results} provides the consistency and asymptotic results for our estimator, which can be used as guidance to select the weight in target-PCA. Section \ref{sec: discussion} discusses the extensions of our method, whose good performance compared to other methods is demonstrated by the extensive simulations in Section \ref{sec: simulation}. Section \ref{sec:empirical} shows the practical relevance of target-PCA through empirical examples of imputing missing entries in macroeconomic data. Section \ref{sec: conclusion} concludes the paper.

\section{Model Setup}\label{sec:model}

\subsection{Model}  \label{sec: sub_model}

We partially observe a target panel data set $Y$ with $T$ time periods and $N_y$ cross-sectional units, where both $T$ and $N_y$ are large. $Y \in \+R^{T \times N_y}$ has an approximate latent factor structure with $k_y$ common factors,
\[  
Y_{ti} = (F_y)_t^\top ({\Lam}^Y_y)_i + (\eY)_{ti}, \quad t=1,\cdots, T, \ i=1,\cdots, \Ny.
\]
Here, $Y_{ti}$ denotes the data for the $i$-th cross section at time $t$, $(F_y)_t$ is a $k_y\times 1$ vector of latent factors, $(\Lam^Y_y)_i$ is its corresponding latent factor loadings, $C_{ti}=(F_y)_t^\top (\Lam^Y_y)_i$ is the common component of $Y_{ti}$ and $(\eY)_{ti}$ is the idiosyncratic component of $Y_{ti}$. 
This factor model can also be written in a matrix notation as
\[
\underbrace{Y}_{T \times \Ny} =  \underbrace{ F_y}_{T\times k_y}  \underbrace{ {\Lam_y^Y}^\top}_{k_y\times \Ny}+  \underbrace{ \eY}_{T\times \Ny}.
\]

Our goal is to estimate the latent factor structure in $Y$. We are particularly interested in the important case, where some entries in $Y$ can be missing. In many practical applications, the panel $Y$ might not be informative enough to estimate the latent factor model. First, the factors in $Y$ can be weak, that is, affect only a small subset of cross-sectional units. In this case, it is possible that the factors cannot be separated from the noise, and conventional principal component analysis (PCA) fails to estimate them. Second, the observed data might be insufficient to estimate the full factor model, either because certain times are never observed in the panel, for example with low-frequency data, or because the missingness depends on the factor structure.

Our solution is to use additional information from auxiliary panel data $X$. Suppose we observe auxiliary data that contain relevant information and can be helpful for the estimation of latent factors in $Y$. For the exposition, we focus on the case where there is only one auxiliary panel data $X$ and study how to optimally use the information in $X$ to estimate the factors in $Y$. Our results can easily be extended to the case of multiple auxiliary panels, as discussed in Section \ref{subsec:multiple-panels}.

Suppose $X$ has an approximate latent factor structure with $T$ time periods and $N_x$ cross-section units, where $N_x$ is large. Both panels $X$ and $Y$ are observed for the same time periods. $X$ has $k_x$ factors   
\[
   \underbrace{X}_{T \times {\Nx}} = \underbrace{F_x}_{T \times k_x} \underbrace{ {\Lam_x^X}^\top}_{k_x \times {\Nx}} + \underbrace{\eX}_{T \times {\Nx}},
\]
where $(F_x)_t$ is a $k_x\times 1$ vector of latent factors, $(\Lam^X_x)_i$ is its corresponding latent factor loadings and $(\eX)_{ti}$ is the idiosyncratic error of $X$.

The auxiliary panel data $X$ can be useful when it has some factors in common with $Y$. Without loss of generality, we can use the rotation of the factor models such that the concatenated factors in $F_y$ and in $F_x$ are orthogonal. 
We denote by $F$ the union of the factors in $F_y$ and $F_x$. 
The total number of non-redundant latent factors in $F$ equals the rank of the concatenated panel matrix $[F_x,F_y]$. 
Given our normalization, $F^\top F/T$ is a diagonal matrix. 
It is notationally convenient to write both the target panel data $Y$ and auxiliary panel $X$ in terms of $F$, that is the union of all the non-redundant factors. Hence, the target $Y$ follows
\[
\underbrace{Y}_{T \times \Ny} =  \underbrace{ F}_{T\times k}  \underbrace{ \LamY^\top}_{k\times \Ny}+  \underbrace{ \eY}_{T\times \Ny},
\]
where $k$ is the total number of non-redundant latent factors in $F$, and $(\LamY)_i$ is a $k\times 1$ vector of loadings on $F$. Note that we explicitly allow the loadings to be zero for some factors. More specifically, for factors that are unique to $X$ and do not appear in $F_y$, the corresponding loadings in $\LamY$ are zero. Similarly, we express the auxiliary panel data $X$ as  
\[
  \underbrace{X}_{T \times {\Nx}} = \underbrace{F}_{T \times k} \underbrace{\LamX^{\top}}_{k \times {\Nx}} + \underbrace{\eX}_{T \times {\Nx}},
\]
where $(\LamX)_i$ is a vector of factor loadings for $F$ in $X$.

In an approximate factor model, a large part of the variation is explained by the factors, while the noise is only weakly dependent. We allow for the empirically relevant case, where not all factors in $Y$ are strong, but only affect a subset of the panel. 
Let $(\Lam_y)_{ij}$ be unit $i$'s loading for the $j$-th factor of $F$.
We refer to the $j$-th factor as a strong factor in $Y$, if it affects a large number of units in $Y$. Formally, factor $j$ is strong factor if it satisfies $N_y^{-1} \sum_{i=1}^{N_y} (\Lam_y)^2_{ij} >0$ for $N_y \rightarrow \infty$.  On the other hand, if $\sum_{i=1}^{N_y} (\Lam_y)^2_{ij}$ grows at a rate smaller than $\Ny$, then we refer to the $j$-th factor as a weak factor in $Y$. The weak factor only affects a small fraction of units, with the fraction converging to zero as $\Ny$ grows. Without loss of generality, we assume that all factors $F_x$ are strong factors in $X$. {As our objective is to estimate the factor model for $Y$, we are not interested in estimating the weak factors in $X$.} However, it is straightforward to extend our results to the case of weak factors in $X$.

For the cross-sectional dimension of $X$ and $Y$, we focus on the setting where $ \Ny/{\Nx}\rightarrow c\in [0,\infty)$. This includes two cases: $\Nx$ and $\Ny$ are of the same order (i.e., $c > 0$) and $\Nx$ is much larger than $\Ny$ (i.e., $c = 0$). When $c = 0$, we consider the finite $\Ny$ case in Section \ref{subsec:finite-dimension-target}. The setting of $\Ny/\Nx\rightarrow \infty$ is analogous and can be studied by similar arguments.

\subsection{Observation Patterns}\label{subsec:observation-pattern}

We allow the target data $Y$ to have missing observations. Let $\WY \in \{0,1\}^{T \times \Ny}$ be the observation pattern of $Y$, where $\WY_{ti}=1$ if $Y_{ti}$ is observed and zero otherwise. 
For simplicity, we assume the auxiliary data $X$ is fully observed, but our methods and results can be easily generalized to the case where $X$ is only partially observed as well.

We allow for very general observation patterns in $Y$. Whether an entry is observed or not can depend on whether other entries are observed, and on the factor model itself. The formal assumptions on the observation patterns are introduced in Section \ref{subsec: assumptions}. To provide some intuition, Figure \ref{fig: missing patterns} shows three important examples of the observation patterns that we are interested in. In the first example, entries are missing completely at random, that is, whether an entry is observed or not does not depend on whether other entries are observed or not. 

The second example is the observation pattern for control panels with staggered treatment adoption. Once a unit adopts the treatment, it stays treated afterwards, which can be modeled as missing values. This pattern is widely assumed in the literature on causal inference in panel data.

The third and attentive example shows the observation pattern of low-frequency time-series variables in $Y$, where variables in other data sets (including $X$) and applications are at a higher frequency. 
For example, assume that $Y$ is a macroeconomic panel data set, where the time series are only available at a quarterly frequency, but a downstream application requires these time series as inputs at a monthly frequency. The monthly observations in between the quarters are modeled as missing observations. Importantly, in this example, there is no information in $Y$ for the months where the quarterly variables are not reported, invalidating the existing latent factor estimation methods from partially observed $Y$ only \citep{bai2019matrix,cahan2021factor,xiong2019large}. In contrast, our proposed method can identify the latent factor values in these months by using auxiliary data of higher frequency.

\begin{figure}[t!]
    \tcapfig{Examples of observation patterns}
	\centering
	\begin{tabular}{L{5cm} L{5cm} L{5cm}}
	\includegraphics[width=0.3\textwidth,height=35mm]{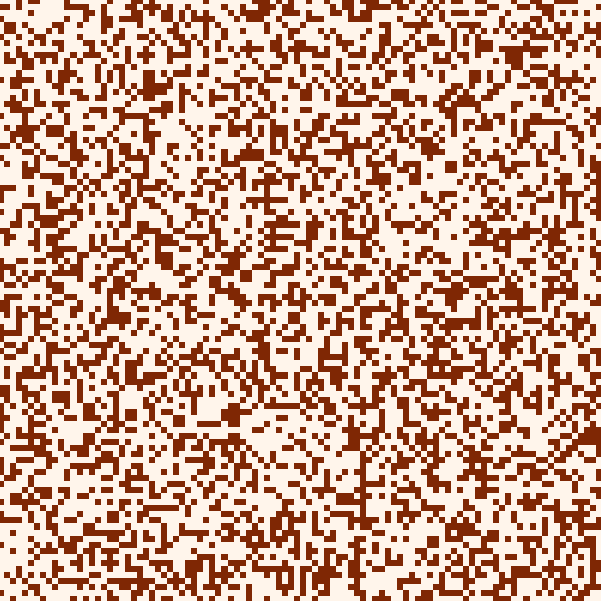} &	\includegraphics[width=0.3\textwidth, height=35mm]{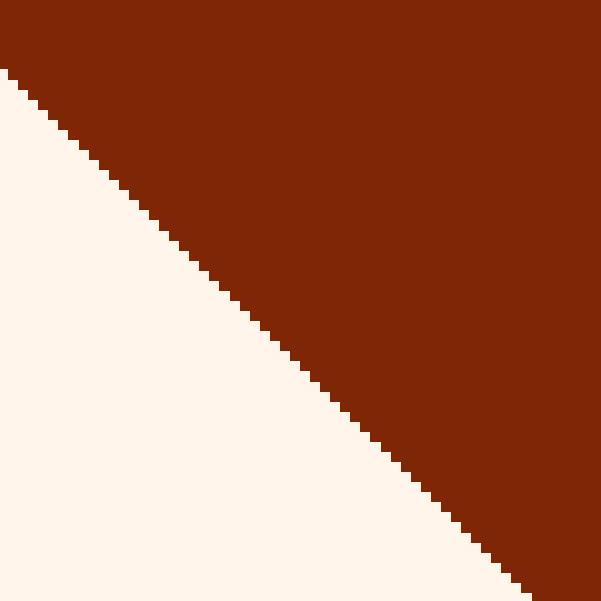} &
	\includegraphics[width=0.3\textwidth, height=35mm]{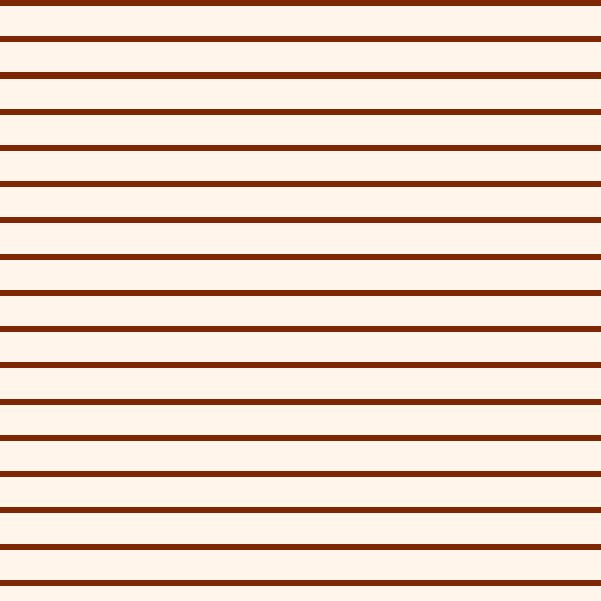} \\
	(a) Randomly missing & (b) Staggered treatment adoption & (c) Low-frequency observation
	\end{tabular}
	\floatfoot{These figures show examples of observation patterns. The entries with dark color denote observed entries, while the entries with light color denote missing entries. Each row represents the observation pattern for a time period and each column represents the observation pattern for a unit.}
	\label{fig: missing patterns}
\end{figure}

{\subsection{Main Objective and Key Challenges}

The main objective of this paper is to estimate a complete factor model, that is, all the relevant factors for the target $Y$, their loadings on the target, and the implied common component for the target, and to provide a complete inferential theory for all components of the factor model. There are three key challenges for this objective, specifically in the presence of missing observations, that invalidate the standard approaches for latent factor estimation.}
In particular, applying conventional PCA to either $Y$, $X$, or a simple concatenated panel of $X$ and $Y$ will either be infeasible, inconsistent, or inefficient.

First, the information in $Y$ may not be sufficient to estimate the factors at all time periods with conventional PCA methods. If some factors in $F_y$ are weak in $Y$, then the identification assumptions underlying PCA are violated, and a conventional PCA estimate cannot separate those factors from the noise. This problem can arise even for fully observed data. If the panel $Y$ has missing data, then this can pose additional challenges. For only partially observed data in $Y$, it is possible that even the strong factors in $F_y$ cannot be estimated for all time periods. The leading example is when $Y$ is only observed at a lower frequency, for example, annual data. In that case, we cannot infer the factor realizations for a higher frequency, for example, monthly, even for strong factors. As another complication, missingness can ``weaken'' the factor signal in the observed data, as the missing pattern can depend on the factor model and affect the effective sample size. In all these three cases, we can only learn the strong factors in $F_y$ from the target $Y$ for the time periods with sufficiently many observations. Hence, in these situations, we need to take advantage of the additional information in $X$.

Second, the auxiliary panel $X$ may not contain all the factors in $Y$, and hence applying PCA only to $X$ can fail to consistently estimate the factors for $Y$. In applications, there is no reason to assume a priori that the strong factors in a supplementary data set, that is not specifically targeted for $Y$, include all the factors $F_y$. 
Hence, we cannot ignore the information in $Y$, and need a way to combine the two panels.

Third, the dimensions of the panels $X$ and $Y$ can be very different. It is natural to assume that in many applications, the auxiliary panel $X$ is much larger in the hope that it contains some useful information for $Y$. If $\Nx$ is much larger than $\Ny$, then applying PCA to the concatenated panel of $X$ and $Y$ can only identify the factors $F_x$ in $X$, but not the factors that are unique to $Y$ and not included in $F_x$. This is because the target panel can receive a weight that is too low in a naively concatenated panel.

Therefore, the presence of these three challenges requires new methods to estimate all the factors in $Y$. Our target-PCA proposed in Section \ref{sec: estimator} below provides a solution to simultaneously address these three challenges. Target-PCA extracts the factors from $X$ that are the most useful for explaining $Y$, and efficiently combines them with the information in $Y$.

In the following, we focus on the practically relevant and challenging case where the auxiliary data $X$ does not contain all the factors in $F_y$, which is formalized in Assumption \ref{assump: factor shift}.
\begin{assumpG}[
] \label{assump: factor shift} \texttt{}  
There exist some strong factors in the target $Y$ that are not strong or not contained in the auxiliary panel $X$, that is, 
\[
{\rm rank}\Lp  \lim_{N_x \rightarrow \infty} \frac{1}{\Nx}\Lambda_{x,y_s}^\top \Lambda_{x,y_s}\Rp < k_{y_s},
\]
where $\Lambda_{x,y_s} \in \+R^{N_x \times k_{y_s}}$ denotes the loadings in $X$ that correspond to the strong factors in $F_y$, and $k_{y_s}$ denotes the number of strong factors in $Y$.
\end{assumpG}
Section \ref{sec: auxilliary is sufficient} discusses the simpler case, where the strong factors in $X$ include all factors $F_y$ for $Y$, and hence the auxiliary panel $X$ is sufficient to learn the factors in $Y$. This special case is a subset of our more comprehensive analysis.

\subsection{Estimator}
\label{sec: estimator}

We propose a novel estimator, target-PCA, to estimate the latent factors in $Y$ by combining the information in $Y$ and $X$. Then we use the estimated factor model to impute missing observations in $Y$.

To illustrate the intuition of our target-PCA estimator, we start with the case where $Y$ is fully observed. Recall, that in the conventional setup where we have only one fully-observed panel, we can estimate the latent factors that explain most variation in the data by minimizing the objective function of PCA. Our target-PCA estimator combines the PCA objective function for the auxiliary panel $X$ and the PCA objective function for the target $Y$ with a positive target weight $\gamma$ as
\begin{equation} \label{equ: target obj}
\min_{F,\LamX,\LamY} \ \  \underbrace{\sum_{i=1}^{\Nx} \sum_{t=1}^T \left(X_{ti}-F_t^\top (\LamX)_i\right)^2}_{\text{auxiliary error}}\ + \ \gamma \cdot  \underbrace{\sum_{j=1}^{\Ny} \sum_{t=1}^T \left(Y_{tj}-F_t^\top (\LamY)_j\right)^2}_{\text{target error}}.
\end{equation}
This combination aims to extract the latent factors of $Y$, by optimally weighting the information in $X$ and $Y$ through the parameter $\gamma$. 
{The target parameter $\gamma$ can be interpreted as a reward for factors that help to explain the target panel $Y$. In this sense, $\gamma$ is similar to a regularization parameter and controls the reward given to factors that reduce the error in explaining $Y$.}
As we will show next, $\gamma$ can also be interpreted as a relative weight for the information in $X$ and $Y$ in a concatenated panel.

We can write \eqref{equ: target obj} in matrix notation as
\begin{equation}\label{equ: target obj matrix}
    \min_{F,\LamX,\LamY} \text{trace}\left( \left(\begin{bmatrix} X^\T \\ \sqrt{\gamma} Y^\T \end{bmatrix} - \begin{bmatrix} \Lambda_x \\ \sqrt{\gamma} \Lambda_y \end{bmatrix} F^\T  \right) \left( \begin{bmatrix} X & \sqrt{\gamma} Y \end{bmatrix} - F \begin{bmatrix}  \Lambda_x^\T &\sqrt{\gamma} \Lambda_y^\T \end{bmatrix} \right)  \right).
\end{equation}
For exposition, we introduce the notation $Z^{(\gamma)} \in \R^{T \times ({\Nx}+\Ny)}$ as the combination of $X$ and $Y$ with target weight $\gamma$, i.e.,
\[
\begin{aligned}
Z^{(\gamma)} \coloneqq \begin{bmatrix} X & \sqrt{\gamma} Y \end{bmatrix} = F \Lam^{(\gamma)\top} + e^{(\gamma)}, \quad \text{where}\ & \Lam^{(\gamma)\top} \coloneqq \begin{bmatrix}  \Lambda_x^\T &\sqrt{\gamma} \Lambda_y^\T \end{bmatrix}  \in \R^{k \times ({\Nx}+\Ny)}, \\
&e^{(\gamma)} \coloneqq \begin{bmatrix}
\eX & \sqrt{\gamma}\eY\end{bmatrix} \in \R^{T\times ({\Nx}+\Ny)}.
\end{aligned}   
\]
With the identifying assumption $\Lam^{(\gamma)\top} \Lam^{(\gamma)}/({\Nx}+\Ny)=I_k$, we can concentrate out $F$ in the objective function \eqref{equ: target obj matrix}, and obtain $\Lam^{(\gamma)}$ from the following objective function\footnote{An alternative approach to obtain the optimal $F$, $\Lambda_x$ and $\Lambda_y$ in the objective function \eqref{equ: target obj} is based on the identifying assumption $F^{\top} F/T=I_k$. Under this identifying assumption, we concentrate out both $\Lambda_x$ and $\Lambda_y$ in the objective function \eqref{equ: target obj matrix}, and solve for $F$ from the objective function 
\begin{align}\label{equ: obj}
    \max_{F} \text{trace}\left( F^{\top} \big( Z^{(\gamma)} Z^{(\gamma)\top} \big) F \right).
\end{align}
{Appendix \ref{appendix:alter-estimator} provides more details about how to adapt this alternative approach to allow for missing observations.}
}
\begin{equation}
    \max_{\Lam^{(\gamma)}} \text{trace}\left( \Lam^{(\gamma)\top} \big(Z^{(\gamma)\top} Z^{(\gamma)} \big) \Lam^{(\gamma)} \right).
\end{equation} 
Hence, when $Y$ is fully observed, we can estimate $\Lam^{(\gamma)}$ by applying PCA to $Z^{(\gamma)\top} Z^{(\gamma)}$. Note that $Z^{(\gamma)\top} Z^{(\gamma)}/T$ is an estimator of the second moment of the weighted concatenated panel $ Z^{(\gamma)}$. We denote the second population moment of $ Z^{(\gamma)}$ as $\Sigma^{Z^{(\gamma)}}$, which equals the cross-sectional covariance matrix of $Z^{(\gamma)}$ in the case of demeaned data. Hence, target-PCA is equivalent to PCA on the weighted data $Z^{(\gamma)}$.

When $Y$ has missing observations, we can also apply PCA to an estimator of $\Sigma^{Z^{(\gamma)}}$, but we need to account for the missing observations in $Z^{(\gamma)}$ in the estimation of $\Sigma^{Z^{(\gamma)}}$.
Conceptually, we can use the time periods when both cross-sectional units $i$ and $j$ are observed to estimate $\Sigma^{Z^{(\gamma)}}_{ij}$, where $\Sigma^{Z^{(\gamma)}}_{ij}$ is the $(i,j)$-th entry in $\Sigma^{Z^{(\gamma)}}$. Formally, we introduce a new observation matrix $W^Z$, where $W_{ti}^Z = 1$ if $Z_{ti}^{(\gamma)}$ is observed and $0$ otherwise. For any two cross-sectional units $i$ and $j$, let $Q_{ij}^Z = \{t: W^Z_{ti} = W^Z_{tj} = 1\}$ be the set of time periods where both $i$ and $j$ are observed. We use the time periods in $Q_{ij}^Z$ to estimate $\Sigma^{Z^{(\gamma)}}_{ij}$
\[
    \tilde{\Sigma}^{Z^{(\gamma)}}_{ij} = \frac{1}{|Q_{ij}^Z|} \sum_{t \in Q_{ij}^Z} Z_{ti}^{(\gamma)}Z_{tj}^{(\gamma)}, \quad i,j=1, \cdots, {\Nx}+\Ny\, .
\]

With the identifying assumption $\tilde{\Lam}^{(\gamma)\top} \tilde{\Lam}^{(\gamma)}/({\Nx}+\Ny)=I_k$, the estimated loadings from PCA, denoted as $\tilde{\Lam}^{(\gamma)}$, are $\sqrt{{\Nx}+\Ny}$ times the eigenvectors of the $k$ largest eigenvalues of $\tilde{\Sigma}^{Z^{(\gamma)}}/({\Nx}+\Ny)$, that is,
\[
    \frac{1}{{\Nx}+\Ny} \tilde{\Sigma}^{Z^{(\gamma)}} \tilde{\Lam}^{(\gamma)} = \tilde{\Lam}^{(\gamma)} \tilde{D}^{(\gamma)}\, ,
\]
where $\tilde{D}^{(\gamma)}$ is the diagonal matrix of the largest $k$ eigenvalues of $\tilde{\Sigma}^{Z^{(\gamma)}}/({\Nx}+\Ny)$. In this step, we simultaneously estimate the factor loadings in $X$ and $Y$. 

In a second step, we regress the observed $Z^{(\gamma)}$ on $\tilde{\Lam}^{(\gamma)}$ to estimate the factors, that is, 
\[
\tilde{F}_t = \Bigg( \sum_{i=1}^{{\Nx}+\Ny} W^Z_{ti}\tilde{\Lam}_i^{(\gamma)} \tilde{\Lam}_i^{(\gamma)\top}  \Bigg)^{-1} \cdot \Bigg(  \sum_{i=1}^{{\Nx}+\Ny} W^Z_{ti}  Z_{ti}^{(\gamma)}\tilde{\Lam}_i^{(\gamma)}  \Bigg)\, , \quad t=1,\cdots, T.
\]
This step only uses the units that are observed at time $t$, and runs a weighted regression of the observed units' outcomes on the estimated loadings. The weight for units in $Y$ is $W_{ti}^Y \cdot \gamma$, while the weight for units in $X$ is $1$.

 Next we estimate the common component of $Y$ with the plug-in estimator $\tilde{C}_{ti} = \tilde{F}_t^\top (\tilde{\Lam}_y)_i$, and use the estimated common components to impute the missing entries in $Y$. Note that our proposed estimator for the latent factor model in the partially observed $Y$ is the same type of estimator as in \cite{xiong2019large}. The key distinction is that our estimator is applied to the weighted concatenated panel $Z^{(\gamma)}$, while the estimator in \cite{xiong2019large} is applied to $Y$ only.\footnote{{It is possible to use other methods to estimate the factor model when the missing pattern has specific structures. For example, we can use \cite{jin2020factor} when observations are missing-at-random or use \cite{bai2019matrix} and \cite{cahan2021factor} for block-missing. The asymptotic results for alternative methods would require a case-by-case analysis, but we expect the general insights for the choice of the target weight to stay the same.}}

As shown in the original objective function \eqref{equ: target obj}, the key of target-PCA is the target weight $\gamma$. There are three special cases of target-PCA with different values of $\gamma$: First, when $\gamma=0$, target-PCA degenerates to applying PCA to $X$. Second, when $\gamma=\infty$, target-PCA degenerates to applying PCA to $Y$. Third, when $\gamma = 1$, target-PCA is equivalent to applying PCA to the concatenated panel ${Z}^{(1)} = \big[X \,\,\, Y \big]$. The key problem is to select the target weight appropriately. Intuitively, the target weight $\gamma$ should not be too small in order to ensure that the selected factors are relevant for $Y$, and $\gamma$ cannot be too large in order to take advantage of the supplementary information in $X$.

The asymptotic distribution results for target-PCA require novel derivations and are not simply an application of \cite{xiong2019large} applied to $Z^{(\gamma)}$. A key challenge is that the target weight can go to infinity with the sample size, which would lead to exploding moments in some of the components of $Z^{(\gamma)}$ and hence violate the assumptions in existing frameworks. Therefore, we need to carefully consider the separate components of $Z^{(\gamma)}$, while allowing for joint asymptotics of $\gamma$ growing with the sample size.

The main insight of target-PCA is that properly weighting the covariance matrices of multiple panels allows for the consistent and efficient estimation of latent factors. As will be shown in the next sections, the optimality of combining auxiliary data boils down to two aspects: (a) the detection of weak signals in the target panel; and (b) the efficient estimation of the factor structure. Our target-PCA estimator simultaneously achieves them in one step through the target weight $\gamma$. In fact, these two aspects correspond to two important effects of the target weight $\gamma$, which will be thoroughly discussed in Section \ref{sec: A Simplified Factor Model}.
		
		\section{Two Fundamental Effects of Target Weight $\gamma$} \label{sec: A Simplified Factor Model}

In this section, we illustrate and highlight the two important effects of the target weight $\gamma$ in target-PCA: the consistency effect in factor identification and the efficiency effect in the estimation of factors and loadings. It is crucial to account for these two effects when choosing $\gamma$ in target-PCA. In this section, we use simple models to explain these fundamental effects, and then provide the results for the general case in Section \ref{sec: results}.

\subsection{Effect 1: Consistency Effect of Target Weight $\gamma$}\label{subsec: consistency effect}

The first important effect is the consistent estimation of the factors. This effect is relevant when some factors in the observed $Y$ are weak and cannot be identified by applying PCA to $Y$. For this case, if the weak factors in $Y$ are strong factors in $X$, then it is possible to identify these factors with target-PCA. Specifically, if we choose $\gamma = r\cdot N_x/N_y$ for some positive constant $r$, then target-PCA can consistently estimate these weak factors. The intuition is as follows.

Target-PCA essentially estimates the factors from a matrix that combines $X X^\T$ and $\gamma \tilde{Y} \tilde{Y}^\T$, where $\tilde{Y} = Y \odot W^Y$ replaces the missing entries in $Y$ with $0$. Suppose that the auxiliary panel $X$ is fully observed and the number of common observations between any two units of the target $Y$ is proportional to $T$ (the main setting studied in this paper). Then the top eigenvalues in $X X^\T$ and $\tilde{Y} \tilde{Y}^\T$ are at the order of $N_x$ and $ N_y$, respectively. If we select $\gamma = r\cdot N_x/N_y$ for some positive constant $r$, then the top eigenvalues in $X X^\T$ and $\gamma \tilde{Y} \tilde{Y}^\T$ are at the same order, and equivalently, the factor strengths of the strong factors in $X$ and in $Y$ are the same in the combined matrix of target-PCA. As the weak factors in $\tilde{Y}$ are strong factors in $X$, target-PCA can identify both the strong and weak factors in $\tilde{Y}$, leading to the consistency effect.

Next, we flesh out the intuition with a toy example. This two-factor example has the following key elements: (a) The dimension of the panel $X$ is much larger than $Y$, i.e., $N_y / N_x \rightarrow 0$. (b) In panel $Y$, factor 1 is strong, but factor 2 is weak. Hence, we can only identify factor 1 from panel $Y$. (c) In panel $X$, factor 2 is strong, but factor 1 has zero exposure to the units. Hence, we can only identify factor 2 from panel $X$. Specifically, the loadings in $Y$ follow 
\begin{align*}
 & \text{first factor loadings: }  (\LamY)_{i1}  \overset{\iid}{\sim} (0,\sigma_{\LamY}^2) \text{ for all } i\, ,\\
 & \text{second factor loadings: }  (\LamY)_{i2}  \overset{\iid}{\sim} (0,\sigma_{\LamY}^2) \text{ for } i < \Ny^{1/2}, \quad ({\LamY})_{i2} = 0 \text{ for } i \geq \Ny^{1/2}\, ,
\end{align*}
and loadings in $X$ follow
\begin{align*}
 & \text{first factor loadings: }  (\LamX)_{11} = \cdots = ({\LamX})_{{\Nx},1} = 0\, , \\
 & \text{second factor loadings: }  (\LamX)_{i2}  \overset{\iid}{\sim} (0,\sigma_{\LamX}^2) \text{ for all } i\, .
\end{align*}
For simplicity, both the factors and idiosyncratic errors are drawn independently as $F_{t1}, F_{t2} \overset{\iid}{\sim} (0,\sigma_F^2), (\eX)_{ti}\overset{\iid}{\sim} (0,\sigma_{\eX}^2),$ and $(\eY)_{ti} \overset{\iid}{\sim} (0,\sigma_{\eY}^2)$ for all $i$ and $t$. 
Furthermore, the factors and loadings have bounded fourth moments, and errors have bounded eighth moments. We assume that we can observe all the entries in $Y$. 
As we show in Section \ref{sec: results}, the insights from this simplified setting carry over to our general model with missing observations.\footnote{Note that the observation pattern affects the asymptotic variance, but not the convergence rate of the estimated factor model or the top eigenvalues of $\tilde{Y} \tilde{Y}^\T$, as shown in \cite{xiong2019large}. Therefore, the intuition for this toy example carries over to other observation patterns.}

In this example, separate PCA on either $X$ or $Y$ is not able to consistently estimate both factors. However, target-PCA can identify both factors with an appropriately chosen $\gamma$. Without error terms and missing observations, target-PCA estimates factors from
\[
\begin{aligned}
\frac{1}{\Nx+\Ny} Z^{(\gamma)} Z^{(\gamma)\T}
& =\frac{1}{\Nx+\Ny} \begin{bmatrix} X & \sqrt{\gamma} Y \end{bmatrix} \begin{bmatrix} X^\T \\ \sqrt{\gamma} Y^\T \end{bmatrix}  =\frac{1}{\Nx+\Ny} \left[ XX^\top +  \gamma Y Y^\top\right] \\
&= \begin{bmatrix}F^{(1)} & F^{(2)} \end{bmatrix} \Bigggl(\underbrace{\begin{bmatrix}\frac{\gamma \cdot \Ny \sigma_{\LamY}^2}{\Nx+\Ny}  & \\ & \frac{\Nx\sigma_{\LamX}^2+ \gamma \cdot \Ny^{1/2}\sigma_{\LamY}^2}{\Nx+\Ny} \end{bmatrix}}_{\coloneqq \hat\Sigma_{\Lam,t}^{(\gamma)}}  + o_p(1) \Bigggr) \begin{bmatrix} F^{(1)\top}   \\F^{(2)\top}   \end{bmatrix},
\end{aligned}    
\]
where $F^{(1)} \in \+R^{T\times 1}$ and $F^{(2)} \in \+R^{T\times 1}$ denote the vector of the first and second factors, respectively. As $\Nx, \Ny\rightarrow \infty$, $\hat\Sigma_{\Lam,t}^{(\gamma)}$ converges to
\begin{align*}
\lim_{\Nx, \Ny\rightarrow\infty} \hat\Sigma_{\Lam,t}^{(\gamma)}&=\lim_{\Nx, \Ny\rightarrow\infty} \frac{\Nx}{\Nx+\Ny}\Bigggl(\ \underbrace{\begin{bmatrix} 0 & 0  \\ 0 & \sigma_{\LamX}^2 \end{bmatrix}}_{\Sigma_{\LamX}} + \gamma \frac{\Ny}{\Nx} \cdot \underbrace{\begin{bmatrix} \sigma_{\LamY}^2  & 0\\ 0 & \frac{\sigma_{\LamY}^2}{\Ny^{1/2}}  \end{bmatrix}}_{{\Sigma_{\LamY},t}}\ \Bigggr) \\
&= \lim_{\Nx, \Ny\rightarrow\infty}\begin{pmatrix}  \gamma\frac{\Ny}{\Nx}\cdot \sigma_{\LamY}^2 & 0\\ 0 & \sigma_{\LamX}^2+\gamma\frac{\Ny^{1/2}}{\Nx}\cdot \sigma_{\LamY}^2 \end{pmatrix},
\end{align*}
where $\Sigma_{\LamX}$ and $\Sigma_{{\LamY},t}$ are the second-moment matrices of the loadings of $X$ and ${Y}$.\footnote{We use subscript $t$ in $\Sigma_{{\LamY},t}$ to account for the case where there are missing observations in $Y$ and the second-moment loading matrices of $\tilde{Y}=Y\odot W^Y$ is time-varying. } \footnote{{If $\Ny/\Nx \rightarrow \infty$, we can estimate the second moment matrix $Z^{(\gamma)} Z^{(\gamma)\T}/\Ny$. Selecting $\gamma = r \cdot \Nx/\Ny$ for a positive constant $r$ can still ensure that $\hat\Sigma_{\Lam,t}^{(\gamma)}$ converges to a full rank matrix, which allows for the consistent estimation of both factors. }}

The key idea for selecting $\gamma$ is to obtain full rank for the limit of the second-moment matrix $\hat\Sigma_{\Lam,t}^{(\gamma)}$. Choosing $\gamma = r\cdot N_x/N_y$ with some positive constant $r$ ensures that both eigenvalues in the limit of $\hat\Sigma_{\Lam,t}^{(\gamma)}$ are of the same order, and therefore, both factors can be identified from $ Z^{(\gamma)} Z^{(\gamma)\T}/(\Nx+\Ny)$. If $\gamma$ is not chosen at this rate, for example, $\gamma = 1$, then only the second factor can be identified as a strong factor in the concatenated panel. We formalize the above discussion in the following proposition, which is a special case of the general Theorem \ref{thm: consistency for loadings}, {and provide further details in Section \ref{sec: results}.}

\begin{proposition} \label{prop: consistency example}
Under the data generating process and observation pattern described in this section, let $\delta_{\Ny,T} = \min(\Ny,T)$ and assume that $\Ny/\Nx \rightarrow 0$. Target-PCA with $\gamma = r\cdot {\Nx}/{\Ny}$ for some positive scaling constant $r$ can consistently estimate the latent factors. As $T, {\Nx}, {\Ny} \rightarrow \infty$, there exists some rotation matrix $H$ such that
\[
\delta_{{\Ny},T} \left( \frac{1}{T} \sum_{t=1}^T \left\|\tilde{F}_t - H F_t \right\|^2 \right) = O_p(1).
\]
If $\gamma = O(1)$, then $\tilde{F}_t$ can be inconsistent.
\end{proposition}

\subsection{Effect 2: Efficiency Effect of Target Weight $\gamma$}\label{subsec: efficiency effect}

After we have selected $\gamma$ in the right order to ensure consistency, we can improve the efficiency by selecting the optimal scale of $\gamma$. We call this the efficiency effect of $\gamma$. The essence of this effect is to use the target weight $\gamma$ to balance the idiosyncratic noise levels between $X$ and $Y$, thus achieving the smallest asymptotic variance in the estimation of factors, loadings, and common components. 

In this section, we illustrate the efficiency effect of $\gamma$ through a simple one-factor model, where factor identification is not a concern and where we can focus on the efficiency effect. The key element is that the idiosyncratic noise levels in $X$ and $Y$ are different. 

More specifically, we assume that $X$ and $Y$ contain the same latent factor. Factors, loadings and idiosyncratic errors are drawn independently from $F_t \stackrel{\iid}{\sim} (0,\sigma_F^2), ({\LamY})_i\stackrel{\iid}{\sim}(0,\sigma_{\LamY}^2), (\LamX)_i\stackrel{\iid}{\sim}(0,\sigma_{\LamX}^2)$, $(\eX)_{ti}\stackrel{\iid}{\sim}(0,\sigma_{\eX}^2),$ and $({\eY})_{ti}\stackrel{\iid}{\sim}(0,\sigma_{\eY}^2)$. Furthermore, the factors and loadings have bounded fourth moments, and errors have bounded eighth moments. Suppose all entries in $Y$ are missing at random with observation probability $\PP(\WY_{ti}=1)=p$ and the number of units in $X$ and $Y$ are at the same order, i.e., $\Ny/\Nx \rightarrow c$ for some $c$ bounded away from $0$. As shown in our general model, the main conclusions do not depend on these specific assumptions.

In this setting, Proposition \ref{prop: efficiency example} provides the asymptotic distribution of the common components of $Y$. The result is a special case of the general Theorem \ref{thm: asymptotic distribution} in Section \ref{sec: results}. We optimize $\gamma$ by minimizing the asymptotic variance.

\begin{proposition}\label{prop: efficiency example}
Under the data generating process and observation pattern described in this section, let $\delta_{\Ny, T} = \min(\Ny,T)$ and suppose $\Ny / \Nx\rightarrow c\in (0,\infty)$. As $T,{\Nx},{\Ny} \rightarrow \infty$, the asymptotic distribution of the estimated common component of $Y$ for any $i$ and $t$ is 
\[
\sqrt{\delta_{\Ny, T}} (\Sigma_{C,ti}^{(\gamma)})^{-1/2} \Lp\tilde{C}_{ti} - C_{ti} \Rp \overset{d}{\rightarrow} \Ncal(0,1),
\]
where 
\begin{align*}
\Sigma_{C,ti}^{(\gamma)} = & \frac{\delta_{{\Ny}T}}{T}\frac{\sigma_{\eY}^2}{p\sigma_F^2}F_t^2 + \frac{\delta_{{\Ny}T}}{T}\left(\frac{1}{p}-1\right) \sigma_F^{-4}\Var(F_t^2)({\LamY})_i^2F_t^2 \\
& + \frac{\delta_{{\Ny}T}}{{\Ny}}({\LamY})_i^2\left(\sigma_{\LamX}^2 + \gamma\frac{{\Ny}}{{\Nx}} p\sigma_{\LamY}^2\right)^{-2} \left(\frac{{\Ny}}{{\Nx}}\sigma_{\LamX}^2\sigma_{\eX}^2 + \gamma^2\frac{\Ny^2}{\Nx^2} p\sigma_{\LamY}^2\sigma_{\eY}^2 \right).
\end{align*}
The optimized $\gamma$ that minimizes $\Sigma_{C,ti}^{(\gamma)}$ is $\gamma^* = \sigma_{\eX}^2/\sigma_{\eY}^2$ for any $i$ and $t$.
\end{proposition}
The optimized $\gamma^*$ minimizes $\Sigma_{C,ti}^{(\gamma)}$, which is the asymptotic variance of the common components as a function of $\gamma$. In this case, the optimized $\gamma^*$ is the ratio of the variances of idiosyncratic errors $\sigma_{\eX}^2$ and $\sigma_{\eY}^2$. We up-weight the target panel $Y$ if $\sigma_{\eY}^2 <\sigma_{\eX}^2$; otherwise, we down-weight the target panel $Y$.  Interestingly and counter-intuitively, when observations are missing at random, the optimized $\gamma^*$ is the same for all $i$ and $t$, and does not depend on any of the following parameters, even though $\Sigma_{C,ti}^{(\gamma)}$ is a function of these parameters: the observation probability $p$, the value of factors $F_t$ and loadings $\Lambda_i$, variance of factors $\sigma_F^2$ and loadings $\sigma_{\LamX}^2$ and $\sigma_{\LamY}^2$, and number of cross-section units $\Nx$ and $\Ny$. The intuition in this setting is the same as that for choosing the optimal weighting matrix for a weighted least squares (WLS) estimator. Specifically, suppose the heteroscedasticity of residuals does not depend on values of covariates as in this setting; then the optimal weighting matrix in WLS is proportional to the inverse variance of residuals and does not directly depend on the covariates and model parameters.

\begin{figure}[t!]
    \tcapfig{Relative MSE for different $\sigma_{\eX}/ \sigma_{\eY}$ and ${\Nx}/{\Ny}$}
     \centering
     \begin{subfigure}[b]{0.32\textwidth}
         \centering
         \includegraphics[width=\textwidth]{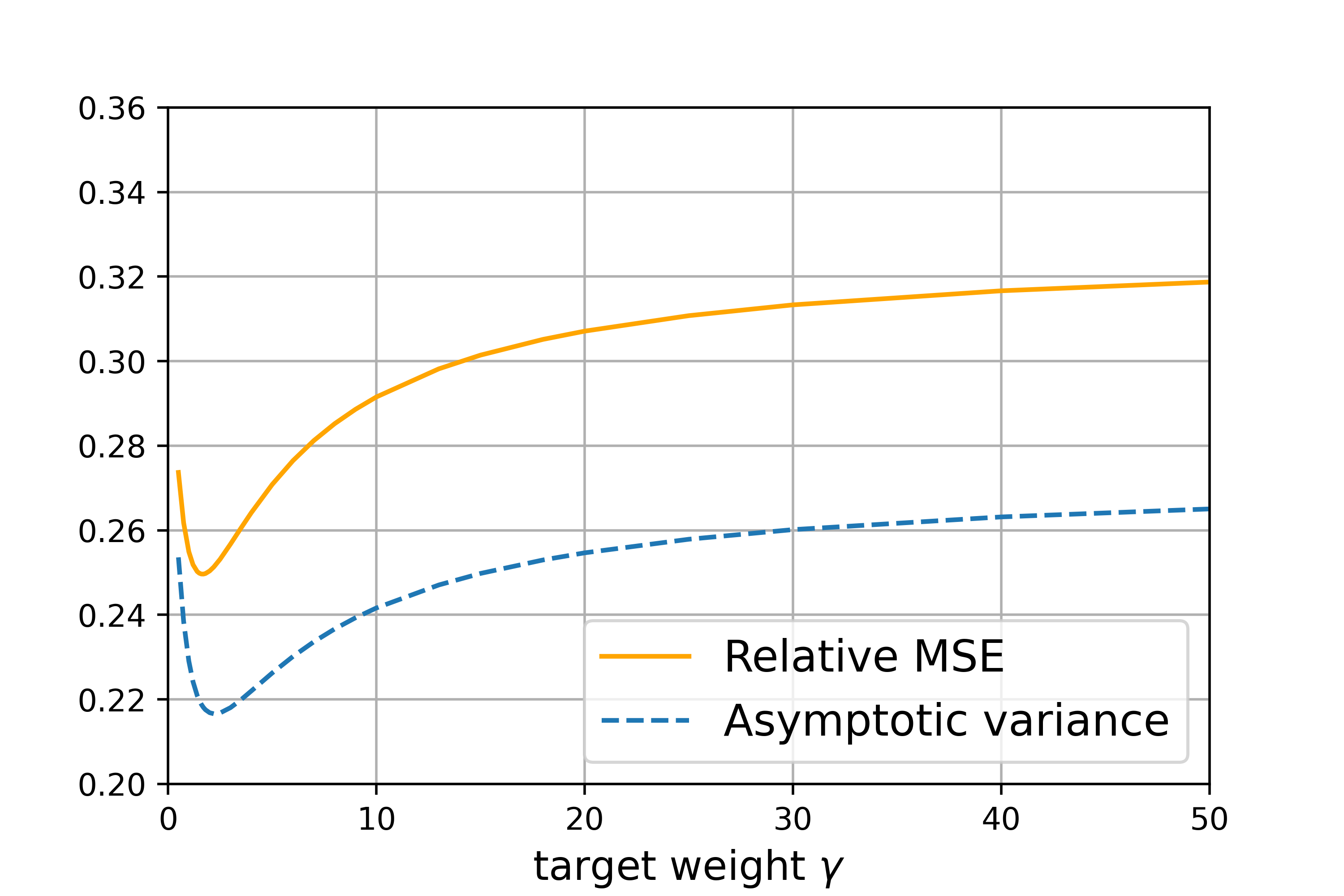}
         \includegraphics[width=\textwidth]{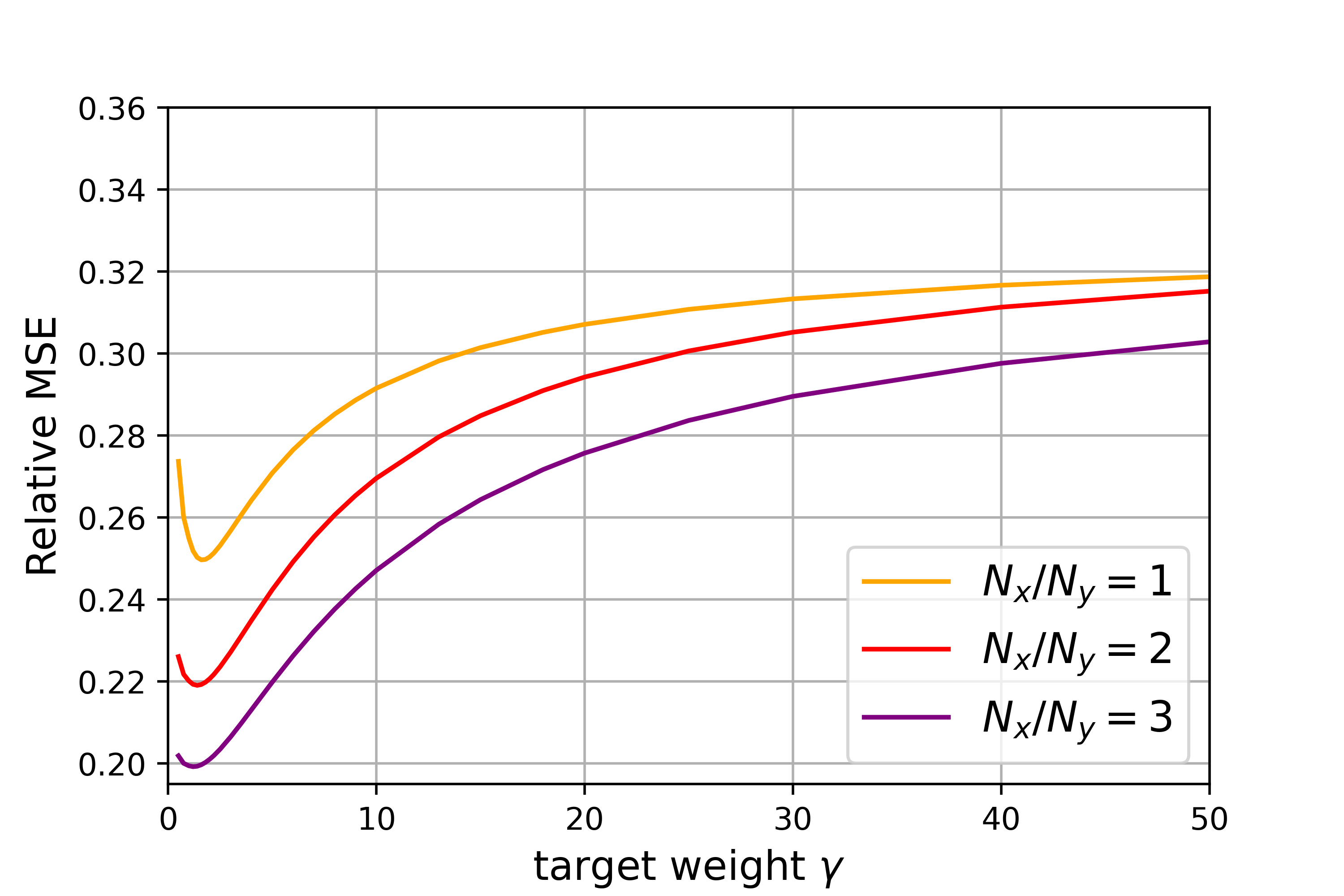}
         \caption{$\sigma_{\eX}=6$}
         \label{fig: sub_sigma_e=6}
     \end{subfigure}
     \hfill
     \begin{subfigure}[b]{0.32\textwidth}
         \centering
         \includegraphics[width=\textwidth]{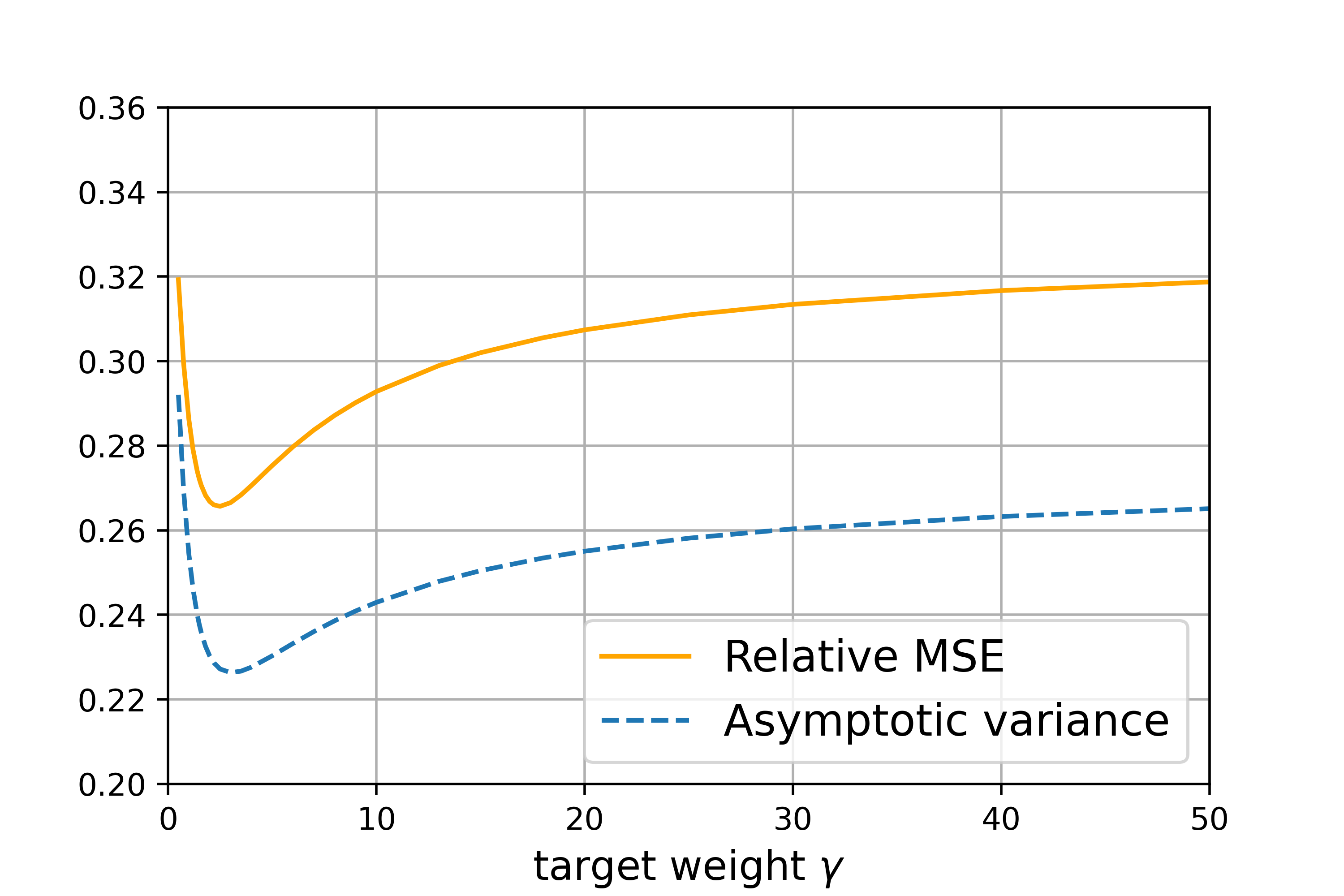}
         \includegraphics[width=\textwidth]{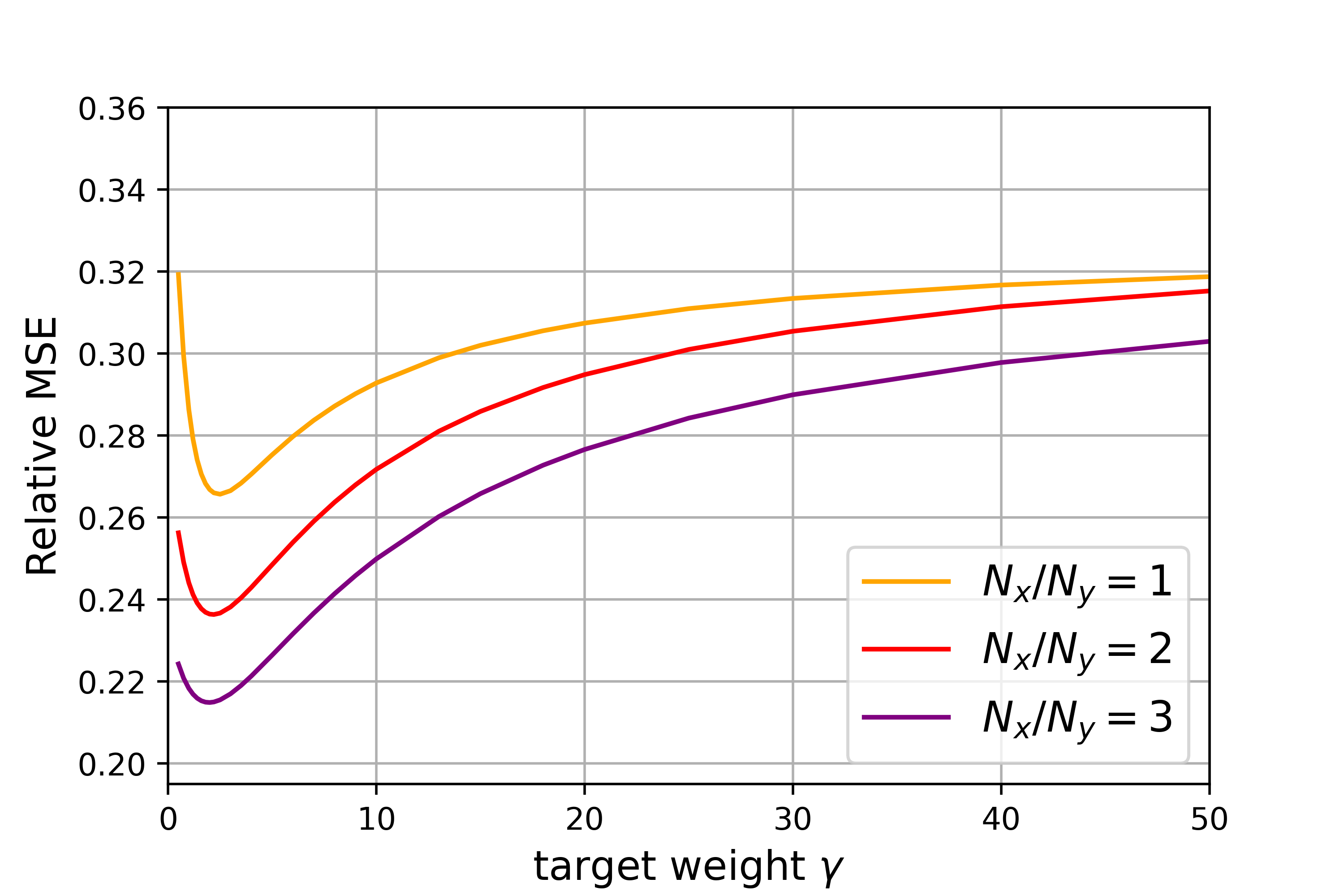}
         \caption{$\sigma_{\eX}=7$}
         \label{fig: sub_sigma_e=7}
     \end{subfigure}
     \hfill
     \begin{subfigure}[b]{0.32\textwidth}
         \centering
         \includegraphics[width=\textwidth]{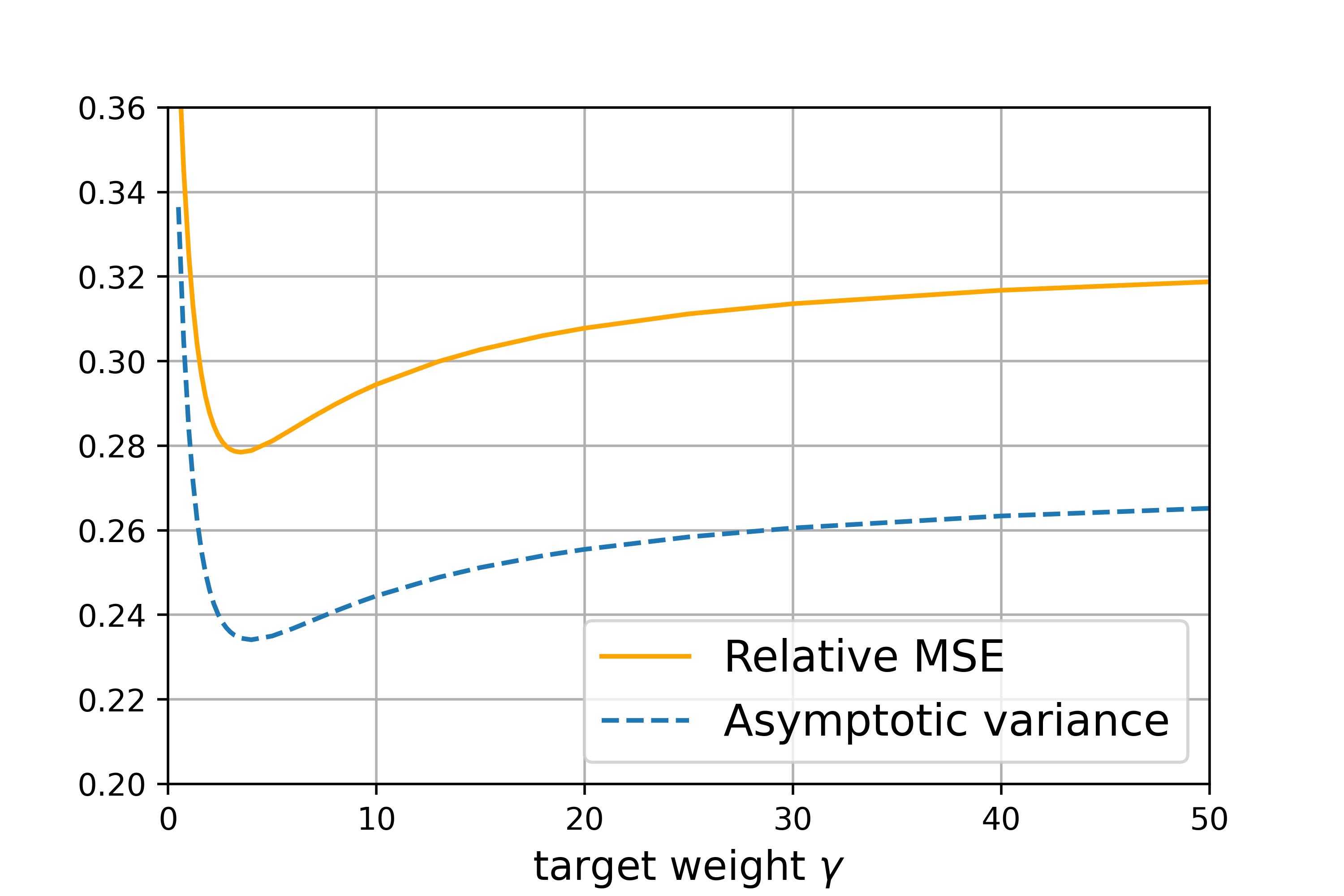}
         \includegraphics[width=\textwidth]{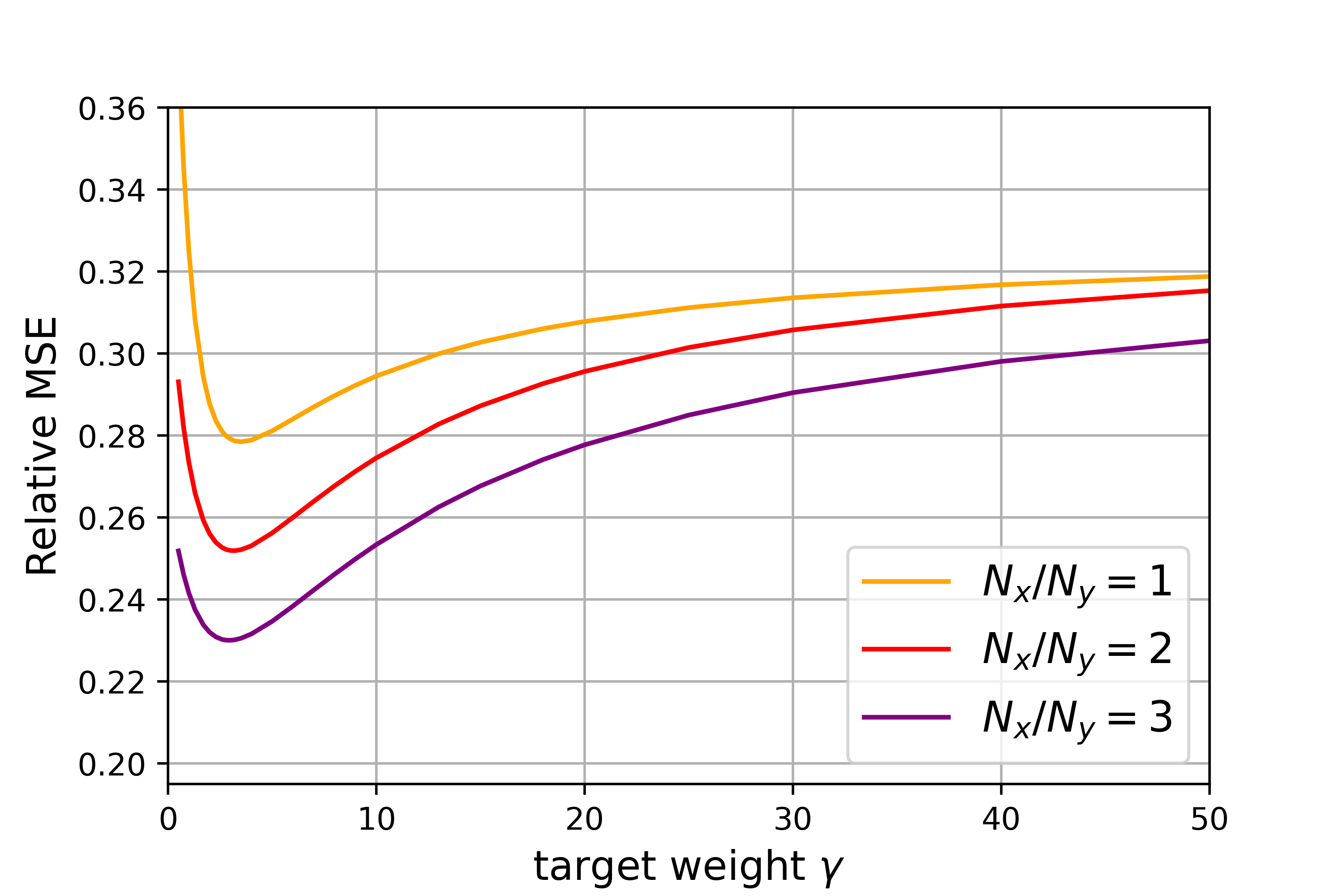}
         \caption{$\sigma_{\eX}=8$}
         \label{fig: sub_sigma_{\eX}=8}
     \end{subfigure}
    \label{fig: shape of gamma curve 2}
    \floatfoot{These figures show the relative MSE of the estimated common components of all entries in $Y$. The data generating process and the observation pattern follow the example described in Section \ref{subsec: efficiency effect}. Specifically, the factors, loadings, and errors are generated from normal distributions with mean zero and $\sigma_F=\sigma_{\LamY}=\sigma_\Lam=1$ and $\sigma_{\eY}=4$, ${\Ny}=T=200$, and observation probability $p=0.6$. For the subplots in the first row, we set ${\Nx}=200$. We run 200 simulations for each setup.}
\end{figure}
 
We illustrate the efficiency effect of this simple model in a simulation. Figure \ref{fig: shape of gamma curve 2} shows the relative mean squared error (relative MSE) of the estimated common components and the theoretical asymptotic covariance $\Sigma_{C,ti}^{(\gamma)}$ averaged over all entries in $Y$ as a function of $\gamma$. We consider different combinations of the noise ratio $\sigma_{\eX}/\sigma_{\eY}$ and the dimension ratio $\Nx/\Ny$ for this one-factor example.\footnote{Section \ref{sec: simulation} presents a comprehensive simulation analysis and provides further details for the simulation setup, including the formal definition of relative MSE.} 

There are three main takeaways from Figure \ref{fig: shape of gamma curve 2}. First, and most importantly, the target weight $\gamma$ that minimizes the relative MSE equals the optimized $\gamma^*$ that minimizes the asymptotic variance. Hence, we can use the inferential theory as a guidance for selecting the target weight $\gamma$. We will further elaborate on the optimal selection in Section \ref{sec: results}. Second, we confirm that the optimized $\gamma^*$ only varies with $\sigma_{\eX}/\sigma_{\eY}$, but not ${\Nx}/{\Ny}$ in this missing at random example. Third, the effect of optimizing $\gamma$ is larger if panel $Y$ and $X$ are more different in terms of noise and dimension. In all cases, PCA on $Y$ only (target PCA with $\gamma = \infty$) or on $X$ only (target PCA with $\gamma = 0$) is not optimal. The difference in relative MSE between target PCA with the optimized $\gamma^\ast$ and the corner case of PCA on $Y$ increases with $\Nx/\Ny$ and $\sigma_{\eY}/\sigma_{\eX}$ (i.e., the information in $X$ is more useful). The difference in relative MSE between target PCA with the optimized $\gamma^\ast$ and the corner case of PCA on $X$ only increases with $\Ny/\Nx$ and $\sigma_{\eX}/\sigma_{\eY}$ (i.e., the information in $X$ is less useful).         

In summary, the consistency and efficiency effects combined provide the protocol for selecting $\gamma$. First, we select the rate of $\gamma$ as $\gamma = r\cdot {\Nx}/{\Ny}$ to ensure consistency under general conditions. In the second step, we select the scaling positive constant $r$ to minimize the asymptotic variance of the common component to ensure an efficient estimation.

\section{Assumptions}\label{subsec: assumptions}

In this section, we lay out the assumptions on the observation patterns of $Y$ and the approximate factor models on both $X$ and $Y$. First, we introduce the assumptions on the observation pattern.

\begin{assumpG}[Observation pattern]\label{assump: obs pattern} \texttt{} 
\begin{enumerate}
    \item The observation matrix $\WY$ is independent of the factors $F$ and idiosyncratic errors $e_y$.
    \item Let $\QY_{ij}= \{t: \WY_{ti}=\WY_{tj}=1\}$ be the set of time periods when both units $i$ and $j$ of $Y$ are observed. For any given observation matrix $\WY$, there exists a positive constant $q$ such that ${|\QY_{ij}|}/{T}\geq q>0$ for all $i,j$. Let $q_{ij} = \lim_{T\rightarrow \infty}{|\QY_{ij}|}/{T}$ and $q_{ij,hl} = \lim_{T\rightarrow \infty}{|\QY_{ij}\cap \QY_{hl}|}/{T}$. For any $i,j,h,l$, $q_{ij}$ and $q_{ij,hl}$ are positive constants bounded away from 0.
    Furthermore, the number of observed units in $Y$ at any time period $t$ is proportional to $\Ny$, i.e., there exists a positive constant $q'$ such that $N_y^{-1}\sum_{i=1}^{\Ny}W^Y_{ti} \geq q' >0$ for all $t$.
\end{enumerate}
\end{assumpG}

Assumption \ref{assump: obs pattern} allows for very general observation patterns. 
The observation matrix $\WY$ can depend on the cross-sectional information, for example, the factor loadings $\Lambda_y$ or time-invariant observed covariates of the units. For the purpose of identification, Assumption \ref{assump: obs pattern} rules out the dependence between $\WY$ and $F$. Note that the time and cross-section dimensions are symmetric in the estimation of common components. Therefore, the symmetric case where $\WY$ depends on $F$ but is independent of $\Lambda_y$ is allowed by swapping the roles of $N$ and $T$ in the estimation of the factor model. We assume that $\WY$ is independent of $e_y$, which is conceptually similar to the unconfoundedness assumption in \cite{rosenbaum1983central}.

In Assumption \ref{assump: obs pattern}, we assume that the number of time periods, for which any four units are simultaneously observed, grows at the order of $T$. This assumption implies that every entry in the covariance matrix $\Sigma^{Z^{(\gamma)}}$ can be consistently estimated at the rate $\sqrt{T}$, and ensures that the asymptotic variances of the estimated factors, loadings and common components from Target-PCA are well-defined. In Appendix \ref{appendix: sparser observation pattern}, we generalize Assumption \ref{assump: obs pattern} to allow $|\QY_{ij}|$ to grow sub-linearly in $T$. This does not change the conceptual arguments, but leads to a more complex notation.

We assume that both, $X$ and $Y$, follow an approximate factor model similar to \cite{bai2003inferential}. We allow for non-trivial time-series dependency of factors, and non-trivial cross-sectional dependency of loadings in $X$, in $Y$, and between $X$ and $Y$. In addition, the idiosyncratic errors can be weakly correlated in $X$, $Y$, and between $X$ and $Y$, in both the time-series and cross-sectional dimensions. The asymptotic distributions are based on general martingale central limit theorems. However, we make a key relaxation for the factor model relative to \cite{bai2003inferential}. Some factors in $Y$ are allowed to be weak and only affect a small subset of units in $Y$. However, those factors have to be strong in $X$ in order to identify them from $Z^{(\gamma)}$ with a properly chosen $\gamma$ (as suggested in Section \ref{sec: A Simplified Factor Model}).

The assumptions for the general factor model are stated in Assumptions \ref{assump: factor model} and \ref{assump: additional assumptions} in Appendix \ref{subsec: general assumptions}. They are delegated to the Appendix as most elements are standard, but fairly technical. Our consistency result in Theorem \ref{thm: consistency for loadings} and asymptotic distribution result in Theorem \ref{thm: asymptotic distribution} are derived under Assumptions \ref{assump: obs pattern}, \ref{assump: factor model} and \ref{assump: additional assumptions}.

We use the simplified model in Example \ref{example: simplified factor model} to illustrate how our factor model generalizes the conventional approximate factor models. It allows us to highlight the relaxation of weak factors. Appendix \ref{appendix: simplified model} shows how our main results are simplified for this example and shows that this example is a special case of our general framework.
\begin{example}[Simplified factor model]  \label{example: simplified factor model}
There exists constant $C<\infty$ such that
\begin{enumerate}
    \item Factors: $F_t \overset{\text{i.i.d.}}{\sim} (0,\Sigma_F)$ and $\E\|F_t\|^4 \leq C$ for any $t.$
    \item Loadings: 
    $(\LamX)_{i} \overset{\text{i.i.d.}}{\sim} (0,\Sigma_{\LamX})$, where $\Sigma_{\LamX}$ is positive semidefinite. 
    $(\LamY^{\mathrm{full}})_{i} \overset{\text{i.i.d.}}{\sim} (0,\Sigma^{\mathrm{full}}_{\LamY})$  and the loading of the $j$-th factor $(\LamY)_{ij} = (\LamY^{\mathrm{full}})_{ij} \cdot (U_{y})_{ij}$, where $\Sigma^{\mathrm{full}}_{\LamY}$ is positive definite and the Bernoulli random variable $(U_{y})_{ij} \in \{0,1\}$ is independent in $i$ with $\PP((U_{y})_{ij} = 1) = p_{j}$ for some $p_{j} \in [0,1]$.
    Furthermore, $\E\|(\LamX)_i\|^4\leq C$, $\E\|(\LamY)_i\|^4\leq C$, $N_y^{-1}\sum_{i=1}^{\Ny}(\LamY)_i(\LamY)_i^\top \overset{p}{\rightarrow} \Sigma_{\LamY}$, and $ \Sigma_{\LamX}+\Sigma_{\LamY}$ is positive definite. For any $t,$ $N_y^{-1}\sum_{i=1}^{\Ny}\WY_{ti}(\LamY)_i (\LamY)_i^\top \overset{p}{\rightarrow} \Sigma_{\LamY,t}$ and  $\Sigma_{\LamX}+\Sigma_{\LamY,t}$ is positive definite.
    \item Idiosyncratic errors: $(\eX)_{ti}\overset{\text{i.i.d.}}{\sim} (0,\sigma_{\eX}^2), (\eY)_{ti}\overset{\text{i.i.d.}}{\sim} (0,\sigma_{\eY}^2)$, $\E(\eX)^8_{ti}\leq C, \E(\eY)^8_{ti}\leq C$.
    \item Independence: $F,\LamX,\LamY,\eX$ and $\eY$ are independent.
\end{enumerate}
\end{example}

The main difference to the conventional factor models is the loadings, while the factors and errors capture the stylized properties in a usual factor setup. The simplified model in this example assumes that all observations are i.i.d. Allowing for more complex dependencies as in our general model, does not change the arguments, but makes the notation more complex. The key element is the strength of the factors measured by their loadings. Specifically, we measure the strength of the factors by the fraction $p_j$ of units in $Y$ that are affected by the corresponding factor. The error terms are non-systematic with bounded eigenvalues in the covariance matrix. 

The assumptions on the loadings $(\LamY)_{i}$ account for three cases of factor strength in $Y$. First, if $p_{j}$ is bounded away from $0$ as $\Ny$ grows, then the $j$-th factor is a strong factor in $Y$. Second, if $p_{j}$ decays to $0$ but is nonzero as $\Ny$ grows, then the $j$-th factor is a weak factor in $Y$. Third, if $p_{j}$ is $0$ for all $\Ny$, then $Y$ does not contain the $j$-th factor. Note that $\Sigma_{\LamX}$ can be rank deficient, implying that the loadings of some factors can be zero for units in $X$. However, our assumption on $(\LamX)_{i}$ rules out the case of weak factors in $X$ as the estimation of weak factors in $X$ is not our objective. We assume that $ \Sigma_{\LamX}+\Sigma_{\LamY}$ is positive definite to ensure that each factor in $F$ is strong in at least one of the two panels $X$ and $Y$. Specifically, weak factors in $Y$ are strong in $X$. Hence, all factors can be identified with target-PCA with a properly chosen $\gamma$. 

The loading assumption also imposes assumptions on the missing pattern in $Y$ to identify all factors when combining the partially observed $Y$ and $X$. More specifically, the second-moment matrix $\Sigma_{\LamY,t}$ does not need to be full rank in Assumption \ref{assump: simplified factor model}.2, which relaxes the full-rank assumption of $\Sigma_{\LamY,t}$ in \cite{xiong2019large}. However, we assume that $\Sigma_{\LamX}+\Sigma_{\LamY,t}$ is positive definite, so that target-PCA can identify all factors from $X$ and partially observed $Y$.

We assume that the number of factors $k$ can be consistently estimated. Given a consistent estimator for the number of factors, we can treat $k$ as known. After selecting $\gamma$ at the rate $r \cdot N_x/N_y$ for some positive scaling constant $r$, the estimation of the number of factors from the weighted concatenated panel $Z^{(\gamma)}$ is the same as in \cite{xiong2019large}. Hence, given the various bounds and expansions derived in this paper, it seems possible to extend the estimator for the number of factors developed in \cite{bai2002determining} to our case of general missing values. A promising alternative is to use cross-validation arguments. However, this is non-trivial for complex missing patterns that can depend on the factor model itself. In summary, given our analysis and selecting $\gamma$ at the right rate, we can transform the problem of estimating the number of factors into a familiar setup. Furthermore, in our empirical analysis, we show that our estimator can be robust to the number of factors once $\gamma$ is selected appropriately.

\section{Inferential Theory} \label{sec: results}

In this section, we provide the asymptotic results of the estimated factor model from target-PCA under general assumptions on the approximate factor model and missing patterns. We present the consistency result in Section \ref{subsec: consistency} and asymptotic normality results in Section \ref{subsec: asymptotic normality}.

\subsection{Consistency} \label{subsec: consistency} 

The loadings and factors can be consistently estimated only if $\gamma$ is properly chosen. The consistency result is an important intermediate step to show the inferential theory in Section \ref{subsec: asymptotic normality}.

\begin{theorem} \label{thm: consistency for loadings}
Let $\delta_{\Ny, T} = \min(\Ny,T)$ and suppose that $\Ny / \Nx \rightarrow c\in [0,\infty)$. Under Assumptions \ref{assump: obs pattern} and \ref{assump: factor model}, for $T,\Nx,\Ny \rightarrow \infty$:
\begin{enumerate}
    \item If $\gamma = r \cdot \Nx/\Ny$ for some positive constant $r$, then $\Sigma_{\Lam}^{(\gamma)}:=\lim_{\Nx, \Ny\rightarrow \infty} {\Nx}{(\Nx+\Ny)^{-1}} (\Sigma_{\LamX} + \gamma {\Ny}/{\Nx} \cdot \Sigma_{{\LamY}})$ and $\Sigma_{\Lam,t}^{(\gamma)}:=\lim_{\Nx, \Ny\rightarrow \infty} {\Nx}{(\Nx+\Ny)^{-1}} (\Sigma_{\LamX} + \gamma {\Ny}/{\Nx} \cdot \Sigma_{{\LamY},t})$ are positive definite. It holds that
    \begin{equation}
        \delta_{\Ny,T} \left( \frac{1}{\Nx+\Ny }\sum_{i=1}^{\Nx+\Ny } \left\|\tilde{\Lam}_i^{(\gamma)} - H^{(\gamma)}\Lam_i^{(\gamma)}\right\|^2 \right) = O_p(1),
    \end{equation}
    \begin{equation}
        \delta_{\Ny,T} \left( \frac{1}{T }\sum_{t=1}^{T } \left\|\tilde{F}_t - (H^{(\gamma)\top})^{-1} F_t\right\|^2 \right) = O_p(1),  
    \end{equation}
    where $H^{(\gamma)} = T^{-1}(\Nx+\Ny )^{-1}(\tilde{D}^{(\gamma)})^{-1}\tilde{\Lam}^{(\gamma)\top} \Lam^{(\gamma)} F^\top F $. This implies that the estimated loadings $(\tilde{\Lambda}_y)_i$ and common components $\tilde{C}_{ti}$ of $Y$ are consistent. 
    \item Under Assumption \ref{assump: factor shift}, if $\gamma$ and $N_x/N_y$ are not of the same order, then $\Sigma_{\Lam,t}^{(\gamma)}$ is not positive definite when $\Ny/\Nx\rightarrow 0$.
    If $\Sigma_{\Lam,t}^{(\gamma)}$ is not positive definite, then $\tilde{F}_t$ does not converge at the rate $\delta_{\Ny, T}$ and can be inconsistent for {the factors that are strong in $Y$ but not in $X$.}
\end{enumerate}
\end{theorem}

Theorem \ref{thm: consistency for loadings} states the consistency effect of $\gamma$ and generalizes Proposition \ref{prop: consistency example} to general factor models. According to Theorem \ref{thm: consistency for loadings}, choosing $\gamma = r\cdot \Nx/\Ny$ with some positive constant $r$ ensures that factors, loadings, and common components of $Y$ can be consistently estimated up to a rotation matrix $H^{(\gamma)}$. The convergence rate is at the smaller of $\sqrt{\Ny}$ and $\sqrt{T}$. This rate is the same as the convergence rate in \cite{bai2002determining} that applies PCA to $Y$ when all factors are strong in $Y$. This rate makes sense for target-PCA: When we up-weight $Y$ by a rate of $ \Nx/\Ny$, the error from $Y$ is always a leading term in target-PCA.  

{If $\gamma$ is not selected at the right order, the estimates of the factors that are strong in $Y$, but not in $X$, are inconsistent unless stronger assumptions are imposed. Specifically, if $N_y = N_x^\alpha$ for some $\alpha \in (0, 1)$ and $\gamma = 1$, then the consistency requires $(N_x/N_x^{\alpha}) \cdot (1/T) \rightarrow 0$, that is, a larger $T$, analogous to Assumption A4 in \cite{bai2021approximate}. However, this assumption is not required if $\gamma$ is chosen at the order of $N_x/N_y$.}

\subsection{Asymptotic Normality} \label{subsec: asymptotic normality}

Based on the consistency results in Theorem \ref{thm: consistency for loadings}, we develop the inferential theory for target-PCA in this section. Theorem \ref{thm: asymptotic distribution} shows the asymptotic distribution of estimated factors, estimated loadings of $Y$, and estimated common components of $Y$ with target weight $\gamma = r\cdot \Nx/\Ny$ for every positive scaling constant $r$ in target-PCA under general assumptions. We show the asymptotic distribution of $Y$ because the factor model and common components of $Y$ are of our primary interest. The asymptotic distribution of $X$ can be shown analogously.

\begin{theorem} \label{thm: asymptotic distribution}
Define $\delta_{\Ny, T} = \min(\Ny ,T)$. Suppose that $\Ny / \Nx\rightarrow c\in [0,\infty)$ and $\gamma = r\cdot \Nx/\Ny$ for some positive constant $r$. If the eigenvalues of $\Sigma_F \Lp\Sigma_{\LamX}+r\Sigma_{\LamY}\Rp$ are distinct, then under Assumptions \ref{assump: obs pattern}, \ref{assump: factor model} and \ref{assump: additional assumptions}, as $T,\Nx,\Ny  \rightarrow \infty$ we have for each $i$ and $t$:
\begin{enumerate}
\item For $\sqrt{T}/\Ny  \rightarrow 0$, the asymptotic distribution of the estimated loadings of $Y$ is
\begin{equation} \label{equ: dist of loadings}
    \sqrt{T} (\Sigma_{\LamY,i}^{(\gamma)})^{-1/2} \left((H^{(\gamma)})^{-1} (\tilde{\Lam}_y)_i - (\LamY)_i\right) \overset{d}{\rightarrow} \Ncal(0,I_k),
\end{equation}
where 
\[
\Sigma_{\LamY,i}^{(\gamma)} = \Sigma_{\LamY,i}^{(\gamma), {\rm obs}} + \Sigma_{\LamY,i}^{(\gamma), {\rm miss}},
\]
$\Sigma_{\LamY,i}^{(\gamma), {\rm obs}} = \Sigma_F^{-1}(\Sigma_{\Lam}^{(\gamma)})^{-1}  \Gamma_{\LamY,i}^{(\gamma),{\rm obs}}(\Sigma_{\Lam}^{(\gamma)})^{-1} \Sigma_F^{-1}$, 
$\Gamma_{\LamY,i}^{(\gamma),{\rm obs}}$ is defined in Assumption \ref{assump: additional assumptions}.6, $\Sigma_{\LamY,i}^{(\gamma), {\rm miss}} = \Sigma_F^{-1}(\Sigma_{\Lam}^{(\gamma)})^{-1}  h_{i+\Nx}^{(\gamma)}((\LamY)_i)(\Sigma_{\Lam}^{(\gamma)})^{-1} \Sigma_F^{-1}$, and the function $h_i^{(\gamma)}(\cdot)$ is defined in Assumption \ref{assump: additional assumptions}.8.
\item For $\sqrt{T}/\Ny  \rightarrow 0$ and $\sqrt{\Ny}/T  \rightarrow 0$,
\begin{itemize}
\item Case 1: If all the factors in $F_y$ are strong factors in $Y$, then the asymptotic distribution of the estimated factors is
\begin{equation}  \label{equ: dist of factors}
    \sqrt{\delta_{\Ny,T}} (\Sigma_{F,t}^{(\gamma)})^{-1/2} (H^{(\gamma)\top} \tilde{F}_t -F_t )\overset{d}{\rightarrow} \Ncal(0,I_k),
\end{equation}
where 
\[
\Sigma_{F,t}^{(\gamma)} =  \frac{\delta_{\Ny,T }}{N_y}\Sigma_{F,t}^{(\gamma), {\rm obs}}  + \frac{\delta_{\Ny,T}}{T}\Sigma_{F,t}^{(\gamma), {\rm miss}},
\]
$\Sigma_{F,t}^{(\gamma), {\rm obs}} = ({\Sigma}_{\Lam,t}^{(\gamma)})^{-1}\Gamma^{(\gamma),{\rm obs}}_{F,t}({\Sigma}_{\Lam,t}^{(\gamma)})^{-1}$, $\Gamma^{(\gamma),{\rm obs}}_{F,t}$ is defined in Assumption \ref{assump: additional assumptions}.7, $\Sigma_{F,t}^{(\gamma), {\rm miss}} = ({\Sigma}_{\Lam,t}^{(\gamma)})^{-1} \cdot g_t^{(\gamma)}\big((\Sigma_{\Lam}^{(\gamma)})^{-1}\Sigma_F^{-1}F_t\big)\cdot({\Sigma}_{\Lam,t}^{(\gamma)})^{-1}$, and the function $g_t^{(\gamma)}(\cdot)$ is defined in Assumption \ref{assump: additional assumptions}.8;
\item Case 2: Suppose some factors in $F_y$ are weak factors in $Y$. Let $F_{t,w}$ be the weak factors in $Y$, and let $F_{t,s} = F_t \backslash {F_{t,w}}$ be the remaining strong factors in $F_t$. For simplicity of notation, we assume that the loadings of the weak factors $F_{t,w}$ are asymptotically orthogonal to the loadings of the strong factors $F_{t,s}$. 
The asymptotic distribution of the estimated weak factors $(H^{(\gamma)\top} \tilde{F}_t)_{w}$ corresponding to $F_{t,w}$ is 
\begin{equation}  \label{equ: dist of factors weak}
    \sqrt{\delta_{N_w,T}} (\Sigma_{F_{w},t}^{(\gamma)})^{-1/2} ((H^{(\gamma)\top} \tilde{F}_t)_{w} -F_{t,w})\overset{d}{\rightarrow} \Ncal(0,I_k),
\end{equation}
where $\delta_{N_w,T} = \min(N_w,T)$, $N_w = \min({\Ny^2}/{{g(\Ny)}},{\Nx})$, $g(\Ny)$ is the rate at which $\sum_{i=1}^{\Ny}(\LamY)_{i,w}^2$ grows,
\[
\Sigma_{F_{w},t}^{(\gamma)} =  \frac{\delta_{N_w,T}}{N_w }\Sigma_{F_{w},t}^{(\gamma),{\rm obs}}  + \frac{\delta_{N_w,T}}{\Talpha}\Sigma_{F_{w},t}^{(\gamma),{\rm miss}}, \]
$\Sigma_{F_w,t}^{(\gamma), {\rm obs}} = ({\Sigma}_{\Lam,t,w}^{(\gamma)})^{-1}\Gamma^{(\gamma),{\rm obs}}_{F_w,t}({\Sigma}_{\Lam,t,w}^{(\gamma)})^{-1}$, $\Gamma^{(\gamma),{\rm obs}}_{F_w,t}$ is defined in Assumption \ref{assump: additional assumptions}.7, ${\Sigma}_{\Lam,t,w}^{(\gamma)}$ and $\Sigma^{(\gamma),{\rm miss}}_{F_{w},t}$ are respectively the diagonal blocks of  ${\Sigma}_{\Lam,t}^{(\gamma)}$ and $\Sigma^{(\gamma),{\rm miss}}_{F,t}$ corresponding to the weak factors.\footnote{The asymptotic distribution of the estimated strong factors $(H^{(\gamma)\top} \tilde{F}_t)_{s}$ is
$
\sqrt{\delta_{N_y,T}} (\Sigma_{F_{s},t}^{(\gamma)})^{-1/2} ((H^{(\gamma)\top} \tilde{F}_t)_{s} -F_{t,s})\overset{d}{\rightarrow} \Ncal(0,I_k),
$
where $\Sigma_{F_{s},t}^{(\gamma)}$ is the diagonal block of $\Sigma_{F,t}^{(\gamma)}$ corresponding to the strong factors $F_{t,s}$.
}
\end{itemize}
\item For $\sqrt{T }/\Ny  \rightarrow 0$ and $\sqrt{\Ny}/T  \rightarrow 0$, the asymptotic distribution of the estimated common components of $Y$ is 
\begin{equation}  \label{equ: dist of common component}
    \sqrt{\delta_{\Ny, T }} (\Sigma_{C,ti}^{(\gamma)})^{-1/2} (\tilde{C}_{ti} - C_{ti} ) \overset{d}{\rightarrow} \Ncal(0,1),
\end{equation}
where 
\begin{align*}
\Sigma_{C,ti}^{(\gamma)} =&\frac{\delta_{\Ny, T }}{T } F_t^\top \Lp \Sigma_{\LamY,i}^{(\gamma),{\rm obs}} + \Sigma_{\LamY,i}^{(\gamma),{\rm miss}}\Rp F_t + \frac{\delta_{\Ny, T }}{\Ny } (\LamY)_i^\top \Sigma_{F,t}^{(\gamma),{\rm obs}} (\LamY)_i \\
& + \frac{\delta_{\Ny, T }}{T } (\LamY)_i^\top \Sigma_{F,t}^{(\gamma),{\rm miss}}(\LamY)_i  -2 \frac{\delta_{\Ny, T }}{T } (\LamY)_i^\top   \Sigma_{\LamY,F,i,t}^{(\gamma),{\rm miss, cov}}  F_t,    
\end{align*}
$\Sigma_{\LamY,F,i,t}^{(\gamma),{\rm miss, cov}} = (\Sigma_{\Lam,t}^{(\gamma)})^{-1} \cdot g_{i,t}^{(\gamma),{\rm cov}}\big((\LamY)_i, (\Sigma_\Lam^{(\gamma)})^{-1}\Sigma_F^{-1}F_t\big) \cdot(\Sigma_\Lam^{(\gamma)})^{-1}  \Sigma_{F}^{-1}$, and function $g_{i,t}^{(\gamma),{\rm cov}}(\cdot,\cdot)$ is defined in Assumption \ref{assump: additional assumptions}.8.
\end{enumerate}
\end{theorem}

The factors, loadings, and common components are asymptotically normally distributed. The asymptotic variance differs from the conventional PCA in \cite{bai2003inferential} in three aspects. First, the asymptotic variances $\Sigma_{\LamY,i}^{(\gamma)}$, $\Sigma_{F,t}^{(\gamma)}$ and $\Sigma_{C,ti}^{(\gamma)}$ depend on $\gamma$ (or equivalently $r$ when $\gamma = r \cdot \Nx/\Ny$). This will be important, as we will use it as the criteria to select the scale of $\gamma$. Second, in the case of missing data, we have additional correction terms to capture the additional uncertainty due to missingness. These correction terms follow the same structure and arguments as in \cite{xiong2019large}. Without missing data, the correction matrices in the variance disappear. Third, as shown in Theorem \ref{thm: asymptotic distribution}.2, the strong and weak factors in $Y$ have different convergence rates.\footnote{The estimated weak factors $F_{t,w}$ have faster convergence rates than the estimated strong factors $F_{t,s}$. This is a consequence of $g(\Ny)$ growing at a smaller rate than $\Ny$, which implies that $\delta_{N_w,T}^{1/2}$ grows at a larger rate than $\delta_{N_y,T}^{1/2}$. This result might seem to be counterintuitive at first glance, but makes sense after a careful analysis of the factor estimation errors. Note that the weak factors $F_{t,w}$ are essentially estimated from $X$. Their estimation errors are dominated by the cross-sectional average of the errors of $X$. In contrast, the strong factors $F_{t,s}$ are mainly estimated from $Y$ and their estimation errors are dominated by those of $Y$, which then converge to zero at a slower rate as $Y$ has fewer units than $X$. The asymptotic distribution of the estimated common components is dominated by that of the strong factors in $Y$, as strong factors have a slower convergence rate than weak factors. Therefore, the two cases in Theorem \ref{thm: asymptotic distribution}.2 lead to the same case in Theorem \ref{thm: asymptotic distribution}.3.} Importantly, this separation between strong and weak factors does not affect the asymptotic distribution of the estimated common components of $Y$. Hence, the asymptotic variance of the common components allows us to select an efficient target weight $\gamma$ independent of the factor strength.

The distribution results of Theorem \ref{thm: asymptotic distribution} simplify under Example \ref{example: simplified factor model}, and we can provide explicit expressions for the asymptotic variances. Appendix \ref{appendix: simplified model} shows the analytical expression of the asymptotic variances under the simplified factor model, which allows us to gain intuition on how $\gamma$ affects the efficiency of the estimation. The asymptotic variances depend on the following key quantities: the target weight $\gamma$, the noise ratio (NR) $\sigma_{\eX}/\sigma_{\eY}$, the dimension ratio (DR) $\Nx/\Ny$ and the dependency structure in the missing pattern. We illustrate this dependency and the effect on the optimal choice of $\gamma$ with a simulation example based on the model in Example \ref{example: simplified factor model} and different observation patterns. 

We {aim to select} the optimized $\gamma^\ast$ as the efficient target weight for the common components. 
The numerical examples in Table \ref{tab: gamma} illustrate how the optimized $\gamma^\ast$ that minimizes $\sum_{t,i} \Sigma_{C,ti}^{(\gamma)}$ depends on the observation pattern, noise ratio and dimension ratio, in a one-factor model. 
The fraction of observed data $p$ in $Y$ varies between 60\%, 75\% and 90\% with the four missing patterns: missing-at-random, block-missing, staggered-missing and mixed-frequency. The dimension ratio is either $\Nx/\Ny=1$ or 4, and the noise ratio $\sigma_{\eX}/\sigma_{\eY}$varies between 0.25, 1 and 4.

\begin{table}[t!]
    \centering
	\tcaptab{Optimized target weight $\gamma^*$ for different missing patterns and noise ratios} 
	\label{tab: gamma}
	\makebox[0pt][c]{\parbox{\textwidth}{%
	\begin{minipage}[c]{\linewidth}
	\centering
    \begin{tabular}{cc|ccc|ccc}
	\toprule
	& $p$ & 	\multicolumn{3}{c|}{ $\Nx/\Ny = 1$} &\multicolumn{3}{c}{ $\Nx/\Ny = 4$} \\ 
	& & NR=0.25 & NR=1 & NR=4
	 & NR=0.25 & NR=1 & NR=4 \\ 
	\midrule
	 \multirow{3}{*}{\ \includegraphics[width=15mm, height=14.5mm]{fig/thm/random.png} } 
	 &60\% &0.25 &1.00 &4.00 &0.25 &1.00 &4.00  \\
	 &75\% &0.25 &1.00 &4.00 &0.25 &1.00 &4.00  \\
	 &90\% &0.25 &1.00 &4.00 &0.25 &1.00 &4.00  \\
	 \midrule
	 \multirow{3}{*}{\ \includegraphics[width=15mm, height=14.5mm]{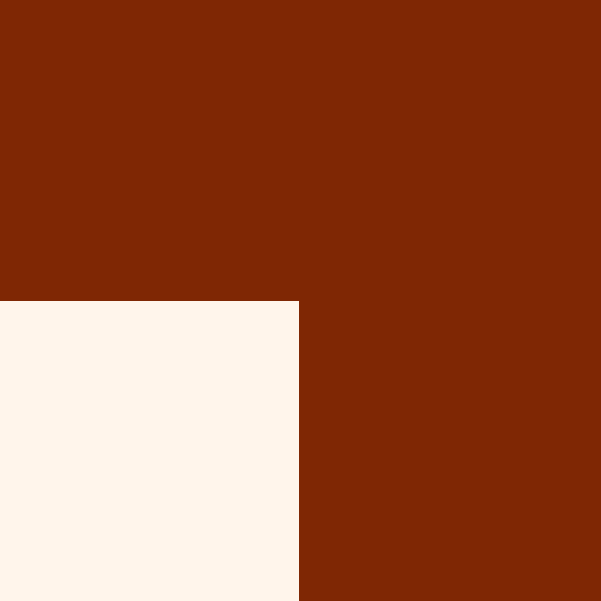} } 
	 &60\% &0.61 &1.75 &4.25 &1.95 &5.09 &7.00 \\
	 &75\% &0.42 &1.53 &4.35 &1.06 &3.62 &6.12 \\
	 &90\% &0.28 &1.15 &4.18 &0.40 &1.62 &4.61 \\ 
	 \midrule
	 \multirow{3}{*}{\ \includegraphics[width=15mm, height=14.5mm]{fig/thm/stagger.png} } 
	 &60\% &0.55 &1.96 &4.66 &1.69 &5.52 &7.84\\
	 &75\% &0.39 &1.47 &4.35 &0.93 &3.26 &5.87\\
	 &90\% &0.28 &1.13 &4.13 &0.40 &1.58 &4.51\\
	 \midrule
	 \multirow{3}{*}{\ \includegraphics[width=15mm, height=14.5mm]{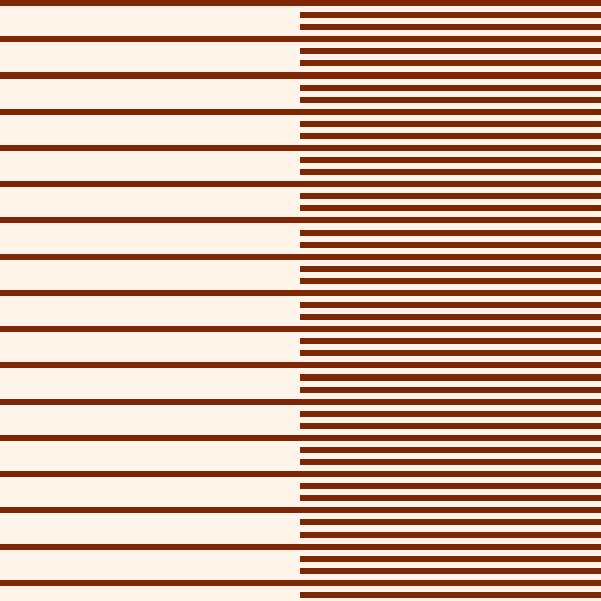} } 
	 &60\% &0.35 &1.72 &4.41 &0.99 &4.43 &6.89\\
	 &75\% &0.35 &1.41 &4.26 &0.79 &3.00 &5.64\\
	 &90\% &0.30 &1.15 &4.10 &0.48 &1.78 &4.60\\
	\bottomrule
    \end{tabular}
	\end{minipage}
	\hfill
	\floatfoot{This table reports the optimized $\gamma^*$ for different missing patterns, dimension ratios $\Nx/\Ny$ and noise ratios (NR) $\sigma_{\eX}^2/\sigma_{\eY}^2$. In this table, the optimized $\gamma^*$ equals $\Nx/\Ny$ multiply by the efficient scaling $r^*$ obtained by minimizing $\sum_{t,i} \Sigma_{C,ti}^{(\gamma)}$ in Corollary \ref{thm: simplified factor model}. The figures on the left show the observation patterns, with the shaded entries indicating the missing entries. We set $\Ny=T=50$ and the fraction of observed entries $p$ to $60\%, 75\%$ and $90\%$. We generate a one-factor model where factors, loadings, and errors are drawn from normal distributions with $\Sigma_F=\Sigma_{\LamX}=\Sigma_{\LamY}=1$. We let $\sigma_{\eX}^2=1$ and $\sigma_{\eY}^2=4$ for $\text{NR}=0.25$, $\sigma_{\eX}^2=\sigma_{\eY}^2=1$ for $\text{NR}=1$, and $\sigma_{\eX}^2=4$ and $\sigma_{\eY}^2=1$ for $\text{NR}=4$. The missing patterns in this table are generated as follows: (a) Missing uniformly at random: Entries are independently observed with probability $p$. (b) Block-missing pattern: $2(1-p)$ fraction of randomly selected units are missing from time $0.5\cdot T$. (c) Staggered treatment pattern: All units are in the control group for $t<c\cdot T$. Starting from time $t=c\cdot T,$ $(t/T-c)$ fraction of randomly selected units are in the treated group at time $t$. (d) Mixed-frequency observation: Entries in the first half of the units are simultaneously observed at every $t_1$ time period and entries in the second half of the units are simultaneously observed at every $t_2$ time period. For $p=60\%$, $t_1 = 1/35\%$ and $t_2=1/85\%$; for $p=75\%$, $t_1 = 1/60\%$ and $t_2=1/90\%$; for $p=90\%$, $t_1 = 1/80\%$ and $t_2=1$.}
	}}

\end{table}

This numerical example illustrates three points. First, the optimized $\gamma^\ast$ can substantially deviate from the naive concatenating weight 1, even when the dimensions of $X$ and $Y$ are the same. In particular, for complex missing patterns, the optimized target weight can deviate substantially from equally weighting the panels. Second, when observations are missing at random, the optimized $\gamma^\ast$ only depends on the noise NR, but not on the dimension ratio $\Nx/\Ny$ or fraction of observed entries $p$, confirming Proposition \ref{prop: efficiency example} in the illustration of the efficiency effect. However, for more complex observation patterns, the optimized $\gamma^\ast$ can depend on $\Nx/\Ny$, $p$, and other quantities related to the observation pattern. Specifically, the optimized $\gamma^\ast$ generally increases with $\Nx/\Ny$ and $1- p$, implying that when the number of observations (or effective sample size) on $Y$ is small compared to $X$, the optimized $\gamma^\ast$ increases to balance the relative contributions of the two panels. Third, the optimized $\gamma^\ast$ grows with a larger dependency in the missing pattern.  For the four observation patterns considered in Table \ref{tab: gamma}, the dependency between the entries of $W^Y$ is generally the highest for the block-missing pattern, followed by the staggered treatment pattern, then the mixed-frequency pattern, while missing-at-random has the lowest dependency. A higher dependency implies a smaller effective sample size in $Y$ and hence a larger $\gamma^\ast$.

One application of our asymptotic distribution theory is to test causal effects. The fundamental problem in causal inference is that we observe an outcome either for the control or the treated data, but not for both at the same time. The unknown counterfactual of what the treated observations could have been without treatment can be naturally modeled as a data imputation problem. The same arguments as in \cite{xiong2019large} for how to use the results for causal inference apply to target-PCA. We can test in an analogous way for point-wise treatment effects that can be heterogeneous and time-dependent under general adoption patterns where the units can be affected by unobserved factors. Importantly, by optimally leveraging auxiliary data, we allow for more general adoption patterns and more precise estimates of the counterfactual outcomes.

\subsection{Selection of Target Weight $\gamma$} \label{sec: selection of gamma}

As we have seen, selecting $\gamma$ appropriately is crucial for the consistent and efficient estimation of the latent factor model on $Y$. We suggest a two-stage approach for choosing $\gamma$ based on our inferential theory. 

In the first stage, we select $\gamma^{\text{first-stage}} = \Nx/\Ny$ to consistently estimate the latent factor model. Based on Theorem \ref{thm: consistency for loadings}, we can consistently estimate all the factors and loadings using target-PCA with $\gamma^{\text{first-stage}}$. In the second stage, we estimate the asymptotic variance of the common component $\Sigma_{C,ti}^{(\gamma)}$ using the estimated factor model from the first stage and the inferential theory in Theorem \ref{thm: asymptotic distribution} and Corollary \ref{thm: simplified factor model}. We then select the scaling constant to achieve efficiency by minimizing a linear combination of the estimated $\hat \Sigma_{C,ti}^{(\gamma)}$. If our goal is to estimate all common components as precisely as possible, then the objective function is to minimize $\sum_{t,i} \Sigma_{C,ti}^{(\gamma)}$. This is the objective that we use in our applications. If we want to impute the missing entries in $Y$ as precisely as possible, then an appropriate objective function would be to minimize 
$\sum_{t,i} (1- W_{ti}^Y) \cdot \Sigma_{C,ti}^{(\gamma)}$. 

Generally, cross-validation could be an alternative approach for selecting $\gamma$. The idea of a cross-validation approach is to mask some observed entries in $Y$ and select the value of $\gamma$ that can most precisely estimate these masked out-of-sample entries. The challenge lies in how to mask the observed entries, and an appropriate implementation is more complicated than what it might appear to be. In particular, some naive masking schemes, such as random masking, may not be appropriate if the actual missing pattern is more complex. Intuitively, the masking should replicate the missing pattern, in order to ensure that the imputation errors of the masked entries can (unbiasedly) estimate the imputation errors of the missing entries in $Y$ (if imputing missing entries in $Y$ is our main goal). However, estimating the propensity of missingness is challenging and can be sensitive to the specification of a propensity model. In addition, the observation pattern can depend on the latent factor model itself, which further complicates the estimation. Our selection criterion based on the inferential theory avoids these issues. 
		
		\section{Extensions} \label{sec: discussion}

\subsection{Auxiliary Data $X$ Sufficient to Estimate All Factors} \label{sec: auxilliary is sufficient}

So far, we have focused on the important setting where the auxiliary data $X$ is not sufficient to estimate all the factors for the target $Y$, which is assumed in Assumption \ref{assump: factor shift}. A simpler case is where $X$ already contains all the information to learn the factors in $Y$, that is, all the factors in $Y$ are strong in $X$. This is a special case of our more general analysis, where we can consistently estimate the factors by applying PCA to $X$ and we would only use target-PCA to achieve higher efficiency. Theorem \ref{thm: asymptotic distribution} applies with minor modifications. 

For this case, if $\Ny/\Nx\rightarrow 0$, then choosing $\gamma = r$ with any constant $r$ is equivalent to applying PCA on $X$, which has the convergence rate of $\Nx$. Here, choosing $\gamma = r\cdot \Nx/\Ny$ is sub-optimal, since it up-weights $Y$ and thus slows down the convergence rate of the estimated factors to $\Ny$.
If $\Ny/\Nx\rightarrow c$ with $c$ bounded away from 0, then it is beneficial to include the target panel $Y$ in the estimation of factors to improve the efficiency. Specifically, we can select the optimized $\gamma^*=r$ based on the asymptotic normality results similar to Theorem \ref{thm: asymptotic distribution}.

This setting is relevant 
when the target panel $Y$ has a low-frequency observation pattern with no available information in $Y$ for some time periods. If $X$ contains all the necessary information to estimate the latent factors in $Y$, then we can accurately impute the value of $Y$ in the periods with no observations and, hence, obtain an imputed target panel with higher frequency observations.

We recommend the choice of target weight $\gamma=r\cdot \Nx/\Ny$ with a positive constant $r$ from the main setting for robustness. Even in the case where all relevant factors are strong in $X$, selecting  $\gamma=r\cdot \Nx/\Ny$ ensures consistent estimation. Importantly, this rate also guarantees consistency, when not all factors can be estimated from $X$. 
In practice, we do not know if all factors for $Y$ are strong factors in $X$. In many applications, $X$ might not be selected in a targeted way, and hence it is likely that factors needed for $Y$ can be missing or weak in $X$. 
The selection procedure in Section \ref{sec: selection of gamma} provides a robust solution.

\subsection{Finite Cross-Sectional Dimension of Target Data}\label{subsec:finite-dimension-target}

So far, we have focused on the case where $\Ny \rightarrow \infty$. In some practical applications, $\Ny$ may be finite \citep{huang2021scaled}.
Our results can be extended to the case of finite $\Ny$ with minor technical modifications.\footnote{We need to modify the definition of $\Sigma_{\LamY,t}$ in Assumption \ref{assump: simplified obs pattern}.3  to $\Sigma_{\LamY,t} = \frac{1}{{\Ny}}\sum_{i=1}^{\Ny}\WY_{ti}(\LamY)_i (\LamY)_i^\top$, when $\Ny$ is finite. } Specifically, we consider the choice of $\gamma$ in two different settings depending on the nature of units in $Y$. 

In the first setting, the units in $Y$ are similar to the units in $X$, and the idiosyncratic noise level in $X$ and $Y$ are at the same scale, that is, $\+E[(\eY)_{ti}^2] / \+E[(\eX)_{ti}^2] = O(1)$. For this case, if all the factors in $Y$ can be identified by applying PCA to $X$, then selecting $\gamma = 1$ (i.e., PCA on $X$) is optimal. This is the degenerate case that asymptotically does not require $Y$ for target-PCA.

In the second setting, units in $Y$ are (weighted) averages of $M$ units, where $M$ is much larger than $\Ny$ and can be of the same order as $\Nx$. An important example is when $Y$ are the principal components from another panel. In this case, if $\+E[(\eY)_{ti}^2]/ \+E[(\eX)_{ti}^2]  = O(1/M)$, then we should choose $\gamma = r\cdot \Nx$, such that target-PCA can identify the factors in $Y$ that are either weak or nonexistent in $X$.\footnote{This case is the same as choosing $\gamma =  O(\Nx/\Ny) = O(\Nx)$ for the consistency effect, that is, our general rule for selecting the rate of $\gamma$ still applies.} 

The following two propositions formalize the above discussion about choosing $\gamma$ in each of the two different settings. 

\begin{proposition}  \label{prop: finite Y 1}
Suppose $\+E[(\eY)_{ti}^2] / \+E[(\eX)_{ti}^2] = O(1)$. If all the factors can be identified in $X$ and $\gamma=r$ for some positive constant $r$, then Theorem \ref{thm: consistency for loadings} holds with convergence rate $\delta_{\Nx,T} = \min(\Nx, T)$, 
Theorem \ref{thm: asymptotic distribution}.2 holds with convergence rate $\sqrt{\delta_{\Nx,T}}$, and the asymptotic variance of the factors is independent of $Y$.
\end{proposition}

\begin{proposition}   \label{prop: finite Y 2}
Suppose $\+E[(\eY)_{ti}^2] / \+E[(\eX)_{ti}^2] = O(1/M)$ for $M \rightarrow \infty$. If  $\gamma = r\cdot \Nx$ for some positive constant $r$, then Theorem \ref{thm: consistency for loadings} holds with convergence rate $\delta_{\Nx,T,M} = \min(\delta_{\Nx,T},M)$, and Theorem \ref{thm: asymptotic distribution}.2 holds with convergence rate $\sqrt{\delta_{\Nx,T,M}}$.
\end{proposition}

\subsection{Multiple Panels}\label{subsec:multiple-panels}

Target-PCA can be generalized to the setting with multiple auxiliary panels $X_1,\cdots, X_{m_x}$ or/and multiple target panels $Y_1,\cdots,Y_{m_y}$, where $m_x$ and $m_y$ are the numbers of auxiliary and target panels.

In the case of multiple auxiliary panels $X_1,\cdots, X_{m_x}$, we can combine all auxiliary panels and the target panel into one panel with weights $\eta_1, \cdots, \eta_{m_x}$ for the auxiliary panels
\[ Z^{(\eta_1, \cdots, \eta_{m_x})} = \begin{bmatrix} \eta_1 X_1 & \eta_2 X_2 & \cdots & \eta_{m_x} X_m & Y   \end{bmatrix} \]
and apply our proposed estimator in Section \ref{sec: estimator} to $Z^{(\eta_1, \cdots, \eta_{m_x})}$. If $m_x= 1$, then the problem collapses to the setup in the main setting, and selecting the target weight as $\gamma$ is the same as choosing the source weight as $\eta_1 = 1/\gamma$. Based on our previous discussion, we should choose $\eta_1 = r_1 \cdot \Ny/N_{x,1}$ for some positive constant $r_1$, where $N_{x,1}$ is the number of units in $X_1$. If there are multiple auxiliary panels, then the natural generalization is to choose $\eta_j = r_j \cdot \Ny/N_{x,j}$ for any $j \in \{1,\cdots, m_x\}$ and some positive constant $r_j$. Such a choice of $\eta_j$ accounts for the case where different auxiliary panels may have different sets of factors that are useful for identifying the factors in $Y$ (consistency effect), and may have different idiosyncratic noise levels (efficiency effect).

In the case of multiple target panels $Y_1,\cdots,Y_{m_x}$, we can use a sequential approach for estimating the factor model and imputing missing values in each target panel. Concretely, we first apply target-PCA to $X$ and $Y_1$, and impute the missing observations in $Y_1$. Second, we treat $X$ and the imputed panel $Y_1$ as two auxiliary panels, and combine them with $Y_2$ to estimate the factor model and impute missing values in $Y_2$. We repeat this procedure until the factor model is estimated and missing values are imputed for all the targets. Essentially, this sequential approach is conceptually the same as target-PCA with only one target panel, but with a more complicated notation. 

\subsection{Anchored Time Series} \label{sec: anchor}
So far, all information for the factor model and the data imputation was based on the cross-sectional dependency in the contemporaneous outcomes in $Y$ and $X$. However, in the case of persistent time series, the prior realizations in $Y$ can provide useful information for imputing missing entries. We propose an extension of target-PCA that takes advantage of prior realizations and the contemporaneous dependency structure. It can be interpreted as anchoring the imputed values around the prediction of a time-series model and correcting them with the innovations around the time-series model estimated from the contemporaneous auxiliary data.

We focus on the practically relevant case of low-frequency data in $Y$, which is combined with high-frequency supplementary data $X$. In this case, the contemporaneous low-frequency outcomes in $Y$ are not sufficient to impute the higher frequency missing values. As a concrete example, we refer to our empirical study, where we consider a panel $Y$ of annually observed macroeconomic time series, while $X$ is a panel of monthly observed supplementary data. Many of these macroeconomic time series are highly persistent, and hence their differenced time series fluctuate around their mean value. Hence, the prior realizations of these time series can serve as anchor points.  

A simple modification of target-PCA allows us to include prior values as anchor points, while the formal theoretical results continue to hold. Concretely, we replace missing entries in $Y$ by their most recent observed values, and then apply target-PCA as before. For example, for yearly observed $Y$, we only observe entries in January, 
while observations from February to December are completely missing.
In this case, we construct the anchored time series by filling in the missing entries from February to December with the observations in January of the corresponding year. 
{This is a valid approach, if the common component based on prior factor realizations is an unbiased estimator of the common component in the next period. This is the case under the assumption that the factors have constant means, while errors have zero means.}  
Under this assumption, the consistency result of Theorem \ref{thm: consistency for loadings} and the asymptotic distribution of Theorem \ref{thm: asymptotic distribution} continue to hold.  In particular, the optimal choice of $\gamma$ follows the same arguments as for the conventional target-PCA.

Intuitively, the target weight implies a weighted average between the estimate of the common component of $Y$ from the last period and the common component estimated on $X$ from the current period. Without the contemporaneous observations in $X$ or past values of $Y$, we cannot make any statements about the higher frequency observations in $Y$. The prior value of the common component of $Y$ can be a noisy estimate of the next period's value, and hence benefits from the cross-sectional contemporaneous information in $X$ to correct the variation and reduce the variance. The weight on the anchored $Y$ panel could be interpreted as a prior for the past observation. If we put no weight on $X$, we would simply use the prior common component to impute the missing values. In the other extreme case, where we put all weight on $X$, we only use the factor model in $X$ for imputation, but ignore the prior values. The optimal choice of $\gamma$ minimizes the variance of the weighted average of these two extreme estimators. 

The extension to more complex time-series models is beyond the scope of this paper, but the general logic still applies. The weight $\gamma$ would imply a weighted average of a time-series model forecast of the common component and the contemporaneous realization of the common component based on auxiliary data.
		
		\section{Simulation} \label{sec: simulation}

In simulations, we show the superior performance of our target-PCA method relative to benchmarks under a variety of settings. For comparison, we include the three natural benchmark methods that apply PCA either to $Y$, a simple concatenated panel of $X$ and $Y$, or $X$ and $Y$ separately and combine those factors. These are naive estimation methods for a target panel with auxiliary data. In more detail, we compare the following estimators:
\begin{enumerate}
\item {\bf T-PCA}: Target-PCA with optimized $\gamma^*$ selected as $\gamma= r \cdot N_y/N_x$ with $r$ minimizing $\sum_{t,i} \sum_{C,ti}^{(\gamma)}$.
\item {\bf  $\text{XP}_Y$}: PCA estimator of \cite{xiong2019large} applied only to $Y$ (special case of target-PCA with $\gamma=\infty$).
\item {\bf $\text{XP}_{Z^{(1)}}$}: PCA estimator of \cite{xiong2019large} applied to the concatenated panel $Z^{(1)}=[X,Y]$ (special case of target-PCA with $\gamma=1$).
\item {\bf SE-PCA}: Separately estimate factors from $X$ and $Y$ with the method of \cite{xiong2019large}, combine the two sets of factors and estimate loadings of the combined factors to impute missing values on $Y$.\footnote{Note that there is no simple way to determine the number of factors extracted from $X$ and $Y$ respectively. When the factor number is $k$ for other methods, we simply combine $k$ factors from $X$ and $k$ factors from $Y$ in SE-PCA. Note that SE-PCA uses $2k$ factors in total to estimate the common components and impute missing entries in $Y$. Therefore, SE-PCA is more likely to identify all the factors, but at the cost of an efficiency loss compared to other methods.} 
\end{enumerate}

We generate data from a two-factor model $Y_{ti} = F_t^\top (\LamY)_i+({\eY})_{ti}$ and $X_{ti} = F_t^\top (\LamX)_i+(\eX)_{ti}$, where $F_t \overset{\iid}{\sim} \Ncal(0,I_2), (\LamY)_i \overset{\iid}{\sim} \Ncal(0,I_2), (\LamX)_i \overset{\iid}{\sim} \Ncal(0,I_2), (\eX)_{ti}\overset{\iid}{\sim} \Ncal(0,\sigma_{\eX}^2)$ and $(\eY)_{ti} \overset{\iid}{\sim} \Ncal(0,\sigma_{\eY}^2).$ We consider three missing patterns for target $Y$: 
\vspace{-1mm}
\begin{enumerate} 
    \item {\it Missing-at-random:} Entries of $Y$ are missing uniformly at random.
    \item {\it Low-frequency observation:} Entries in $Y$ are observed at a lower frequency and only every second time-series observation is available.
    \item {\it Missingness depends on loadings:} Entries of $Y$ are missing conditional on a unit-specific characteristic $S_i = \mathbbm{1}(|(\LamY)_{i2}|>\text{threshold}).$ This means that units that are more exposed to the second factor are more likely to be missing in $Y$. In this case, we assume $(\LamX)_{i1}=0$, that is, factor 1 is not included in $X$.
\end{enumerate}
\vspace{-2mm}

The detailed description of observation patterns and data-generating processes is in Table \ref{tab: RMSE}.

Table \ref{tab: RMSE} compares the performance of the four methods in estimating the common components of the target $Y$. Specifically, it reports the relative mean squared error (relative MSE) of the estimated common components of the observed, missing, and all entries in $Y$, which is defined as
\[
\text{relative MSE}_{\mathcal{M}} = \frac{\sum_{(t,i)\in \mathcal{M}}(\tilde{C}_{ti}-C_{ti})^2}{\sum_{(t,i)\in \mathcal{M}}(C_{ti})^2},
\]
where $\mathcal{M}$ denotes the set of either observed, missing, or all entries in $Y$. The MSE for the observed values can be interpreted as an in-sample evaluation, while the MSE for the imputed values serves as an out-of-sample evaluation as it evaluates the model on data that was not used in the estimation.

\begin{table*}[t!]
\centering
\begin{threeparttable}
\tcaptab{Relative MSE for different estimators}
\label{tab: RMSE}
\begin{tabular}{cccccc}
\toprule
Observation Pattern & \ \ $\mathcal{M}$ \ \ &\ \ T-PCA \ \ &\ \ $\text{XP}_Y$ \ \ &\ \ $\text{XP}_{Z^{(1)}}$ \ \ &\ \ SE-PCA \ \ \\
\midrule
\multirow{3}{*}{ \includegraphics[width=0.1\textwidth, height=14.5mm]{fig/thm/random.png} } 
&obs  &{\bf0.184} &0.408 &0.224 &0.530\\
&miss &{\bf0.182} &0.414 &0.220 &0.564\\
&all  &{\bf0.183} &0.411 &0.222 &0.547\\
\midrule
\multirow{3}{*}{\ \includegraphics[width=0.1\textwidth, height=14.5mm]{fig/thm/low_fre_pattern.png} } 
&obs  &{\bf0.291} &- &0.846 &1.059\\
&miss &{\bf1.029} &- &1.119 &1.104\\
&all  &{\bf0.656} &- &0.979 &1.080\\
\midrule
\multirow{3}{*}{\includegraphics[width=0.1\textwidth, height=14.5mm]{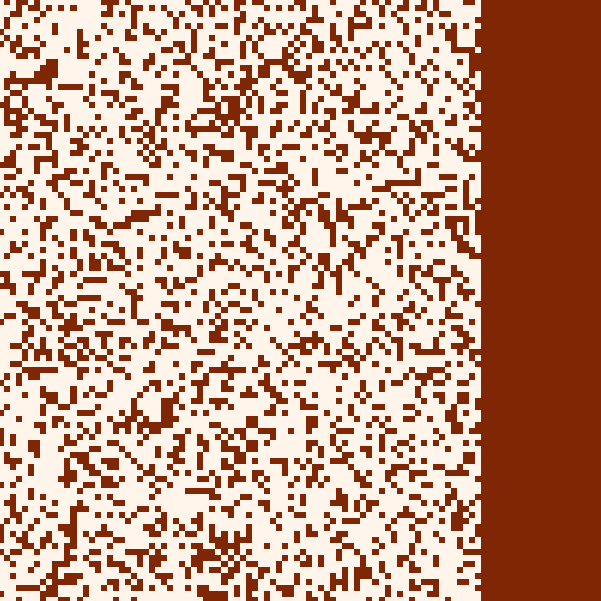}} 
&obs  &{\bf0.219} &0.238 &0.262 &0.280\\
&miss &{\bf0.252} &0.293 &0.287 &0.356\\
&all  &{\bf0.244} &0.280 &0.281 &0.338\\
\bottomrule
\end{tabular}
\floatfoot{This table reports the relative MSE of T-PCA (our benchmark method), $\text{XP}_Y$ (PCA on $Y$), $\text{XP}_{Z^{(1)}}$ (PCA on concatenated panel) and SE-PCA (separate PCA). The figures on the left show patterns of missing observations with each row representing the observation pattern for a  specific time period, and the shaded entries indicating the observed entries. Bold numbers indicate the best relative model performance. We generate a two-factor model and the observation patterns are generated as follows: (a) Missing uniformly at random:  $\sigma_{\eX}=1$, $\sigma_{\eY}=4,$ and entries of $Y$ are missing independently with observation probability $p=\PP(\WY_{ti}=1)=0.5.$  (b) Low-frequency observation: $\sigma_{\eX}=16, \sigma_{\eY}=4$, and entries in $Y$ are only observed every second time period. (c) Missingness depends on loadings: $(\LamX)_{i1}=0, \sigma_{\eY}=2, \sigma_{\eX}=4$ and define a unit-specific characteristic $S_i = \mathbbm{1}(|(\LamY)_{i,2}|>0.1).$ Entries are missing independently with observation probability $p=0.2$ if $S_i=1$, and $p=1$ if $S_i=0.$ We assume $\Nx=\Ny=T=200$ and run 200 simulations for each setup.}
\end{threeparttable}
\end{table*}

As shown in Table \ref{tab: RMSE}, target-PCA performs well under different observation patterns and dominates the other benchmarks. Our estimator has the smallest relative MSEs as compared to the three benchmark methods, whereas the three benchmark methods are either infeasible or inefficient in different settings. These results hold for the observed and imputed observations for all types of missing patterns.

In the case of missing-at-random, target-PCA has the smallest relative MSEs as compared to the three benchmark methods. This is the setting where all benchmark methods can identify all the factors in $Y$, but target-PCA is more efficient by appropriately using the information in the auxiliary panel. The gain is particularly large relative to using only $Y$ or estimating the factors separately from $X$ and $Y$. Both cases have a much smaller effective sample size than that of target-PCA, leading to more than double of the relative MSE.

In the setting of low-frequency observations, target-PCA continues to dominate the benchmark methods. This is a particularly interesting case, as the panel $Y$ is not sufficient for estimating the full factor model, and hence $\text{XP}_Y$ is not feasible. As estimating latent factors on $Y$ is not feasible in some periods, SE-PCA degenerates to PCA on $X$, and therefore the performance of SE-PCA solely depends on the auxiliary panel $X$. When $X$ has a low signal-to-noise ratio (i.e., $\sigma_{\eX}$ is large), SE-PCA can perform poorly. $\text{XP}_Z$ simultaneously uses the information in both $X$ and $Y$, and therefore performs the best among the three benchmark methods. Target-PCA further improves upon $\text{XP}_Z$ by efficiently weighting the two panels $X$ and $Y$.

Target-PCA also has the smallest relative MSE when missingness depends on the loadings. This setting can be viewed as endogenously missing data: 
The second factor has a weak signal on the observed entries of $Y$, but is important to model the missing data in $Y$.
This setting is similar to, but more complicated than, our toy example in Section \ref{subsec: consistency effect}. The auxiliary panel $X$ only contains the second factor, but not the first one; target $Y$ contains both of the two factors, but the second factor is relatively weak on the observed entries of $Y$ because the missing pattern depends on the factor loadings of $Y$. In this setting, $\text{XP}_Y$ performs worse than target-PCA because $\text{XP}_Y$ can hardly detect the second factor using the observed entries in $Y$. $\text{XP}_{Z^{(1)}}$ performs worse than target-PCA mainly because $\text{XP}_{Z^{(1)}}$ does not properly weight the two panels to account for their differences in the idiosyncratic noise levels. SE-PCA also performs worse than target-PCA because each separate estimation is noisier than our combined estimation. 

Our results are robust to modifying the parameters of the simulations. The Internet Appendix collects extensive robustness results, where we vary the noise variances and the fraction of observed entries for the models. Target-PCA continues to perform well and to dominate the benchmark methods. We conclude that target-PCA with its more comprehensive use of the target and auxiliary panels is better than the conventional benchmarks under various settings.

		\section{Empirical Results}\label{sec:empirical}
\subsection{Data} \label{sec: empirical data}
In our empirical study, we show the good performance of target-PCA for imputing missing values in popular macroeconomic panels. Our empirical analysis uses two standard data sets from the Federal Reserve Economic Data (FRED) of the St. Louis Fed. 

Our first macroeconomic panel is the FRED-MD macroeconomic database introduced by \cite{FRED_MD}.\footnote{We use the version of December 2021.} This dataset consists of 127 monthly macroeconomic variables which are classified into 8 groups: (a) output and income, (b) labor market, (c) housing, (d) consumption, orders, and inventories, (e) money and credit, (f) bond and exchange rates, (g) prices, and (h) stock market. Among them, we use the 120 time series that are fully observed over the time window from 01/1960 to 12/2020. We obtain stationary time series by applying the standard data transformation suggested by \cite{FRED_MD} to each time series and then normalize them to have zero mean and unit standard deviation. In Section \ref{sec: empirical 1}, we use this data set to evaluate the precision of different imputation methods. 

Our second macroeconomic panel contains 58 quarterly-observed macroeconomic time series from two categories, the national income \& product accounts category and the flow of funds category of the  FRED database. These 58 time series are represented as percentage changes relative to the prior year. Table IA.2 in the Internet Appendix provides a complete description of the data. We normalize the time series to have zero mean and unit standard deviation. In Section \ref{sec: empirical 2}, we combine this low-frequency panel with the higher-frequency panel of monthly-observed FRED-MD data. This is an example where the higher frequency data is not available and our imputation results provide a nowcasting time series of higher frequency.

\subsection{Comparison with Benchmark Methods} \label{sec: empirical 1}
In this section, we compare the imputation accuracy of target-PCA with the previously considered benchmark methods using the 120 monthly macroeconomic variables from the FRED-MD macroeconomic database. We take the 19 variables in the interest and exchange rates group as our target panel $Y$, and the remaining 101 variables in the other 7 groups as the auxiliary panel $X$. Table IA.1 in the Internet Appendix provides a list of the variables in $Y$. Having a good model to explain and impute interest and exchange rate time series is of economic interest.

We compare the out-of-sample imputation results for the target panel $Y$. We mask some entries in $Y$ as missing values with different missing patterns, and compare for various methods the imputation accuracy of the masked entries. Specifically, we consider the following four masking/missing patterns: 
\begin{enumerate} 
    \item {\it Missing-at-random:} The entries of $Y$ are missing uniformly at random.
    \item {\it Block-missing:} A subset of the macroeconomic variables in $Y$ is completely missing during some time periods.
    \item {\it Low-frequency observation:} All variables in $Y$ are observed at a lower than the desired frequency (annually instead of monthly).
    \item {\it Censoring:} The entries of $Y$ are missing if their values exceed a certain threshold.
\end{enumerate}

Table \ref{tab: empirical OOS RMSE} provides a detailed description of the masking mechanism and the percentage of observed values. This table compares the imputation accuracy of our target-PCA estimator with the three benchmark estimators $\text{XP}_Y$ (using only $Y$), $\text{XP}_{Z^{(1)}}$ (naive concatenation), and $\text{SE-PCA}$ (separate PCA). Target-PCA selects the optimized $\gamma^*$ based on our theory as $\gamma= r \cdot N_y/N_x$ with $r$ minimizing $\sum_{t,i} \sum_{C,ti}^{(\gamma)}$. We report the relative MSE of using estimated common components to impute masked entries in $Y$ for various estimators. The relative MSE in this case represents the out-of-sample imputation performance.

\begin{table}[t!] 
\centering
\begin{threeparttable}
\tcaptab{Out-of-sample relative MSE for different methods on FRED-MD}
\label{tab: empirical OOS RMSE}
\begin{tabular}{ccccccc}
\toprule
&  {\begin{tabular}{@{}c@{}}
Observation Pattern \\ (missing ratio)
\end{tabular}}  & {\begin{tabular}{@{}c@{}}
Factor \\ Number
\end{tabular}}  &\  T-PCA \  & \ $\text{XP}_Y$ \ & \ $\text{XP}_{Z^{(1)}}$  \ & \ SE-PCA  \ \\
\midrule
\multirow{5}{*}{ \includegraphics[width=0.13\textwidth, height=2.4cm]{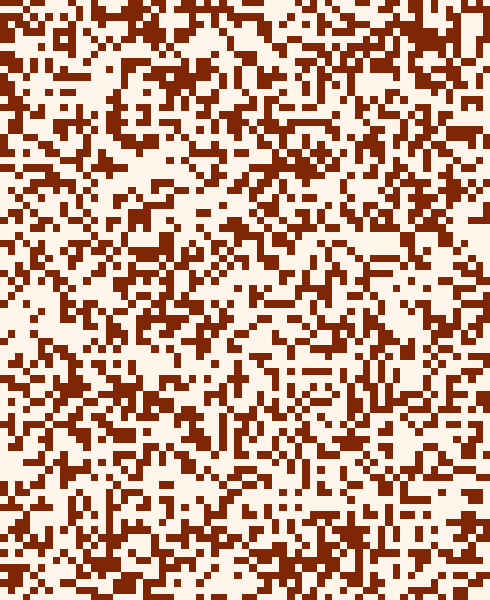}}
&\multirow{5}{*}{\begin{tabular}{@{}c@{}}
missing-at-random  \\  (40\%)
\end{tabular}}
&$k=1$  &\textbf{0.785}  &0.789 &0.986 &0.800 \\
& &$k=2$  &\textbf{0.488} &0.503 &0.968 &0.500 \\
& &$k=3$  &\textbf{0.485} &0.683 &0.926 &0.675 \\
& &$k=4$  &\textbf{0.491} &0.813 &0.797 &0.795 \\
& &$k=5$  &\textbf{0.483} &1.363 &0.615 &1.355 \\
\midrule
\multirow{5}{*}{ \includegraphics[width=0.13\linewidth, height=2.4cm]{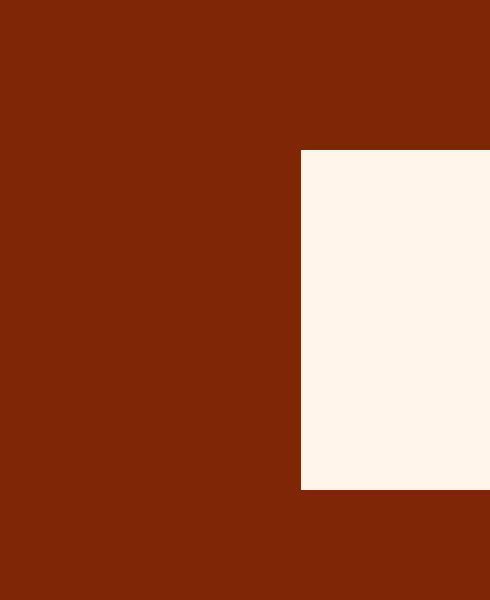}}
&\multirow{5}{*}{\begin{tabular}{@{}c@{}}
block-missing  \\   (19\%)
\end{tabular}} 
&$k=1$  &\textbf{0.958}  &1.018 &0.971 &1.003 \\
& &$k=2$  &\textbf{0.710} &0.805 &0.961 &0.852 \\
& &$k=3$  &\textbf{0.713} &0.796 &0.974 &0.803 \\
& &$k=4$  &\textbf{0.778} &0.783 &0.974 &0.781 \\
& &$k=5$  &\textbf{0.792} &2.601 &0.935 &2.584 \\
\midrule
\multirow{5}{*}{ \includegraphics[width=0.13\linewidth, height=2.4cm]{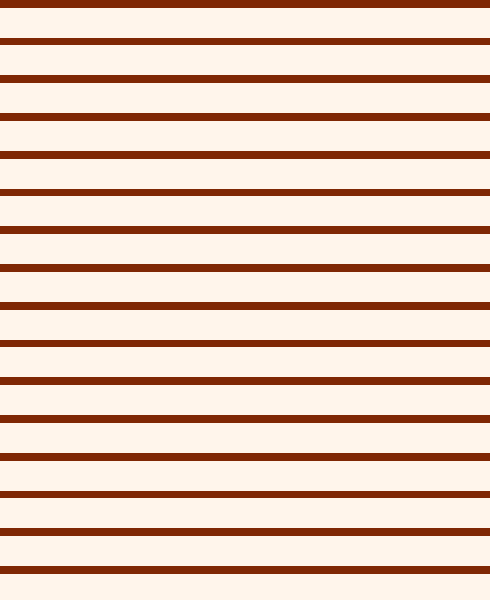}}
&\multirow{5}{*}{\begin{tabular}{@{}c@{}}
low-frequency \\   (92\%)
\end{tabular}} 
&$k=1$  &\textbf{0.942} &0.949 &1.019 &1.009 \\
& &$k=2$  &\textbf{0.927} &1.140 &0.931 &1.149 \\
& &$k=3$  &\textbf{0.926} &1.213 &0.936 &1.223 \\
& &$k=4$  &\textbf{0.910} &1.212 &1.095 &1.234 \\
& &$k=5$  &\textbf{1.017} &1.251 &1.092 &1.280 \\
\midrule
\multirow{5}{*}{ \includegraphics[width=0.13\linewidth, height=2.4cm]{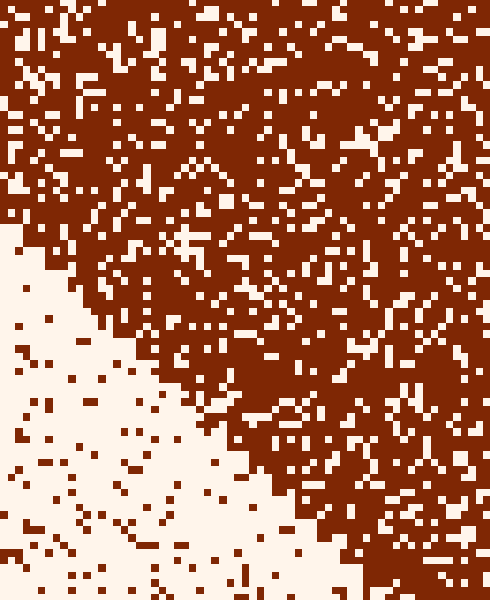}}
&\multirow{5}{*}{\begin{tabular}{@{}c@{}}
Censoring \\ (40\%)
\end{tabular}}
   &$k=1$  &\textbf{0.927}  &- &0.996 &0.995 \\
&  &$k=2$  &\textbf{0.881}  &- &0.996 &0.994 \\
&  &$k=3$  &\textbf{0.892}  &- &0.993 &0.992 \\
&  &$k=4$  &\textbf{0.885}  &- &0.990 &0.987 \\
& &$k=5$  &\textbf{0.869}   &- &0.984 &0.981 \\
\bottomrule
\end{tabular}
\end{threeparttable}
\floatfoot{This table reports the relative out-of-sample MSEs of the target panel $Y$ on the FRED-MD macroeconomic panel for three benchmark methods with different numbers of latent factors. The target panel $Y$ contains the 19 variables in the interest and exchange rates group, and the remaining 101 variables from $X$ from 01/1960 to 12/2020.
The masking/ missing patterns in this table are generated as follows. (a) Missing uniformly at random: We randomly mask each entry in $Y$ with probability $0.4$. We repeat this random masking 100 times and report the average relative MSE. (b) Block-missing: we mask the period from 01/1980 to 12/2009 for the 7 time series in $Y$ related to bond prices: TB3MS, TB6MS, GS1, GS5, GS10, AAA, and BAA. (c) Low-frequency observation: We mask the observations of $Y$ from February to December each year. Only for this case, the input target data for the latent factor model estimation is anchored at the most recent observed value as described in Section \ref{sec: anchor}. This means we augment $Y$ with the anchored time series that use the observation in January of the corresponding year for February to December for estimating the various latent factor models. 
(d) Censoring: We mask entries whose absolute value exceeds a threshold in $Y$. We set the threshold to 0.6, in order to mask approximately 40\% of the entries. 
The plots on the left column of this table illustrate these four missing patterns. The entries with dark color denote observed entries, and the entries with light color indicate that entries are missing. 
}
\end{table}

Target-PCA method dominates the benchmark methods for all missing patterns. As shown in Table \ref{tab: empirical OOS RMSE}, target-PCA achieves the smallest MSEs compared to other methods. In the case of missing-at-random and block-missing patterns, our target-PCA is more efficient since it can appropriately combine the information from the target and auxiliary panel. This is also supported by our theoretical results in Section \ref{sec: results}. In the low-frequency observation setting, target-PCA performs better by leveraging the information from the auxiliary panel. For the low-frequency masked data, we use the time series anchored at the most recent observed entries in $Y$ as described in Section \ref{sec: anchor}. In more detail, for this specific case, we fill the missing entries in $Y$ (February to December) with the observations in January of the corresponding year before applying target-PCA. 
Similarly, we also use these anchored time series for the $\text{XP}_Y$ and $\text{SE-PCA}$ estimators, which otherwise would not be applicable. The anchoring improves the performance for the low-frequency case, but the relative qualitative results are the same without anchoring. Our target-PCA estimator is also the best for the case of censoring, which could cause weak signals in the observed entries. In this case, $\text{XP}_Y$ is not applicable since there are time periods where all units have large absolute values and thus are missing. For this censored masking, the missing pattern can depend on the noise and factor realization, and hence violate our assumptions. Nevertheless, our method still performs very well and better than the benchmark methods.

Our target-PCA estimator is robust to the number of latent factors. Table \ref{tab: empirical OOS RMSE} compares the out-of-sample results for different numbers of factors. First, target-PCA dominates the benchmarks uniformly in the number of factors, that is, for the same number of factors it always performs better. Second, the results of target-PCA are very robust to the choice of $k$. Target-PCA with two factors is close to optimal and outperforms all other estimators even if they use more factors.
The benchmark methods are not stable in the number of factors. For a larger number of factors ($k=5$), $\text{XP}_Y$ and SE-PCA can perform substantially worse than the other two methods, especially for the missing-at-random and block-missing patterns. A plausible reason is that both $\text{XP}_Y$ and SE-PCA estimate the factors from regressing the observed $Y$ on estimated loadings, which can be very noisy for higher-order factors in $Y$ which are weaker. Both target-PCA and $\text{XP}_{Z^{(1)}}$ do not have this problem as factors are estimated from the regression on the combined panel of $X$ and $Y$. Target-PCA further improves upon $\text{XP}_{Z^{(1)}}$ by choosing an appropriate $\gamma$ for the consistency and efficiency effects.

\begin{figure}[t!]
    \tcapfig{Imputed time series of 6-month treasury yield for block-missing pattern}
        \centering
    \begin{tabular}{c}
         \includegraphics[width=0.805\textwidth]{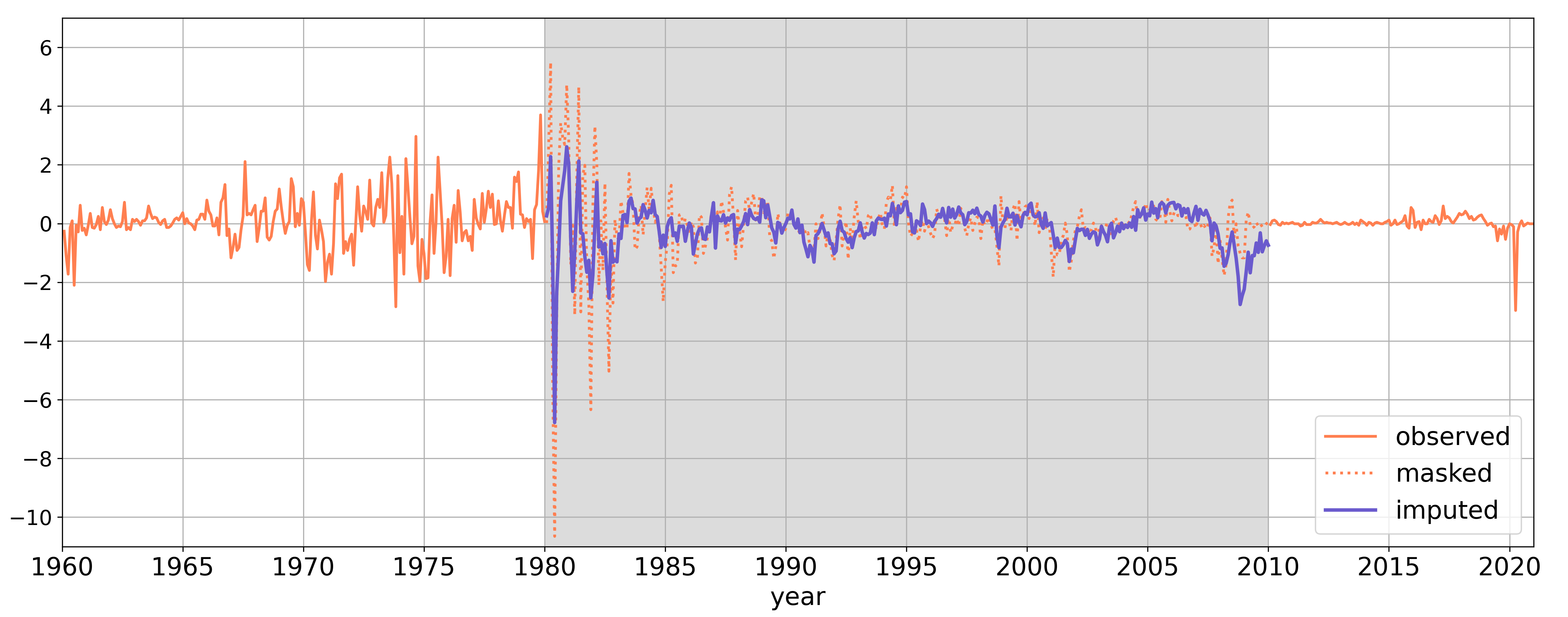}
    \end{tabular}
    \floatfoot{This figure shows the time series of the 6-month treasury yield for the block-missing pattern. The orange line shows the actual time series and the true masked values. The purple line denotes the imputed values with target-PCA with $k=3$ latent factors. The gray block indicates the missing period for the out-of-sample imputation. }
    \label{fig: time series}
\end{figure}

\begin{figure}[t!]
    \tcapfig{Imputed time series of rate spread for low-frequency missing pattern}
          \begin{subfigure}[b]{1.0\textwidth}
          \centering
\includegraphics[width=0.83\textwidth]{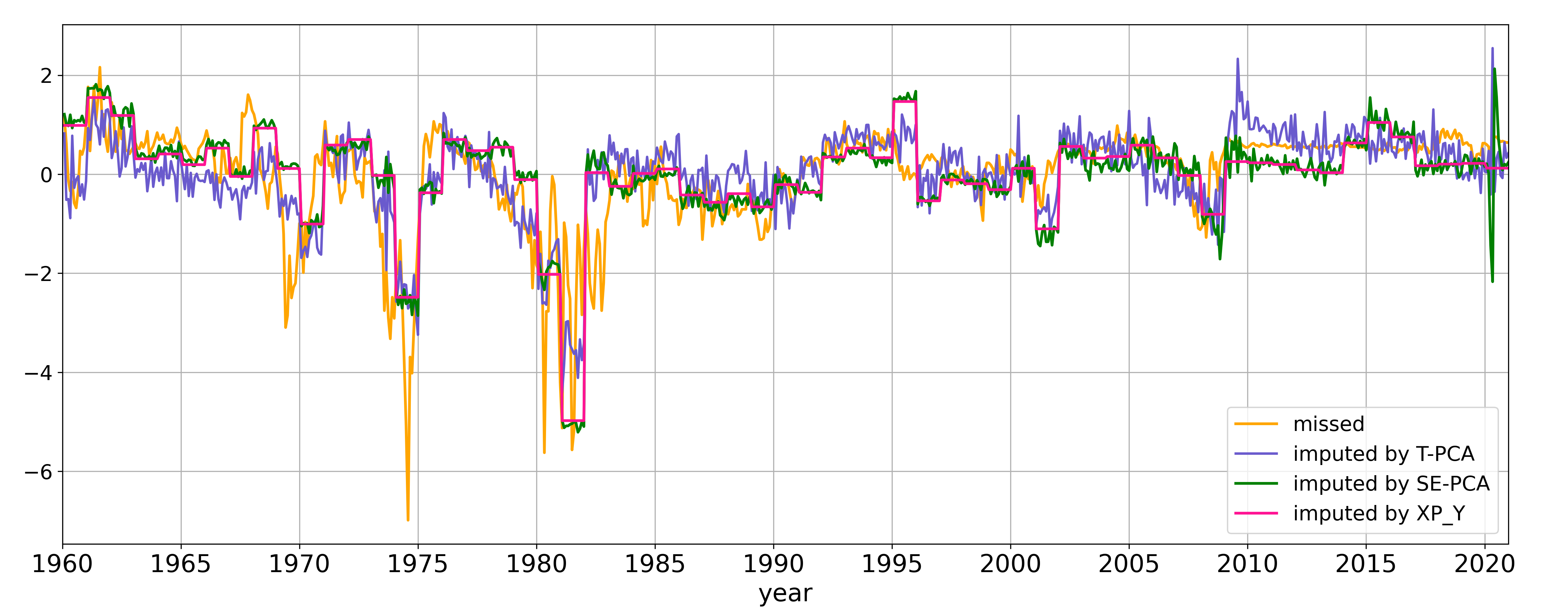}         
\caption{Full time period (01/1960 - 12/2020)}
    \end{subfigure}
          \begin{subfigure}[b]{1.0\textwidth}
          \centering
\includegraphics[width=0.83\textwidth]{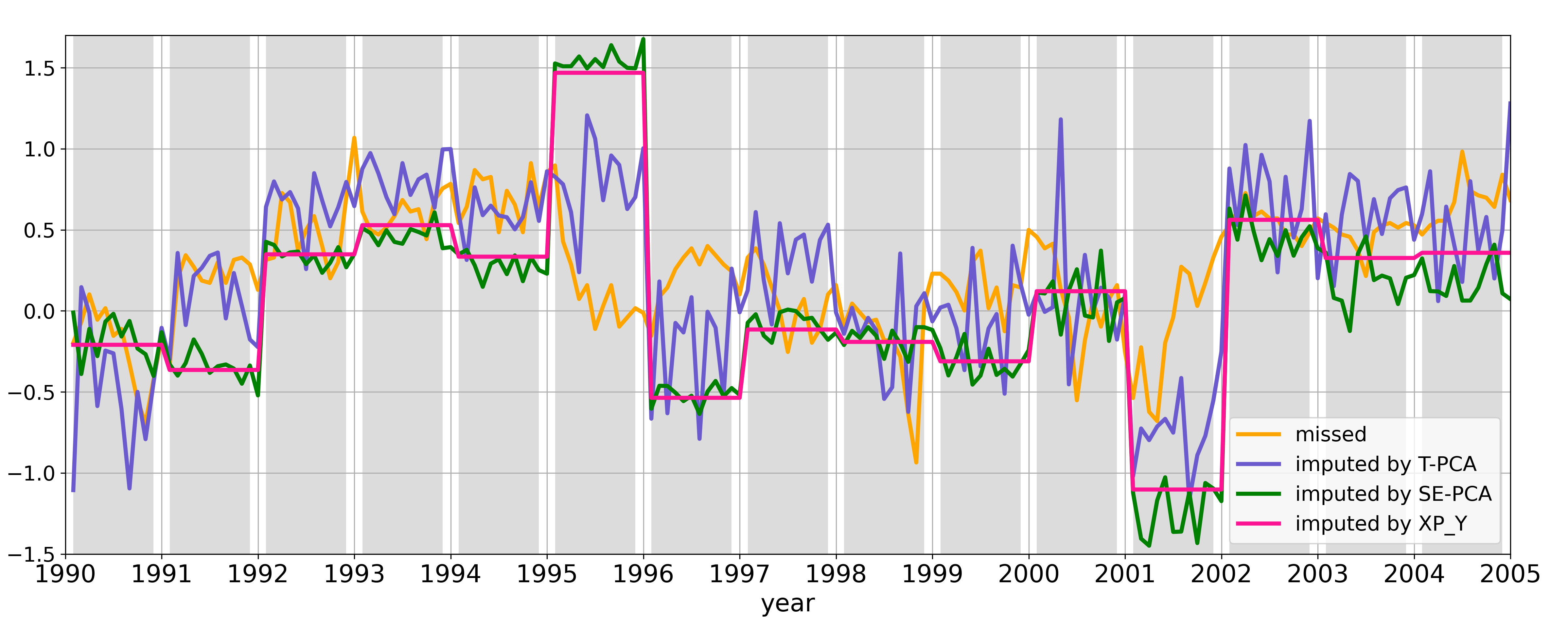}         
\caption{Zoom-in (01/1990 - 12/2005)}
    \end{subfigure}
    \vspace{-0.5cm}
    \floatfoot{This figure shows the time series of the spread between the 3-month treasury and Fed Funds rate for the low-frequency missing pattern. The orange line shows the actual time series and the true masked values. Observations are only observed in January and missing from February to December each year. The gray blocks in the zoom-in plot indicate the missing periods for the out-of-sample imputation with different methods. The number of factors equals $k=4$ for each imputation method.}
    \label{fig: time series_low_frequency}
\end{figure}

Target-PCA results in a meaningful time series of imputed values. Figures \ref{fig: time series} and \ref{fig: time series_low_frequency} illustrate the precise imputation of target-PCA for representative examples. Figure \ref{fig: time series} plots the imputed time series with target-PCA and the actual values of the 6-month treasury bill for the block-missing pattern. Our target-PCA imputation is quite accurate for the masked/missing time period as it captures well the fluctuation in the real time series. Figure \ref{fig: time series_low_frequency} plots the imputed time series with target-PCA and the actual time series of the spread between the 3-month treasury and Fed Funds rate for the low-frequency observation pattern. As a reference, we also include the time series imputed by the benchmark methods $\text{XP}_Y$ and $\text{SE-PCA}$. The lower panel zooms in to show the time period between 1990 and 2005. Target-PCA imputes the real time series reasonably well and visibly better than the benchmarks. In particular, a proper weight on $X$ seems to be essential to capture the contemporaneous fluctuations in the missing blocks.

\subsection{Nowcasting with Target-PCA} \label{sec: empirical 2}

An important practical problem is nowcasting macroeconomic panel data. We use target-PCA to impute unbalanced low-frequency macroeconomic panel data that is not available at a high frequency. The imputed high-frequency data represents a nowcasted panel, that is of interest to itself and also for downstream applications that require higher frequency data.\footnote{{A widely used approach for nowcasting is based on state-space models. Those models are complementary to our approach. State-space models are usually dynamic models for low dimensions that impose distributional assumptions. In contrast, our method is designed for large panels and focuses on leveraging cross-sectional information. For example, when a panel is of mixed frequency, our method can take advantage of high-frequency observations to ``nowcast'' the low-frequency variables. An interesting direction for future work is to extend our framework to dynamic factor models that optimally use auxiliary data.}}

In this section, we illustrate the time series of imputed nowcast data that are not available at a higher frequency. We impute the monthly values of only quarterly-observed time series of the FRED database. Our target panel consists of 58 macroeconomic time series of the national income and product accounts category and the flow of funds category. It includes several fundamental macroeconomic time series needed in economic research. To impute the monthly values of our target, we 
construct our auxiliary data by the 120 monthly macroeconomic indicators of FRED-MD from the first empirical study.

Target-PCA imputation performs well for the observed quarterly time periods and results in economically meaningful imputed values. Figure \ref{fig: time series in empirical 2} illustrates the imputation with target-PCA with the representative Gross Domestic Product (GDP) time series. We only observe quarterly-frequency GDP data, which is denoted by the orange ``+''. The monthly observed auxiliary panel $X$, allows us to use target-PCA to impute the monthly values of the GDP time series, which is plotted as the purple curve in Figure \ref{fig: time series in empirical 2}.\footnote{Our imputation can be used for economic stock and flow variables. In the case of a flow variable, the sum of monthly observations has to add up to lower frequencies like quarterly and annual values. For example, for GDP growth, the quarterly growth is the sum of the monthly growths within the quarter. Our imputed monthly values represent an appropriate moving average of the GDP growth; that is, the imputed values represent the growth over the previous three months. Obviously, the imputed values can always be transformed into monthly GDP growth. In summary, in our imputation process, we do not need to distinguish between flow and stock variables, as we are not directly imputing the monthly GDP growth. In practice, it is important to consider the interpretation of imputed values.}
Target-PCA captures the unknown and unobserved variation between two quarterly GDP observations using the monthly observed auxiliary data. 
Figure \ref{fig:example2} in the Appendix collects further examples, which all demonstrate the good performance of target-PCA for meaningful imputed values between the quarters. Our results use $k=5$ latent factors. As discussed in Section \ref{sec: empirical 1}, the target-PCA estimator is robust to the number of latent factors, and the results are very similar for different numbers of factors.\footnote{The factor model is estimated only once on the unbalanced panel. Applying target-PCA on an expanding window to avoid using future data gives essentially identical results. The results are available upon request.}

\begin{figure}[t!]
    \tcapfig{Quarterly observed GDP vs. monthly imputed GDP with target-PCA}
    \centering
        \includegraphics[width=0.9\textwidth]{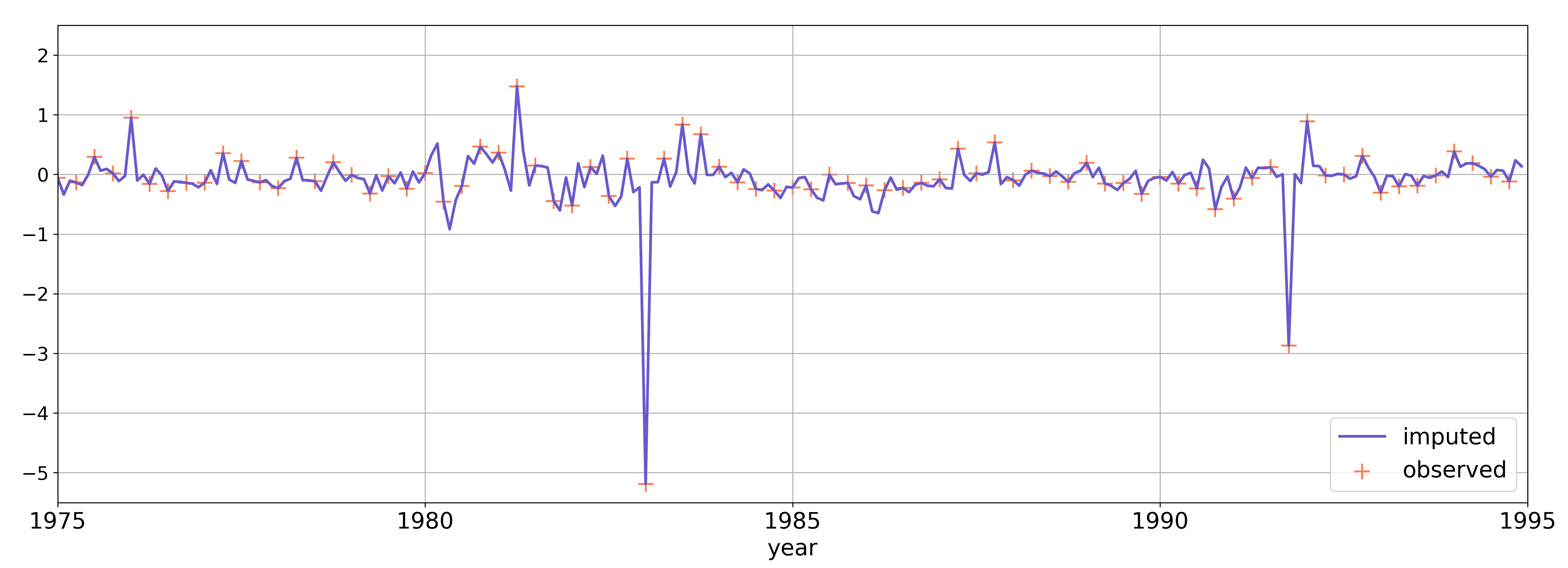}
    \floatfoot{This figure shows the time series of observed GDP and monthly imputed GDP with target-PCA. The orange ``+''  denotes the quarterly observed values of the GDP time series, and the purple curve denotes the imputed values of the GDP time series with target-PCA for $k=5$ latent factors. We only impute for the months when the observation is missing and take the actual values when the data is observed. The time series represent percentage changes relative to the prior year.
    }
    \label{fig: time series in empirical 2}
\end{figure}
		
		\section{Conclusion} \label{sec: conclusion}
This paper proposes our novel method target-PCA to estimate a latent factor model for target data by optimally weighting and combining the information from auxiliary data that is relevant to our target. This method is broadly applicable but easy to implement: It applies principal component analysis to a weighted average of the covariance matrices of the target and auxiliary panels. Target-PCA is particularly beneficial when the target data have missing observations, which could handle some of the scenarios that conventional methods cannot solve. A leading example is the mixed-frequency observation pattern, where conventional PCA cannot estimate the factors in times without any observations.

To optimally combine the auxiliary data, we need to overcome the differences in dimensionality and noise ratio between the two panels. Target-PCA tackles these problems by introducing a target weight $\gamma$ and combining information from the auxiliary data with the weighted target data. We show two essential effects of this target weight for target-PCA: the consistency and efficiency effects. First, by selecting the target weight at the right rate, we can ensure the consistent estimation of all factors in the target panel, including weak factors. Second, by selecting the scaling of the target weight to account for the noise ratios between the panels, we improve the efficiency of the estimated factor model.

We develop the inferential theory for the estimated factors, loadings, and imputed values of target-PCA under very general assumptions on the approximate factor model and missing patterns. The asymptotic results are used to construct confidence intervals and provide guidance on choosing the optimal target weight for target-PCA. In an empirical analysis, we illustrate the benefit of our approach with the imputation of unbalanced macroeconomic panel data.

 \end{onehalfspacing}

\singlespacing
\bibliographystyle{econometrica}
{\small
\bibliography{reference}
}
\onehalfspacing

\appendix

\renewcommand{\thesubsection}{\Alph{section}.\arabic{subsection}}
\setcounter{table}{0}
\setcounter{figure}{0}
\renewcommand{\thetable}{A.\arabic{table}}
\renewcommand{\thefigure}{A.\arabic{figure}}
	    	    

%
%

\section{General Assumptions}\label{subsec: general assumptions}

\paragraph{Notation.} Let $C<\infty$ denote a generic constant. Let $\norm{v} $ denote the vector norm and $\norm{A} = trace(A^\top A)^{1/2}$ the Frobenius norm of matrix $A$.

\begin{assumpG}[Factor model]  \label{assump: factor model}
\texttt{}
 \begin{enumerate}
 \item  Factors: $\forall t$, $\E\|F_t\|^4 \leq C$. There exists a positive definite $k\times k$ matrix $\Sigma_F$ such that $\frac{1}{T}\sum_{t=1}^T F_t F_t^\top \overset{p}{\rightarrow} \Sigma_F$ and $\E \left\|\sqrt{T}\left(\frac{1}{T}\sum_{t=1}^T F_t F_t^\top -\Sigma_F\right)\right\|^2 \leq C$. Furthermore, for any $\QY_{ij}$, $\frac{1}{|Q^Z_{ij}|} \sum_{t \in Q_{ij}^Z} F_t F_t^\top \overset{p}{\rightarrow} \Sigma_F$ and $\E \left\|\sqrt{|Q^Z_{ij}|}\left(\frac{1}{|Q^Z_{ij}|}\sum_{t \in Q^Z_{ij}} F_t F_t^\top -\Sigma_F\right)\right\|^2 \leq C$.

 \item  Loadings: loadings are independent of factors and errors, and $\LamX$ is independent with $\LamY$. $\E\|({\LamX})_i\|^4 \leq C$, $\E\|({\LamY})_i\|^4 \leq C,$ $\frac{1}{\Nx}\sum_{i=1}^{\Nx}({\LamX})_i({\LamX})_i^\top \overset{p}{\rightarrow}\Sigma_{{\LamX}}$, $\frac{1}{\Ny}\sum_{i=1}^{\Ny}({\LamY})_i({\LamY})_i^\top \overset{p}{\rightarrow}\Sigma_{{\LamY}}$, and $\Sigma_{{\LamX}}+ \Sigma_{\LamY}$ is a positive definite matrix. Furthermore, for any $t$, $\frac{1}{\Ny}\sum_{i=1}^{\Ny} \WY_{ti}({\LamY})_i ({\LamY})_i^\top \overset{p}{\rightarrow} \Sigma_{{\LamY},t},$ and $\Sigma_{\LamX} +  \Sigma_{{\LamY},t}$ is a positive definite matrix. 
\item  Idiosyncratic errors: 
\begin{enumerate}
\item $\E(\eX)_{ti}=\E({\eY})_{ti}=0, \E(\eX)_{ti}^8 \leq C$ and $\E({\eY})_{ti}^8 \leq C$.
\item $\E\Ls(\eX)_{ti}(\eX)_{si}\Rs = \gamma^{(\eX)}_{st,i}$ with $|\gamma_{st,i}^{(\eX)}| \leq \gamma_{st}$ and $\E\Ls({\eY})_{ti}({\eY})_{si}\Rs = \gamma^{({\eY})}_{st,i}$ with $|\gamma_{st,i}^{({\eY})}| \leq \gamma_{st}$ for some $\gamma_{st}$ and all $i$, $\sum_{s=1}^T \gamma_{st} \leq C$ for all $t$. Furthermore, $\gamma_{tt,i}^{(\eX)}/\gamma_{tt,i}^{(\eY)}$ is bounded away from zero for all $t,i$.
\item $\E\Ls(\eX)_{ti}(\eX)_{tj}\Rs = \tau_{ij,t}^{(\eX)}$ with $|\tau_{ij,t}^{(\eX)}| \leq \tau_{ij}^{(\eX)}$ for some $\tau_{ij}^{(\eX)}$ and all $t$, $\sum_{j=1}^{\Nx} \tau_{ij}^{(\eX)} \leq C$ for all $i$; $\E\Ls({\eY})_{ti}({\eY})_{tj}\Rs = \tau_{ij,t}^{({\eY})}$ with $|\tau_{ij,t}^{({\eY})}| \leq \tau_{ij}^{({\eY})}$ for some $\tau_{ij}^{({\eY})}$ and all $t$, $\sum_{j=1}^{\Ny} \tau_{ij}^{({\eY})} \leq C$ for all $i$; $\E\Ls({\eY})_{ti}(\eX)_{tj}\Rs = \tau_{ij,t}^{({\eY},\eX)}$ with $|\tau_{ij,t}^{({\eY},\eX)}| \leq \tau_{ij}^{({\eY},\eX)}$ for some $\tau_{ij}^{({\eY},\eX)}$ and all $t$, $\sum_{i=1}^{\Ny} \tau_{ij}^{({\eY},\eX)} \leq C$ for all $j$, and $\sum_{j=1}^{\Nx} \tau_{ij}^{({\eY},\eX)} \leq C$ for all $i$.
\item $\E\Ls(\eX)_{ti}(\eX)_{sj}\Rs = \tau_{ij,ts}^{(\eX)},$ and $\sum_{j=1}^{\Nx} \sum_{s=1}^T|\tau_{ij,ts}^{(\eX)}|\leq C$ for all $i$ and $t;$ $\E\Ls({\eY})_{ti}({\eY})_{sj}\Rs = \tau_{ij,ts}^{({\eY})},$ and $\sum_{j=1}^{\Ny} \sum_{s=1}^T|\tau_{ij,ts}^{({\eY})}|\leq C$ for all $i$ and $t$; $\E\Ls(\eX)_{ti}({\eY})_{sj}\Rs = \tau_{ij,ts}^{({\eY},\eX)},$ for all $i$ and $t$, $\sum_{i=1}^{\Nx} \sum_{s=1}^T|\tau_{ij,ts}^{({\eY},\eX)}|\leq C$.
\item $\E \Big [ \frac{1}{|\QY_{ij}|^{1/2}}\sum_{t \in \QY_{ij}}\Lp({\eY})_{ti}({\eY})_{tj} - \E\Ls({\eY})_{ti} ({\eY})_{tj}\Rs\Rp \Big]^4 \leq C$, $\E \Big [ \frac{1}{T^{1/2}}\sum_{t =1}^T \big((\eX)_{ti}(\eX)_{tj}-\E[(\eX)_{ti} \\ (\eX)_{tj}]\big) \Big]^4 \leq C$ and $\E \Big [ \frac{1}{|\QY_{jj}|^{1/2}}\sum_{t \in \QY_{jj}}\Lp(\eX)_{ti}({\eY})_{tj} - \E\Ls(\eX)_{ti}({\eY})_{tj}\Rs\Rp \Big]^4 \leq C$.
\end{enumerate}
\item  Weak dependence between factor and idiosyncratic errors: $\E \left\|\frac{1}{\sqrt{|Q^Y_{ij}|}}\sum_{t \in Q^Y_{ij}}F_t ({\eY})_{tj} \right\|^2 \leq C$ for any $i,j=1,\cdots,\Ny$. $\E \left\|\frac{1}{\sqrt{|Q^Y_{ii}|}}\sum_{t \in Q^Y_{ii}}F_t (\eX)_{tj'} \right\|^2 \leq C$ and $\E \left\|\frac{1}{\sqrt{T}}\sum_{t =1}^T F_t (\eX)_{tj'} \right\|^2 \leq C$ for any $j'=1,\cdots,\Nx$.
  \end{enumerate}
\end{assumpG}

\begin{assumpG}[Moment conditions and central limit theorems] \label{assump: additional assumptions} \texttt{}
\begin{enumerate}
\item For any $i=1,\cdots,\Ny$, there is $\E\Big\|\sqrt{\frac{T}{\Nx}}\sum_{j=1}^{\Nx} ({\LamX})_j \frac{1}{|\QY_{ii}|}\sum_{s\in \QY_{ii}}F_s^\top (\eX)_{sj}\Big\|^2\leq C$,
and \\ $\E\Big\|\sqrt{\frac{T}{\Ny}}\sum_{j=1}^{\Ny} ({\LamY})_j \frac{1}{|\QY_{ij}|}\sum_{s\in \QY_{ij}}F_s^\top ({\eY})_{sj}\Big\|^2\leq C$.
 \item  $\E \Big\|\frac{\sqrt{\Nx T}}{\Nx^2}\sum_{i,j=1}^{\Nx }({\LamX})_j({\LamX})_j^\top \frac{1}{T}\sum_{s=1}^T F_s({\LamX})_i^\top (\eX)_{si} \Big\|^2\leq C$,\\
 $\E \Big\|\frac{\sqrt{\Ny T}}{\Ny^2}\sum_{i,j=1}^{\Ny}({\LamY})_j({\LamY})_j^\top \frac{1}{{|\QY_{ij}|}}\sum_{s\in \QY_{ij}} F_s({\LamY})_i^\top ({\eY})_{si} \Big\|^2\leq C$, \\
 $\E \Big\|\frac{\sqrt{\Nx T}}{\Ny\Nx }\sum_{i=1}^{\Nx }\sum_{j=1}^{\Ny}({\LamY})_j({\LamY})_j^\top \frac{1}{{|\QY_{jj}|}}\sum_{s\in \QY_{jj}} F_s({\LamX})_i^\top (\eX)_{si} \Big\|^2\leq C$,\\
 $\E \Big\|\frac{\sqrt{\Ny T}}{\Ny\Nx }\sum_{i=1}^{\Ny}\sum_{j=1}^{\Nx }({\LamX})_j({\LamX})_j^\top \frac{1}{{|\QY_{ii}|}}\sum_{s\in \QY_{ii}} F_s({\LamY})_i^\top ({\eY})_{si} \Big\|^2\leq C$.
 \item For any $i=1,\cdots,\Ny,$ $\E\Big\| \sqrt{\frac{T}{\Nx }}
 \sum_{j=1}^{\Nx}  ({\LamX})_j \frac{1}{|\QY_{ii}|}\sum_{s\in \QY_{ii}}\Lp(\eX)_{sj}({\eY})_{si}-\E\Ls(\eX)_{sj}({\eY})_{si}\Rs\Rp \Big\|^2\leq C,$ and $\E\Big\| \sqrt{\frac{T}{\Ny}}
 \sum_{j=1}^{\Ny} ({\LamY})_j \frac{1}{|\QY_{ij}|}\sum_{s\in \QY_{ij}}\Lp({\eY})_{si}({\eY})_{sj}-\E\Ls({\eY})_{si}({\eY})_{sj}\Rs\Rp \Big\|^2\leq C$.
 \item  For any $t=1,\cdots, T$, $\E \Big\|\sqrt{\frac{T}{\Nx^3}}\sum_{i,j=1}^{\Nx }\frac{1}{T}\sum_{s=1}^T\phi_{ij,st} \Lp(\eX)_{si} (\eX)_{sj} - \E\Ls(\eX)_{si}(\eX)_{sj}\Rs\Rp \Big\|^2\leq C$ 
 with $\phi_{ij,st} = ({\LamX})_j({\LamX})_i^\top$, $({\LamX})_j({\LamX})_i^\top F_s$, $({\LamX})_j(\eX)_{ti}$, $\E \Big\|\sqrt{\frac{T}{\Ny^3}}\sum_{i,j=1}^{\Ny}\frac{1}{|\QY_{ij}|}\sum_{s\in \QY_{ij}}\phi_{ij,st} (({\eY})_{si} ({\eY})_{sj} - \E\Ls({\eY})_{si}({\eY})_{sj}\Rs) \Big\|^2 \leq C$ with $\phi_{ij,st} = ({\LamY})_j({\LamY})_i^\top$, $\WY_{ti}({\LamY})_j({\LamY})_i^\top F_s$, $\WY_{ti}({\LamY})_j({\eY})_{ti}$, and moreover, $\E \Big\|\sqrt{\frac{T}{\Nx^2\Ny}}\sum_{i=1}^{\Nx} \sum_{j=1}^{\Ny}\frac{1}{|\QY_{jj}|}\sum_{s\in \QY_{jj}}\phi_{ij,st} \Lp(\eX)_{si} ({\eY})_{sj} - \E\Ls(\eX)_{si}({\eY})_{sj}\Rs\Rp \Big\|^2\leq C$ with $\phi_{ij,st} = ({\LamY})_j({\LamX})_i^\top$, $({\LamY})_j({\LamX})_i^\top F_s,$ $\WY_{tj}({\LamX})_i({\LamY})_j^\top F_s$, $({\LamY})_j(\eX)_{ti}$, $\WY_{tj}({\LamX})_i({\eY})_{tj}$.
 \item For any $t$ and $j=1,\cdots,\Ny,$ $\E\Big\|\sqrt{\frac{{T}}{\Nx}}\sum_{i=1}^{\Nx} \big(\frac{1}{|\QY_{jj}|}\sum_{s\in \QY_{jj}}F_sF_s^\top - \frac{1}{T}\sum_{s=1}^TF_s F_s^\top\big) ({\LamX})_i (\eX)_{ti}\Big\|^4\leq C$,
and $\E\Big\|\sqrt{\frac{{T}}{\Ny}} \sum_{i=1}^{\Ny} \big(\frac{1}{|(Q_{ij}^*|}\sum_{s\in Q_{ij}^*}F_sF_s^\top - \frac{1}{T}\sum_{s=1}^TF_s F_s^\top\big)\WY_{ti} ({\LamY})_i ({\eY})_{ti}\Big\|^4\leq C$ with $Q_{ij}^*=\QY_{ii}$ or $\QY_{ij}.$
 \item $\frac{\Nx}{\Nx+\Ny}\Big( \frac{\sqrt{T}}{\Nx}\sum_{j=1}^{\Nx} ({\LamX})_j({\LamX})_j^\top \frac{1}{|\QY_{ii}|}\sum_{t\in \QY_{ii}}F_t({\eY})_{ti} +\gamma\cdot \frac{\sqrt{T}}{\Nx}\sum_{j=1}^{\Ny} ({\LamY})_j({\LamY})_j^\top \frac{1}{|\QY_{ij}|}\sum_{t\in \QY_{ij}}F_t ({\eY})_{ti}\Big)$  $\overset{d}{\rightarrow} \Ncal(0,\Gamma_{{\LamY},i}^{(\gamma),{\rm obs}})$ for any $\gamma = r\cdot \Nx/\Ny$ and $i=1,\cdots,\Ny.$
 \item  For any $t$, $\frac{1}{\sqrt{\Nx}}\sum_{i=1}^{\Nx} ({\LamX})_i (\eX)_{ti} \overset{d}{\rightarrow}\Ncal(0,\Sigma_{{\LamX} \eX,t})$, $\frac{1}{\sqrt{\Ny}}\sum_{i=1}^{\Ny} \WY_{ti}({\LamY})_i ({\eY})_{ti} \overset{d}{\rightarrow}\Ncal(0,\Sigma_{{\LamY}{\eY},t})$. For any $\gamma = r\cdot \Nx/\Ny$, $\Gamma_{F,t}^{(\gamma),{\rm obs}}:=\lim\frac{\Nx \Ny}{(\Nx+\Ny)^2}\Sigma_{{\LamX} \eX,t} + (\gamma \frac{\Ny}{\Nx+\Ny})^2 \Sigma_{{\LamY}{\eY},t}$. If there is some weak factor $F_{w}$ in $Y$ whose loading $\sum_{i=1}^{\Ny}(\LamY)^2_{i,w}$ grows at the rate $g(\Ny)$, then let $N_w= \min(\Ny^2/g(\Ny),\Nx)$. For $F_w$, $\frac{1}{\sqrt{g(\Ny)}} \sum_{i=1}^{\Ny} W^Y_{ti} (\Lam_y)_{i,w} (\eY)_{ti}  \overset{d}{\rightarrow}\Ncal(0,\Sigma_{{\Lam_{y,w}}{\eY},t})$, and for any $\gamma = r\cdot \Nx/\Ny$, $\Gamma_{F_w,t}^{(\gamma),{\rm obs}}:=\lim\frac{\Nx N_w}{(\Nx+\Ny)^2}\Sigma_{{\LamX} \eX,t} +\gamma^2 \frac{{ g(\Ny)N_w}}{(\Nx+\Ny)^2} \Sigma_{{\Lam_{y,w}}{\eY},t}$.
 \item  Define the filtration $\mathcal{G}^t = \sigma(\cup_{s=1}^T \mathcal{G}^t_{s})$ with $\mathcal{G}^t_{s} = \sigma(\{\WY_{ij},i \leq s,\text{all}\ j\},\LamY, v_t^{(\gamma)})$ generated by $\{\WY_{ij},i \leq s,\text{all}\ j\}, \LamY$ and $v_t^{(\gamma)} = (\Sigma_\Lam^{(\gamma)})^{-1}\Sigma_F^{-1}F_t$. For every $i$, $t$, and $u_i = (\LamY)_i,$ it holds
\[
\sqrt{T} \begin{bmatrix}
X_i^{(\gamma)}u_i \\ \mathbf{X}_t^{(\gamma)}v_t^{(\gamma)}
\end{bmatrix} 
\overset{d}{\rightarrow} \mathcal{N} \left( 0, \begin{bmatrix} h_i^{(\gamma)}(u_i) & g_{i,t}^{(\gamma),\text{cov}}(u_i,v_t^{(\gamma)})^\top \\ g_{i,t}^{(\gamma),\text{cov}}(u_i,v_t^{(\gamma)}) & g_t^{(\gamma)}(v_t^{(\gamma)})
\end{bmatrix} \right)\quad \mathcal{G}^t - \text{stably}.
\]
When $i=1,\cdots,\Nx$, $X_i^{(\gamma)} = \gamma\cdot \frac{1}{\Nx+\Ny}\sum_{j=1}^{\Ny} ({\LamY})_j({\LamY})_j^\top \big(\frac{1}{|\QY_{jj}|}\sum_{t\in \QY_{jj}}F_t F_t^\top - \frac{1}{T}\sum_{t=1}^TF_t F_t^\top\big)$, and when $i=\Nx+1,\cdots,\Nx+\Ny$, $X_i^{(\gamma)} =  \frac{1}{\Nx+\Ny}\sum_{j=1}^{\Nx} ({\LamX})_j({\LamX})_j^\top \big(\frac{1}{|\QY_{i'i'}|}\sum_{t\in \QY_{i'i'}}F_t F_t^\top - \frac{1}{T}\sum_{t=1}^TF_t F_t^\top\big) + \gamma\cdot \frac{1}{\Nx+\Ny}\sum_{j=1}^{\Ny} ({\LamY})_j({\LamY})_j^\top \big(\frac{1}{|\QY_{i'j}|}\sum_{t\in \QY_{i'j}}F_t F_t^\top - \frac{1}{T}\sum_{t=1}^TF_t F_t^\top\big)$ with $i'=i-\Nx$. For any $t$, $\mathbf{X}_t^{(\gamma)} = \frac{1}{\Nx+\Ny}\sum_{i=1}^{\Nx}X^{(\gamma)}_i({\LamX})_i({\LamX})_i^\top + \gamma\cdot \frac{1}{\Nx+\Ny}\sum_{i=1}^{\Ny} X^{(\gamma)}_{i+\Nx} \WY_{ti}({\LamY})_i ({\LamY})_i^\top$. 
\end{enumerate}
\end{assumpG}

\vspace{3mm}

\onehalfspacing

The assumptions on the factor structure in $Y$ are on the same level of generality as in \cite{xiong2019large} with the generalization that factors in $Y$ can be weak. The approximate factor model in $X$ is also on a similar level of generality as in \cite{bai2003inferential} with the generalization that not all factors in $Y$ are included in $X$. The assumptions essentially assume that for a properly selected $\gamma$ the combined weighted panel $Z^{(\gamma)}$ satisfies the assumptions in \cite{xiong2019large}.

In more detail, Assumption \ref{assump: factor model} describes an approximate factor structure and is at a similar level of generality as \cite{bai2003inferential}: 
	(1) Assumption \ref{assump: factor model}.1 states that each factor has a nontrivial variation for the observed time periods.
	(2) We assume loadings are random but independent of factors and errors in Assumption \ref{assump: factor model}.2. This is a standard assumption. 
  Assumption \ref{assump: factor model}.2 ensures that the loadings are systematic for both the full observed weighted matrix $Z^{(\gamma)}$ and the partially observed $\Tilde{Z}^{(\gamma)}$ with our proposed target weight $\gamma = r\cdot \Nx/\Ny$. This assumption is needed to show the consistency of loadings and factors. This assumption deviates from the usual factor model assumptions as neither $\Sigma_{{\LamX}}$ nor $\Sigma_{\LamY}$ has to be full rank. 
(3) Assumption \ref{assump: factor model}.3 is a standard weak dependency assumption in the noise and allows the errors to be time-series and cross-sectionally weakly correlated.
	(4) Assumption \ref{assump: factor model}.4 allows the factors and idiosyncratic errors to be weakly correlated.

Assumption \ref{assump: additional assumptions} is not necessary to show the consistency of loadings and factors but is only used to show the asymptotic normality of the estimators. The assumptions are closely related to the moment and CLT assumptions in \cite{bai2003inferential}. Assumptions \ref{assump: additional assumptions}.1-\ref{assump: additional assumptions}.5 bound the second moments of certain averages. Assumption \ref{assump: additional assumptions}.6 and \ref{assump: additional assumptions}.7 state the necessary central limit theorems. Assumption \ref{assump: additional assumptions}.8 is specific to the missing value problem and introduces the correction terms that appear in the asymptotic distribution. These terms emerge because our estimator averages over different numbers of observations for different entries in the covariance matrix. 
Assumption \ref{assump: additional assumptions}.8 assumes a central limit theorem for $X_i^{(\gamma)}$ and $\mathbf{X}_t^{(\gamma)}$. The conventional CLT of the form 
		\begin{align*}
			\sqrt{T} \begin{bmatrix}
				\text{vec}(X_i^{(\gamma)}) \\ \text{vec}(\mathbf{X}_t^{(\gamma)})
			\end{bmatrix} \xrightarrow{d} \Ncal \Lp 0, \begin{bmatrix}
				\Psi_i & ({\Psi}^{\text{cov}}_{i,t})^\T  \\ {\Psi}^{\text{cov}}_{i,t}  & {\Psi}_t
			\end{bmatrix} \Rp
		\end{align*}
		would not be sufficient as $X_i^{(\gamma)}$ and $\mathbf{X}_t^{(\gamma)}$ are multiplied with the random variables $u_i$ and $v_t^{(\gamma)}$ in $\sqrt{T} \begin{bmatrix}
			(X_{i}^{(\gamma)} u_i)^\T & (\mathbf{X}_t^{(\gamma)} v_t^{(\gamma)})^\T
		\end{bmatrix}$. Hence, the asymptotic variances of these products are quadratic functions in the elements of those random variables given by $h_i^{(\gamma)}(u_i)$ and $g_t^{(\gamma)}(v_t^{(\gamma)})$ and take the form of $h_i^{(\gamma)}(u_i) = (u_i^\T \otimes I_k) \Psi_i (u_i \otimes I_k)$ and $g_t^{(\gamma)}(v_t^{(\gamma)}) = (v_t^{(\gamma)\T} \otimes I_k) {\Psi}_t (v_t^{(\gamma)} \otimes I_k)$.

In Assumption \ref{assump: additional assumptions}.8 we require a central limit theorem for stable convergence in law which is stronger than the conventional central limit theorem for convergence in distribution. This is because the asymptotic variance in Assumption \ref{assump: additional assumptions}.8 depends on $(\LamY)_i$, $(\LamX)_i$ and $F_t$, which are random variables. Thus, we deal with a mixed normal limit, and stable convergence in law ensures that the normal distribution of the central limit theorem will be independent of $(\LamY)_i$, $(\LamX)_i$, and $F_t$. The stable convergence in law result implies that the estimated factors and common components normalized by their random standard deviation converge to a standard normal distribution. More specifically, Assumption \ref{assump: additional assumptions}.8 implies that $X_i^{(\gamma)}$ and $\textbf{X}_t^{(\gamma)}$ jointly converge $\mathcal{G}^t$-stably for $(N_x, N_y, T) \rightarrow \infty$ to a mixed normal distribution, whose asymptotic variance is random but measurable with respect to the sigma-field $\mathcal{G}^t$. Assumption \ref{assump: additional assumptions}.8 is needed for the asymptotic distribution of the variance correction term in Theorem \ref{thm: asymptotic distribution} as its asymptotic variance is random. Our simplified factor model specified by Assumptions \ref{assump: simplified factor model} satisfies a central limit theorem for stable convergence in law.

\section{A Simplified Factor Model} \label{appendix: simplified model}
We present a simplified factor model with the stronger Assumptions \ref{assump: simplified factor model} and \ref{assump: simplified obs pattern}, which substantially simplifies the notation but conveys the main conceptual insights of the general model. It allows us to highlight the effect of the target weight, the relaxation of weak factors and missing observations. In particular, under these assumptions, we can provide explicit expressions for the asymptotic variances in Theorem \ref{thm: asymptotic distribution}.

The consistency results are based on the simplified Assumption \ref{assump: simplified factor model} that assumes that all observations are i.i.d. The key element is the strength of the factors measured by their loadings. Specifically, we measure the strength of the factors by the fraction $p_j$ of units in $Y$ that are affected by the corresponding factor. The error terms are non-systematic with bounded eigenvalues in the covariance matrix. Allowing for more complex dependency as in our general model does not change the arguments, but makes the notation more complex.


\begin{assumpS}[Simplified factor model]  \label{assump: simplified factor model}
There exists constant $C<\infty$ such that
\begin{enumerate}
    \item Factors: $F_t \overset{\text{i.i.d.}}{\sim} (0,\Sigma_F)$ and $\E\|F_t\|^4 \leq C$ for any $t.$
    \item Loadings: 
    %
    $(\LamX)_{i} \overset{\text{i.i.d.}}{\sim} (0,\Sigma_{\LamX})$, where $\Sigma_{\LamX}$ is positive semidefinite. 
    $(\LamY^{\mathrm{full}})_{i} \overset{\text{i.i.d.}}{\sim} (0,\Sigma^{\mathrm{full}}_{\LamY})$  and the loading of the $j$-th factor $(\LamY)_{ij} = (\LamY^{\mathrm{full}})_{ij} \cdot (U_{y})_{ij}$, where $\Sigma^{\mathrm{full}}_{\LamY}$ is positive definite and the Bernoulli random variable $(U_{y})_{ij} \in \{0,1\}$ is independent in $i$ with $\PP((U_{y})_{ij} = 1) = p_{j}$ for some $p_{j} \in [0,1]$.
    Furthermore, $\E\|(\LamX)_i\|^4\leq C$, $\E\|(\LamY)_i\|^4\leq C$, $N_y^{-1}\sum_{i=1}^{\Ny}(\LamY)_i(\LamY)_i^\top \overset{p}{\rightarrow} \Sigma_{\LamY}$, and $ \Sigma_{\LamX}+\Sigma_{\LamY}$ is positive definite. For any $t,$ $N_y^{-1}\sum_{i=1}^{\Ny}\WY_{ti}(\LamY)_i (\LamY)_i^\top \overset{p}{\rightarrow} \Sigma_{\LamY,t}$ and  $\Sigma_{\LamX}+\Sigma_{\LamY,t}$ is positive definite.
    \item Idiosyncratic errors: $(\eX)_{ti}\overset{\text{i.i.d.}}{\sim} (0,\sigma_{\eX}^2), (\eY)_{ti}\overset{\text{i.i.d.}}{\sim} (0,\sigma_{\eY}^2)$, $\E(\eX)^8_{ti}\leq C, \E(\eY)^8_{ti}\leq C$.
    \item Independence: $F,\LamX,\LamY,\eX$ and $\eY$ are independent.
\end{enumerate}
\end{assumpS}

Our generative model for $(\LamY)_{i}$ accounts for three cases of factor strength on $Y$. First, if $p_{j}$ is bounded away from $0$ as $\Ny$ grows, then the $j$-th factor is a strong factor in $Y$. Second, if $p_{j}$ decays to $0$ but is nonzero as $\Ny$ grows, then the $j$-th factor is a weak factor in $Y$. Third, if $p_{j}$ is $0$ for all $\Ny$, then $Y$ does not contain the $j$-th factor. Note that $\Sigma_{\LamX}$ can be rank deficient, implying that the loadings of some factors can be zero for units in $X$. However, our assumption on $(\LamX)_{i}$ rules out the case of weak factors in $X$ as the estimation of weak factors in $X$ is not our objective. We assume that $ \Sigma_{\LamX}+\Sigma_{\LamY}$ is positive definite to ensure that each factor in $F$ is strong in at least one of the two panels $X$ and $Y$. Specifically, weak factors in $Y$ are strong in $X$. Hence, all factors can be identified with target-PCA with a properly chosen $\gamma$.

Furthermore, Assumption \ref{assump: simplified factor model}.2 imposes assumptions on the missing pattern in $Y$ to identify all factors when combining the partially observed $Y$ and $X$. More specifically, the second-moment matrix $\Sigma_{\LamY,t}$ does not need to be full rank in Assumption \ref{assump: simplified factor model}.2, which relaxes the full-rank assumption of $\Sigma_{\LamY,t}$ in \cite{xiong2019large}. However, we assume that $\Sigma_{\LamX}+\Sigma_{\LamY,t}$ is positive definite, so that target-PCA can identify all factors from $X$ and partially observed $Y$. 

The asymptotic results require additional restrictions on the observation pattern, as stated in Assumption \ref{assump: simplified obs pattern} below. 

\begin{assumpS}[{Moment conditions under partial observations}] \label{assump: simplified obs pattern} \texttt{}
For any $i$, there exist constants $\omega_i^{(1)}$, $\omega_i^{(2,1)}$, $\omega_i^{(2,2)}$, $\omega_i^{(2,3)}$ and $\omega_i^{(3)}$, such that $N_y^{-1}\sum_{j=1}^{\Ny}\frac{q_{ij}}{q_{ii}q_{jj}}\overset{p}{\rightarrow} \omega_i^{(1)}$, $N_y^{-2}\sum_{j,l=1}^{\Ny} \frac{q_{ii,jl}}{q_{ii}q_{jl}}\overset{p}{\rightarrow}\omega_i^{(2,1)}$, $N_y^{-2}\sum_{j,l=1}^{\Ny} \frac{q_{jj,il}}{q_{jj}q_{il}}\overset{p}{\rightarrow}\omega_i^{(2,2)}$, $N_y^{-2}\sum_{j,l=1}^{\Ny}  \frac{q_{ij,il}}{q_{ij}q_{il}}\overset{p}{\rightarrow}\omega_i^{(2,3)}$, and $\Ny^{-3}\sum_{j,l,h=1}^{\Ny} \frac{q_{il,jh}}{q_{il}q_{jh}}\overset{p}{\rightarrow}\omega_i^{(3)}$. Furthermore, there exist constants $\omega^{(1)}$, $\omega^{(2)}$ and $\omega^{(3)}$, such that $N_y^{-1}\sum_{i=1}^{\Ny}\omega_i^{(1)}\overset{p}{\rightarrow}\omega^{(1)}$, $N_y^{-1}\sum_{i=1}^{\Ny}\omega_i^{(2,1)}=N_y^{-1}\sum_{i=1}^{\Ny}\omega_i^{(2,2)}\overset{p}{\rightarrow}\omega^{(2)}$, and $N_y^{-1}\sum_{i=1}^{\Ny}\omega_i^{(3)}\overset{p}{\rightarrow}\omega^{(3)}$.\end{assumpS}

Assumption \ref{assump: simplified obs pattern} introduces interpretable parameters that fully capture the information in the observation pattern, analogous to \cite{xiong2019large}. These key parameters are $\omega_i^{(1)}$,  
$\omega_i^{(2,1)}$, $\omega_i^{(2,2)}$, $\omega_i^{(2,3)}$, $\omega_i^{(3)}$, $\omega^{(1)}$, $\omega^{(2)}$, and $\omega^{(3)}$, and they are relevant for the asymptotic distribution. All these
 quantities are the averages of $\frac{q_{il,jh}}{q_{il} q_{jh}}$, where $\frac{q_{il,jh}}{q_{il} q_{jh}}$ roughly measures the correlation in the observation patterns for unit $i$, $j$, $h$ and $l$. The superscript $m \in \{1,2,3\}$ in $\omega_i^{(m)}$, $\omega_i^{(m,m^\prime)}$ and $\omega^{(m)}$ refers to the number of indices over which $\frac{q_{il,jh}}{q_{il} q_{jh}}$ is averaged. Roughly speaking, $\omega_i^{(m)}$ and $\omega_i^{(m,m^\prime)}$ measure the average correlation in the observation patterns between unit $i$ and any other $m$ units. $\omega^{(m)}$ is the average of $\omega_i^{(m)}$ or $\omega_i^{(m,m^\prime)}$ over $i$. Note that
 $\omega_i^{(2,1)}$,  $\omega_i^{(2,2)}$ and  $\omega_i^{(2,3)}$ are closely connected, but they are slightly different in that the two indices over which $\frac{q_{il,jh}}{q_{il} q_{jh}}$ are averaged are different. 

In the special case, when observations are missing at random with observed probability $p$, we have $\frac{q_{il,jh}}{q_{il} q_{jh}} = 1$ for four distinct units $i$, $j$, $h$ and $l$, and we can show that
$\omega_{i}^{(2,3)} = {1}/{p}$, $\omega_i^{(1)} = \omega_i^{(2,1)} = \omega_i^{(2,2)}= \omega_i^{(3)} = 1$, and $\omega^{(1)}=\omega^{(2)}=\omega^{(3)} = 1$. For other observation patterns, $\omega_i^{(\cdot)}$, $\omega_i^{(\cdot,\cdot)}$ and $\omega^{(\cdot)}$ tend to increase if there are stronger correlations in whether entries are observed across units and across time. In addition, these quantities tend to decrease with the fraction of observed entries.\footnote{See Table 5 in \cite{xiong2019large} for the values of these quantities under different observation patterns and fractions of observed entries.} In Corollary \ref{thm: simplified factor model} below, we show that these quantities can summarize all the information in the observation pattern that is relevant for the efficiency of the estimated factor model. 

The simplified factor model specified by Assumptions \ref{assump: simplified factor model} and \ref{assump: simplified obs pattern} satisfies the general Assumptions \ref{assump: factor model} and \ref{assump: additional assumptions}, as stated in Proposition \ref{prop: link of assumptions}.

\begin{proposition}  \label{prop: link of assumptions}
The simplified factor model specified in Assumptions \ref{assump: simplified factor model} and \ref{assump: simplified obs pattern} is a special case of the general approximate factor model assumed in Assumptions \ref{assump: factor model} and \ref{assump: additional assumptions}. Specifically, Assumptions \ref{assump: obs pattern} and \ref{assump: simplified factor model} imply Assumption \ref{assump: factor model}, and Assumptions \ref{assump: obs pattern}, \ref{assump: simplified factor model} and \ref{assump: simplified obs pattern} imply Assumption \ref{assump: additional assumptions}. 
\end{proposition}


The distribution results of Theorem \ref{thm: asymptotic distribution} simplify under Assumptions \ref{assump: simplified factor model} and \ref{assump: simplified obs pattern}, and we can provide explicit expressions for the asymptotic variances. Corollary \ref{thm: simplified factor model} shows the analytical expression of the asymptotic variances under the simplified factor model, which allows us to gain intuition on how $\gamma$ affects the efficiency of the estimation.

\begin{corollary} \label{thm: simplified factor model}
Under Assumptions $\ref{assump: obs pattern}$, $\ref{assump: simplified factor model}$ and $\ref{assump: simplified obs pattern}$, and for $\gamma = r\cdot \Nx/\Ny$ with some positive constant $r$, the asymptotic distributions in Theorem \ref{thm: asymptotic distribution} hold as $T,\Nx,\Ny \rightarrow \infty$. 
In addition, if we assume that $q_{ij}$ and $q_{ij,hl}$ are independent of $(\LamX)_m(\LamX)_m^\top$ and $(\LamY)_m(\LamY)_m^\top$ for any $i,j,h,l,m$, then the asymptotic variances in Theorem \ref{thm: asymptotic distribution} can be explicitly written as follows:
\begin{enumerate}
\item The asymptotic variance of the estimated loadings of $Y$ in \eqref{equ: dist of loadings} is
\[
\Sigma_{\LamY,i}^{(\gamma)} = \Sigma_{\LamY,i}^{(\gamma),{\rm obs}} +  \Sigma_{\LamY,i}^{(\gamma),{\rm miss}},
\]
where 
\begin{align*}
    &\Sigma_{\LamY,i}^{(\gamma),{\rm obs}} = \frac{1}{q_{ii}}\sigma_{\eY}^2\Sigma_F^{-1}, \\
    & \Sigma_{\LamY,i}^{(\gamma),{\rm miss}} = \Big(\frac{1}{q_{ii}}-1\Big) \Sigma_F^{-1}((\LamY)_i^\top \otimes I_k)\Xi_F ((\LamY)_i \otimes I_k)\Sigma_F^{-1} + \Big(\omega_{i}^{(2,3)}-\frac{1}{q_{ii}}\Big)\Sigma_F^{-1}\Lp \Sigma_{\LamX} + r\Sigma_{\LamY}\Rp^{-1}\\
    & \qquad \quad \quad \ \ \Ls \sigma_{\eY}^2r^2\Sigma_{\LamY}\Sigma_F\Sigma_{\LamY} + \Lp (\LamY)_i^\top \otimes r \Sigma_{\LamY}\Rp\Xi_F\Lp (\LamY)_i \otimes r \Sigma_{\LamY}\Rp \Rs\Lp \Sigma_{\LamX} + r\Sigma_{\LamY}\Rp^{-1}\Sigma_F^{-1},
\end{align*}
and $\Xi_F = \Var(\text{vec}(F_tF_t^\top))$.

\item 
For case 1 in Theorem \ref{thm: asymptotic distribution}.2, the asymptotic variance of the estimated factors in \eqref{equ: dist of factors} is
\begin{align*}
\Sigma_{F,t}^{(\gamma)} & = \frac{\delta_{\Ny,T}}{\Ny}\cdot \Sigma_{F,t}^{(\gamma),{\rm obs}}  + \frac{\delta_{\Ny,T}}{T}\cdot \Sigma_{F,t}^{(\gamma),{\rm miss}},
\end{align*}
where
\begin{align*}
&\Sigma_{F,t}^{(\gamma),{\rm obs}} =  (\Sigma_{\LamX}+r\Sigma_{{\LamY},t})^{-1}\Big(\frac{\Ny}{\Nx}\sigma_{\eX}^2\Sigma_{\LamX}+ r^2\sigma_{\eY}^2\Sigma_{{\LamY},t}\Big) (\Sigma_{\LamX}+r\Sigma_{{\LamY},t})^{-1}, \\
&\Sigma_{F,t}^{(\gamma),{\rm miss}}  = 
(\Sigma_{\LamX}+r\Sigma_{{\LamY},t})^{-1} \Lp I_k \otimes F_t^\top\Sigma_{F}^{-1}\big(\Sigma_{\LamX}+ r \Sigma_{\LamY}\big)^{-1}\Rp\bigg[  \Lp \Sigma_{\LamX} \otimes r\Sigma_{\LamY} + r\Sigma_{{\LamY},t} \otimes \Sigma_{\LamX} \Rp \\
&\qquad \qquad \quad \Xi_F \Big[(\omega^{(1)}-1)\Lp \Sigma_{\LamX} \otimes r\Sigma_{\LamY} + r\Sigma_{{\LamY},t} \otimes \Sigma_{\LamX} \Rp + (\omega^{(2)}-1)\Lp  r\Sigma_{{\LamY},t} \otimes r\Sigma_{\LamY} \Rp\Big] + \\
&\qquad \qquad \quad\Lp r\Sigma_{\LamY,t} \otimes r\Sigma_{\LamY}\Rp \Xi_F \Big[(\omega^{(2)}-1) \Lp \Sigma_{\LamX} \otimes r\Sigma_{\LamY} + r\Sigma_{{\LamY},t} \otimes \Sigma_{\LamX} \Rp+ (\omega^{(3)}-1) \\
&\qquad \qquad \quad\Lp r\Sigma_{\LamY,t} \otimes r\Sigma_{\LamY}\Rp\Big]\bigg]
\Lp I_k \otimes (\Sigma_{\LamX}+r\Sigma_{\LamY})^{-1}\Sigma_{F}^{-1}F_t\Rp\Lp\Sigma_{\LamX}+r\Sigma_{{\LamY},t}\Rp^{-1}. 
\end{align*}
For case 2 in Theorem \ref{thm: asymptotic distribution}.2, the asymptotic variance of the estimated weak factors in \eqref{equ: dist of factors weak} is
\begin{align*}
\Sigma_{F_w,t}^{(\gamma)} & = \frac{\delta_{N_w,T}}{N_w}\cdot \Sigma_{F_w,t}^{(\gamma),{\rm obs}}  + \frac{\delta_{N_w,T}}{T}\cdot \Sigma_{F_w,t}^{(\gamma),{\rm miss}},
\end{align*}
where
\[
\Sigma_{F_w,t}^{(\gamma),{\rm obs}} =  (\Sigma_{\LamX}+r\Sigma_{{\LamY},t})_w^{-1}\Big(\frac{N_w}{\Nx}\sigma_{\eX}^2\Sigma_{\LamX,w}+ r^2\frac{p_w N_w}{\Ny}\sigma_{\eY}^2\Sigma_{{\LamY},t,w}\Big) (\Sigma_{\LamX}+r\Sigma_{{\LamY},t})^{-1}_w,
\]
$N_w = \min({\Ny}/{p_w}, \Nx)$, $p_w$ is defined in Assumption \ref{assump: simplified factor model}.2 as the fraction of units in $Y$ that are affected by the weak factors,
$(\Ny p_w)^{-1}\sum_{i=1}^{\Ny}W^Y_{ti}(\LamY)_{i,w}(\LamY)_{i,w}^\top \overset{p}{\rightarrow} \Ncal(0,\Sigma_{{\LamY},t,w})$, $(\Sigma_{\LamX}+r\Sigma_{{\LamY},t})_w^{-1}$, $\Sigma_{\LamX,w}$ and $\Sigma^{(\gamma),{\rm miss}}_{F_{w},t}$ are respectively the diagonal block of  $(\Sigma_{\LamX}+r\Sigma_{{\LamY},t})^{-1}$, $\Sigma_{\LamX}$ and $\Sigma^{(\gamma),{\rm miss}}_{F,t}$ corresponding to the weak factors.

\item  The asymptotic variance of the estimated common components of $Y$ in \eqref{equ: dist of common component} is
\begin{align*}
\Sigma_{C,ti}^{(\gamma)} =&\frac{\delta_{\Ny, T }}{T } F_t^\top \Lp \Sigma_{\LamY,i}^{(\gamma),{\rm obs}} + \Sigma_{\LamY,i}^{(\gamma),{\rm miss}}\Rp F_t + \frac{\delta_{\Ny, T }}{\Ny } (\LamY)_i^\top \Sigma_{F,t}^{(\gamma),{\rm obs}} (\LamY)_i \\
& + \frac{\delta_{\Ny, T }}{T } (\LamY)_i^\top \Sigma_{F,t}^{(\gamma),{\rm miss}}(\LamY)_i  -2 \frac{\delta_{\Ny, T }}{T } (\LamY)_i^\top   \Sigma_{\LamY,F,i,t}^{(\gamma),{\rm miss, cov}}  F_t,    
\end{align*}
where 
\begin{align*}
\Sigma_{\LamY,F,i,t}^{(\gamma),{\rm miss,cov}} = &\Lp\Sigma_{\LamX}+r\Sigma_{\LamY,t}\Rp^{-1} \Lp I_k \otimes F_t^\top\Sigma_{F}^{-1}\big(\Sigma_{\LamX}+ r \Sigma_{\LamY}\big)^{-1}\Rp\Big[ \Lp\Sigma_{\LamX}\otimes r\Sigma_{\LamY} + r\Sigma_{\LamY,t}\otimes  \Sigma_{\LamX}\Rp \\
&\Xi_F\Lp (\omega_i^{(1)}-1)((\LamY)_i\otimes \Sigma_{\LamX}) + (\omega_i^{(2,2)}-1) ((\LamY)_i \otimes r\Sigma_{\LamY})\Rp+ (r\Sigma_{\LamY,t}\otimes r\Sigma_{\LamY})\Xi_F\\ &   \Lp(\omega_i^{(2,1)}-1)((\LamY)_i\otimes \Sigma_{\LamX}) + (\omega_i^{(3)}-1) \Lp(\LamY)_i\otimes r\Sigma_{\LamY}\Rp\Rp \Big]\Lp\Sigma_{\LamX}+r\Sigma_{\LamY}\Rp^{-1}\Sigma_F^{-1}.
\end{align*}
\end{enumerate}
\end{corollary}

The simplified model provides a clear interpretation of how the asymptotic variance is affected by missingness and the target weight. The explicit expressions in Corollary \ref{thm: simplified factor model} show how the asymptotic variances depend on the following key quantities: the target weight $\gamma$, the noise ratio (NR) $\sigma_{\eX}/\sigma_{\eY}$, the dimension ratio (DR) $\Nx/\Ny$ and the dependency structure in the missing pattern captured by the parameters defined in Assumption \ref{assump: simplified obs pattern}. Without missing data, the correction matrices in the variance disappear. Some correction terms also disappear when data is missing at random, as in this case it holds that $\omega_{i}^{(2,3)} = 1/p$, $\omega_i^{(1)} = \omega_i^{(2,1)} = \omega_i^{(2,2)}= \omega_i^{(3)} = 1$, and $\omega^{(1)}=\omega^{(2)}=\omega^{(3)} = 1$.

\section{Sparser Observation Pattern} \label{appendix: sparser observation pattern}

In the main text, we have considered the case where the number of observed entries in $Y$ is of the order $\Ny T$. In this section, we generalize our results to the case where the target $Y$ has even more missing observations. For this case, we can still estimate the factor model using our target-PCA method, but we have a slower convergence rate due to more missingness. Specifically, we assume that the number of observed time periods in $Y$ can grow sublinearly in $T$.


\edef\oldassumption{\the\numexpr\value{assumption}+1}

\renewcommand{\theassumption}{G\oldassumption$^\prime$}
\begin{assumption}
\label{assump: obs pattern weak signals}
(Sparser Observation Pattern) 
\begin{enumerate}
    \item The observation matrix $\WY$ is independent of the factors $F$ and idiosyncratic errors $e_y$.
    \item For any given observation matrix $\WY$, there exist positive constants $q_1$, $q_2$ and $\alpha \in (0,1]$ such that $q_1 \geq {|\QY_{ij}|}/{T^\alpha}\geq q_2>0$ for all $i,j$. Furthermore, let $q_{ij} = \lim_{T\rightarrow \infty}{|\QY_{ij}|}/{T^\alpha}$ and $q_{ij,hl} = \lim_{T\rightarrow \infty}{|\QY_{ij}\cap \QY_{hl}|}/{T^\alpha}$. For any $i,j,h,l$, $q_{ij}$ and $q_{ij,hl}$ are positive constants bounded away from 0.
\end{enumerate}
\end{assumption}


\let\theassumption\origtheassumption

Assumption \ref{assump: obs pattern weak signals} generalizes Assumption \ref{assump: obs pattern}.
More specifically, Assumption \ref{assump: obs pattern weak signals}.1 is the same as Assumption \ref{assump: obs pattern}.1. Assumption \ref{assump: obs pattern weak signals}.2 is more general than Assumption \ref{assump: obs pattern}.2, which is the special case for $\alpha = 1$. We assume that the number of time periods when any two and any four units are observed are both proportional to $T^\alpha$. This assumption includes the important example of low-frequency observations in Figure \ref{fig: missing patterns}(c), where the observed time periods are the same for all units and are at the order of $T^\alpha$. Under Assumption \ref{assump: obs pattern weak signals}, if $X$ contains all the factors in $Y$, the latent factor model on $Y$ can still be consistently estimated from target-PCA and the estimators are asymptotically normal, but with a slower convergence rate.


\begin{proposition} \label{thm: asymptotic distribution weak signals}
Define $\delta_{\Ny, T^\alpha} = \min(\Ny ,T^\alpha)$ and select $\gamma$ such that $\Sigma^{(\gamma)}_{\Lam,t}$ is a positive definite matrix and the eigenvalues of $\Sigma_F\Lp \Sigma_{\LamX}+\gamma{\Ny}/{\Nx}\cdot\Sigma_{\LamY}\Rp$ are distinct. Suppose Assumptions \ref{assump: obs pattern weak signals} and \ref{assump: factor model} hold, and Assumption \ref{assump: additional assumptions} holds with $\sqrt{T}$ replaced by $\sqrt{T^\alpha}$. As $T,\Nx,\Ny  \rightarrow \infty$ we have for each $i$ and $t$:
\begin{enumerate}
\item For $\sqrt{T^\alpha}/\Ny \rightarrow 0,$ the asymptotic distribution of the loadings of $Y$ in Theorem \ref{thm: asymptotic distribution}.1 continues to hold with convergence rate $\sqrt{T^\alpha}$.

\item For $\sqrt{T^\alpha}/\Ny  \rightarrow 0$ and $\sqrt{\Ny}/T^\alpha \rightarrow 0$, the asymptotic distribution of the factors in Theorem \ref{thm: asymptotic distribution}.2 and the asymptotic distribution of the common components of $Y$ in Theorem \ref{thm: asymptotic distribution}.3 continue to hold with convergence rate $\sqrt{\delta_{\Ny, T^\alpha}}$.
\end{enumerate}
\end{proposition}

Note that the convergence rate of the estimated loadings for each unit in $Y$ is $\sqrt{T^\alpha}$, and the convergence rate of the estimated factors and common components equals the smaller value of $\sqrt{T^\alpha}$ and $\sqrt{\Ny}$. This convergence rate makes sense, as for each unit in $Y$, we only have $T^\alpha$ time-series observations to estimate its factor loadings. 

Choosing $\gamma$ under Assumption \ref{assump: obs pattern weak signals} is conceptually similar to choosing $\gamma$ in our base case. Target-PCA essentially estimates the latent factor model from a matrix that combines the sample covariance matrix of $X$ (that is $X^\T X/T$) and the sample covariance matrix of $Y$ weighted by $\gamma$ (that is $\gamma \cdot \tilde{Y}^\T \tilde{Y}/T^\alpha$ if $Y$ has a low-frequency observation pattern). As the top eigenvalues in $X^\T X/T$ and of $\gamma \cdot \tilde{Y}^\T \tilde{Y}/T^\alpha$ are at the order of $\Nx$ and $\gamma \Ny$,\footnote{If $\alpha = 1$, then this is equivalent to selecting $\gamma$ such that the top eigenvalues of $X^\T X$ and of $\gamma \cdot \tilde{Y}^\T \tilde{Y}$ are at the same order (or equivalently, $X X^\T$ and $\gamma \cdot \tilde{Y} \tilde{Y}^\T$ given that for any matrix $A$, the nonzero eigenvalues of $A A^\T$ are the same as the nonzero eigenvalues of $A^\T A$). } selecting $\gamma = r\cdot \Nx/\Ny$ ensures that the top eigenvalues in two sample covariance matrices are of the same order, and the strong factors in either $X$ or $Y$ can be identified with target-PCA. In summary, the general logic and selection rules from the main setting carry over to this more general case.


{\section{An Alternative Estimator}\label{appendix:alter-estimator}
An alternative approach to estimate $F$, $\Lambda_x$ and $\Lambda_y$ is based on the identification assumption $F^{\top} F/T=I_k$, which allows us solve for $F$ from the objective function 
\begin{align}\label{equ: obj}
    \max_{F} \text{trace}\left( F^{\top} \big( Z^{(\gamma)} Z^{(\gamma)\top} \big) F \right),
\end{align}
when $Y$ is fully observed. The solution to this optimization problem is obtained by applying PCA to $Z^{(\gamma)}Z^{(\gamma)\top}$. This is essentially the same as applying PCA to the $T \times T$ time-series second moment matrix of $Z^{(\gamma)}$, denoted by $\Sigma^{Z(\gamma),\mathsf{time}}$, with the $(s,t)$-th entry
(suppose zero idiosyncratic noise for simplicity)
\[\Sigma^{Z(\gamma),\mathsf{time}}_{s,t} = F_s^\T \left(\lim_{N_x, N_y\rightarrow \infty} \frac{N_x}{N_x + N_y} \Sigma_{\LamX} + \gamma \frac{N_y}{N_x + N_y} \Sigma_{\LamY} \right) F_t. \]
Without missing values, the estimation of the latent factor model by applying PCA to either the cross-sectional or time-series second-moment matrix of $Z^{(\gamma)}$ is equivalent. However, this changes in the case of only partially observed data.

When $Y$ has missing observations, we can estimate $\Sigma^{Z(\gamma),\mathsf{time}}_{(s,t)}$ using an approach analogous to the one in Section \ref{sec: estimator}:
\[\hat{\Sigma}^{Z(\gamma),\mathsf{time}}_{s,t} =  \frac{N_x}{N_x + N_y} \cdot 
\underbrace{\frac{1}{N_x}  \sum_{i = 1}^{N_x} X_{si} X_{ti}}_{\mathrm{estimate~}F_s^\T \Sigma_{\LamX} F_t }  + \gamma \frac{N_y}{N_x + N_y} \underbrace{\frac{1}{|Q_{st}|}  \sum_{j \in Q_{st}} Y_{sj} Y_{tj}}_{\mathrm{estimate~}F_s^\T \Sigma_{\LamY} F_t },    \]
where $Q_{st}$ is the set of units in $Y$ that are observed at both time $s$ and $t$. Then, we can estimate $F$ by applying PCA to $\hat\Sigma^{Z(\gamma),\mathsf{time}}$. 

In order to ensure consistency of the PCA factors based on $\hat{\Sigma}^{Z(\gamma),\mathsf{time}}$, we need to require that the following holds for any $s$ and $t$, 
\[\frac{1}{|Q_{st}|}  \sum_{j \in Q_{st}} Y_{sj} Y_{tj} = F_t^\T  \left(\frac{1}{|Q_{st}|} \sum_{j \in Q_{st}} \left(\LamY \right)_{tj}\left(\LamY \right)_{tj}^\T  \right) F_s + O_P(1) =   F_s^\T \Sigma_{\LamY} F_t + O_P(1), \]
where the dominating eigenvalues have to come from $F_s^\T \Sigma_{\LamY} F_t$, that is, the eigenvectors of the largest eigenvalues of $|Q_{st}|^{-1}  \sum_{j \in Q_{st}} Y_{sj} Y_{tj}$ converge to the factors up to a rotation. This implies that the observation pattern of a unit has to be independent of its loadings. For causal inference applications, this assumption is less benign than Assumption \ref{assump: obs pattern} that the observation pattern of a unit is independent of the factors. Specifically, the treatment adoption pattern is commonly allowed to depend on the characteristics of units (e.g. the loadings) in observational studies. }


\newpage

\section{Empirical Study}\label{appendix: empirical}

\begin{figure}[H]
    \tcapfig{Examples of quarterly observed vs. monthly imputed time series with target-PCA}
          \begin{subfigure}[b]{1.0\textwidth}
          \centering
        \includegraphics[width=0.9\textwidth]{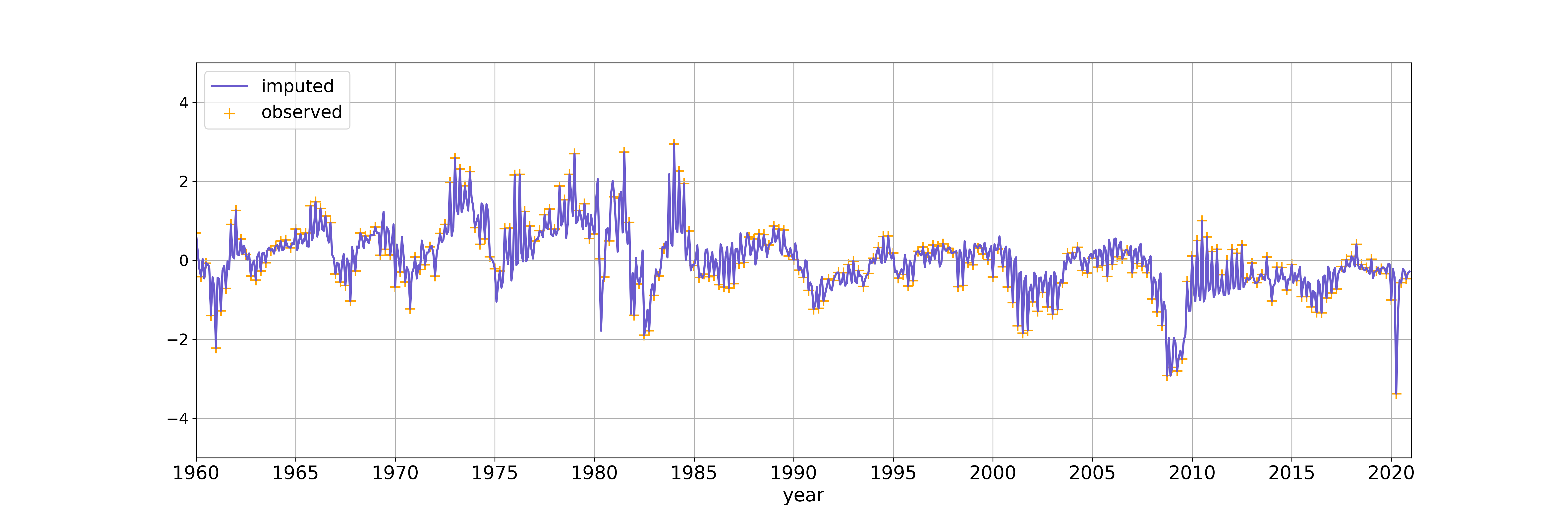}
\caption{Goods}
    \end{subfigure}
          \begin{subfigure}[b]{1.0\textwidth}
          \centering
        \includegraphics[width=0.9\textwidth]{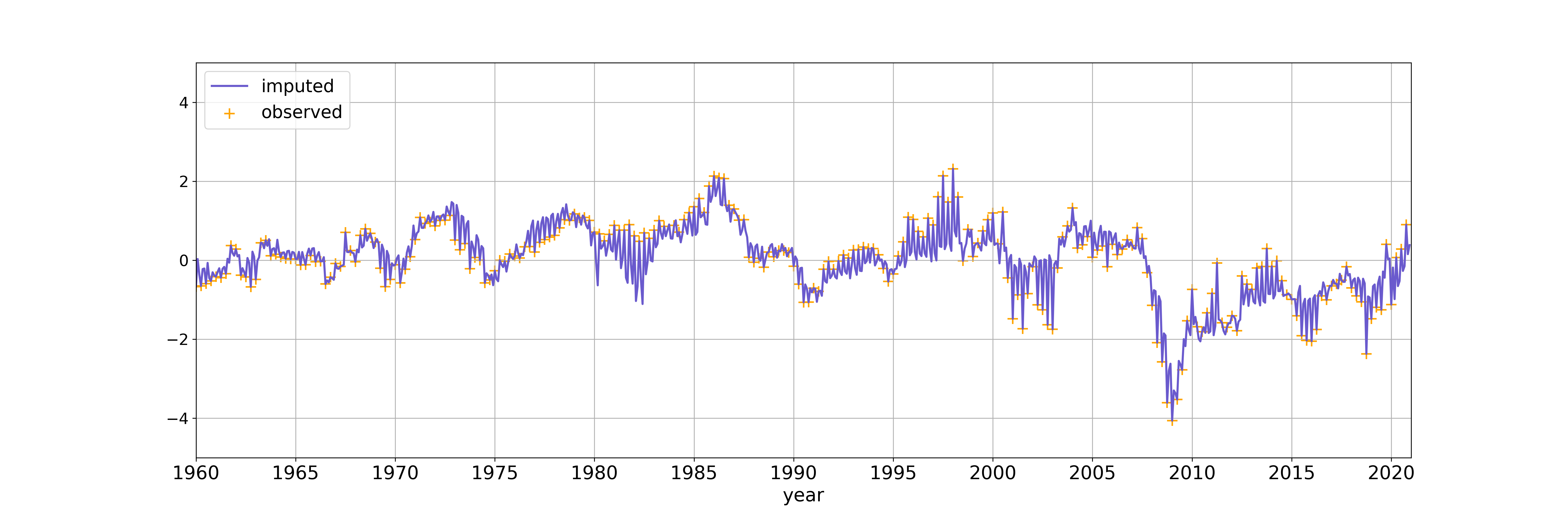}
\caption{Total liabilities and equity in domestic financial sectors}
    \end{subfigure}    
          \begin{subfigure}[b]{1.0\textwidth}
          \centering
        \includegraphics[width=0.9\textwidth]{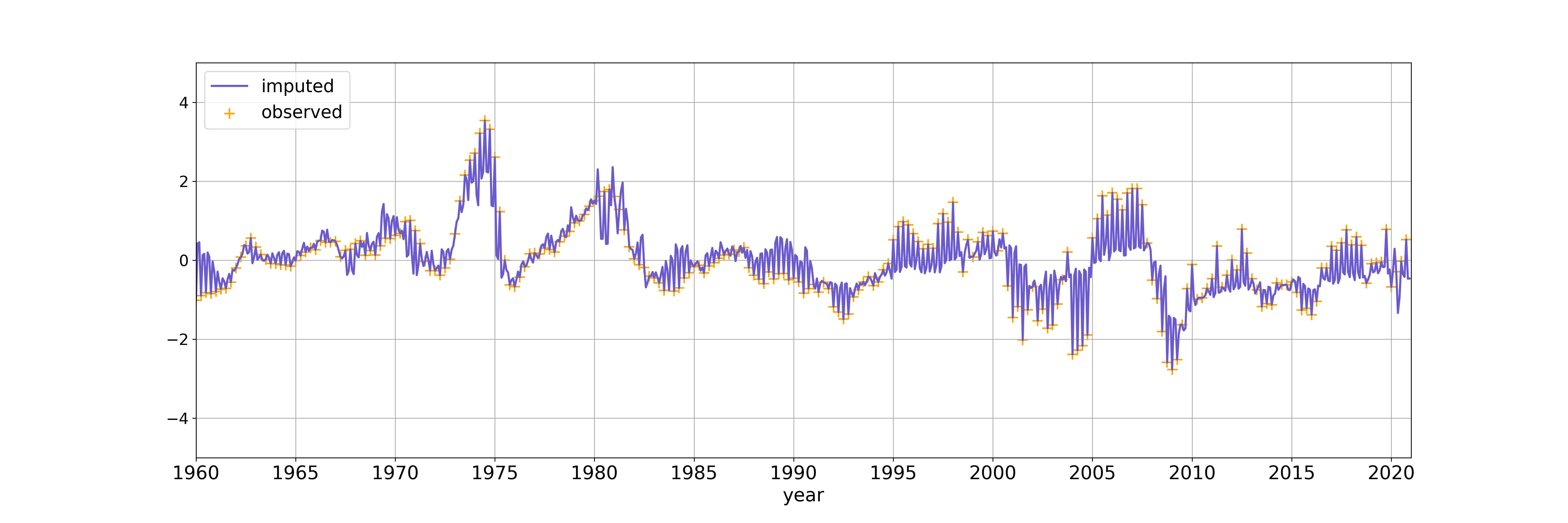}
\caption{Net worth (IMA) of state and local governments}
    \end{subfigure}        
    \floatfoot{This figure shows examples of observed monthly imputed time series with target-PCA. The quarterly observed target consists of macroeconomic time series of the national income \& product accounts category and the flow of funds category from the FRED database. The orange ``+''  denotes the quarterly observed values of the macroeconomic time series, and the purple curve displays the imputed values of the time series with target-PCA for $k=5$ latent factors. The time series are percentage changes relative to the prior year.
    }
    \label{fig:example2}
\end{figure}


\end{document}


\title{Internet Appendix to Target PCA: Transfer Learning Large Dimensional Panel Data}

\date{June 7, 2023}
	
        \author{Junting Duan\thanks{ \scriptsize Stanford University, Department of Management Science \& Engineering, Email: \texttt{duanjt@stanford.edu}.}
                \and
		Markus Pelger\thanks{\scriptsize Stanford University, Department of Management Science \& Engineering, Email: \texttt{mpelger@stanford.edu}.}
               \and
              Ruoxuan Xiong\thanks{ \scriptsize Emory University, Department of Quantitative Theory and Methods, Email: \texttt{ruoxuan.xiong@emory.edu}.}
	}
	
	\onehalfspacing
	
	\begin{titlepage}
    \maketitle
    \thispagestyle{empty}
    
    \begin{abstract}
    This Internet Appendix collects the detailed proofs for all the theoretical statements in the main text, the data description and additional simulation results.

    \vspace{1cm}
	\noindent\textbf{Keywords:} Factor Analysis, Principal Components, Transfer Learning, Multiple Data Sets, Large-Dimensional Panel Data, Large $N$ and $T$, Missing Data, Weak Factors, Causal Inference
	\noindent\textbf{JEL classification:} C14, C38, C55, G12
	\end{abstract}

\end{titlepage}

\newpage

    \tableofcontents

\setcounter{page}{1}
\begin{onehalfspacing}
\setcounter{section}{0}
\renewcommand\thesection{IA.\Alph{section}}

\setcounter{table}{0}
\setcounter{figure}{0}
\renewcommand{\thetable}{IA.\arabic{table}}
\renewcommand{\thefigure}{IA.\arabic{figure}}

\section{Data for Empirical Study}

\begin{table}[H] \small
    \tcaptab{Selected target time series from the interest and exchange rates category}
    \centering
    \begin{tabular}{ll} 
    \toprule
        abbreviation \qquad \qquad & description  \\
        \midrule
        FEDFUNDS & Effective Federal Funds Rate\\
        TB3MS    & 3-Month Treasury Bill \\
        TB6MS    & 6-Month Treasury Bill \\
        GS1      & 1-Year Treasury Rate  \\
        GS5      & 5-Year Treasury Rate  \\
        GS10     & 10-Year Treasury Rate \\
        AAA      & Moody's Seasoned Aaa Corporate Bond Yield \\
        BAA      & Moody's Seasoned Baa Corporate Bond Yield\\
        TB3SMFFM & 3-Month Treasury C Minus FEDFUNDS \\
        TB6SMFFM & 6-Month Treasury C Minus FEDFUNDS \\
        T1YFFM & 1-Year Treasury C Minus FEDFUNDS \\
        T5YFFM & 5-Year Treasury C Minus FEDFUNDS\\
        T10YFFM & 10-Year Treasury C Minus FEDFUNDS\\
        AAAFFM & Moody's Aaa Corporate Bond Minus FEDFUNDS \\
        BAAFFM & Moody's Baa Corporate Bond Minus FEDFUNDS \\
        EXSZUSx & Switzerland / U.S. Foreign Exchange Rate\\
        EXJPUSx & Japan / U.S. Foreign Exchange Rate \\
        EXUSUKx & U.S. / U.K. Foreign Exchange Rate\\
        EXCAUSx & Canada / U.S. Foreign Exchange Rate\\
        \bottomrule
    \end{tabular}
    \label{tab: empirical 1 target}
    \floatfoot{This table lists the selected target time series from the interest and exchange rates category of the FRED-MD data. Further details are in the appendix of \cite{FRED_MD}. The time series selected here are the ones without missing values from 01/1960 to 12/2020.}
\end{table}

\begin{table}[H] 
\tcaptab{Quarterly-observed target time series from the national income \& product accounts and the flow of funds categories}
\label{tab: flow of funds} 
      \renewcommand{\arraystretch}{1.1}
\centering
\begin{adjustbox}{max height=1.1\textheight}
\renewcommand{\arraystretch}{0.75}
        \setlength{\tabcolsep}{1.8pt}
        {\scriptsize
\begin{tabular}{l}
\toprule
Description \\
\midrule
Real Gross National Product \\
Gross National Product \\
Real Gross Domestic Product \\
Gross Domestic Product \\
Gross Value Added: GDP: Business \\
Gross Domestic Product: Services \\
Gross Domestic Product: Goods \\
Gross Domestic Product: Durable Goods \\
Gross Domestic Product: Nondurable Goods \\
Gross Domestic Purchases \\
Gross Domestic Product: Terms of Trade Index \\
Gross Domestic Product: Research and Development \\
All Sectors; Corporate and Foreign Bonds; Liability \\
All Sectors; Total Debt Securities; Liability \\
All Sectors; U.S. Government Agency Securities; Liability \\
Banks in U.S.-Affiliated Areas; Debt Securities; Asset \\
Domestic Nonfinancial Sectors; Debt Securities; Asset \\
Domestic Financial Sectors; Corporate Equities; Liability \\
Domestic Financial Sectors; U.S. Direct Investment Abroad; Asset \\
Domestic Financial Sectors; Debt Securities; Liability \\
Domestic Financial Sectors; Debt Securities; Asset \\
Domestic Financial Sectors; Total Liabilities and Equity \\
Domestic Financial Sectors; Short-Term Loans Including Security Repurchase Agreements; Liability \\
Domestic Financial Sectors; Total Assets (Does Not Include Land) \\
Domestic Financial Sectors; Total Currency and Deposits; Asset \\
Domestic Financial Sectors; Total Financial Assets \\
Federal Government; Checkable Deposits and Currency; Asset\\
Federal Government; Debt Securities; Liability \\
Federal Government Retirement Funds; Debt Securities; Asset \\
Federal Government; Total Financial Assets \\
Federal Government; Total Liabilities \\
Federal Government; Total Assets (Does Not Include Land) \\
Federal Government; Treasury Securities; Liability \\
Federal Government; Total Mortgages; Asset \\
Finance Companies; Debt Securities; Liability \\
GSEs and Agency- and GSE-Backed Mortgage Pools; U.S. Government Agency Securities; Liability \\
Households and Nonprofit Organizations; Corporate Equities; Asset, Market Value Levels \\
Mutual Funds; Debt Securities; Asset (Market Value) \\
Nonfinancial Corporate Business; Corporate Equities; Liability, Market Value Levels \\
Nonfinancial Noncorporate Business; Real Estate at Market Value, Market Value Levels \\
Rest of the World; Corporate Bonds; Asset, Transactions \\
Rest of the World; U.S. Corporate Equities; Asset, Transactions \\
Rest of the World; U.S. Corporate Equities; Asset \\
Rest of the World; Foreign Direct Investment in U.S.; Asset (Current Cost) \\
Rest of the World; Total Financial Assets \\
Rest of the World; Total Liabilities and Equity \\
Rest of the World; Treasury Securities; Asset \\
Security Brokers and Dealers; Debt Securities; Asset \\
State and Local Governments; Debt Securities; Asset \\
State and Local Governments; U.S. Government Loans; Liability\\
State and Local Governments; Total Liabilities \\
State and Local Governments; Municipal Securities; Liability \\
State and Local Governments; Total Currency and Deposits; Asset \\
State and Local Governments; Total Financial Assets \\
State and Local Governments; Total Mortgages; Asset \\
State and Local Governments; Net Worth (IMA) \\
State and Local Governments; Trade Payables; Liability, Transactions \\
State and Local Governments; Treasury Securities, Including SLGS; Asset \\
\bottomrule
\end{tabular}
}
\end{adjustbox}
    \floatfoot{This table lists the selected quarterly target time series from the national income \& product accounts and the flow of funds categories from the FRED database.}
\end{table}

\section{Additional Simulation Results}

\begin{table}[H] \small
\centering
\begin{threeparttable}
\tcaptab{Relative MSE for missing-at-random pattern}
\begin{tabular}{cccccc}
\toprule
Observation Pattern of $Y$ & \ \ $\mathcal{M}$ \ \ &\ \ T-PCA \ \ &\ \ $\text{XP}_Y$ \ \ &\ \ $\text{XP}_{Z^{(1)}}$ \ \ &\ \ SE-PCA \ \ \\
\midrule
\multirow{3}{*}{$p=0.3$, $\sigma_{\eX}=16$} 
 &obs &0.720 &0.760 &1.088 &1.015 \\
&miss &0.727 &0.787 &1.139 &1.132 \\
&all  &0.725 &0.778 &1.124 &1.096 \\
\midrule
\multirow{3}{*}{$p=0.5$, $\sigma_{\eX}=16$}
&obs  &0.396 &0.408 &0.932 &0.559\\
&miss &0.398 &0.414 &0.964 &0.596\\
&all  &0.397 &0.411 &0.948 &0.578\\
\midrule
\multirow{3}{*}{$p=0.7$, $\sigma_{\eX}=16$} 
&obs  &0.268 &0.272 &0.867 &0.383\\
&miss &0.269 &0.275 &0.891 &0.401\\
&all  &0.268 &0.273 &0.875 &0.388\\
\midrule[1.1pt]
\multirow{3}{*}{$p=0.3$, $\sigma_{\eX}=4$} 
 &obs &0.440 &0.760 &0.440 &0.946 \\
&miss &0.430 &0.787 &0.430 &1.049 \\
&all  &0.433 &0.778 &0.433 &1.018 \\
\midrule
\multirow{3}{*}{$p=0.5$, $\sigma_{\eX}=4$}
&obs  &0.256 &0.408 &0.256 &0.538\\
&miss &0.253 &0.414 &0.253 &0.573\\
&all  &0.255 &0.411 &0.255 &0.556\\
\midrule
\multirow{3}{*}{$p=0.7$, $\sigma_{\eX}=4$} 
&obs  &0.183 &0.272 &0.183 &0.373\\
&miss &0.182 &0.275 &0.182 &0.391\\
&all  &0.183 &0.273 &0.183 &0.379\\
\midrule[1.1pt]
\multirow{3}{*}{$p=0.3$, $\sigma_{\eX}=1$} 
 &obs &0.325 &0.760 &0.397 &0.927 \\
&miss &0.314 &0.787 &0.386 &1.028 \\
&all  &0.317 &0.778 &0.390 &0.998 \\
\midrule
\multirow{3}{*}{$p=0.5$, $\sigma_{\eX}=1$}
&obs  &0.184 &0.408 &0.224 &0.530\\
&miss &0.182 &0.414 &0.220 &0.564\\
&all  &0.183 &0.411 &0.222 &0.547\\
\midrule
\multirow{3}{*}{$p=0.7$, $\sigma_{\eX}=1$} 
&obs  &0.127 &0.272 &0.158 &0.368\\
&miss &0.127 &0.275 &0.157 &0.385\\
&all  &0.127 &0.273 &0.158 &0.373\\
\bottomrule
\end{tabular}
\floatfoot{This table reports the relative MSE of different estimation methods for missing-at-random and different parameter choices. We compare T-PCA (our benchmark method), $\text{XP}_Y$ (PCA on $Y$), $\text{XP}_{Z^{(1)}}$ (PCA on concatenated panel) and SE-PCA (separate PCA). We generate a two-factor model as follows: Factors $F_t \overset{i.i.d.}{\sim} \Ncal(0,I_2)$, loadings $(\LamX)_i \overset{i.i.d.}{\sim} \Ncal(0,I_2), (\LamY)_i \overset{i.i.d.}{\sim} \Ncal(0,I_2)$, and errors $(\eX)_{ti}\overset{i.i.d.}{\sim} \Ncal(0,\sigma_{\eX}^2), (\eY)_{ti} \overset{i.i.d.}{\sim} \Ncal(0,\sigma_{\eY}^2)$ with different $\sigma_{\eX}$ and fixed $\sigma_{\eY}=4.$ The entries of $Y$ are missing independently at random with observation probability $p$. We set $\Nx=\Ny=T=200$ and run 200 simulations for each setup.}
\end{threeparttable}
\end{table}

\begin{table}[H] \small
\centering
\begin{threeparttable}
\tcaptab{Relative MSE for low-frequency observation pattern}
\begin{tabular}{cccccc}
\toprule
Observation Pattern of $Y$ & \ \ $\mathcal{M}$ \ \ &\ \ T-PCA \ \ &\ \ $\text{XP}_Y$ \ \ &\ \ $\text{XP}_{Z^{(1)}}$ \ \ &\ \ SE-PCA \ \ \\
\midrule
\multirow{3}{*}{$p=0.3$, $\sigma_{\eX}=16$} 
 &obs &0.450 &- &1.096 &1.160 \\
&miss &1.058 &- &1.346 &1.234 \\
&all  &0.872 &- &1.264 &1.210 \\
\midrule
\multirow{3}{*}{$p=0.5$, $\sigma_{\eX}=16$}
&obs  &0.291 &- &0.846 &1.059\\
&miss &1.029 &- &1.119 &1.104\\
&all  &0.656 &- &0.979 &1.080\\
\midrule
\multirow{3}{*}{$p=0.7$, $\sigma_{\eX}=16$} 
&obs  &0.227 &- &0.816 &1.016\\
&miss &1.011 &- &1.033 &1.051\\
&all  &0.460 &- &0.879 &1.026\\
\midrule[1.1pt]
\multirow{3}{*}{$p=0.3$, $\sigma_{\eX}=8$} 
 &obs &0.431 &- &0.537 &0.591 \\
&miss &0.565 &- &0.611 &0.622 \\
&all  &0.522 &- &0.586 &0.611 \\
\midrule
\multirow{3}{*}{$p=0.5$, $\sigma_{\eX}=8$}
&obs  &0.273 &- &0.329 &0.480\\
&miss &0.476 &- &0.491 &0.502\\
&all  &0.373 &- &0.408 &0.490\\
\midrule
\multirow{3}{*}{$p=0.7$, $\sigma_{\eX}=8$} 
&obs  &0.210 &- &0.258  &0.436\\
&miss &0.434 &- &0.436  &0.450\\
&all  &0.276 &- &0.310  &0.440\\
\midrule[1.1pt]
\multirow{3}{*}{$p=0.3$, $\sigma_{\eX}=4$} 
 &obs &0.374 &- &0.372 &0.355 \\
&miss &0.363 &- &0.366 &0.369 \\
&all  &0.365 &- &0.367 &0.364 \\
\midrule
\multirow{3}{*}{$p=0.5$, $\sigma_{\eX}=4$}
&obs  &0.230 &- &0.230 &0.243\\
&miss &0.256 &- &0.256 &0.250\\
&all  &0.242 &- &0.242 &0.246\\
\midrule
\multirow{3}{*}{$p=0.7$, $\sigma_{\eX}=4$} 
&obs  &0.173 &- &0.173 &0.196\\
&miss &0.208 &- &0.208 &0.202\\
&all  &0.183 &- &0.183 &0.198\\
\bottomrule
\end{tabular}
\floatfoot{This table reports the relative MSE of different estimation methods for a low-frequency observation pattern and different parameter choices. We compare T-PCA (our benchmark method), $\text{XP}_Y$ (PCA on $Y$), $\text{XP}_{Z^{(1)}}$ (PCA on concatenated panel) and SE-PCA (separate PCA). We generate a two-factor model as follows: Factors $F_t \overset{i.i.d.}{\sim} \Ncal(0,I_2)$, loadings $(\LamX)_i \overset{i.i.d.}{\sim} \Ncal(0,I_2), (\LamY)_i \overset{i.i.d.}{\sim} \Ncal(0,I_2)$, and errors $(\eX)_{ti}\overset{i.i.d.}{\sim} \Ncal(0,\sigma_{\eX}^2)$ and $(\eY)_{ti} \overset{i.i.d.}{\sim} \Ncal(0,\sigma_{\eY}^2)$ with different $\sigma_{\eX}$ and fixed $\sigma_{\eY}=4.$ The rows of $Y$ are missing-at-random with observation probability $p$. We set $\Nx=\Ny=T=200$ and run 200 simulations for each setup. Note that in this setting, $\text{XP}_Y$ is not applicable and separate-PCA degenerates to PCA using only the panel $X$.}
\end{threeparttable}
\end{table}

\begin{table}[H] \small
\centering
\begin{threeparttable}
\caption{Relative MSE for missingness depends on loadings pattern}
\begin{tabular}{cccccc}
\toprule
Observation Pattern of $Y$ & \ \ $\mathcal{M}$ \ \ &\ \ T-PCA \ \ &\ \ $\text{XP}_Y$ \ \ &\ \ $\text{XP}_{Z^{(1)}}$ \ \ &\ \ SE-PCA \ \ \\
\midrule
\multirow{3}{*}{$p_1=0.2$, $\sigma_{\eX}=8$} 
 &obs &1.014 &1.104 &1.161 &1.384 \\
&miss &1.097 &1.286 &1.213 &1.738 \\
&all  &1.077 &1.243 &1.200 &1.655 \\
\midrule
\multirow{3}{*}{$p_1=0.4$, $\sigma_{\eX}=8$}
&obs  &0.466 &0.503 &0.629 &0.675\\
&miss &0.473 &0.522 &0.636 &0.740\\
&all  &0.470 &0.513 &0.633 &0.712\\
\midrule
\multirow{3}{*}{$p_1=0.6$, $\sigma_{\eX}=8$} 
&obs  &0.305 &0.322 &0.451  &0.444\\
&miss &0.303 &0.324 &0.452  &0.465\\
&all  &0.304 &0.323 &0.451  &0.452\\
\midrule[1.1pt]
\multirow{3}{*}{$p_1=0.2$, $\sigma_{\eX}=6$} 
&obs  &0.499 &0.545 &0.614 &0.688 \\
&miss &0.546 &0.641 &0.643 &0.866 \\
&all  &0.535 &0.619 &0.636 &0.824 \\
\midrule
\multirow{3}{*}{$p_1=0.4$, $\sigma_{\eX}=6$}
&obs  &0.245 &0.264 &0.319 &0.354\\
&miss &0.251 &0.277 &0.325 &0.389\\
&all  &0.249 &0.272 &0.323 &0.374\\
\midrule
\multirow{3}{*}{$p_1=0.6$, $\sigma_{\eX}=6$} 
&obs  &0.162 &0.172 &0.216 &0.238\\
&miss &0.162 &0.174 &0.217 &0.249\\
&all  &0.162 &0.173 &0.216 &0.242\\
\midrule[1.1pt]
\multirow{3}{*}{$p_1=0.2$, $\sigma_{\eX}=4$} 
&obs  &0.219 &0.238 &0.262 &0.280\\
&miss &0.252 &0.293 &0.287 &0.356\\
&all  &0.244 &0.280 &0.281 &0.338\\
\midrule
\multirow{3}{*}{$p_1=0.4$, $\sigma_{\eX}=4$}
&obs  &0.111 &0.119 &0.132 &0.150\\
&miss &0.117 &0.129 &0.138 &0.165\\
&all  &0.114 &0.125 &0.135 &0.158\\
\midrule
\multirow{3}{*}{$p_1=0.6$, $\sigma_{\eX}=4$} 
&obs  &0.072 &0.077 &0.087 &0.102\\
&miss &0.074 &0.080 &0.089 &0.107\\
&all  &0.073 &0.078 &0.087 &0.104\\
\bottomrule
\end{tabular}
\floatfoot{This table reports the relative MSE of different estimation methods when missingness depends on the loadings and for different parameter choices. We compare T-PCA (our benchmark method), $\text{XP}_Y$ (PCA on $Y$), $\text{XP}_{Z^{(1)}}$ (PCA on concatenated panel) and SE-PCA (separate PCA). We generate a two-factor model as follows: Factors $F_t \overset{i.i.d.}{\sim} \Ncal(0,I_2)$, loadings $(\LamY)_i \overset{i.i.d.}{\sim} \Ncal(0,I_2), (\LamX)_{i2}\overset{i.i.d.}{\sim} \Ncal(0,1),$ and errors $(\eX)_{ti}\overset{i.i.d.}{\sim} \Ncal(0,\sigma_{\eX}^2)$ and $(\eY)_{ti} \overset{i.i.d.}{\sim} \Ncal(0,\sigma_{\eY}^2)$. We assume $(\LamX)_{i1}=0$ and $\sigma_{\eY} = 0.5\cdot\sigma_{\eX}.$ Furthermore, we define a unit-specific characteristic $S_i = \mathbbm{1}(|(\LamY)_{i,2}|>0.1)$. The entries of $Y$ are missing independently with observation probability $p_1$ if $S_i=1$ and $p_2=1$ if $S_i=0.$ We set $\Nx=\Ny=T=200$ and run 200 simulations for each setup.}
\end{threeparttable}
\end{table}

\section{Proofs}

\subsection{Proof of Theorem \ref{thm: consistency for loadings}}
According to Assumptions \ref{assump: factor model}.2, $\Sigma_{\LamX}$ and $\Sigma_{{\LamY},t}$ are positive semi-definite, and $\Sigma_{\LamX} + \Sigma_{{\LamY},t}$ is positive definite. Therefore, when $\gamma = r \cdot \Nx/\Ny$ with some positive constant $r$, the weighted second moment matrix
\[
\Sigma_{\Lam,t}^{(\gamma)}=\lim \frac{\Nx}{\Nx+\Ny}\Lp\Sigma_{\LamX} + \gamma \frac{\Ny}{\Nx}\cdot \Sigma_{{\LamY},t}\Rp = \lim \frac{\Nx}{\Nx+\Ny}\Lp\Sigma_{\LamX} + r\cdot \Sigma_{{\LamY},t}\Rp
\]
is positive definite. Similarly, we can show that the matrix $\Sigma_{\Lam}^{(\gamma)}=\lim \frac{\Nx}{\Nx+\Ny}\Lp\Sigma_{\LamX} + r\cdot \Sigma_{{\LamY}}\Rp$ is positive definite. By Assumption \ref{assump: factor shift}, $\Sigma_{\LamX}$ is not positive definite. When $\gamma \neq r\cdot \Nx/\Ny$ for any constant $r$ and $\Ny/\Nx\rightarrow 0$, the second moment matrix $\Sigma_{\Lam,t}^{(\gamma)}$ is not positive definite as well.

In the following, we prove that when $\gamma = r\cdot \Nx/\Ny$, we can consistently estimate the loadings, factors, and thus the common components of $Y$. Based on the definition of target-PCA in Section \ref{sec: estimator}, we can plug the expression of $\tilde{\Sigma}^{Z^{(\gamma)}}$ into $\frac{1}{\Nx+\Ny } \tilde{\Sigma}^{Z^{(\gamma)}} \tilde{\Lam}^{(\gamma)} = \tilde{\Lam}^{(\gamma)} \tilde{D}^{(\gamma)}$ and obtain the decomposition for the estimated combined loadings $\tilde{\Lam}^{(\gamma)}$ as
\[
\begin{aligned}
\tilde{\Lam}_i^{(\gamma)} &= \frac{1}{\Nx+\Ny } (\tilde{D}^{(\gamma)})^{-1} \sum_{j=1}^{\Nx+\Ny }  \tilde{\Lam}^{(\gamma)}_j \tilde{\Sigma}_{ji}^{Z^{(\gamma)}} \\
& = H_i^{(\gamma)}{\Lam}_i^{(\gamma)} + \underbrace{ \frac{1}{\Nx+\Ny} (\tilde{D}^{(\gamma)})^{-1}   \sum_{j=1}^{\Nx+\Ny}  \tilde{\Lam}^{(\gamma)}_j  \Lam_i^{(\gamma)\top} \frac{1}{|Q^Z_{ij}|}\sum_{t\in Q^Z_{ij}} F_t e_{tj}^{(\gamma)}}_{(a)} \\
& \ \  \qquad \qquad \ +\underbrace{\frac{1}{\Nx+\Ny} (\tilde{D}^{(\gamma)})^{-1} \sum_{j=1}^{\Nx+\Ny}  \tilde{\Lam}_j^{(\gamma)}  \Lam_j^{(\gamma)\top} \frac{1}{|Q^Z_{ij}|}\sum_{t\in Q^Z_{ij}} F_t e_{ti}^{(\gamma)}}_{(b)}\\
& \qquad \qquad \ \ \ +\underbrace{\frac{1}{\Nx+\Ny} (\tilde{D}^{(\gamma)})^{-1}  \sum_{j=1}^{\Nx+\Ny}  \tilde{\Lam}_j^{(\gamma)}   \frac{1}{|Q^Z_{ij}|}\sum_{t\in Q^Z_{ij}} e_{tj}^{(\gamma)} e_{ti}^{(\gamma)}}_{(c)},
\end{aligned}
\]
where $H_i^{(\gamma)} = (\tilde{D}^{(\gamma)})^{-1} \frac{1}{\Nx+\Ny}  \sum_{j=1}^{\Nx+\Ny}  \tilde{\Lam}_j^{(\gamma)}  \Lam_j^{(\gamma)\top} \frac{1}{|Q^Z_{ij}|}\sum_{t\in Q^Z_{ij}} F_t F_t^\top.$
As is shown, the estimated combined loadings are related to the true combined loadings through $\tilde{\Lam}_i^{(\gamma)} = H_i^{(\gamma)}{\Lam}_i^{(\gamma)} + (a) + (b)+(c)$ up to the rotation matrix $H_i^{(\gamma)}$ for index $i.$ Since the estimation of factors requires the same rotation matrix for all $\Lam_i^{(\gamma)}$, we consider the unified rotation matrix $H^{(\gamma)}$ defined as 
\begin{align*}
H^{(\gamma)} = \frac{1}{T(\Nx+\Ny)}(\tilde{D}^{(\gamma)})^{-1} \tilde{\Lam}^{(\gamma)\top} \Lam^{(\gamma)} F^\top F.  
\end{align*}
This yields the decomposition $\tilde{\Lam}_i^{(\gamma)} = H^{(\gamma)}{\Lam}_i^{(\gamma)} + (H^{(\gamma)}_i - H^{(\gamma)}){\Lam}_i^{(\gamma)} + (a) + (b)+(c)$, based on which we derive the consistency result of the estimated loadings.

To simplify notation, we define the following four terms
\begin{align*}
&\eta_{ij} = \frac{1}{|Q^Z_{ij}|} \sum_{t \in Q^Z_{ij}}  \Lam_i^{(\gamma)\top} F_t e_{tj}^{(\gamma)}, 
\qquad \quad \xi_{ij} = \frac{1}{|Q^Z_{ij}|} \sum_{t \in Q^Z_{ij}}  \Lam_j^{(\gamma)\top} F_t e_{ti}^{(\gamma)},\\
&\gamma(i,j) =\frac{1}{|Q^Z_{ij}|} \sum_{t \in Q^Z_{ij}} \E[e_{ti}^{(\gamma)} e_{tj}^{(\gamma)}], 
\qquad \ \ \zeta_{ij} = \frac{1}{|Q^Z_{ij}|} \sum_{t \in Q^Z_{ij}} e^{(\gamma)}_{ti}e^{(\gamma)}_{tj} - \gamma(i,j).
\end{align*}
We omit the superscript $(\gamma)$ in these four terms to save space. It holds that
\begin{align*}
  \tilde{\Lam}_i^{(\gamma)} - H_i^{(\gamma)} \Lam_i^{(\gamma)} = \frac{1}{\Nx+\Ny} (\tilde{D}^{(\gamma)})^{-1} \sum_{j=1}^{\Nx+\Ny}   \left( \tilde{\Lam}_j^{(\gamma)} \eta_{ij} + \tilde{\Lam}_j^{(\gamma)} \xi_{ij} + \tilde{\Lam}_j^{(\gamma)} \zeta_{ij} + \tilde{\Lam}_j^{(\gamma)} \gamma(i,j) \right).
\end{align*}
We provide bounds for $\eta_{ij},\xi_{ij},\zeta_{ij}$ and $\gamma(i,j)$ in the following.

\begin{lemma} \label{lemma: order of the four terms}
Under Assumption \ref{assump: obs pattern} and Assumption \ref{assump: factor model}, suppose $\Ny /\Nx\rightarrow c\in [0,\infty)$ and $\gamma = r\cdot \Nx/\Ny$, then as $T, \Nx, \Ny \rightarrow \infty$, we have 
\begin{enumerate}
    \item $\frac{1}{(\Nx+\Ny)^2}\sum_{i,j=1}^{\Nx+\Ny}\eta_{ij}^2 = O_p\Lp\frac{1}{\Talpha}\Rp$;
    \item $\frac{1}{(\Nx+\Ny)^2}\sum_{i,j=1}^{\Nx+\Ny}\xi_{ij}^2 = O_p\Lp\frac{1}{\Talpha}\Rp$;
    \item $\frac{1}{(\Nx+\Ny)^2}\sum_{i,j=1}^{\Nx+\Ny}\zeta_{ij}^2 = O_p\Lp\frac{1}{\Talpha}\Rp;$
    \item $\frac{1}{(\Nx+\Ny)^2}\sum_{i,j=1}^{\Nx+\Ny}\gamma^2(i,j) \leq  \frac{C}{\Ny}$.
\end{enumerate}
\end{lemma}
\begin{proof}
1. By Assumptions \ref{assump: obs pattern}, \ref{assump: factor model}.2 and \ref{assump: factor model}.4, there is 
\begin{align*}
\E \Ls\frac{1}{(\Nx+\Ny)^2} \sum_{i,j=1}^{\Nx+\Ny} \eta^2_{ij} \Rs &= \frac{1}{(\Nx+\Ny)^2}\sum_{i,j=1}^{\Nx+\Ny} \E \Ls\Lam_i^{(\gamma)\top} \frac{1}{|Q^Z_{ij}|}\sum_{t\in Q^Z_{ij}} F_t \ee_{tj} \Rs^2  \\
& \leq \frac{1}{(\Nx+\Ny)^2}\sum_{i,j=1}^{\Nx+\Ny} \E\Lv\Lam_i^{(\gamma)}\Rv^2 \cdot \E\left\| \frac{1}{|Q^Z_{ij}|}\sum_{t\in Q^Z_{ij}} F_t \ee_{tj} \right\|^2 \\
& \leq \frac{C}{\Talpha}.
\end{align*}
 As a result, it holds that $\frac{1}{(\Nx+\Ny)^2}\sum_{i,j}\eta_{ij}^2 = O_p\Lp\frac{1}{\Talpha}\Rp$. \\
2. By the same arguments, we can show that $\frac{1}{(\Nx+\Ny)^2}\sum_{i,j}\xi_{ij}^2 = O_p\Lp\frac{1}{\Talpha}\Rp$. \\
3. Following from Assumption \ref{assump: factor model}.3(e), it holds that
\begin{align*}
\E \Ls\frac{1}{(\Nx+\Ny)^2} \sum_{i,j=1}^{\Nx+\Ny} \zeta^2_{ij} \Rs &= \frac{1}{(\Nx+\Ny)^2}\sum_{i,j=1}^{\Nx+\Ny} \E \left [ \frac{1}{|Q_{ij}^Z|}\sum_{t \in Q^Z_{ij}}\Lp\ee_{ti}\ee_{tj} - \E(\ee_{ti} \ee_{tj})\Rp    \right]^2  \\
& \leq  \frac{C}{(\Nx+\Ny)^2} \left( \sum_{i,j=\Nx+1}^{\Nx+\Ny} \gamma^2 \frac{1}{|Q^Z_{ij}|}+ \sum_{i=1}^{\Nx} \sum_{j=\Nx+1}^{\Nx+\Ny} \gamma \frac{1}{|Q^Z_{ij}|} + \sum_{i,j=1}^{\Nx} \frac{1}{|Q^Z_{ij}|}  \right) \\
& \leq \frac{C}{\Talpha}.
\end{align*}
Therefore, $\frac{1}{(\Nx+\Ny)^2}\sum_{i,j}\zeta_{ij}^2 = O_p\Lp\frac{1}{\Talpha}\Rp$ as claimed.\\
4. By definition, we have
\begin{align*}
\frac{1}{(\Nx+\Ny)^2} \sum_{i,j=1}^{\Nx+\Ny} \gamma^2(i,j) = \frac{1}{(\Nx+\Ny)^2} \sum_{i,j=1}^{\Nx+\Ny} \frac{1}{|Q^Z_{ij}|^2}\left(\sum_{t\in Q^Z_{ij}} \E\Ls\ee_{ti}\ee_{tj}\Rs\right)^2.
\end{align*}
According to Assumption \ref{assump: factor model}.3(c), the RHS of the above equation can be bounded by
\begin{align*}
\text{RHS} \leq &\frac{1}{(\Nx+\Ny)^2}  \sum_{i,j=1}^{\Ny} \gamma^2 
(\tau_{ij}^{(\eY)})^2
+ \frac{2}{(\Nx+\Ny)^2} \sum_{i=1}^{\Nx} \sum_{j=1}^{\Ny} \gamma (\tau_{ij}^{(\eY,\eX)})^2  + \frac{1}{(\Nx+\Ny)^2}\sum_{i,j=1}^{\Nx}  (\tau_{ij}^{(\eX)})^2 \\
\leq &  \frac{C}{(\Nx+\Ny)^2} \left( \Ny\frac{\Nx^2}{\Ny^2} + \Nx \frac{\Nx}{\Ny} + \Nx
\right)\\
\leq & \frac{C}{\Ny}.
\end{align*}
\end{proof}

\vspace{5mm}
\begin{lemma}   \label{lemma: D} 
Under Assumption \ref{assump: obs pattern} and Assumption \ref{assump: factor model}, suppose $\Ny/\Nx\rightarrow c\in [0,\infty)$ and $\gamma = r\cdot \Nx/\Ny$. Then as $T, \Nx, \Ny \rightarrow \infty$, we have
\begin{enumerate}
    \item $\frac{1}{(\Nx+\Ny)^2} \tilde{\Lam}^{(\gamma)\top} \left(  (\tilde{Z}^{(\gamma)\top} \tilde{Z}^{(\gamma)} ) \odot \left[ \frac{1}{|Q^Z_{ij}|} \right] \right) \tilde{\Lam}^{(\gamma)} = \tilde{D}^{(\gamma)} \overset{p}{\rightarrow} D^{(\gamma)}$;
    \item $\frac{1}{(\Nx+\Ny)^2} \tilde{\Lam}^{(\gamma)\top} \left(  ((F \Lam^{(\gamma)\top}) \odot W^Z)^\top(( F\Lam^{(\gamma)\top}) \odot W^{Z}) \odot \left[ \frac{1}{|Q^Z_{ij}|} \right] \right) \tilde{\Lam}^{(\gamma)} \overset{p}{\rightarrow} D^{(\gamma)}$;
    \item $\frac{1}{(\Nx+\Ny)^2} \tilde{\Lam}^{(\gamma)\top} \left( \Lam^{(\gamma)}  \frac{F^\top F}{T}  \Lam^{(\gamma)\top} \right) \tilde{\Lam}^{(\gamma)}  \overset{p}{\rightarrow} D^{(\gamma)}$;
\end{enumerate}
\noindent where $\tilde{Z}^{(\gamma)} = Z^{(\gamma)} \odot W^Z$ and $D^{(\gamma)} = \text{diag}\Lp d_1^{(\gamma)}, \cdots, d_k^{(\gamma)}\Rp$ are the eigenvalues of $ \Sigma_F \Sigma_{\Lam}^{(\gamma)}$.
\end{lemma}

\begin{proof}
Let $\lambda \in \R^{(\Nx+\Ny)\times 1}$. Define $\Gamma = \{ \lambda| \lambda^\top \lambda =\Nx+\Ny\}$, and let
\begin{align*}
    &R(\lambda) = \frac{1}{(\Nx+\Ny)^2} \lambda^\top  \left(  (\tilde{Z}^{(\gamma)\top} \tilde{Z}^{(\gamma)} ) \odot \left[ \frac{1}{|Q^Z_{ij}|} \right] \right) \lambda, \\
    &\tilde{R}(\lambda) = \frac{1}{(\Nx+\Ny)^2} \lambda^{(\gamma)\top} \left(  ((F \Lam^{(\gamma)\top}) \odot W^Z)^\top(( F\Lam^{(\gamma)\top}) \odot W^Z) \odot \left[ \frac{1}{|Q^Z_{ij}|} \right] \right) \lambda, \\
    &R^*(\lambda) = \frac{1}{(\Nx+\Ny)^2} \lambda^\top \left( \Lam^{(\gamma)} \frac{F^\top F}{T}  \Lam^{(\gamma)\top} \right) \lambda.
\end{align*}
First of all, we prove that as $T, \Nx, \Ny \rightarrow \infty$,
\begin{align*}
&(1) = \sup_{\lambda \in \Gamma} \frac{1}{(\Nx+\Ny)^2} \left| \lambda^\top \left(  ((\ee)^\top \odot (W^Z)^\top)(\ee \odot W^Z) \odot \left[ \frac{1}{|Q^Z_{ij}|} \right] \right) \lambda  \right| \overset{p}{\rightarrow} 0, \\
&(2) = \sup_{\lambda \in \Gamma}  \frac{1}{(\Nx+\Ny)^2} \left| \lambda^\top \left(  ((\ee)^\top \odot (W^Z)^\top))(F\Lam^{(\gamma)\top} \odot W^Z) \odot \left[ \frac{1}{|Q^Z_{ij}|} \right] \right) \lambda \right| \overset{p}{\rightarrow} 0.
\end{align*}
Observe that
\begin{align*}
    (1) &= \sup_{\lambda \in \Gamma}\frac{1}{(\Nx+\Ny)^2} \left|\sum_{i,j=1}^{\Nx+\Ny}  \lambda_i\lambda_j \frac{1}{|Q^Z_{ij}|}\sum_{t\in Q^Z_{ij}}\ee_{ti}\ee_{tj} \right| \\
    & \leq \sup_{\lambda \in \Gamma} \underbrace{\left(\frac{1}{(\Nx+\Ny)^2}\sum_{i,j=1}^{\Nx+\Ny}\lambda_i^2 \lambda_j^2\right)^{1/2}}_{=1}\cdot \left(\frac{1}{(\Nx+\Ny)^2}\sum_{i,j=1}^{\Nx+\Ny} \Bigg(\frac{1}{|Q^Z_{ij}|}\sum_{t\in Q^Z_{ij}}\ee_{ti}\ee_{tj}\Bigg)^2\right)^{1/2} \\
    & =  \left(\frac{1}{(\Nx+\Ny)^2}\sum_{i,j=1}^{\Nx+\Ny} (\zeta_{ij}+\gamma(i,j))^2\right)^{1/2}.
\end{align*}
According to Lemma \ref{lemma: order of the four terms}, $\frac{1}{(\Nx+\Ny)^2}\sum_{i,j=1}^{\Nx+\Ny}\zeta_{ij}^2 = o_p(1)$ and $ \frac{1}{(\Nx+\Ny)^2}\sum_{i,j=1}^{\Nx+\Ny}\gamma^2(i,j) = o(1)$. As a result, $(1)\overset{p}{\rightarrow} 0.$ Consider (2), we have
\begin{align*}
    (2) & = \sup_{\lambda \in \Gamma}\frac{1}{(\Nx+\Ny)^2}\left| \sum_{i,j=1}^{\Nx+\Ny} \lambda_i \lambda_j \frac{1}{|Q^Z_{ij}|}\sum_{t\in Q^Z_{ij}} \ee_{ti}F_t^\top \Lam_j^{(\gamma)} \right| \\
    &\leq \sup_{\lambda \in \Gamma} \underbrace{\left(\frac{1}{(\Nx+\Ny)^2}\sum_{i,j=1}^{\Nx+\Ny}\lambda_i^2 \lambda_j^2\right)^{1/2}}_{=1}\cdot \left(\frac{1}{(\Nx+\Ny)^2}\sum_{i,j=1}^{\Nx+\Ny} \xi_{ij}^2 \right)^{1/2} = o_p(1)
\end{align*}
following from Lemma \ref{lemma: order of the four terms}. Combining these two terms, it holds that
\begin{align*}
\sup_{\lambda \in \Gamma} |R(\lambda) - \tilde{R}(\lambda)| \leq 
& \sup_{\lambda \in \Gamma} \frac{1}{(\Nx+\Ny)^2} \left| \lambda^\top \left(  ((\ee)^\top \odot (W^Z)^\top)(\ee \odot W^Z) \odot \left[ \frac{1}{|Q^Z_{ij}|} \right] \right) \lambda  \right| \\
&+  \sup_{\lambda \in \Gamma} \frac{2}{(\Nx+\Ny)^2} \left| \lambda^\top \left(  ((\ee)^\top \odot (W^Z)^\top))(F\Lam^\top \odot W^Z) \odot \left[ \frac{1}{|Q^Z_{ij}|} \right] \right) \lambda \right| \\
&\overset{p}{\rightarrow} 0.
\end{align*}
Furthermore, for any $\lambda \in \Gamma$,
\begin{align*}
\tilde{R}(\lambda) - R^*(\lambda) &= \frac{1}{(\Nx+\Ny)^2} \sum_{i,j=1}^{\Nx+\Ny} \lambda_i \lambda_j \Lam_i^{(\gamma)\top} \underbrace{\left(  \frac{1}{|Q^Z_{ij}|} \sum_{t \in Q^Z_{ij}} F_t F_t^\top - \frac{1}{T}F^\top F \right)}_{\Delta_{F,ij}} \Lam_j^{(\gamma)} \\
& \leq \underbrace{\left( \frac{1}{(\Nx+\Ny)^2} \sum_{i,j=1}^{\Nx+\Ny} \lambda_i^2 \lambda_j^2 \right)^{1/2}}_{=1}  \cdot \left(  \frac{1}{(\Nx+\Ny)^2} \sum_{i,j=1}^{\Nx+\Ny} \Big( \Lam_i^{(\gamma)\top}  \Delta_{F,ij} \Lam_j^{(\gamma)} \Big)^2 \right)^{1/2}.
\end{align*}
By Assumptions \ref{assump: factor model}.1 and \ref{assump: factor model}.2, we have
\begin{align*}
    &\E \left[\frac{1}{(\Nx+\Ny)^2} \sum_{i,j=1}^{\Nx+\Ny} \Big( \Lam_i^{(\gamma)\top}  \Delta_{F,ij} \Lam_j^{(\gamma)} \Big)^2  \right] \\
    \leq & \frac{1}{(\Nx+\Ny)^2} \sum_{i,j=1}^{\Nx+\Ny} \E\Ls \Lv\Lam_i^{(\gamma)}\Rv^2 \Lv\Lam_j^{(\gamma)}\Rv^2\Rs\cdot \E\| \Delta_{F,ij}\|^2\ \leq \frac{C}{\Talpha}.
\end{align*}
Therefore, $\sup_{\lambda \in \Gamma} |\tilde{R}(\lambda) - R^*(\lambda)| \overset{p}{\rightarrow} 0 $. We can also derive $\sup_{\lambda \in \Gamma} |R(\lambda) - R^*(\lambda)| \overset{p}{\rightarrow} 0 $ by the decomposition $R(\lambda) - R^*(\lambda) = R(\lambda) - \tilde{R}(\lambda) + \tilde{R}(\lambda) - R^*(\lambda)$. As a result, $|\sup_{\lambda \in \Gamma} R(\lambda) - \sup_{\lambda \in \Gamma} R^*(\lambda)| \leq \sup_{\lambda \in \Gamma} |R(\lambda) - R^*(\lambda)| \overset{p}{\rightarrow} 0$. Since $\sup_{\lambda \in \Gamma} R^*(\lambda) \overset{p}{\rightarrow}d_1^{(\gamma)}$ where $d_1^{(\gamma)}$ is the largest eigenvalue of $ \Sigma_F  \Sigma_{\Lam}^{(\gamma)},$ we have $\sup_{\lambda \in \Gamma} R(\lambda) \overset{p}{\rightarrow}d_1^{(\gamma)}$. By definition, $\tilde{\Lam}_1^{(\gamma)} = \arg \sup_{\lambda \in \Gamma} R(\lambda)$, thus $R(\tilde{\Lam}_1^{(\gamma)})=\tilde{d}_1^{(\gamma)} \overset{p}{\rightarrow} d_1^{(\gamma)}$ and $\tilde{R}(\tilde{\Lam}_1^{(\gamma)}), R^*(\tilde{\Lam}_1^{(\gamma)}) \overset{p}{\rightarrow} d_1^{(\gamma)}$. We repeat this procedure sequentially using the orthonormal subspace of $\Gamma$ and finally complete our proof.
\end{proof}

\vspace{5mm}
\begin{lemma}   \label{lemma: H,Q}
Under Assumption \ref{assump: obs pattern} and Assumption \ref{assump: factor model}, suppose $\Ny/\Nx\rightarrow c\in [0,\infty)$ and $\gamma = r\cdot \Nx/\Ny$. For $T, \Nx, \Ny \rightarrow \infty$ it holds that
\begin{enumerate}
\item $\frac{1}{\Nx+\Ny} \Lam^{(\gamma)\top} \tilde{\Lam}^{(\gamma)} \overset{p}{\rightarrow} Q^{(\gamma)} = \Sigma_{F}^{-1/2} \Upsilon^{(\gamma)} (D^{(\gamma)})^{1/2} ,$ where the diagonal entries of $D^{(\gamma)}$ are eigenvalues of $\Sigma_F^{1/2}\Sigma_{\Lam}^{(\gamma)}\Sigma_F^{1/2}$, and $\Upsilon^{(\gamma)}$ is the corresponding eigenvector matrix such that $\Upsilon^{(\gamma)\top} \Upsilon^{(\gamma)} = I$;
\item $H^{(\gamma)} \overset{p}{\rightarrow} (Q^{(\gamma)})^{-1}$, where $H^{(\gamma)} = T^{-1}(\Nx+\Ny)^{-1} (\tilde{D}^{(\gamma)})^{-1} \tilde{\Lam}^{(\gamma)\top} \Lam^{(\gamma)} F^\top F$.
\end{enumerate}

\end{lemma}

\begin{proof}
Left-multiplying $\frac{1}{\Nx+\Ny} \tilde{\Sigma}^{Z^{(\gamma)}} \tilde{\Lam}^{(\gamma)} = \tilde{\Lam}^{(\gamma)} \tilde{D}^{(\gamma)}$ on both sides by $\frac{1}{\Nx+\Ny}(\frac{F^\top F}{T})^{1/2} \Lam^{(\gamma)\top}$, we have
\[
\left( \frac{F^\top F}{T} \right)^{1/2}  \frac{ \Lam^{(\gamma)\top}}{\Nx+\Ny}  \tilde{\Sigma}^{Z^{(\gamma)}} \frac{ \tilde{\Lam}^{(\gamma)} }{\Nx+\Ny}= \left(\frac{ F^\top F }{T} \right)^{1/2} \left( \frac{\Lam^{(\gamma)\top} \tilde{\Lam}^{(\gamma)}}{\Nx+\Ny} \right) \tilde{D}^{(\gamma)}.
\]
This can be rewritten as
\[
\left( \frac{F^\top F}{T} \right)^{1/2}  \frac{ \Lam^{(\gamma)\top}\Lam^{(\gamma)}}{\Nx+\Ny}  \left(\frac{F^\top F}{T}\right) \frac{ \Lam^{(\gamma)\top} \tilde{\Lam}^{(\gamma)} }{\Nx+\Ny} + d_{T,\Ny\Nx}^{(\gamma)}= \left(\frac{ F^\top F }{T} \right)^{1/2} \left( \frac{\Lam^{(\gamma)\top} \tilde{\Lam}^{(\gamma)}}{\Nx+\Ny} \right) \tilde{D}^{(\gamma)},
\]
where 
\[
d_{T, \Nx\Ny}^{(\gamma)} = \left( \frac{F^\top F}{T} \right)^{1/2} \frac{\Lam^{(\gamma)\top}}{\Nx+\Ny} 
\left( \tilde{\Sigma}^{Z^{(\gamma)}} - \Lam^{(\gamma)} \frac{F^\top F}{T} \Lam^{(\gamma)\top}  \right) 
\frac{\tilde{\Lam}^{(\gamma)}}{\Nx+\Ny}.
\]
Plugging the expansion of $\tilde{\Sigma}^{Z^{(\gamma)}}_{ij}$ into $d_{T,\Nx \Ny}^{(\gamma)}$, we obtain for any $i$ and $j$,
\begin{align*}
\left(\tilde{\Sigma}^{Z^{(\gamma)}} - \Lam^{(\gamma)} \frac{F^\top F}{T} \Lam^{(\gamma)\top}   \right)_{ij} = & \Lam_i^{(\gamma)\top} \Bigg[ \underbrace{\frac{1}{|Q^Z_{ij}|} \sum_{t \in Q^Z_{ij}} F_t F_t^\top -  \frac{1}{T} F^\top F}_{\Delta_{F,ij}} \Bigg] \Lam_j^{(\gamma)} +  \frac{1}{|Q^Z_{ij}|} \sum_{t \in Q^Z_{ij}} F_t^\top \Lam_i^{(\gamma)} \ee_{tj} \\
&   +  \frac{1}{|Q^Z_{ij}|} \sum_{t \in Q^Z_{ij}} F_t^\top \Lam_j^{(\gamma)} \ee_{ti}  + \frac{1}{|Q^Z_{ij}|} \sum_{t \in Q^Z_{ij}} \ee_{ti} \ee_{tj}\\
= & \Lam_i^{(\gamma)\top} \Delta_{F,ij} \Lam_j^{(\gamma)}  + \eta_{ij} + \xi_{ij} + \zeta_{ij} + \gamma(i,j).
\end{align*}
Each $(r,s)$-th entry of the component
$\frac{\Lam^{(\gamma)\top}}{\Nx+\Ny} 
\left( \tilde{\Sigma}^{Z^{(\gamma)}} - \Lam^{(\gamma)} \frac{F^\top F}{T} \Lam^{(\gamma)\top}  \right)\frac{\tilde{\Lam}^{(\gamma)}}{\Nx+\Ny}$ of $d_{T,\Nx \Ny}^{(\gamma)}$ can be bounded by
\begin{align*}
    & \frac{1}{(\Nx+\Ny)^2}\sum_{i,j=1}^{\Nx+\Ny}\Lam_{ir}^{(\gamma)}\tilde{\Lam}^{(\gamma)}_{js}\left(\Lam_i^{(\gamma)\top} \Delta_{F,ij}\Lam_j^{(\gamma)} + \eta_{ij}+ \xi_{ij}+ \zeta_{ij}+ \gamma(i,j)\right) \\
\leq & \left( \frac{1}{(\Nx+\Ny)^2}\sum_{i,j=1}^{\Nx+\Ny} (\Lam_{ir}^{(\gamma)})^2 (\tilde{\Lam}_{js}^{(\gamma)})^2\right)^{1/2}\cdot \left[\bigg(\frac{1}{(\Nx+\Ny)^2}\sum_{i,j=1}^{\Nx+\Ny} \big(\Lam_i^{(\gamma)\top} \Delta_{F,ij}\Lam_j^{(\gamma)}\big)^2 \bigg)^{1/2}\right. \\
 & \qquad \qquad+ \bigg(\frac{1}{(\Nx+\Ny)^2}\sum_{i,j=1}^{\Nx+\Ny} \eta_{ij}^2 \bigg)^{1/2} + \bigg(\frac{1}{(\Nx+\Ny)^2}\sum_{i,j=1}^{\Nx+\Ny} \xi_{ij}^2 \bigg)^{1/2}  \\
&\left. \qquad \qquad+  \bigg(\frac{1}{(\Nx+\Ny)^2}\sum_{i,j=1}^{\Nx+\Ny} \zeta_{ij}^2 \bigg)^{1/2}+\bigg(\frac{1}{(\Nx+\Ny)^2}\sum_{i,j=1}^{\Nx+\Ny} \gamma^2(i,j) \bigg)^{1/2} \right].
\end{align*} 
Based on Lemma \ref{lemma: order of the four terms} and the proof of Lemma \ref{lemma: D}, we can deduce that $d_{T,\Nx \Ny}^{(\gamma)} = o_p(1)$. Let
\[
B_{T,\Nx \Ny}^{(\gamma)} = \left( \frac{F^\top F}{T} \right)^{1/2} \frac{\Lam^{(\gamma)\top} \Lam^{(\gamma)}}{\Nx+\Ny}   \left( \frac{F^\top F}{T} \right)^{1/2},
\]
and 
\[
R_{T,\Nx \Ny}^{(\gamma)} =  \left(\frac{ F^\top F }{T} \right)^{1/2} \frac{\Lam^{(\gamma)\top} \tilde{\Lam}^{(\gamma)}}{\Nx+\Ny} .
\]
It holds that 
\[
(B_{T,\Nx \Ny}^{(\gamma)} + d_{T,\Nx \Ny}^{(\gamma)}(R_{T,\Nx \Ny}^{(\gamma)})^{-1}) R_{T,\Nx \Ny}^{(\gamma)} = R_{T,\Nx \Ny}^{(\gamma)} \tilde{D}^{(\gamma)}.
\]
Note that $B_{T,\Nx \Ny}^{(\gamma)} + d_{T,{\Nx}\Ny}^{(\gamma)}(R_{T,{\Nx}\Ny}^{(\gamma)})^{-1} \overset{p}{\rightarrow} B^{(\gamma)}=\Sigma_F^{1/2}\Sigma_\Lam^{(\gamma)} \Sigma_F^{1/2}$ by Assumption \ref{assump: factor model} and the fact that $d_{T,{\Nx}\Ny}^{(\gamma)}=o_p(1).$ Because the eigenvalues of $B^{(\gamma)}$ are distinct, the eigenvalues of $B_{T,{\Nx}\Ny}^{(\gamma)} + d_{T,{\Nx}\Ny}^{(\gamma)}(R_{T,{\Nx}\Ny}^{(\gamma)})^{-1}$ will also be distinct for large $T,{\Nx}$ and $\Ny$ by the continuity of eigenvalues. With similar arguments as the proof of Proposition 1 in \cite{bai2003inferential}, it holds that
\[
\frac{\Lam^{(\gamma)\top} \tilde{\Lam}^{(\gamma)}}{{\Nx}+\Ny} \overset{p}{\rightarrow} \Sigma_{F}^{-1/2} \Upsilon^{(\gamma)} (D^{(\gamma)})^{1/2} =:Q^{(\gamma)},
\]
where $D^{(\gamma)}$ is the diagonal matrix consisting of eigenvalues of $\Sigma_{F}^{1/2}  \Sigma_\Lam^{(\gamma)}  \Sigma_{F}^{1/2}$, and $\Upsilon^{(\gamma)}$ is the corresponding eigenvector matrix such that $\Upsilon^{(\gamma)\top} \Upsilon^{(\gamma)} = I$. Note that the eigenvalues of $\Sigma_{F}^{1/2}  \Sigma_\Lam^{(\gamma)}  \Sigma_{F}^{1/2}$ are the same as the eigenvalues of $\Sigma_F \Sigma_{\Lam}^{(\gamma)}$. Furthermore, it holds that 
\begin{align*}
H^{(\gamma)} &= T^{-1}(\Nx+\Ny)^{-1} (\tilde{D}^{(\gamma)})^{-1} \tilde{\Lam}^{(\gamma)\top} \Lam^{(\gamma)} F^\top F\\
&\overset{p}{\rightarrow}  (D^{(\gamma)})^{-1} (D^{(\gamma)})^{1/2} \Upsilon^{(\gamma)\top} \Sigma_{F}^{-1/2} \Sigma_F 
=  (D^{(\gamma)})^{-1/2} \Upsilon^{(\gamma)\top} \Sigma_{F}^{1/2} = (Q^{(\gamma)})^{-1}.
\end{align*}
\end{proof}

\vspace{5mm}
\begin{proof}
\noindent {Proof of Theorem \ref{thm: consistency for loadings}.1:}

(1) Proof of $\frac{1}{\Nx+\Ny }\sum_{i=1}^{\Nx+\Ny } \left\|\tilde{\Lam}_i^{(\gamma)} - H^{(\gamma)}\Lam_i^{(\gamma)}\right\|^2 = O_p(\frac{1}{\delta_{\Ny,T}})$.

Observe that 
\begin{align*}
 \frac{1}{\Nx+\Ny} \sum_{i=1}^{\Nx+\Ny} \Lv\tilde{\Lam}^{(\gamma)}_i - H^{(\gamma)}\Lam^{(\gamma)}_i\Rv^2 \leq\ & \frac{1}{\Nx+\Ny} \sum_{i=1}^{\Nx+\Ny} \Lv\tilde{\Lam}^{(\gamma)}_i - H^{(\gamma)}_i \Lam^{(\gamma)}_i\Rv^2 \\
 &+ \frac{1}{\Nx+\Ny} \sum_{i=1}^{\Nx+\Ny} \Lv(H_i^{(\gamma)} - H^{(\gamma)}) \Lam_i^{(\gamma)}\Rv^2.   
\end{align*}
We bound the two terms on the RHS respectively in the following. For the first term, we have the decomposition 
\begin{align*}
 &\tilde{\Lam}_i^{(\gamma)} - H_i^{(\gamma)} \Lam_i^{(\gamma)} \\
 = \ & \frac{1}{\Nx+\Ny} (\tilde{D}^{(\gamma)})^{-1}  \left( \sum_{j=1}^{\Nx+\Ny} \tilde{\Lam}_j^{(\gamma)} \eta_{ij} +   \sum_{j=1}^{\Nx+\Ny} \tilde{\Lam}_j^{(\gamma)} \xi_{ij} + \sum_{j=1}^{\Nx+\Ny} \tilde{\Lam}_j^{(\gamma)} \zeta_{ij} +  \sum_{j=1}^{\Nx+\Ny} \tilde{\Lam}_j^{(\gamma)} \gamma(i,j) \right),   
\end{align*}
which follows that
\begin{align*}
& \frac{1}{\Nx+\Ny} \sum_{i=1}^{\Nx+\Ny} \Lv\tilde{\Lam}_i^{(\gamma)} - H_i^{(\gamma)} \Lam_i^{(\gamma)} \Rv^2 \\
\leq & 4 \big \|(\tilde{D}^{(\gamma)})^{-1} \big \|^2  \cdot \frac{1}{\Nx+\Ny} \sum_{i=1}^{\Nx+\Ny} \frac{1}{(\Nx+\Ny)^2} \left( \Bigg \|\sum_{j=1}^{\Nx+\Ny} \tilde{\Lam}_j^{(\gamma)} \eta_{ij}\Bigg \|^2 + \Bigg \|\sum_{j=1}^{\Nx+\Ny} \tilde{\Lam}_j^{(\gamma)} \xi_{ij}\Bigg \|^2\right. \\
&\qquad \qquad \qquad\qquad\qquad \qquad\qquad\qquad\qquad\left.+ \Bigg \|\sum_{j=1}^{\Nx+\Ny} \tilde{\Lam}_j^{(\gamma)} \zeta_{ij}\Bigg \|^2+  \Bigg \|\sum_{j=1}^{\Nx+\Ny} \tilde{\Lam}_j^{(\gamma)} \gamma(i,j)\Bigg \|^2 \right).
\end{align*}
Each term $\frac{1}{\Nx+\Ny}\sum_{i=1}^{\Nx+\Ny}\frac{1}{(\Nx+\Ny)^2} \big\|\sum_{j=1}^{\Nx+\Ny} \tilde{\Lam}_j^{(\gamma)} \phi_{ij} \big\|^2$ with $\phi_{ij} = \eta_{ij}, \xi_{ij}, \zeta_{ij}, \gamma(i,j)$ on the RHS can be bounded by
\begin{align*}
\frac{1}{\Nx+\Ny}\sum_{i=1}^{\Nx+\Ny}\frac{1}{(\Nx+\Ny)^2} \Bigg \|\sum_{j=1}^{\Nx+\Ny} \tilde{\Lam}_j^{(\gamma)} \phi_{ij}\Bigg \|^2  \leq \underbrace{ \frac{1}{\Nx+\Ny} \sum_{j=1}^{\Nx+\Ny} \|\tilde{\Lam}_j^{(\gamma)}\|^2}_{O_p(1)} \cdot  \frac{1}{(\Nx+\Ny)^2} \sum_{i,j=1}^{\Nx+\Ny} \phi_{ij}^2. 
\end{align*} 
By Lemma \ref{lemma: order of the four terms} and $\Lv(\tilde{D}^{(\gamma)})^{-1}\Rv=O_p(1)$ proved in Lemma \ref{lemma: D}, we conclude that 
\[
\frac{1}{\Nx+\Ny}\sum_{i=1}^{\Nx+\Ny}  \Lv\tilde{\Lam}_i^{(\gamma)} - H_i^{(\gamma)} \Lam_i^{(\gamma)}  \Rv^2 = O_p\big(\frac{1}{\Talpha}\big) + O_p\big(\frac{1}{\Ny}\big) =  O_p\big(\frac{1}{\delta_{\Ny, \Talpha}}\big),
\]
where $\delta_{\Ny, \Talpha} = \min(\Ny,T).$

The second part $\frac{1}{\Nx+\Ny} \sum_{i=1}^{\Nx+\Ny} \Lv(H_i^{(\gamma)} - H^{(\gamma)}) \Lam_i^{(\gamma)}\Rv^2$ can be bounded by 
\begin{align*}
& \frac{1}{\Nx+\Ny} \sum_{i=1}^{\Nx+\Ny} \Lv ( H_i^{(\gamma)}-H^{(\gamma)})\Lam_i^{(\gamma)} \Rv^2 \\
 = & \frac{1}{\Nx+\Ny} \sum_{i=1}^{\Nx+\Ny}  \left\|  \frac{1}{\Nx+\Ny} (\tilde{D}^{(\gamma)})^{-1} \sum_{j=1}^{\Nx+\Ny}  \tilde{\Lam}^{(\gamma)}_j  \Lam_j^{(\gamma)\top} \Lp \frac{1}{|Q^Z_{ij}|}\sum_{t \in Q^Z_{ij}} F_t F_t^\top -  \frac{1}{T} \sum_{t=1}^{T} F_t F_t^\top \Rp \Lam_i^{(\gamma)} \right\|^2 \\
 \leq  & \underbrace{\left\| (\tilde{D}^{(\gamma)})^{-1}\right\|^2}_{O_p(1)} \cdot \underbrace{\bigg( \frac{1}{\Nx+\Ny} \sum_{j=1}^{\Nx+\Ny} \|\tilde{\Lam}_j^{(\gamma)}\|^2 \bigg)}_{O_p(1)} \cdot \Delta^{(\gamma)}, 
\end{align*}
where $\Delta^{(\gamma)} \coloneqq \frac{1}{(\Nx+\Ny)^2} \sum_{i,j=1}^{\Nx+\Ny} \|\Lam_i^{(\gamma)}\|^2\|\Lam_j^{(\gamma)}\|^2 \cdot \left\| \frac{1}{|Q^Z_{ij}|}\sum_{t \in Q^Z_{ij}} F_t F_t^\top -  \frac{1}{T} \sum_{t=1}^{T} F_t F_t^\top \right\|^2.$
Since there is $\E\Ls\Delta^{(\gamma)}\Rs\leq \frac{C}{\Talpha}$, we conclude that $ \frac{1}{\Nx+\Ny} \sum_{i=1}^{\Nx+\Ny} \Lv ( H_i^{(\gamma)}-H^{(\gamma)})\Lam_i^{(\gamma)} \Rv^2 = O_p\big(\frac{1}{\Talpha} \big)$. 

Combining these two parts, we have
\[
\frac{1}{\Nx+\Ny} \sum_{i=1}^{\Nx+\Ny}  \Lv\tilde{\Lam}_i^{(\gamma)} - H^{(\gamma)} \Lam_i^{(\gamma)}  \Rv^2 = O_p \big( \frac{1}{\Talpha} \big) + O_p \big( \frac{1}{\delta_{\Ny, \Talpha}} \big)=O_p \big( \frac{1}{\delta_{\Ny, \Talpha}} \big)
\]
as claimed.

\vspace{5mm}
{(2) Proof of $\frac{1}{T}\sum_{t=1}^{T} \left\|\tilde{F}_t - (H^{(\gamma)\top})^{-1}F_t\right\|^2 = O_p(\frac{1}{\delta_{\Ny,T}})$.}

We derive the estimated factors $\tilde{F}_t$ by regressing the observed $Z_{ti}^{(\gamma)}$ on $\tilde{\Lam}_i^{(\gamma)}$ as
\[
\tilde{F}_t = \Lp \sum_{i=1}^{\Nx+\Ny} W_{ti}^Z \tilde{\Lam}_i^{(\gamma)} \tilde{\Lam}_i^{(\gamma)\top} \Rp^{-1} \Lp \sum_{i=1}^{\Nx+\Ny} W_{ti}^Z  Z_{ti}^{(\gamma)} \tilde{\Lam}_i^{(\gamma)} \Rp, \quad t = 1, \cdots, T.
\] 
For analysis, we define an auxiliary $\tilde{F}_t^*$ as
\[
\tilde{F}_t^* := \Lp \sum_{i=1}^{\Nx+\Ny} W_{ti}^Z H^{(\gamma)}\Lam_i^{(\gamma)} \Lam_i^{(\gamma)\top} H^{(\gamma)\top} \Rp^{-1} \Lp \sum_{i=1}^{\Nx+\Ny} W_{ti}^Z  Z_{ti}^{(\gamma)} \tilde{\Lam}_i^{(\gamma)} \Rp.
\]
We have the decomposition
\[
\frac{1}{T}\sum_{t=1}^{T} \left\|\tilde{F}_t - (H^{(\gamma)\top})^{-1}F_t\right\|^2 \leq \frac{1}{T}\sum_{t=1}^{T} \left\|\tilde{F}_t^* - (H^{(\gamma)\top})^{-1}F_t\right\|^2 + \frac{1}{T}\sum_{t=1}^{T} \left\|\tilde{F}_t^* - \tilde{F}_t\right\|^2.
\]
We bound the two terms on the RHS in the following. For the first term, $\tilde{F}^*_t$ can be decomposed as
\begin{align*}
\tilde{F}_t^* =& (H^{(\gamma)\top})^{-1} F_t + (H^{(\gamma)\top})^{-1}  (\tilde{\Sigma}_{\Lam,t}^{(\gamma)})^{-1} \Lp  \frac{1}{\Nx+\Ny} \sum_{i=1}^{\Nx+\Ny} W_{ti}^Z \Lam_i^{(\gamma)} \ee_{ti} \Rp  \\
& + (H^{(\gamma)\top})^{-1}  (\tilde{\Sigma}_{\Lam,t}^{(\gamma)})^{-1}(H^{(\gamma)})^{-1}  \Lp \frac{1}{\Nx+\Ny}\sum_{i=1}^{\Nx+\Ny} W_{ti}^Z (\tilde{\Lam}_i^{(\gamma)}-H^{(\gamma)}\Lam_i^{(\gamma)})Z_{ti}^{(\gamma)}\Rp,
\end{align*}
where $ \tilde{\Sigma}_{\Lam,t}^{(\gamma)} = \frac{1}{\Nx+\Ny} \sum_{i=1}^{\Nx+\Ny} W_{ti}^Z \Lam_i^{(\gamma)} \Lam_i^{(\gamma)\top} \overset{p}{\rightarrow}\Sigma_{\Lam,t}^{(\gamma)} \succ 0$. Therefore, it holds that
\begin{align*}
    \frac{1}{T}\sum_{t=1}^{T} \Lv\tilde{F}_t^* - (H^{(\gamma)\top})^{-1}F_t\Rv^2 \leq & \ \frac{C}{T}\sum_{t=1}^T \Lv  \frac{1}{\Nx+\Ny} \sum_{i=1}^{\Nx+\Ny} W_{ti}^Z \Lam_i^{(\gamma)} \ee_{ti}\Rv^2 \\
    &+ \frac{C}{T}\sum_{t=1}^T \Lv \frac{1}{\Nx+\Ny}\sum_{i=1}^{\Nx+\Ny} W_{ti}^Z (\tilde{\Lam}_i^{(\gamma)}-H^{(\gamma)}\Lam_i^{(\gamma)})Z_{ti}^{(\gamma)}\Rv^2.
\end{align*}
By Assumption \ref{assump: factor model}.3, there is
\begin{align*}
    &\E\Lv \frac{1}{\Nx+\Ny} \sum_{i=1}^{\Nx+\Ny} W_{ti}^Z \Lam_i^{(\gamma)} \ee_{ti}\Rv^2 =\sum_{r=1}^k \frac{1}{(\Nx+\Ny)^2}\sum_{i,j=1}^{\Nx+\Ny}\E\Ls W_{ti}^ZW_{tj}^Z\Lam_{ir}^{(\gamma)}\Lam_{jr}^{(\gamma)}\Rs\E\Ls \ee_{ti}\ee_{tj}\Rs \leq \frac{C}{\Ny},
\end{align*}
and based on the first step,
\begin{align*}
    &\frac{1}{T}\sum_{t=1}^T \Lv \frac{1}{\Nx+\Ny}\sum_{i=1}^{\Nx+\Ny} W_{ti}^Z (\tilde{\Lam}_i^{(\gamma)}-H^{(\gamma)}\Lam_i^{(\gamma)})Z_{ti}^{(\gamma)}\Rv^2\\
    \leq & \frac{1}{T}\sum_{t=1}^T \frac{1}{\Nx+\Ny}\sum_{i=1}^{\Nx+\Ny}\Lp Z_{ti}^{(\gamma)}\Rp^2 \cdot \frac{1}{\Nx+\Ny}\sum_{i=1}^{\Nx+\Ny}\Lv \tilde{\Lam}^{(\gamma)}_i - H^{(\gamma)}\Lam_i^{(\gamma)}\Rv^2 = O_p(\frac{1}{\delta_{\Ny,T}}).
\end{align*}
Thus, $\frac{1}{T}\sum_{t=1}^{T} \Lv\tilde{F}_t^* - (H^{(\gamma)\top})^{-1}F_t\Rv^2 = O_p(\frac{1}{\delta_{\Ny,T}})$. 

Consider the second term $\frac{1}{T}\sum_{t=1}^{T} \left\|\tilde{F}_t^* - \tilde{F}_t\right\|^2$. The difference $\tilde{F}_t^* - \tilde{F}_t$ can be expanded as
\begin{align*}
\tilde{F}_t^* - \tilde{F}_t   = &  \Lp \frac{1}{\Nx+\Ny} \sum_{i=1}^{\Nx+\Ny} W_{ti}^Z \tilde{\Lam}_i^{(\gamma)} \tilde{\Lam}_i^{(\gamma)\top} \Rp^{-1}\cdot\\
&\Ls  \frac{1}{\Nx+\Ny} \sum_{i=1}^{\Nx+\Ny} W_{ti}^Z \tilde{\Lam}_i^{(\gamma)} \tilde{\Lam}_i^{(\gamma)\top}  - \frac{1}{\Nx+\Ny}  \sum_{i=1}^{\Nx+\Ny} W_{ti}^Z H^{(\gamma)}\Lam_i^{(\gamma)} \Lam_i^{(\gamma)\top} H^{(\gamma)\top}  \Rs
 \tilde{F}^*_t.
\end{align*} 
Observe that 
\begin{align*}
&\Lv \frac{1}{\Nx+\Ny} \sum_{i=1}^{\Nx+\Ny} W_{ti}^Z \tilde{\Lam}_i^{(\gamma)} \tilde{\Lam}_i^{(\gamma)\top}  - \frac{1}{\Nx+\Ny}  \sum_{i=1}^{\Nx+\Ny} W_{ti}^Z H^{(\gamma)}\Lam_i^{(\gamma)} \Lam_i^{(\gamma)\top} H^{(\gamma)\top} \Rv \\
\leq \ &\frac{1}{\Nx+\Ny} \sum_{i=1}^{\Nx+\Ny}\Lp\Lv \tilde{\Lam}_i^{(\gamma)}\Rv + \Lv H^{(\gamma)}\Lam_i^{(\gamma)}\Rv\Rp \Lv \tilde{\Lam}_i^{(\gamma)}-H^{(\gamma)}{\Lam}_i^{(\gamma)} \Rv \\
\leq \ & \Lp \frac{1}{\Nx+\Ny} \sum_{i=1}^{\Nx+\Ny}\Lp\Lv \tilde{\Lam}_i^{(\gamma)}\Rv + \Lv H^{(\gamma)}\Lam_i^{(\gamma)}\Rv\Rp^2\Rp^{1/2} \cdot \Lp \frac{1}{\Nx+\Ny} \sum_{i=1}^{\Nx+\Ny}\Lv \tilde{\Lam}_i^{(\gamma)}-H^{(\gamma)}{\Lam}_i^{(\gamma)} \Rv^2\Rp^{1/2}.
\end{align*} 
By $\frac{1}{\Nx+\Ny}\tilde{\Lam}^{(\gamma)\top}\tilde{\Lam}^{(\gamma)} = I_k$ and the consistency results of loadings, we have 
\[
\Lv \frac{1}{\Nx+\Ny} \sum_{i=1}^{\Nx+\Ny} W_{ti}^Z \tilde{\Lam}_i^{(\gamma)} \tilde{\Lam}_i^{(\gamma)\top}  - \frac{1}{\Nx+\Ny}  \sum_{i=1}^{\Nx+\Ny} W_{ti}^Z H^{(\gamma)}\Lam_i^{(\gamma)} \Lam_i^{(\gamma)\top} H^{(\gamma)\top} \Rv = O_p\big(\frac{1}{\sqrt{\delta_{\Ny,T}}}\big).
\]
Therefore, $\frac{1}{\Nx+\Ny} \sum_{i=1}^{\Nx+\Ny} W_{ti}^Z \tilde{\Lam}_i^{(\gamma)} \tilde{\Lam}_i^{(\gamma)\top} \overset{p}{\rightarrow}(Q^{(\gamma)})^{-1}\Sigma_{\Lam,t}^{(\gamma)}((Q^{(\gamma)})^{-1})^\top \succ 0$. To simplify notation, we let $\Delta_{\Lam,t} =\frac{1}{\Nx+\Ny} \sum_{i=1}^{\Nx+\Ny} W_{ti}^Z \tilde{\Lam}_i^{(\gamma)} \tilde{\Lam}_i^{(\gamma)\top}  - \frac{1}{\Nx+\Ny}  \sum_{i=1}^{\Nx+\Ny} W_{ti}^Z H^{(\gamma)}\Lam_i^{(\gamma)} \Lam_i^{(\gamma)\top} H^{(\gamma)\top}$. Similar to the first term, the second term can be bounded by
\[
\frac{1}{T}\sum_{t=1}^{T} \left\|\tilde{F}_t^* - \tilde{F}_t\right\|^2 \leq \frac{C}{T}\sum_{t=1}^{T} \Lv\Delta_{\Lam,t}\Rv^2 \Lv \tilde{F}^*_t\Rv^2.
\]
For $\Delta_{\Lam,t}$, we have
\begin{align*}
\Delta_{\Lam,t} = &  \frac{1}{\Nx+\Ny} \sum_{i=1}^{\Nx+\Ny} W_{ti}^Z \tilde{\Lam}_i^{(\gamma)} \tilde{\Lam}_i^{(\gamma)\top}- \frac{1}{\Nx+\Ny} \sum_{i=1}^{\Nx+\Ny} W_{ti}^Z H^{(\gamma)} \Lam_i^{(\gamma)} \Lam_i^{(\gamma)\top} H^{(\gamma)\top}\\
= & \frac{1}{\Nx+\Ny} \sum_{i=1}^{\Nx+\Ny}\left[H^{(\gamma)}  W_{ti}^Z \Lam_i^{(\gamma)} (\tilde{\Lam}_i^{(\gamma)} - H^{(\gamma)}_i\Lam_i^{(\gamma)} )^\top+  W_{ti}^Z  (\tilde{\Lam}_i^{(\gamma)} - H^{(\gamma)}_i\Lam_i^{(\gamma)} )\Lam_i^{(\gamma)\top} H^{(\gamma)\top}  \right]
\\
 & + \frac{1}{\Nx+\Ny}  \sum_{i=1}^{\Nx+\Ny} \left[ W_{ti}^Z H^{(\gamma)} \Lam_i^{(\gamma)} (H_i^{(\gamma)} \Lam_i^{(\gamma)} - H^{(\gamma)} \Lam_i^{(\gamma)})^\top  + W_{ti}^Z  (H^{(\gamma)}_i\Lam_i^{(\gamma)} - H^{(\gamma)} \Lam_i^{(\gamma)})(H^{(\gamma)}\Lam_i^{(\gamma)})^\top \right]  \\
& + \frac{1}{\Nx+\Ny} \sum_{i=1}^{\Nx+\Ny} W_{ti}^Z  ( \tilde{\Lam}_i^{(\gamma)} - H^{(\gamma)}\Lam_i^{(\gamma)})(\tilde{\Lam}_i^{(\gamma)} - H^{(\gamma)}\Lam_i^{(\gamma)})^\top,
\end{align*}
which implies that
\begin{align*}
    \Lv\Delta_{\Lam,t}\Rv^2 \leq & 10\cdot \Lv \frac{1}{\Nx+\Ny} \sum_{i=1}^{\Nx+\Ny}H^{(\gamma)}  W_{ti}^Z \Lam_i^{(\gamma)} (\tilde{\Lam}_i^{(\gamma)} - H^{(\gamma)}_i\Lam_i^{(\gamma)} )^\top\Rv^2 \\
    &+ 10\cdot \Lv\frac{1}{\Nx+\Ny}  \sum_{i=1}^{\Nx+\Ny}  W_{ti}^Z H^{(\gamma)} \Lam_i^{(\gamma)}\Lam_i^{(\gamma)\top} (H_i^{(\gamma)}  - H^{(\gamma)})^\top \Rv^2 \\
    &+ 5\cdot \Lv \frac{1}{\Nx+\Ny} \sum_{i=1}^{\Nx+\Ny} W_{ti}^Z  ( \tilde{\Lam}_i^{(\gamma)} - H^{(\gamma)}\Lam_i^{(\gamma)})(\tilde{\Lam}_i^{(\gamma)} - H^{(\gamma)}\Lam_i^{(\gamma)})^\top\Rv^2.
\end{align*}
According to the proof of the first term, we have
\[
\frac{1}{T}\sum_{t=1}^T \Lv \tilde{F}^*_t\Rv^2 \leq \frac{1}{T}\sum_{t=1}^T \Lv \tilde{F}^*_t-(H^{(\gamma)\top})^{-1}F_t\Rv^2 + \frac{1}{T}\sum_{t=1}^T \Lv (H^{(\gamma)\top})^{-1}F_t\Rv^2 = O_p(1).
\]
Based on the proof of the consistency results of loadings, 
\begin{align*}
    \frac{1}{T}\sum_{t=1}^{T} \left\|\tilde{F}_t^* - \tilde{F}_t\right\|^2 \leq & \  \frac{C}{T}\sum_{t=1}^{T} \Lv \tilde{F}^*_t\Rv^2 \Bigg[ \frac{1}{\Nx+\Ny}\sum_{i=1}^{\Nx+\Ny}\Lv H^{(\gamma)} \Lam_i^{(\gamma)} \Rv^2 \cdot \frac{1}{\Nx+\Ny}\sum_{i=1}^{\Nx+\Ny} \Lv\tilde{\Lam}_i^{(\gamma)} - H^{(\gamma)}_i\Lam_i^{(\gamma)} \Rv^2\\
    & + \frac{1}{\Nx+\Ny}\sum_{i=1}^{\Nx+\Ny}\Lv H^{(\gamma)} \Lam_i^{(\gamma)} \Rv^2 \cdot \frac{1}{\Nx+\Ny}\sum_{i=1}^{\Nx+\Ny} \Lv \Lp H^{(\gamma)}_i-H^{(\gamma)}\Rp\Lam_i^{(\gamma)} \Rv^2 \\
    & + \frac{1}{\Nx+\Ny}\sum_{i=1}^{\Nx+\Ny}\Lv \tilde{\Lam}_i^{(\gamma)} - H^{(\gamma)}\Lam_i^{(\gamma)}\Rv^2\cdot \frac{1}{\Nx+\Ny}\sum_{i=1}^{\Nx+\Ny}\Lv \tilde{\Lam}_i^{(\gamma)} - H^{(\gamma)}\Lam_i^{(\gamma)}\Rv^2\Bigg]\\
    = & \ O_p\big( \frac{1}{\delta_{\Ny,T}}\big).
\end{align*}
Combining this with the first term, we complete our proof.
\end{proof}

\vspace{8mm}
\subsection{Proof of Theorem \ref{thm: asymptotic distribution}}

\subsubsection{Proof of Theorem \ref{thm: asymptotic distribution}.1}
\begin{lemma}   \label{lemma: HDH}
Under Assumptions \ref{assump: obs pattern} and \ref{assump: factor model}, suppose that ${\Ny}/ \Nx\rightarrow c \in [0,\infty)$ and $\gamma = r\cdot \Nx/\Ny$, it holds that
\[
(H^{(\gamma)})^{-1} (\tilde{D}^{(\gamma)})^{-1} H^{(\gamma)} = \left( \frac{1}{T}F^\top F \right)^{-1} \left( \frac{1}{\Nx+\Ny}\Lam^{(\gamma)\top} \Lam^{(\gamma)} \right)^{-1}    + o_p(1).
\]
\end{lemma}
\begin{proof}
Based on the definition of $H^{(\gamma)}$ and the fact that $H^{(\gamma)} = (Q^{(\gamma)})^{-1} + o_p(1) = \left(\frac{1}{\Nx+\Ny} \Lam^{(\gamma)\top} \tilde{\Lam}^{(\gamma)} \right)^{-1} + o_p(1)$, we obtain
\begin{align*}
(H^{(\gamma)})^{-1} (\tilde{D}^{(\gamma)})^{-1} H^{(\gamma)} & =   \left(\frac{1}{T}F^\top F\right)^{-1} \left( \frac{1}{\Nx+\Ny}\tilde{\Lam}^{(\gamma)\top} \Lam^{(\gamma)}\right)^{-1} \tilde{D}^{(\gamma)}(\tilde{D}^{(\gamma)})^{-1} H^{(\gamma)}\\
& =  \left(\frac{1}{T}F^\top F\right)^{-1}  H^{(\gamma)\top}  H^{(\gamma)} + o_p(1).
\end{align*}
Additionally, by Theorem \ref{thm: consistency for loadings},
\begin{align*}
H^{(\gamma)\top} H^{(\gamma)} & =  \left(\frac{1}{\Nx+\Ny} \tilde{\Lam}^{(\gamma)\top} \Lam^{(\gamma)} \right)^{-1}  \left(\frac{1}{\Nx+\Ny} \Lam^{(\gamma)\top} \tilde{\Lam}^{(\gamma)} \right)^{-1} + o_p(1) \\
& =\left(\frac{1}{\Nx+\Ny} \Lam^{(\gamma)\top} {\Lam}^{(\gamma)} \right)^{-1} (H^{(\gamma)\top} H^{(\gamma)})^{-1} \left(\frac{1}{\Nx+\Ny} \Lam^{(\gamma)\top} {\Lam}^{(\gamma)} \right)^{-1} + o_p(1).
\end{align*}
Left-multiplying both sides by $\frac{1}{\Nx+\Ny} \Lam^{(\gamma)\top} {\Lam}^{(\gamma)}$, we have
\[
H^{(\gamma)\top} H^{(\gamma)} = \left(\frac{1}{\Nx+\Ny}\Lam^{(\gamma)\top} {\Lam}^{(\gamma)} \right)^{-1}  + o_p(1).
\]
Plugging this into $(H^{(\gamma)})^{-1} (\tilde{D}^{(\gamma)})^{-1} H^{(\gamma)}$, we complete our proof.
\end{proof}

\vspace{5mm}
\begin{lemma} \label{lemma: order of 4 terms factor first}
Suppose $\Ny / \Nx \rightarrow c\in [0,\infty)$ and let $\gamma = r\cdot \Nx/\Ny$. Under Assumptions \ref{assump: obs pattern}, \ref{assump: factor model} and \ref{assump: additional assumptions}, we have for any $i=\Nx+1,\cdots,\Nx+\Ny$, 
\begin{enumerate}
    \item $\frac{1}{\sqrt{\gamma}}\cdot \frac{1}{\Nx+\Ny} \sum_{j=1}^{\Nx+\Ny} \tilde{\Lam}_j^{(\gamma)} \eta_{ij} = O_p \big(\frac{1}{\sqrt{ T\delta_{\Ny, \Talpha}}} \big)$;
    \item $\frac{1}{\sqrt{\gamma}}\cdot \frac{1}{\Nx+\Ny} \sum_{j=1}^{\Nx+\Ny} \tilde{\Lam}_j^{(\gamma)} \xi_{ij} = O_p \big(\frac{1}{\sqrt{ \Talpha}} \big)$;
    \item $\frac{1}{\sqrt{\gamma}}\cdot \frac{1}{\Nx+\Ny} \sum_{j=1}^{\Nx+\Ny} \tilde{\Lam}_j^{(\gamma)} \gamma(i,j) = O_p \big(\frac{1}{\sqrt{ \Ny\delta_{\Ny, \Talpha}}} \big)$;
    \item $\frac{1}{\sqrt{\gamma}}\cdot \frac{1}{\Nx+\Ny} \sum_{j=1}^{\Nx+\Ny} \tilde{\Lam}_j^{(\gamma)} \zeta_{ij} = O_p \big(\frac{1}{\sqrt{ T \delta_{\Ny, \Talpha}}} \big)$.
\end{enumerate}
\end{lemma}
\begin{proof}
Observe that we can decompose each term $\frac{1}{\sqrt{\gamma}}\cdot \frac{1}{\Nx+\Ny } \sum_{j=1}^{\Nx+\Ny } \tilde{\Lam}_j^{(\gamma)} \phi_{ij}$ with $\phi_{ij} = \eta_{ij}, \xi_{ij}, \gamma(i,j)$ and $\zeta_{ij}$ as
\begin{align*}
\frac{1}{\sqrt{\gamma}}\cdot \frac{1}{\Nx+\Ny } \sum_{j=1}^{\Nx+\Ny } \tilde{\Lam}_j^{(\gamma)} \phi_{ij}
= & \frac{1}{\sqrt{\gamma}}\cdot \frac{1}{\Nx+\Ny} \sum_{j=1}^{\Nx+\Ny} (\tilde{\Lam}_j^{(\gamma)}-H^{(\gamma)}\Lam_j^{(\gamma)}) \phi_{ij} \\ & +\frac{1}{\sqrt{\gamma}}\cdot \frac{1}{\Nx+\Ny} \sum_{j=1}^{\Nx+\Ny} H^{(\gamma)}{\Lam}_j^{(\gamma)} \phi_{ij},
\end{align*}
where the first term $ \frac{1}{\Nx+\Ny} \sum_{j=1}^{\Nx+\Ny} (\tilde{\Lam}_j^{(\gamma)}-H^{(\gamma)}\Lam_j^{(\gamma)}) \phi_{ij}$ can be bounded by
\begin{align*}
& \bigg\| \frac{1}{\Nx+\Ny} \sum_{j=1}^{\Nx+\Ny} (\tilde{\Lam}_j^{(\gamma)}-H^{(\gamma)}\Lam_j^{(\gamma)}) \phi_{ij}\bigg\| \\
\leq & \underbrace{\bigg( \frac{1}{\Nx+\Ny} \sum_{j=1}^{\Nx+\Ny} \|\tilde{\Lam}_j^{(\gamma)}-H^{(\gamma)}\Lam_j^{(\gamma)}\|^2 \bigg)^{1/2}}_{O_p\big(\frac{1}{\sqrt{\delta_{\Ny, \Talpha}} }\big)\ \text{by Theorem \ref{thm: consistency for loadings}}}\cdot \bigg(
\frac{1}{\Nx+\Ny}\sum_{j=1}^{\Nx+\Ny}\phi^2_{ij}\bigg)^{1/2} .
\end{align*}
We analyze $\frac{1}{\sqrt{\gamma}} \frac{1}{\Nx+\Ny } \sum_{j=1}^{\Nx+\Ny } \tilde{\Lam}_j^{(\gamma)} \phi_{ij}$ with $\phi_{ij}= \eta_{ij},\xi_{ij},\gamma(i,j),\zeta_{ij}$ respectively in the following.

1. $\phi_{ij}=\eta_{ij}$: For $i=\Nx+1,\cdots,\Nx+\Ny$, it holds that
\[\E\Ls\frac{1}{\Nx+\Ny}\sum_{j=1}^{\Nx+\Ny}\eta^2_{ij} \Rs \leq \frac{1}{\Nx+\Ny} \sum_{j=1}^{\Nx+\Ny} \E\Lv\Lam_i^{(\gamma)}\Rv^2 \cdot \E\Lv\frac{1}{|Q^Z_{ij}|} \sum_{t\in Q^Z_{ij}}F_t \ee_{tj}\Rv^2 \leq \frac{C\gamma}{\Talpha},
\]
where the last inequality follows from Assumption \ref{assump: factor model}.2 and Assumption \ref{assump: factor model}.4. Additionally, by Assumption \ref{assump: additional assumptions}.1, there is
\begin{align*}
\Lv\frac{1}{\Nx+\Ny}\sum_{j=1}^{\Nx+\Ny}\Lam_j^{(\gamma)} \eta_{ij}  \Rv &= \Lv\frac{1}{\Nx+\Ny}\sum_{j=1}^{\Nx+\Ny}\Lam_j^{(\gamma)} \Lam_i^{(\gamma)\top}\frac{1}{|Q^Z_{ij}|}\sum_{t\in Q^Z_{ij}}F_t\ee_{tj}  \Rv \\
& \leq \Lv\frac{1}{\Nx+\Ny}\sum_{j=1}^{\Nx+\Ny}\Lam_j^{(\gamma)} \frac{1}{|Q^Z_{ij}|}\sum_{t\in Q^Z_{ij}}F_t^\top\ee_{tj}  \Rv  \cdot  \Lv\Lam_i^{(\gamma)} \Rv = O_p\big(\frac{\sqrt{\gamma}}{\sqrt{\Ny \Talpha}} \big).
\end{align*}
Combining these parts, we have $\frac{1}{\sqrt{\gamma}}\cdot \frac{1}{\Nx+\Ny } \sum_{j=1}^{\Nx+\Ny} \tilde{\Lam}_j^{(\gamma)} \phi_{ij} = O_p\big( \frac{1}{\sqrt{T\delta_{\Ny, \Talpha}}} \big)$. 

\vspace{3mm}
2. $\phi_{ij}=\xi_{ij}:$ Similar to the previous part, we can show that $\E\Ls\frac{1}{\Nx+\Ny}\sum_{j=1}^{\Nx+\Ny}\xi^2_{ij} \Rs \leq \frac{C\gamma}{\Talpha}$. According to Assumption \ref{assump: additional assumptions}.6, we have
\[
 \frac{1}{\Nx+\Ny} \sum_{j=1}^{\Nx+\Ny} {\Lam}_j^{(\gamma)} \xi_{ij} = \frac{1}{\Nx+\Ny}\sum_{j=1}^{\Nx+\Ny}\Lam_j^{(\gamma)}\Lam_j^{(\gamma)\top} \frac{1}{|Q^Z_{ij}|}\sum_{t\in Q^Z_{ij}}F_t\ee_{ti} = O_p\big(\frac{\sqrt{\gamma}}{\sqrt{\Talpha}}\big).
\]
Therefore, $\frac{1}{\sqrt{\gamma}}\cdot \frac{1}{\Nx+\Ny} \sum_{j=1}^{\Nx+\Ny} \tilde{\Lam}_j^{(\gamma)} \xi_{ij} = O_p\big( \frac{1}{\sqrt{\Talpha}} \big)$ as claimed. 

\vspace{3mm}
3. $\phi_{ij}=\gamma(i,j):$ Based on Assumption \ref{assump: factor model}.3(c), for any $i=\Nx+1,\cdots,\Nx+\Ny,$
\begin{align*}
    \frac{1}{\Nx+\Ny}\sum_{j=1}^{\Nx+\Ny}\gamma^2(i,j) &= \frac{1}{\Nx+\Ny}\sum_{j=1}^{\Nx+\Ny} \Bigg( \frac{1}{|Q^Z_{ij}|}\sum_{t\in Q^Z_{ij}}\E\Ls\ee_{ti}\ee_{tj}\Rs \Bigg)^2 \\
    & \leq \frac{1}{\Nx+\Ny}\sum_{j=1}^{\Ny} \gamma^2 (\tau_{(i-\Nx),j}^{(\eY)})^2 + \frac{1}{\Nx+\Ny}\sum_{j=1}^{\Nx} \gamma (\tau_{(i-\Nx),j}^{(\eY,\eX)})^2  \leq C \frac{\gamma}{\Ny}.
\end{align*}
Moreover, there is
\begin{align*}
    \E\Lv\frac{1}{\Nx+\Ny}\sum_{j=1}^{\Nx+\Ny}\Lam_j^{(\gamma)} \gamma(i,j)\Rv &\leq \frac{1}{\Nx+\Ny}\sum_{j=1}^{\Nx+\Ny}\E\Lv\Lam_j^{(\gamma)}\Rv \cdot |\gamma(i,j)| \\
    & \leq \frac{C}{\Nx+\Ny}\sum_{j=1}^{\Ny} \gamma^{3/2}\tau_{(i-\Nx),j}^{(\eY)} + \frac{C}{\Nx+\Ny}\sum_{j=1}^{\Nx}\gamma^{1/2} \tau_{(i-\Nx),j}^{(\eY,\eX)} \\
    & \leq C\frac{\sqrt{\gamma}}{\Ny}.
\end{align*}
As a result, we conclude that $\frac{1}{\sqrt{\gamma}}\cdot \frac{1}{\Nx+\Ny} \sum_{j=1}^{\Nx+\Ny} \tilde{\Lam}_j^{(\gamma)} \gamma(i,j) = O_p\big( \frac{1}{\sqrt{\Ny \delta_{\Ny, \Talpha}}} \big)$.

\vspace{3mm}
4. $\phi_{ij}=\zeta_{ij}:$ According to Assumption \ref{assump: factor model}.3(e), for any $i=\Nx+1,\cdots,\Nx+\Ny$ it holds that
\[
\E \Ls \frac{1}{\Nx+\Ny }\sum_{j=1}^{\Nx+\Ny }\zeta_{ij}^2 \Rs = \frac{1}{\Nx+\Ny}\sum_{j=1}^{\Nx+\Ny} \E \Ls\frac{1}{|Q^Z_{ij}|}\sum_{t\in Q^Z_{ij}}[\ee_{ti}\ee_{tj}-\E(\ee_{ti}\ee_{tj})] \Rs^2 \leq \frac{C\gamma}{\Talpha},
\]
and by Assumption \ref{assump: additional assumptions}.3, 
\begin{align*}
 \frac{1}{\Nx+\Ny}\sum_{j=1}^{\Nx+\Ny}\Lam_j^{(\gamma)} \zeta_{ij} &= \frac{1}{\Nx+\Ny}\sum_{j=1}^{\Nx+\Ny}\Lam_j^{(\gamma)} \Lp\frac{1}{|Q^Z_{ij}|}\sum_{t\in Q^Z_{ij}}[\ee_{ti}\ee_{tj}-\E(\ee_{ti}\ee_{tj})] \Rp \\ &= O_p\big( \frac{\sqrt{\gamma}}{\sqrt{\Ny \Talpha}}\big).
\end{align*}
Combining these terms, we have $\frac{1}{\sqrt{\gamma}}\cdot \frac{1}{\Nx+\Ny} \sum_{j=1}^{\Nx+\Ny} \tilde{\Lam}_j^{(\gamma)} \zeta_{ij} = O_p\big( \frac{1}{\sqrt{T\delta_{\Ny, \Talpha}}} \big)$. We complete our proof.
\end{proof}

\vspace{5mm}
\begin{proof} Proof of Theorem \ref{thm: asymptotic distribution}.1: 

For any $i=1,\cdots,\Ny$, we have the decomposition
\[
\sqrt{\Talpha}\Lp(\tilde{\Lam}_y)_i - H^{(\gamma)}(\LamY)_i\Rp =\sqrt{\Talpha} \Lp(\tilde{\Lam}_y)_i - H_{i+\Nx}^{(\gamma)}(\LamY)_i\Rp + \sqrt{\Talpha}\Lp H_{i+\Nx}^{(\gamma)} - H^{(\gamma)}\Rp (\LamY)_i.
\]
For simplicity, we denote $i' = i+\Nx$.\\

Step 1 -- For the first term $\sqrt{\Talpha} \Lp(\tilde{\Lam}_y)_i - H_{i+\Nx}^{(\gamma)}(\LamY)_i\Rp$, observe that
\begin{align*}
    (\tilde{\Lam}_y)_i - H_{i+\Nx}^{(\gamma)} (\LamY)_i &= \frac{1}{\sqrt{\gamma}} \left(\tilde{\Lam}^{(\gamma)}_{i'} - H_{i'}^{(\gamma)} \Lam_{i'}^{(\gamma)}\right) \\
    & =(\tilde{D}^{(\gamma)})^{-1}  \frac{1}{\sqrt{\gamma}}\cdot   \frac{1}{\Nx+\Ny}\sum_{j=1}^{\Nx+\Ny}\Big( \tilde{\Lam}_j^{(\gamma)} \eta_{{i'}j} + \tilde{\Lam}_j^{(\gamma)} \xi_{{i'}j} + \tilde{\Lam}_j^{(\gamma)} \zeta_{{i'}j}  + \tilde{\Lam}_j^{(\gamma)} \gamma(i',j) \Big).
\end{align*}
According to Lemma \ref{lemma: order of 4 terms factor first}, when $\sqrt{\Talpha}/\Ny \rightarrow 0$, there is 
\begin{align*}
\sqrt{\Talpha}\Lp(\tilde{\Lam}_y)_i - H_{i+\Nx}^{(\gamma)} (\LamY)_i\Rp &= \sqrt{\Talpha}  (\tilde{D}^{(\gamma)})^{-1} \frac{1}{\sqrt{\gamma}}\cdot \frac{1}{\Nx+\Ny}\sum_{j=1}^{\Nx+\Ny} \tilde{\Lam}_j^{(\gamma)} \xi_{i'j} + o_p(1) \\
& = \sqrt{\Talpha}\cdot \underbrace{(\tilde{D}^{(\gamma)})^{-1} \frac{1}{\Nx+\Ny} \sum_{j=1}^{\Nx+\Ny}  H^{(\gamma)} \Lam_j^{(\gamma)}\Lam_j^{(\gamma)\top} \frac{1}{|Q^Z_{i'j}|} \sum_{t \in Q^Z_{i'j}} F_t (\eY)_{ti}}_{\omega_{\Lam_i,1}} + o_p(1).
\end{align*}
By Assumption \ref{assump: additional assumptions}.6, for any $i=1, \cdots, N_y$ and $i'=i+N_x,$
\[
 \frac{\sqrt{\Talpha}}{\Nx+\Ny} \sum_{j=1}^{\Nx+\Ny}   \Lam_j^{(\gamma)}\Lam_j^{(\gamma)\top} \frac{1}{|Q^Z_{i'j}|} \sum_{t \in Q^Z_{i'j}} F_t (\eY)_{ti}  \overset{d}{\rightarrow} \Ncal(0, \Gamma_{\LamY,i}^{(\gamma),\text{obs}}).
\]
From Lemma \ref{lemma: D} and Lemma \ref{lemma: H,Q}, $(\tilde{D}^{(\gamma)})^{-1} \overset{p}{\rightarrow} (D^{(\gamma)})^{-1}$ and $H^{(\gamma)} \overset{p}{\rightarrow} (Q^{(\gamma)})^{-1}$. Based on Slutsky's Theorem, it holds that
\begin{align*}
    \sqrt{\Talpha}\Lp(\tilde{\Lam}_y)_i - H_{i+\Nx}^{(\gamma)}(\LamY)_i\Rp  &= \sqrt{T}\cdot \omega_{\Lam_i,1}+o_p(1) \\ &\overset{d}{\rightarrow} \Ncal\Lp 0, (D^{(\gamma)})^{-1}(Q^{(\gamma)})^{-1} \Gamma_{\LamY,i}^{(\gamma),\text{obs}} ((Q^{(\gamma)})^{-1})^\top (D^{(\gamma)})^{-1}\Rp.
\end{align*}

\vspace{2mm}
Step 2 -- For the second term $\sqrt{\Talpha}\Lp H_{i+\Nx}^{(\gamma)} - H^{(\gamma)}\Rp (\LamY)_i$, we have
\begin{align*}
\Lp H_{i+\Nx}^{(\gamma)} - H^{(\gamma)}\Rp (\LamY)_i & =   (\tilde{D}^{(\gamma)})^{-1} \frac{1}{\Nx+\Ny} \sum_{j=1}^{\Nx+\Ny}  \tilde{\Lam}_j^{(\gamma)}  \Lam_j^{(\gamma)\top} \underbrace{\Bigg( \frac{1}{|Q^Z_{i'j}|}\sum_{t \in Q^Z_{i'j}} F_t F_t^\top - \frac{1}{T} \sum_{t=1}^{T} F_tF_t^\top  \Bigg) }_{\Delta_{F,i'j}}(\LamY)_i \\
& = (\tilde{D}^{(\gamma)})^{-1}  \frac{1}{\Nx+\Ny}  \sum_{j=1}^{\Nx+\Ny}  (\tilde{\Lam}_j^{(\gamma)} -H^{(\gamma)}\Lam_j^{(\gamma)}) \Lam_j^{(\gamma)\top} \Delta_{F, i'j} (\LamY)_i \\
 & \quad + \underbrace{(\tilde{D}^{(\gamma)})^{-1}  \frac{1}{\Nx+\Ny}  \sum_{j=1}^{\Nx+\Ny}H^{(\gamma)}\Lam_j^{(\gamma)}\Lam_j^{(\gamma)\top} \Delta_{F, i'j} (\LamY)_i}_{\omega_{\Lam_i,2}}.
\end{align*}
The first part on the RHS can be bounded by 
\begin{align*}
    &\Lv \frac{1}{\Nx+\Ny}  \sum_{j=1}^{\Nx+\Ny}  (\tilde{\Lam}_j^{(\gamma)} -H^{(\gamma)}\Lam_j^{(\gamma)}) \Lam_j^{(\gamma)\top} \Delta_{F, i'j}(\LamY)_i\Rv^2 \\
    \leq & \underbrace{\frac{1}{\Nx+\Ny}  \sum_{j=1}^{\Nx+\Ny}  \Lv\tilde{\Lam}_j^{(\gamma)} -H^{(\gamma)}\Lam_j^{(\gamma)}\Rv^2}_{O_p(\frac{1}{\delta_{\Ny,T}}) \ \text{by Theorem \ref{thm: consistency for loadings}}} \cdot \Lv (\LamY)_i\Rv^2 \cdot  \frac{1}{\Nx+\Ny}  \sum_{j=1}^{\Nx+\Ny}  \Lv\Lam_j^{(\gamma)}\Rv^2 \|\Delta_{F,i' j}\|^2 ,
\end{align*}
where 
\[
\E \left[  \frac{1}{\Nx+\Ny}  \sum_{j=1}^{\Nx+\Ny} \|\Lam_j^{(\gamma)}\|^2 \|\Delta_{F,i'j}\|^2 \right] =  \frac{1}{\Nx+\Ny}  \sum_{j=1}^{\Nx+\Ny} \E\Lv\Lam_j^{(\gamma)}\Rv^2 \E\Lv\Delta_{F,i'j}\Rv^2\leq \frac{C}{\Talpha}.
\]
As a result, $\frac{\sqrt{T}}{\Nx+\Ny}  \sum_{j=1}^{\Nx+\Ny}  (\tilde{\Lam}_j^{(\gamma)} -H^{(\gamma)}\Lam_j^{(\gamma)}) \Lam_j^{(\gamma)\top} \Delta_{F, i'j}(\LamY)_i= o_p(1)$. Consider the second part $\omega_{\Lam_i,2}$. By Assumption \ref{assump: additional assumptions}.8 and Slutsky's theorem, we have 
\begin{align*}
\sqrt{T}\omega_{\Lam_i,2} = &\sqrt{\Talpha} (\tilde{D}^{(\gamma)})^{-1}  \frac{1}{\Nx+\Ny}  \sum_{j=1}^{\Nx+\Ny}H^{(\gamma)}\Lam_j^{(\gamma)}\Lam_j^{(\gamma)\top} \Delta_{F, i'j}(\LamY)_i \\  = & \ (\tilde{D}^{(\gamma)})^{-1} H^{(\gamma)} \frac{\sqrt{\Talpha}}{\Nx+\Ny}   \sum_{j=1}^{\Nx+\Ny} \Lam_j^{(\gamma)}\Lam_j^{(\gamma)\top} \Bigg( \frac{1}{|Q^Z_{i'j}|}\sum_{t \in Q^Z_{i'j}} F_tF_t^\top -  \frac{1}{T} \sum_{t=1}^{T} F_tF_t^\top  \Bigg) (\LamY)_i \\
  \overset{d}{\rightarrow} & \  \Ncal \left(0,  (D^{(\gamma)})^{-1}(Q^{(\gamma)})^{-1} \Gamma_{\LamY,i}^{(\gamma),\text{miss}} ((Q^{(\gamma)})^{-1})^\top (D^{(\gamma)})^{-1} \right) \qquad \mathcal{G}^t - \text{stably},
\end{align*}
where $\Gamma_{\LamY,i}^{(\gamma),\text{miss}} =  h_{i+\Nx}^{(\gamma)}((\LamY)_i).$ \\

Step 3 -- Observe that $\omega_{\Lam_i,1}$ and $\omega_{\Lam_i,2}$ are asymptotically independent. Combining the results from the first two steps, we have
\begin{align*}
\sqrt{\Talpha}\Lp(\tilde{\Lam}_y)_i - H^{(\gamma)}(\LamY)_i\Rp &= \sqrt{T}\Lp \omega_{\Lam_i,1}+\omega_{\Lam_i,2}\Rp +o_p(1)\\ &\overset{d}{\rightarrow}  \Ncal \left( 0, (D^{(\gamma)})^{-1}(Q^{(\gamma)})^{-1} \left(\Gamma_{\LamY,i}^{(\gamma),\text{obs}} + \Gamma_{\LamY,i}^{(\gamma),\text{miss}} \right)((Q^{(\gamma)})^{-1})^\top (D^{(\gamma)})^{-1} \right)    
\end{align*}
$\mathcal{G}^t$ - stably.
If we left-multiply $(\tilde{\Lam}_y)_i - H^{(\gamma)}(\LamY)_i$ by $(H^{(\gamma)})^{-1}$, according to Lemma \ref{lemma: H,Q} and Lemma \ref{lemma: HDH}, it holds that
\begin{align*}
&\sqrt{\Talpha} \Lp(H^{(\gamma)})^{-1} (\tilde{\Lam}_y)_i - (\LamY)_i\Rp \\ \overset{d}{\rightarrow}&\ \Ncal \left( 0, \Sigma_F^{-1}(\Sigma_{\Lam}^{(\gamma)})^{-1}  \left(\Gamma_{\LamY,i}^{(\gamma),\text{obs}} + \Gamma_{\LamY,i}^{(\gamma),\text{miss}}\right)(\Sigma_{\Lam}^{(\gamma)})^{-1} \Sigma_F^{-1} \right)\quad \mathcal{G}^t - \text{stably},    
\end{align*}
or equivalently, 
\[
\sqrt{\Talpha} (\Sigma_{\LamY,i}^{(\gamma)})^{-1/2} \left((H^{(\gamma)})^{-1} (\tilde{\Lam}_y)_i - (\LamY)_i\right) \overset{d}{\rightarrow} \Ncal(0,I_k),
\]
where $\Sigma_{\LamY,i}^{(\gamma)} = \Sigma_F^{-1}(\Sigma_{\Lam}^{(\gamma)})^{-1}  \left(\Gamma_{\LamY,i}^{(\gamma),\text{obs}} + \Gamma_{\LamY,i}^{(\gamma),\text{miss}}\right)(\Sigma_{\Lam}^{(\gamma)})^{-1} \Sigma_F^{-1}.$
\end{proof}

\vspace{5mm}
\subsubsection{Proof of Theorem \ref{thm: asymptotic distribution}.2}
\begin{lemma}   \label{lemma: small order terms}
Suppose that $\Ny/ \Nx\rightarrow c \in [0,\infty)$ and $\gamma = r\cdot \Nx/\Ny$ for some constant $r$. Under Assumptions \ref{assump: obs pattern}, \ref{assump: factor model} and Assumption \ref{assump: additional assumptions}, we have for any $t,$
\begin{enumerate}
\item $\frac{1}{\Nx+\Ny} \sum_{i=1}^{\Nx+\Ny} W^Z_{ti} (\tilde{\Lam}_i^{(\gamma)} - H_i^{(\gamma)} \Lam_i^{(\gamma)}) \ee_{ti} = O_p\big(\frac{1}{{\delta_{\Ny, \Talpha}}}\big)$;
\item $\frac{1}{\Nx+\Ny} \sum_{i=1}^{\Nx+\Ny} W^Z_{ti} (\tilde{\Lam}_i^{(\gamma)} - H^{(\gamma)} \Lam_i^{(\gamma)}) \ee_{ti} = O_p\big(\frac{1}{{\delta_{\Ny, \Talpha}}}\big)$;
\item $\frac{1}{\Nx+\Ny} \sum_{i=1}^{\Nx+\Ny} (\tilde{\Lam}_i^{(\gamma)} - H_i^{(\gamma)} \Lam_i^{(\gamma)}) \Lam_i^{(\gamma)\top} = O_p\big(\frac{1}{{\delta_{\Ny, \Talpha}}}\big)$;
\item $\frac{1}{\Nx+\Ny} \sum_{i=1}^{\Nx+\Ny}W_{ti}^Z  (\tilde{\Lam}_i^{(\gamma)} - H_i^{(\gamma)} \Lam_i^{(\gamma)}) \Lam_i^{(\gamma)\top} = O_p\big(\frac{1}{{\delta_{\Ny, \Talpha}}}\big)$.

\end{enumerate}
\end{lemma}
\begin{proof}
\noindent  1. It holds that 
\begin{align*}
&\frac{1}{\Nx+\Ny}\sum_{i=1}^{\Nx+\Ny}W^Z_{ti}\Lp\tilde{\Lam}^{(\gamma)}_i- H^{(\gamma)}_i \Lam^{(\gamma)}_i\Rp\ee_{ti} \\
=& (\tilde{D}^{(\gamma)})^{-1}\Bigg(\frac{1}{(\Nx+\Ny)^2}\sum_{i,j=1}^{\Nx+\Ny} W^Z_{ti} \tilde{\Lam}_j^{(\gamma)} \eta_{ij}\ee_{ti} +  \frac{1}{(\Nx+\Ny)^2}\sum_{i,j=1}^{\Nx+\Ny}W^Z_{ti} \tilde{\Lam}_j^{(\gamma)} \xi_{ij}\ee_{ti} \\
& + \frac{1}{(\Nx+\Ny)^2}\sum_{i,j=1}^{\Nx+\Ny}W^Z_{ti} \tilde{\Lam}_j^{(\gamma)} \zeta_{ij}\ee_{ti} + \frac{1}{(\Nx+\Ny)^2}\sum_{i,j=1}^{\Nx+\Ny}W^Z_{ti} \tilde{\Lam}_j^{(\gamma)} \gamma(i,j)\ee_{ti} \Bigg).
\end{align*}
Each $\frac{1}{(\Nx+\Ny)^2}\sum_{i,j=1}^{\Nx+\Ny} W^Z_{ti} \tilde{\Lam}_j^{(\gamma)} \phi_{ij}\ee_{ti}$ with $\phi_{ij} = \eta_{ij}, \xi_{ij},\zeta_{ij}$ and $\gamma(i,j)$ can be decomposed as
\begin{align*}
&\frac{1}{(\Nx+\Ny)^2}\sum_{i,j=1}^{\Nx+\Ny}W^Z_{ti}\tilde{\Lam}_j^{(\gamma)}  \phi_{ij} \ee_{ti}  \\
= &\ \frac{1}{(\Nx+\Ny)^2}\sum_{i,j=1}^{\Nx+\Ny} W^Z_{ti} \Lp\tilde{\Lam}_j^{(\gamma)} - H^{(\gamma)}\Lam_j^{(\gamma)}\Rp \phi_{ij}\ee_{ti} + \frac{1}{(\Nx+\Ny)^2}H^{(\gamma)}\sum_{i,j=1}^{\Nx+\Ny}W^Z_{ti} \Lam_j^{(\gamma)} \phi_{ij}\ee_{ti},
\end{align*}
where the first part on the RHS can be bounded by
\begin{align*}
    &\left\|\frac{1}{(\Nx+\Ny)^2}\sum_{i,j=1}^{\Nx+\Ny} W^Z_{ti}\Lp\tilde{\Lam}_j^{(\gamma)} - H^{(\gamma)}\Lam_j^{(\gamma)}\Rp\phi_{ij} \ee_{ti}\right\|^2\\
    \leq &  \underbrace{ \frac{1}{\Nx+\Ny}\sum_{j=1}^{\Nx+\Ny}\Lv\tilde{\Lam}_j^{(\gamma)} - H^{(\gamma)}\Lam_j^{(\gamma)}\Rv^2}_{O_p(\frac{1}{\delta_{\Ny, \Talpha}})\ \text{by Theorem \ref{thm: consistency for loadings}}}\cdot   \underbrace{ \frac{1}{\Nx+\Ny}\sum_{i=1}^{\Nx+\Ny}(\ee_{ti})^2}_{O_p(1)}\cdot \underbrace{\frac{1}{(\Nx+\Ny)^2}\sum_{i,j=1}^{\Nx+\Ny}\phi_{ij}^2}_{O_p(\frac{1}{\delta_{\Ny, \Talpha}})\ \text{by Lemma \ref{lemma: order of the four terms}}} = O_p(\frac{1}{\delta_{\Ny, \Talpha}^2}).
\end{align*}
We analyze the second part $\frac{1}{(\Nx+\Ny)^2}H^{(\gamma)}\sum_{i,j=1}^{\Nx+\Ny}W^Z_{ti} \Lam_j^{(\gamma)} \phi_{ij}\ee_{ti}$ for $\phi_{ij} = \eta_{ij}, \xi_{ij},\zeta_{ij}$ and $\gamma(i,j)$ in the following.

For $\phi_{ij}=\eta_{ij}$, we have
\begin{align*}
    &\frac{1}{(\Nx+\Ny)^2} \sum_{i,j=1}^{\Nx+\Ny}W^Z_{ti}\Lam_j^{(\gamma)}\eta_{ij}\ee_{ti}\\
    = & \frac{1}{(\Nx+\Ny)^2}\sum_{i,j=1}^{\Nx+\Ny}W^Z_{ti}\Lam_j^{(\gamma)}\frac{1}{|Q^Z_{ij}|}\sum_{s\in Q^Z_{ij}}\Lam_i^{(\gamma)\top} F_s \E\Ls\ee_{sj}\ee_{ti}\Rs\\
    & + \frac{1}{(\Nx+\Ny)^2}\sum_{i,j=1}^{\Nx+\Ny}W^Z_{ti}\Lam_j^{(\gamma)}\frac{1}{|Q^Z_{ij}|}\sum_{s\in Q^Z_{ij}}\Lam_i^{(\gamma)\top} F_s \Lp\ee_{sj}\ee_{ti} - \E\Ls\ee_{sj}\ee_{ti}\Rs\Rp.
\end{align*}
Based on Assumptions \ref{assump: factor model}.1, \ref{assump: factor model}.2 and \ref{assump: factor model}.3(d), it holds that 
\begin{align*}
    &\E\Lv\frac{1}{(\Nx+\Ny)^2}\sum_{i,j=1}^{\Nx+\Ny}W^Z_{ti}\Lam_j^{(\gamma)}\frac{1}{|Q^Z_{ij}|}\sum_{s\in Q^Z_{ij}}\Lam_i^{(\gamma)\top} F_s \E\Ls\ee_{sj}\ee_{ti}\Rs\Rv  \\
     \leq & \frac{1}{(\Nx+\Ny)^2}\sum_{i,j=1}^{\Nx+\Ny} \E \Ls\|\Lam_j^{(\gamma)}\|\|\Lam_i^{(\gamma)}\|\Rs\cdot \frac{1}{|Q^Z_{ij}|}\sum_{s\in Q^Z_{ij}}\E\|F_s\|\cdot \left|\E\Ls \ee_{sj}\ee_{ti}\Rs\right| \\
     \leq & \frac{C\gamma}{(\Nx+\Ny)^2}\sum_{i,j=1}^{\Nx+\Ny} \frac{1}{T}\sum_{s=1}^T\left|\E\Ls\ee_{sj}\ee_{ti}\Rs\right|
     \leq \frac{C}{\Ny \Talpha}.
\end{align*}
Additionally, by Assumption \ref{assump: additional assumptions}.4, we have
\[
\frac{1}{(\Nx+\Ny)^2}\sum_{i,j=1}^{\Nx+\Ny}W^Z_{ti}\Lam_j^{(\gamma)}\frac{1}{|Q^Z_{ij}|}\sum_{s\in Q^Z_{ij}}\Lam_i^{(\gamma)\top} F_s \Lp\ee_{sj}\ee_{ti} - \E\Ls\ee_{sj}\ee_{ti}\Rs\Rp  = O_p\big(\frac{1}{\delta_{\Ny, \Talpha}}\big).
\]
As a result, $\frac{1}{(\Nx+\Ny)^2} \sum_{i,j=1}^{\Nx+\Ny}W^Z_{ti}\Lam_j^{(\gamma)}\eta_{ij}\ee_{ti} = O_p(\frac{1}{\delta_{\Ny, \Talpha}})$. Combining the first part and the fact $\|H^{(\gamma)}\|=O_p(1),$ we have $ \frac{1}{(\Nx+\Ny)^2}\sum_{i,j=1}^{\Nx+\Ny}W^Z_{ti}\tilde{\Lam}_j^{(\gamma)}  \eta_{ij} \ee_{ti}= O_p(\frac{1}{\delta_{\Ny, \Talpha}})$. By similar arguments, we can show that $\frac{1}{(\Nx+\Ny)^2}\sum_{i,j=1}^{\Nx+\Ny}W^Z_{ti}\tilde{\Lam}_j^{(\gamma)}  \xi_{ij} \ee_{ti}= O_p(\frac{1}{\delta_{\Ny, \Talpha}})$. 

\vspace{2mm}
For $\phi_{ij}=\zeta_{ij}$, we have
\begin{align*}
    &\Lv \frac{1}{(\Nx+\Ny)^2}\sum_{i,j=1}^{\Nx+\Ny}W^Z_{ti} \Lam_j^{(\gamma)}\zeta_{ij}\ee_{ti} \Rv\\
    =  & \ \Lv \frac{1}{(\Nx+\Ny)^2}\sum_{i,j=1}^{\Nx+\Ny}W^Z_{ti} \Lam_j^{(\gamma)}  \frac{1}{|Q^Z_{ij}|}\sum_{s\in Q^Z_{ij}}\Lp\ee_{si} \ee_{sj} - \E\Ls\ee_{si}\ee_{sj}\Rs\Rp \ee_{ti} \Rv = O_p\big(\frac{1}{\delta_{\Ny, \Talpha}}\big)
\end{align*}
following from Assumption \ref{assump: additional assumptions}.4. Therefore, $\frac{1}{(\Nx+\Ny)^2}\sum_{i,j=1}^{\Nx+\Ny}W^Z_{ti}\tilde{\Lam}_j^{(\gamma)} \zeta_{ij} \ee_{ti}= O_p(\frac{1}{\delta_{\Ny, \Talpha}}).$ 

Finally for $\phi_{ij}=\gamma(i,j)$, by Assumption \ref{assump: factor model}.3(c), it holds that 
\begin{align*}
    &\E\Lv\frac{1}{(\Nx+\Ny)^2}\sum_{i,j=1}^{\Nx+\Ny}W^Z_{ti}\Lam_j^{(\gamma)}  \gamma(i,j)\ee_{ti}\Rv \\= & \E \Lv \frac{1}{(\Nx+\Ny)^2}\sum_{i,j=1}^{\Nx+\Ny}W^Z_{ti}\Lam_j^{(\gamma)}  \frac{1}{|Q^Z_{ij}|}\sum_{s\in Q^Z_{ij}} \E\Ls\ee_{si}\ee_{sj}\Rs\ee_{ti} \Rv \\
    \leq &\frac{1}{(\Nx+\Ny)^2}\sum_{i,j=1}^{\Nx+\Ny}\E\Ls\Lv\Lam_j^{(\gamma)}\Rv \left|\ee_{ti}\right|\Rs\cdot \frac{1}{|Q^Z_{ij}|}\sum_{s\in Q^Z_{ij}} \left|\E\Ls\ee_{si}\ee_{sj}\Rs\right| \leq \frac{C}{\Ny}.
\end{align*}
So, $\frac{1}{(\Nx+\Ny)^2}\sum_{i,j=1}^{\Nx+\Ny}W^Z_{ti}\tilde{\Lam}_j^{(\gamma)} \gamma(i,j) \ee_{ti}= O_p(\frac{1}{\delta_{\Ny, \Talpha}}).$

Combining the four terms and the fact that $\|(\tilde{D}^{(\gamma)})^{-1}\|=O_p(1),$ we derive our result. \\

\noindent 2. We have the decomposition
\begin{align*}
\frac{1}{\Nx+\Ny}\sum_{i=1}^{\Nx+\Ny}W^Z_{ti}\Lp\tilde{\Lam}_i^{(\gamma)} - H^{(\gamma)} \Lam_i^{(\gamma)}\Rp \ee_{ti} =& \underbrace{ \frac{1}{\Nx+\Ny}\sum_{i=1}^{\Nx+\Ny}W^Z_{ti}(\tilde{\Lam}_i^{(\gamma)} - H_i^{(\gamma)} \Lam_i^{(\gamma)}) \ee_{ti} }_{O_p(\frac{1}{\delta_{\Ny, \Talpha}})\ \text{by Lemma \ref{lemma: small order terms}.1}} \\ &+  \frac{1}{\Nx+\Ny}\sum_{i=1}^{\Nx+\Ny}W^Z_{ti}(H_i^{(\gamma)} - H^{(\gamma)}) \Lam_i^{(\gamma)} \ee_{ti}.
\end{align*}
The leading term of the second part can be bounded
\begin{align*}
     &\Lv\frac{1}{(\Nx+\Ny)^2}(\tilde{D}^{(\gamma)})^{-1}\sum_{i,j=1}^{\Nx+\Ny}W^Z_{ti}H^{(\gamma)}{\Lam}_j^{(\gamma)}\Lam_j^{(\gamma)\top} \Delta_{F,ij}\Lam_i^{(\gamma)} \ee_{ti}\Rv^2 \\
    \leq &\Lv(\Tilde{D}^{(\gamma)})^{-1}\Rv^2\Lv H^{(\gamma)}\Rv^2 \cdot \frac{1}{\Nx+\Ny}\sum_{j=1}^{\Nx+\Ny}\Lv \Lam^{(\gamma)}_{j}\Rv^2 \cdot \\
    & \frac{1}{\Nx+\Ny}\sum_{j=1}^{\Nx+\Ny}\Lv \Lam^{(\gamma)}_{j}\Rv^2  \Lv\frac{1}{\Nx+\Ny}\sum_{i=1}^{\Nx+\Ny} \Delta_{F,ij} W^Z_{ti}\Lam_i^{(\gamma)} \ee_{ti}\Rv^2,
\end{align*}
where $\Delta_{F,ij} = \frac{1}{|Q^Z_{ij}|}\sum_{s\in Q^Z_{ij}}F_sF_s^\top - \frac{1}{T}\sum_{s=1}^T F_sF_s^\top$. According to Assumption \ref{assump: additional assumptions}.5, it holds that $\E\Lv\frac{1}{\Nx+\Ny}\sum_{i=1}^{\Nx+\Ny} \Delta_{F,ij} W^Z_{ti}\Lam_i^{(\gamma)} \ee_{ti}\Rv^4 \leq C/T^2\Ny^2$. As a result, we have
\begin{align*}
     &\E\Ls\frac{1}{\Nx+\Ny}\sum_{j=1}^{\Nx+\Ny}\Lv \Lam^{(\gamma)}_{j}\Rv^2  \Lv\frac{1}{\Nx+\Ny}\sum_{i=1}^{\Nx+\Ny} \Delta_{F,ij} W^Z_{ti}\Lam_i^{(\gamma)} \ee_{ti}\Rv^2 \Rs\\
     \leq & \frac{1}{\Nx+\Ny}\sum_{j=1}^{\Nx+\Ny}\Lp \E\Lv \Lam^{(\gamma)}_{j}\Rv^4 \cdot \E\Lv\frac{1}{\Nx+\Ny}\sum_{i=1}^{\Nx+\Ny} \Delta_{F,ij} W^Z_{ti}\Lam_i^{(\gamma)} \ee_{ti}\Rv^4\Rp^{1/2} \leq  \frac{1}{T\Ny},
\end{align*}
which implies that 
\[
\frac{1}{\Nx+\Ny}\sum_{i=1}^{\Nx+\Ny}W^Z_{ti}(H_i^{(\gamma)} - H^{(\gamma)}) \Lam_i^{(\gamma)} \ee_{ti} = O_p\big(\frac{1}{\delta_{\Ny,T}}\big).
\]
Thus, $\frac{1}{\Nx+\Ny }\sum_{i=1}^{\Nx+\Ny }W^Z_{ti}\Lp\tilde{\Lam}_i^{(\gamma)} - H^{(\gamma)} \Lam_i^{(\gamma)}\Rp \ee_{ti} = O_p(\frac{1}{\delta_{\Ny, \Talpha}})$ as claimed. \\

\noindent 3. $\frac{1}{\Nx+\Ny }\sum_{i=1}^{\Nx+\Ny }\Lp\tilde{\Lam}_i^{(\gamma)}- H_i^{(\gamma)} \Lam_i^{(\gamma)}\Rp\Lam_i^{(\gamma)\top}$ can be decomposed as
\begin{align*}
&\frac{1}{\Nx+\Ny }\sum_{i=1}^{\Nx+\Ny }\Lp\tilde{\Lam}_i^{(\gamma)}- H_i^{(\gamma)} \Lam_i^{(\gamma)}\Rp\Lam_i^{(\gamma)\top }\\
=& \ (\tilde{D}^{(\gamma)})^{-1}\Bigg(\frac{1}{(\Nx+\Ny)^2}\sum_{i,j=1}^{\Nx+\Ny}\tilde{\Lam}_j^{(\gamma)} \Lam_i^{(\gamma)\top} \eta_{ij} +  \frac{1}{(\Nx+\Ny)^2}\sum_{i,j=1}^{\Nx+\Ny}\tilde{\Lam}_j^{(\gamma)} \Lam_i^{(\gamma)\top} \xi_{ij} \\
& + \frac{1}{(\Nx+\Ny)^2}\sum_{i,j=1}^{\Nx+\Ny}\tilde{\Lam}_j^{(\gamma)} \Lam_i^{(\gamma)\top} \zeta_{ij} + \frac{1}{(\Nx+\Ny)^2}\sum_{i,j=1}^{\Nx+\Ny}\tilde{\Lam}_j^{(\gamma)} \Lam_i^{(\gamma)\top} \gamma(i,j) \Bigg).
\end{align*}
For each $\frac{1}{(\Nx+\Ny)^2}\sum_{i,j=1}^{\Nx+\Ny}\tilde{\Lam}_j^{(\gamma)} \Lam_i^{(\gamma)\top} \phi_{ij}$ with $\phi_{ij} = \eta_{ij}, \xi_{ij},\zeta_{ij}$ and $\gamma(i,j)$, we have
\begin{align*}
&\frac{1}{(\Nx+\Ny)^2}\sum_{i,j=1}^{\Nx+\Ny}\tilde{\Lam}_j^{(\gamma)} \Lam_i^{(\gamma)\top} \phi_{ij}\\ =& \frac{1}{(\Nx+\Ny)^2}\sum_{i,j=1}^{\Nx+\Ny} \Lp\tilde{\Lam}_j^{(\gamma)} - H^{(\gamma)}\Lam_j^{(\gamma)}\Rp\Lam_i^{(\gamma)\top} \phi_{ij} + \frac{1}{(\Nx+\Ny)^2}H^{(\gamma)}\sum_{i,j=1}^{\Nx+\Ny} \Lam_j^{(\gamma)}\Lam_i^{(\gamma)\top} \phi_{ij}.    
\end{align*}
The first term on the RHS can be bounded by
\begin{align*}
    &\left\|\frac{1}{(\Nx+\Ny)^2}\sum_{i,j=1}^{\Nx+\Ny} (\tilde{\Lam}_j^{(\gamma)} - H^{(\gamma)}\Lam_j^{(\gamma)})\Lam_i^{(\gamma)\top} \phi_{ij}\right\|^2\\
    \leq &  \underbrace{ \frac{1}{\Nx+\Ny}\sum_{j=1}^{\Nx+\Ny}\Lv\tilde{\Lam}_j^{(\gamma)} - H^{(\gamma)}\Lam_j^{(\gamma)}\Rv^2}_{O_p(\frac{1}{\delta_{\Ny, \Talpha}})\ \text{by Theorem \ref{thm: consistency for loadings}}}\cdot \underbrace{ \frac{1}{\Nx+\Ny }\sum_{i=1}^{\Nx+\Ny }\Lv\Lam_i^{(\gamma)}\Rv^2}_{O_p(1)}\cdot \frac{1}{(\Nx+\Ny )^2}\sum_{i,j=1}^{\Nx+\Ny }\phi_{ij}^2     = O_p\big(\frac{1}{\delta_{\Ny, \Talpha}^2}\big)
\end{align*}
following from Lemma \ref{lemma: order of the four terms}. We analyze the second term $\frac{1}{(\Nx+\Ny)^2}\sum_{i,j=1}^{\Nx+\Ny} \Lam_j^{(\gamma)}\Lam_i^{(\gamma)\top} \phi_{ij}$ in the following.

For $\phi_{ij}=\eta_{ij}$, by Assumption \ref{assump: additional assumptions}.2,
\begin{align*}
    \Lv \frac{1}{(\Nx+\Ny)^2}\sum_{i,j=1}^{\Nx+\Ny }\Lam_j^{(\gamma)} \Lam_i^{(\gamma)\top} \eta_{ij} \Rv &= \Lv\frac{1}{(\Nx+\Ny )^2}\sum_{i,j=1}^{\Nx+\Ny }\Lam_i^{(\gamma)}\Lam_i^{(\gamma)\top }\frac{1}{|Q^Z_{ij}|}\sum_{t\in Q^Z_{ij}}F_t \Lam_j^{(\gamma)\top} \ee_{tj}\Rv \\
    &= O_p\big(\frac{1}{\sqrt{\Ny \Talpha}}\big). 
\end{align*}
Therefore, $\frac{1}{(\Nx+\Ny)^2}\sum_{i,j=1}^{\Nx+\Ny}\tilde{\Lam}_j^{(\gamma)} \Lam_i^{(\gamma)\top} \eta_{ij} = O_p(\frac{1}{\delta_{\Ny, \Talpha}}).$ By the same arguments, we can show that $\frac{1}{(\Nx+\Ny)^2}\sum_{i,j=1}^{\Nx+\Ny}\tilde{\Lam}_j^{(\gamma)} \Lam_i^{(\gamma)\top} \xi_{ij} = O_p(\frac{1}{\delta_{\Ny, \Talpha}}).$

For $\phi_{ij}=\zeta_{ij}$, by Assumption \ref{assump: additional assumptions}.4 it holds that
\begin{align*}
    \Lv \frac{1}{(\Nx+\Ny)^2}\sum_{i,j=1}^{\Nx+\Ny}\Lam_j^{(\gamma)} \Lam_i^{(\gamma)\top} \zeta_{ij} \Rv =& \Lv \frac{1}{(\Nx+\Ny)^2}\sum_{i,j=1}^{\Nx+\Ny} \Lam_j^{(\gamma)} \Lam_i^{(\gamma)\top} \frac{1}{|Q^Z_{ij}|}\sum_{t\in Q^Z_{ij}} \Lp\ee_{ti} \ee_{tj} - \E\Ls \ee_{ti} \ee_{tj}\Rs\Rp  \Rv \\
    =& O_p\big(\frac{1}{\delta_{\Ny, \Talpha}}\big).
\end{align*}
Thus, $\frac{1}{(\Nx+\Ny)^2}\sum_{i,j}^{\Nx+\Ny}\tilde{\Lam}_j^{(\gamma)} \Lam_i^{(\gamma)\top} \zeta_{ij} = O_p(\frac{1}{\delta_{\Ny, \Talpha}}).$ 

Finally, for $\phi_{ij}=\gamma(i,j)$, we have
\begin{align*}
    \E\Lv\frac{1}{(\Nx+\Ny)^2}\sum_{i,j=1}^{\Nx+\Ny}\Lam_j^{(\gamma)} \Lam_i^{(\gamma)\top} \gamma(i,j)\Rv
    &=  \E\Lv \frac{1}{(\Nx+\Ny)^2}\sum_{i,j=1}^{\Nx+\Ny}\Lam_j^{(\gamma)} \Lam_i^{(\gamma)\top} \frac{1}{|Q^Z_{ij}|}\sum_{t\in Q^Z_{ij}} \E\Ls\ee_{ti}\ee_{tj}\Rs \Rv \\
     & \leq\frac{1}{(\Nx+\Ny)^2}\sum_{i,j=1}^{\Nx+\Ny}\E\Ls\|\Lam_j^{(\gamma)}\| \|\Lam_i^{(\gamma)}\|\Rs \cdot \frac{1}{|Q^Z_{ij}|}\sum_{t\in Q^Z_{ij}} \left|\E\Ls\ee_{ti}\ee_{tj}\Rs\right| \\
     & \leq \frac{C}{\Ny},
\end{align*}
where the last equality follows from Assumption \ref{assump: factor model}.3(c).

Combining the four terms and the fact that $\|(\tilde{D}^{(\gamma)})^{-1}\|=O_p(1),$ we complete our proof. \\

\noindent 4. We omit the proof of $\frac{1}{\Nx+\Ny} \sum_{i=1}^{\Nx+\Ny}W_{ti}^Z  \Lp\tilde{\Lam}_i^{(\gamma)} - H_i^{(\gamma)} \Lam_i^{(\gamma)}\Rp \Lam_i^{(\gamma)\top} = O_p(\frac{1}{{\delta_{\Ny, \Talpha}}})$, which is similar to the proof of Lemma \ref{lemma: small order terms}.3.
\end{proof}

\vspace{5mm}
\begin{proof}
Proof of Theorem \ref{thm: asymptotic distribution}.2:

We derive the estimated factors $\tilde{F}_t$ by regressing the observed $Z_{ti}^{(\gamma)}$ on $\tilde{\Lam}_i^{(\gamma)}$, i.e.
\[
\tilde{F}_t = \Lp \sum_{i=1}^{\Nx+\Ny} W_{ti}^Z \tilde{\Lam}_i^{(\gamma)} \tilde{\Lam}_i^{(\gamma)\top} \Rp^{-1} \Lp \sum_{i=1}^{\Nx+\Ny} W_{ti}^Z  Z_{ti}^{(\gamma)} \tilde{\Lam}_i^{(\gamma)} \Rp, \quad t = 1, \cdots, T.
\] 
Similar with the proof of Theorem \ref{thm: consistency for loadings}.1, we resort to the auxiliary $\tilde{F}_t^*$, which is defined as
\[
\tilde{F}_t^* = \Lp \sum_{i=1}^{\Nx+\Ny} W_{ti}^Z H^{(\gamma)}\Lam_i^{(\gamma)} \Lam_i^{(\gamma)\top} H^{(\gamma)\top} \Rp^{-1} \Lp \sum_{i=1}^{\Nx+\Ny} W_{ti}^Z  Z_{ti}^{(\gamma)} \tilde{\Lam}_i^{(\gamma)} \Rp.
\] \\

\noindent Step 1 -- In the first step, we analyze $\tilde{F}_t^*$. We have the decomposition
\begin{align*}
H^{(\gamma)\top}\tilde{F}_t^* = \ &  F_t +  (\tilde{\Sigma}_{\Lam,t}^{(\gamma)})^{-1} \Lp  \frac{1}{\Nx+\Ny} \sum_{i=1}^{\Nx+\Ny} W_{ti}^Z \Lam_i^{(\gamma)} \ee_{ti} \Rp \\
&+ (\tilde{\Sigma}_{\Lam,t}^{(\gamma)})^{-1}(H^{(\gamma)})^{-1}  \Lp \frac{1}{\Nx+\Ny}\sum_{i=1}^{\Nx+\Ny} W_{ti}^Z \Lp\tilde{\Lam}_i^{(\gamma)}-H^{(\gamma)}\Lam_i^{(\gamma)}\Rp\Lp\Lam_i^{(\gamma)\top} F_t + \ee_{ti}\Rp\Rp,
\end{align*} 
where $ \tilde{\Sigma}_{\Lam,t}^{(\gamma)} = \frac{1}{\Nx+\Ny} \sum_{i=1}^{\Nx+\Ny} W_{ti}^Z \Lam_i^{(\gamma)} \Lam_i^{(\gamma)\top}\overset{p}{\rightarrow}\Sigma_{\Lam,t}^{(\gamma)}$ is positive definite by Theorem \ref{thm: consistency for loadings}.

Consider the first part $\frac{1}{\Nx+\Ny} \sum_{i=1}^{\Nx+\Ny} W^Z_{ti} \Lam_i^{(\gamma)} \ee_{ti}$, we have
\[
\frac{1}{\Nx+\Ny} \sum_{i=1}^{\Nx+\Ny} W^Z_{ti} \Lam_{i}^{(\gamma)} \ee_{ti} = \frac{\sqrt{\Nx}}{\Nx+\Ny} \frac{1}{\sqrt{\Nx}}\sum_{i=1}^{\Nx}(\LamX)_{i}(\eX)_{ti}  + \frac{\gamma\sqrt{\Ny}}{\Nx+\Ny} \frac{1}{\sqrt{\Ny}}\sum_{i=1}^{\Ny}W^Y_{ti}(\LamY)_{i}(\eY)_{ti}.
\]
If all the factors in $F_y$ are strong factors in $Y$, then each entry of $\frac{1}{\Nx+\Ny} \sum_{i=1}^{\Nx+\Ny} W^Z_{ti} \Lam_{i}^{(\gamma)} \ee_{ti}$ will converge at the same rate of $\sqrt{\Ny}$, which is determined by the second term. By Assumption \ref{assump: additional assumptions}.7,
\[
\frac{\sqrt{\Ny}}{\Nx+\Ny} \sum_{i=1}^{\Nx+\Ny} W^Z_{ti} \Lam_{i}^{(\gamma)} \ee_{ti}  \overset{d}{\rightarrow}\mathcal{N}\Lp 0,\Gamma_{F,t}^{(\gamma),\text{obs}}\Rp,
\]
where $\Gamma_{F,t}^{(\gamma),\text{obs}}$ is a positive definite matrix. We let $\omega_{F,1} \coloneqq  (\tilde{\Sigma}_{\Lam,t}^{(\gamma)})^{-1} (  \frac{1}{\Nx+\Ny} \sum_{i=1}^{\Nx+\Ny} W_{ti}^Z \Lam_i^{(\gamma)} \ee_{ti} )$. Then 
$\sqrt{\Ny}\cdot \omega_{F,1} \overset{d}{\rightarrow}\ \Ncal( 0, ({\Sigma}_{\Lam,t}^{(\gamma)})^{-1}\Gamma_{F,t}^{(\gamma),\text{obs}}({\Sigma}_{\Lam,t}^{(\gamma)})^{-1} )$. If some factor $F_{t,w}$ is weak in $Y$ whose loading $\sum_{i=1}^{\Ny}(\Lam_y)_{i,w}^2$ grows at the rate $g(\Ny)$ which is sub-linear or constant in $\Ny$, then for $F_{t,w}$ there is 
\begin{align*}
    \frac{1}{\Nx+\Ny} \sum_{i=1}^{\Nx+\Ny} W^Z_{ti} \Lam_{i,w}^{(\gamma)} \ee_{ti} = & \frac{\sqrt{\Nx}}{\Nx+\Ny} \frac{1}{\sqrt{\Nx}}\sum_{i=1}^{\Nx}(\LamX)_{i,w}(\eX)_{ti}  \\ &+ \frac{\gamma\sqrt{g(\Ny)}}{\Nx+\Ny} \frac{1}{\sqrt{g(\Ny)}}\sum_{i=1}^{\Ny}W^Y_{ti}(\LamY)_{i,w}(\eY)_{ti}.
\end{align*}
If $g(\Ny)\Nx/\Ny^2 \rightarrow \infty$, the convergence rate of $\frac{1}{\Nx+\Ny} \sum_{i=1}^{\Nx+\Ny} W^Z_{ti} \Lam_{i,w}^{(\gamma)} \ee_{ti}$ is $O(\Ny/\sqrt{g(\Ny)})$, which is determined by the second term; otherwise, the convergence rate is $O(\sqrt{\Nx})$. Combining these two cases, Assumption \ref{assump: additional assumptions}.7 assumes that
\[
\frac{\sqrt{N_w}}{\Nx+\Ny} \sum_{i=1}^{\Nx+\Ny} W^Z_{ti} \Lam_{i,w}^{(\gamma)} \ee_{ti}  \overset{d}{\rightarrow}\mathcal{N}\Lp 0,\Gamma_{F_w,t}^{(\gamma),\text{obs}}\Rp,
\]
where $N_w= \min\Lp \Ny^2/g(\Ny),\Nx\Rp$ and $\Gamma_{F_w,t}^{(\gamma),\text{obs}}$ is positive definite. If we can assume that the loadings of $F_{t,w}$ are orthogonal to the loadings of other factors, then we have
\begin{align*}
    &\sqrt{N_w}\Lp\omega_{F,1}\Rp_{w} \overset{d}{\rightarrow}\ \Ncal\Lp 0, ({\Sigma}_{\Lam,t,w}^{(\gamma)})^{-1}\Gamma_{F_w,t}^{(\gamma),\text{obs}}({\Sigma}_{\Lam,t,w}^{(\gamma)})^{-1} \Rp,
\end{align*}
where $\Sigma^{(\gamma)}_{\Lam,t,w}$ is the diagonal entry of $\Sigma_{\Lam,t}^{(\gamma)}$ corresponding to the weak factor.

Now, we consider the second part $\frac{1}{\Nx+\Ny}\sum_{i=1}^{\Nx+\Ny} W_{ti}^Z (\tilde{\Lam}_i^{(\gamma)}-H^{(\gamma)}\Lam_i^{(\gamma)})(\Lam_i^{(\gamma)\top} F_t + \ee_{ti})$. Observe that
\begin{align*}
    &\frac{1}{\Nx+\Ny}\sum_{i=1}^{\Nx+\Ny} W_{ti}^Z \Lp\tilde{\Lam}_i^{(\gamma)}-H^{(\gamma)}\Lam_i^{(\gamma)}\Rp \Lam_i^{(\gamma)\top} F_t\\ = &  \underbrace{\frac{1}{\Nx+\Ny}\sum_{i=1}^{\Nx+\Ny} W_{ti}^Z \Lp\tilde{\Lam}_i^{(\gamma)}-H^{(\gamma)}_i\Lam_i^{(\gamma)}\Rp \Lam_i^{(\gamma)\top} F_t}_{ \Delta_{t,1}} + \frac{1}{\Nx+\Ny }\sum_{i=1}^{\Nx+\Ny} W_{ti}^Z (H^{(\gamma)}_i-H^{(\gamma)})\Lam_i^{(\gamma)} \Lam_i^{(\gamma)\top} F_t.
\end{align*}
Let $\Delta_{F,ij}=\frac{1}{|Q^Z_{ij}|}\sum_{s\in Q^Z_{ij}}F_sF_s^\top - \frac{1}{T}\sum_{s=1}^TF_sF_s^\top$. The second part $\frac{1}{\Nx+\Ny}\sum_{i=1}^{\Nx+\Ny} W_{ti}^Z (H^{(\gamma)}_i-H^{(\gamma)})\Lam_i^{(\gamma)} \Lam_i^{(\gamma)\top} F_t$ can be further decomposed as
\begin{align*}
&\frac{1}{\Nx+\Ny}\sum_{i=1}^{\Nx+\Ny} W_{ti}^Z (H^{(\gamma)}_i-H^{(\gamma)}) \Lam_i^{(\gamma)} \Lam_i^{(\gamma)\top} F_t \\
 = \ & \underbrace{(\tilde{D}^{(\gamma)})^{-1} \frac{1}{(\Nx+\Ny)^2} \sum_{i,j=1}^{\Nx+\Ny} \Lp\tilde{\Lam}_j^{(\gamma)} - H^{(\gamma)}\Lam_j^{(\gamma)}\Rp \Lam_j^{(\gamma)\top}  \Delta_{F,ij}  W^Z_{ti} \Lam_i^{(\gamma)} \Lam_i^{(\gamma)\top} F_t}_{\Delta_{t,2}} \\
  & + \underbrace{ (\tilde{D}^{(\gamma)})^{-1}H^{(\gamma)} \frac{1}{(\Nx+\Ny)^2} \sum_{i,j=1}^{\Nx+\Ny} \Lam_j^{(\gamma)}\Lam_j^{(\gamma)\top}  \Delta_{F,ij} W^Z_{ti} \Lam_i^{(\gamma)} \Lam_i^{(\gamma)\top} F_t}_{\Delta_{t,3}}.
\end{align*}
For $\Delta_{t,3}$, we have
\begin{align*}
    \Delta_{t,3}=(\tilde{D}^{(\gamma)})^{-1}H^{(\gamma)} \frac{1}{(\Nx+\Ny)^2} \sum_{i,j=1}^{\Nx+\Ny} \Lam_j^{(\gamma)}\Lam_j^{(\gamma)\top}  \Delta_{F,ij} W^Z_{ti} \Lam_i^{(\gamma)} \Lam_i^{(\gamma)\top} F_t = (\tilde{D}^{(\gamma)})^{-1} H^{(\gamma)} \mathbf{X}_t^{(\gamma)} F_t,
\end{align*}
where $\mathbf{X}_t^{(\gamma)} = \frac{1}{(\Nx+\Ny)^2} \sum_{i,j=1}^{\Nx+\Ny} \Lam_j^{(\gamma)}\Lam_j^{(\gamma)\top}  \Delta_{F,ij}   W^Z_{ti} \Lam_i^{(\gamma)} \Lam_i^{(\gamma)\top}$ is asymptotically normal with convergence rate $\sqrt{\Talpha}$ from Assumption \ref{assump: additional assumptions}.8. We denote $\omega_{F,2} :=  (\tilde{\Sigma}_{\Lam,t}^{(\gamma)})^{-1}(H^{(\gamma)})^{-1} (\tilde{D}^{(\gamma)})^{-1}H^{(\gamma)} \mathbf{X}_t^{(\gamma)} F_t.$
For $\Delta_{t,2}$, since $\E\Lv \Delta_{F,ij}\Rv^2 \leq  \frac{C}{\Talpha}$ and $\frac{1}{\Nx+\Ny} \sum_{j=1}^{\Nx+\Ny} \Lv\tilde{\Lam}_j^{(\gamma)} - H^{(\gamma)}\Lam_j^{(\gamma)} \Rv^2=O_p(\frac{1}{\delta_{\Ny,T}})$ by Theorem \ref{thm: consistency for loadings}, 
\begin{align*}
 &\Lv (\tilde{D}^{(\gamma)})^{-1} \frac{1}{(\Nx+\Ny)^2} \sum_{i,j=1}^{\Nx+\Ny} \Lp\tilde{\Lam}_j^{(\gamma)} - H^{(\gamma)}\Lam_j^{(\gamma)}\Rp \Lam_j^{(\gamma)\top}  \Delta_{F,ij}  W^Z_{ti} \Lam_i^{(\gamma)} \Lam_i^{(\gamma)\top} F_t\Rv^2 \\
 \leq& \ O_p \big (\frac{1}{\delta_{\Ny, \Talpha}}\big) 
 \cdot \Lp \frac{1}{\Nx+\Ny }\sum_{j=1}^{\Nx+\Ny } \Lv\Lam_j^{(\gamma)}\Rv^2 \Lv \frac{1}{\Nx+\Ny } \sum_{i=1}^{\Nx+\Ny }  \Delta_{F,ij}  W^Z_{ti} \Lam_i^{(\gamma)} \Lam_i^{(\gamma)\top} \Rv^2 \Rp \\
 \leq & \  O_p \big (\frac{1}{\delta_{\Ny, \Talpha}}\big) \cdot \Lp \frac{1}{\Nx+\Ny} \sum_{j=1}^{\Nx+\Ny} \Lv\Lam_j^{(\gamma)}\Rv^2 \cdot \bigg( \frac{1}{\Nx+\Ny} \sum_{ i=1}^{\Nx+\Ny} \Lv\Lam_i^{(\gamma)}\Rv^2 \Lv\Delta_{F,ij}\Rv \bigg)^2 \Rp,
\end{align*}
where 
\begin{align*}
&\E \Ls \frac{1}{\Nx+\Ny} \sum_{j=1}^{\Nx+\Ny} \Lv\Lam_j^{(\gamma)}\Rv^2 \cdot \bigg( \frac{1}{\Nx+\Ny} \sum_{ i=1}^{\Nx+\Ny} \Lv\Lam_i^{(\gamma)}\Rv^2 \|\Delta_{F,ij}\| \bigg)^2 \Rs \\
= & \ \frac{1}{(\Nx+\Ny)^3} \sum_{i,j,l=1}^{\Nx+\Ny} \E \Ls\|\Lam_i^{(\gamma)}\|^2 \|\Lam_j^{(\gamma)}\|^2\|\Lam_l^{(\gamma)}\|^2  \Rs \cdot \E \Ls \|\Delta_{F,ij}\|  \|\Delta_{F, lj}\| \Rs \\
 \leq & \ \frac{1}{(\Nx+\Ny)^3} \sum_{i,j,l=1}^{\Nx+\Ny} \E \Ls\|\Lam_i^{(\gamma)}\|^2 \|\Lam_j^{(\gamma)}\|^2\|\Lam_l^{(\gamma)}\|^2  \Rs \cdot \left( \E  \|\Delta_{F,ij}\|^2 \cdot \E  \|\Delta_{F, lj}\|^2 \right)^{1/2}  \leq   \frac{C}{\Talpha} .
\end{align*}
As a result, $\Delta_{t,2} = O_p(\frac{1}{\delta_{\Ny,T}})$. By Lemma \ref{lemma: small order terms}, $\Delta_{t,1}=\frac{1}{\Nx+\Ny}\sum_{i=1}^{\Nx+\Ny} W_{ti}^Z (\tilde{\Lam}_i^{(\gamma)}-H^{(\gamma)}_i\Lam_i^{(\gamma)}) \Lam_i^{(\gamma)\top} F_t = O_p(\frac{1}{\delta_{\Ny,T}})$ and $\frac{1}{\Nx+\Ny}\sum_{i=1}^{\Nx+\Ny} W^Z_{ti} (\tilde{\Lam}_i^{(\gamma)} - H^{(\gamma)}\Lam_i^{(\gamma)})\ee_{ti} = O_p(\frac{1}{\delta_{\Ny, \Talpha}})$. For $\sqrt{T}/\Ny\rightarrow 0,$ they are all small order terms compared to $\Delta_{t,3}$. Therefore, we have $H^{(\gamma)\top} \tilde{F}^*_t = F_t+\omega_{F,1}+\omega_{F,2}+O_p(\frac{1}{\delta_{\Ny,T}})$.\\

Step 2 -- Next, we analyze the difference between $\tilde{F}_t$ and $\tilde{F}^*_t$. We have the decomposition
\begin{align*}
\tilde{F}_t^* - \tilde{F}_t   = &  \Lp \frac{1}{\Nx+\Ny} \sum_{i=1}^{\Nx+\Ny} W_{ti}^Z \tilde{\Lam}_i^{(\gamma)} \tilde{\Lam}_i^{(\gamma)\top} \Rp^{-1}\cdot\\
&\Ls  \frac{1}{\Nx+\Ny} \sum_{i=1}^{\Nx+\Ny} W_{ti}^Z \tilde{\Lam}_i^{(\gamma)} \tilde{\Lam}_i^{(\gamma)\top}  - \frac{1}{\Nx+\Ny}  \sum_{i=1}^{\Nx+\Ny} W_{ti}^Z H^{(\gamma)}\Lam_i^{(\gamma)} \Lam_i^{(\gamma)\top} H^{(\gamma)\top}  \Rs
 \tilde{F}^*_t.
\end{align*}
We let $\Delta_{\Lam,t} = \frac{1}{\Nx+\Ny} \sum_{i=1}^{\Nx+\Ny} W_{ti}^Z \tilde{\Lam}_i^{(\gamma)} \tilde{\Lam}_i^{(\gamma)\top} - \frac{1}{\Nx+\Ny}  \sum_{i=1}^{\Nx+\Ny} W_{ti}^Z H^{(\gamma)}\Lam_i^{(\gamma)} \Lam_i^{(\gamma)\top} H^{(\gamma)\top}$.
According to the proof of Theorem \ref{thm: consistency for loadings}, we have
\begin{align*}
\Delta_{\Lam,t} 
= & \frac{1}{\Nx+\Ny} \sum_{i=1}^{\Nx+\Ny}\left[H^{(\gamma)}  W_{ti}^Z \Lam_i^{(\gamma)} (\tilde{\Lam}_i^{(\gamma)} - H^{(\gamma)}_i\Lam_i^{(\gamma)} )^\top+  W_{ti}^Z  (\tilde{\Lam}_i^{(\gamma)} - H^{(\gamma)}_i\Lam_i^{(\gamma)} )\Lam_i^{(\gamma)\top} H^{(\gamma)\top}  \right]
\\
 & + \frac{1}{\Nx+\Ny}  \sum_{i=1}^{\Nx+\Ny} \left[ W_{ti}^Z H^{(\gamma)} \Lam_i^{(\gamma)} (H_i^{(\gamma)} \Lam_i^{(\gamma)} - H^{(\gamma)} \Lam_i^{(\gamma)})^\top  + W_{ti}^Z  (H^{(\gamma)}_i\Lam_i^{(\gamma)} - H^{(\gamma)} \Lam_i^{(\gamma)})(H^{(\gamma)}\Lam_i^{(\gamma)})^\top \right]  \\
& + \frac{1}{\Nx+\Ny} \sum_{i=1}^{\Nx+\Ny} W_{ti}^Z  ( \tilde{\Lam}_i^{(\gamma)} - H^{(\gamma)}\Lam_i^{(\gamma)})(\tilde{\Lam}_i^{(\gamma)} - H^{(\gamma)}\Lam_i^{(\gamma)})^\top.
\end{align*}
By Lemma \ref{lemma: small order terms} and Theorem \ref{thm: consistency for loadings}, the first and third parts on the RHS are at the order $O_p(\frac{1}{\delta_{\Ny,T}})$. Thus,
$
\Delta_{\Lam,t} 
=  \frac{1}{\Nx+\Ny}  \sum_{i=1}^{\Nx+\Ny} \left[ W_{ti}^Z H^{(\gamma)} \Lam_i^{(\gamma)} \Lam_i^{(\gamma)\top}(H_i^{(\gamma)}  - H^{(\gamma)} )^\top  + W_{ti}^Z  (H^{(\gamma)}_i - H^{(\gamma)} )(H^{(\gamma)}\Lam_i^{(\gamma)}\Lam_i^{(\gamma)\top})^\top \right]  + O_p\big(\frac{1}{\delta_{\Ny,T}}\big).
$
As analyzed in the first step, the leading term of $H^{(\gamma)}_i - H^{(\gamma)}$ is 
\[
\frac{1}{\Nx+\Ny} (\tilde{D}^{(\gamma)})^{-1} \sum_{j=1}^{\Nx+\Ny} H^{(\gamma)}\Lam_j^{(\gamma)} \Lam_j^{(\gamma)\top} \Delta_{F,ij},
\]
and thus, the leading term of $\frac{1}{\Nx+\Ny} \sum_{i=1}^{\Nx+\Ny} W_{ti}^Z  (H^{(\gamma)}_i- H^{(\gamma)}) \Lam_i^{(\gamma)} \Lam_i^{(\gamma)\top} H^{(\gamma)\top}$ is $(\tilde{D}^{(\gamma)})^{-1} H^{(\gamma)} \mathbf{X}_t^{(\gamma)} H^{(\gamma)\top}$. As a result, 
\[
\Delta_{\Lam,t} = (\tilde{D}^{(\gamma)})^{-1}H^{(\gamma)}\mathbf{X}_t^{(\gamma)}H^{(\gamma)\top} + H^{(\gamma)}\mathbf{X}_t^{(\gamma)}H^{(\gamma)\top} ((\tilde{D}^{(\gamma)})^{-1})^\top + O_p(\frac{1}{\delta_{\Ny,T}}).
\]
Note that $\mathbf{X}_t^{(\gamma)}$ is asymptotically normal with convergence rate $\sqrt{\Talpha}$, we have 
\[
\frac{1}{\Nx+\Ny} \sum_{i=1}^{\Nx+\Ny} W_{ti}^Z \tilde{\Lam}_i^{(\gamma)} \tilde{\Lam}_i^{(\gamma)\top} = \frac{1}{\Nx+\Ny} \sum_{i=1}^{\Nx+\Ny} W_{ti}^Z H^{(\gamma)} \Lam_i^{(\gamma)} \Lam_i^{(\gamma)\top} H^{(\gamma)\top} + O_p(\frac{1}{\sqrt{T}}).
\]
According to the first step, $H^{(\gamma)\top}\tilde{F}^*_t = F_t+O_p(\frac{1}{\sqrt{\delta_{\Ny,T}}})$. Combining this with $\Delta_{\Lam,t}$,
we derive
\begin{align*}
H^{(\gamma)\top}(\tilde{F}^*_t-\tilde{F}_t) = &  \underbrace{ (\tilde{\Sigma}_{\Lam,t}^{(\gamma)})^{-1}(H^{(\gamma)})^{-1} \Big[  (\tilde{D}^{(\gamma)})^{-1}H^{(\gamma)}\mathbf{X}_tH^{(\gamma)\top} + H^{(\gamma)}\mathbf{X}_tH^{(\gamma)\top} ((\tilde{D}^{(\gamma)})^{-1})^\top  \Big] (H^{(\gamma)\top})^{-1} F_t}_{\omega_{F,3}}\\
& + O_p(\frac{1}{\delta_{\Ny,T}}).
\end{align*} \\

Step 3 -- In the final step, we analyze the asymptotic distribution of $\tilde{F}_t$, which is determined by $\omega_{F,1}$, $\omega_{F,2}$ and $\omega_{F,3}$ from the first two steps, i.e.,
\[
H^{(\gamma)\top} \tilde{F}_t - F_t = \omega_{F,1} + \omega_{F,2}-\omega_{F,3} + O_p\big(\frac{1}{\delta_{\Ny,T}}\big).
\]
We have $\omega_{F,2}-\omega_{F,3} =- \tilde{\Sigma}_{\Lam,t}^{(\gamma)} \cdot \mathbf{X}_t^{(\gamma)} H^{(\gamma)\top} ((\tilde{D}^{(\gamma)})^{-1})^\top  (H^{(\gamma)\top})^{-1} F_t.$ According to Lemma \ref{lemma: HDH}, there is $(H^{(\gamma)})^{-1} (\tilde{D}^{(\gamma)})^{-1} H^{(\gamma)} \overset{p}{\rightarrow} \left( \frac{F^\top F}{T} \right)^{-1} \left( \frac{\Lam^{(\gamma)\top} \Lam^{(\gamma)}}{\Nx+\Ny} \right)^{-1} $. By Assumption \ref{assump: additional assumptions}.8 and Slutsky's Theorem, it holds that
\[
\sqrt{\Talpha} \Lp\omega_{F,2}-\omega_{F,3}\Rp \overset{d}{\rightarrow} \Ncal \left(0,  ({\Sigma}_{\Lam,t}^{(\gamma)})^{-1} \Gamma^{(\gamma),\text{miss}}_{F,t} ({\Sigma}_{\Lam,t}^{(\gamma)})^{-1}  \right) \quad \mathcal{G}^t - \text{stably},
\]
where $\Gamma^{(\gamma),\text{miss}}_{F,t} =  g_t^{(\gamma)}((\Sigma_{\Lam}^{(\gamma)})^{-1}\Sigma_F^{-1}F_t)$ with function $g_t^{(\gamma)}(\cdot)$ defined in Assumption \ref{assump: additional assumptions}.8. Additionally, $\omega_{F,1}$ and $\omega_{F,2}-\omega_{F,3}$ are asymptotically independent. If all the factors in $F_y$ are strong factors in $Y$, we can deduce that 
\begin{align*}
&\sqrt{\delta_{\Ny,\Talpha}} \Lp H^{(\gamma)\top} \tilde{F}_t - F_t \Rp \\  \overset{d}{\rightarrow}& \ \Ncal \left(0,  ({\Sigma}_{\Lam,t}^{(\gamma)})^{-1} \left[ \text{plim}\bigg(\frac{\delta_{\Ny,\Talpha}}{\Ny  }\Gamma^{(\gamma),\text{obs}}_{F,t}  + \frac{\delta_{\Ny,  \Talpha}}{\Talpha}\Gamma^{(\gamma),\text{miss}}_{F,t} \bigg)\right] ({\Sigma}_{\Lam,t}^{(\gamma)})^{-1}  \right) \quad \mathcal{G}^t - \text{stably}.
\end{align*}
If some factor $F_{t,w}$ is weak in $Y$ and its loadings are orthogonal to the loadings of the other factors, then 
\[
\sqrt{\Talpha} \Lp\omega_{F,2}-\omega_{F,3}\Rp_{w} \overset{d}{\rightarrow} \Ncal \left(0,  ({\Sigma}_{\Lam,t,w}^{(\gamma)})^{-1} \Gamma^{(\gamma),\text{miss}}_{F_w,t} ({\Sigma}_{\Lam,t,w}^{(\gamma)})^{-1}  \right) \quad \mathcal{G}^t - \text{stably},
\]
where $\Gamma^{(\gamma),\text{miss}}_{F_w,t}$ corresponds to the weak factor in $\Gamma^{(\gamma),\text{miss}}_{F,t}$,
and thus,
\begin{align*}
&\sqrt{\delta_{N_w,T}} \Lp (H^{(\gamma)\top} \tilde{F}_t)_{w} - F_{t,w} \Rp \\  \overset{d}{\rightarrow}& \ \Ncal \left(0,  ({\Sigma}_{\Lam,t,w}^{(\gamma)})^{-1} \left[ \text{plim}\bigg(\frac{\delta_{N_w,T}}{N_w }\Gamma^{(\gamma),\text{obs}}_{F_w,t}  + \frac{\delta_{N_w,T}}{\Talpha}\Gamma^{(\gamma),\text{miss}}_{F_w,t} \bigg)\right] ({\Sigma}_{\Lam,t,w}^{(\gamma)})^{-1}  \right) \quad \mathcal{G}^t - \text{stably},
\end{align*}
where $\delta_{N_w,T} = \min(N_w,T)$ and $N_w=\min(\Ny^2/{g(\Ny)},\Nx)$.
\end{proof}

\vspace{5mm}
\subsubsection{Proof of Theorem \ref{thm: asymptotic distribution}.3}
\begin{proof}
For any $t=1,\cdots,T$ and $i=1,\cdots,\Ny  $, we have the decomposition
\begin{align*}
\tilde{C}_{ti} - C_{ti} &= \tilde{F}_t^\top (\tilde{\Lam}_y)_i - F^\top_t (\LamY)_i \\
&= \tilde{F}_t^\top \Lp(\tilde{\Lam}_y)_i - H^{(\gamma)} (\LamY)_i\Rp + \Lp\tilde{F}^\top_tH^{(\gamma)} - F_t^\top\Rp (\LamY)_i .
\end{align*}
From Theorem \ref{thm: asymptotic distribution}.1 and Theorem \ref{thm: asymptotic distribution}.2, it holds that
\begin{align*}
\sqrt{\delta_{\Ny,\Talpha}} (\tilde{C}_{ti} - C_{ti} ) = & \sqrt{\delta_{\Ny ,\Talpha}} F_t^\top (H^{(\gamma)})^{-1}  \Lp \omega_{\Lam,1}+\omega_{\Lam,2}\Rp \\
& +\sqrt{\delta_{\Ny ,\Talpha}}  (\LamY)_i^\top \Lp \omega_{F,1}+\omega_{F,2}-\omega_{F,3}\Rp + o_p(1).
\end{align*}
Plugging the expression of $\omega_{\Lam,1}, \omega_{\Lam,2}, \omega_{F,1}, \omega_{F,2}$ and $\omega_{F,3}$ into the RHS, we obtain
\begin{align*}
&\sqrt{\delta_{\Ny, \Talpha}} (\tilde{C}_{ti} - C_{ti} )\\
=\ &   F_t^\top (H^{(\gamma)})^{-1} (\tilde{D}^{(\gamma)})^{-1} H^{(\gamma)} \frac{\sqrt{\delta_{\Ny ,\Talpha}}}{\Nx+\Ny} \sum_{j=1}^{\Nx+\Ny}   \Lam_j^{(\gamma)}\Lam_j^{(\gamma)\top} \frac{1}{|Q^Z_{i'j}|} \sum_{t \in Q^Z_{i'j}} F_t (\eY)_{ti}  \\
&+F_t^\top (H^{(\gamma)})^{-1}   (\tilde{D}^{(\gamma)})^{-1} H^{(\gamma)} \sqrt{\delta_{\Ny, \Talpha}}X_{i+\Nx}(\LamY)_i
\\
& + (\LamY)_i^\top \sqrt{\delta_{\Ny, \Talpha}}  \bigg( \frac{1}{\Nx+\Ny } \sum_{i=1}^{\Nx+\Ny } W_{ti}^Z \Lam_i^{(\gamma)} \Lam_i^{(\gamma)\top} \bigg)^{-1} \bigg(  \frac{1}{\Nx+\Ny} \sum_{i=1}^{\Nx+\Ny} W_{ti}^Z \Lam_i^{(\gamma)} \ee_{ti}  \bigg) \\
& -(\LamY)_i^\top \sqrt{\delta_{\Ny, \Talpha}}\tilde{\Sigma}_{\Lam,t}^{(\gamma)} \mathbf{X}_t^{(\gamma)} H^{(\gamma)\top} ((\tilde{D}^{(\gamma)})^{-1})^\top  (H^{(\gamma)\top})^{-1} F_t + o_p(1) ,
\end{align*}
where $X_i, \mathbf{X}_t$ are defined as $X_i = \frac{1}{\Nx+\Ny } \sum_{j=1}^{\Nx+\Ny } \Lam_j^{(\gamma)}\Lam_j^{(\gamma)\top} \Big( \frac{1}{|Q^Z_{ij}|}\sum_{s \in Q^Z_{ij}} F_s F_s^\top - \frac{1}{T} \sum_{s=1}^{T} F_sF_s^\top  \Big)$, $\mathbf{X}_t = \frac{1}{\Nx+\Ny } \sum_{i=1}^{\Nx+\Ny } X_i W_{ti}^Z \Lam_i^{(\gamma)}\Lam_i^{(\gamma)\top}$, and $\tilde{\Sigma}_{\Lam,t}^{(\gamma)} = \frac{1}{\Nx+\Ny }\sum_{j=1}^{\Nx+\Ny }W^Z_{tj}\Lam_j^{(\gamma)}\Lam_j^{(\gamma)\top}$. 
Note that $X_i$ and $\mathbf{X}_t$ are correlated and are asymptotically independent of other terms in $\sqrt{\delta_{\Ny,T}} (\tilde{C}_{ti} - C_{ti} )$. Therefore, by Assumption \ref{assump: additional assumptions}, Lemma \ref{lemma: HDH} and proof of Theorems \ref{thm: asymptotic distribution}.1 and \ref{thm: asymptotic distribution}.2, we can conclude that
\begin{align*}
&\sqrt{\delta_{\Ny, \Talpha}} (\tilde{C}_{ti} - C_{ti} ) \overset{d}{\rightarrow}\\
\Ncal &\left( 0, \text{plim} \left( \frac{\delta_{\Ny, \Talpha}}{\Talpha} F_t^\top \Sigma_{F}^{-1}(\Sigma_\Lam^{(\gamma)})^{-1} \Gamma_{\LamY,i}^{(\gamma),\text{obs}}  (\Sigma_\Lam^{(\gamma)})^{-1} \Sigma_{F}^{-1} F_t + \frac{\delta_{\Ny, \Talpha}}{\Ny } (\LamY)_i^\top (\Sigma_{\Lam,t}^{(\gamma)})^{-1}  \Gamma_{F,t}^{(\gamma),\text{obs}} (\Sigma_{\Lam,t}^{(\gamma)})^{-1} (\LamY)_i \right. \right. \\
& + \frac{\delta_{\Ny, \Talpha}}{\Talpha} (\LamY)_i^\top (\Sigma_{\Lam,t}^{(\gamma)})^{-1} \Gamma_{F,t}^{(\gamma),\text{miss}} (\Sigma_{\Lam,t}^{(\gamma)})^{-1}(\LamY)_i  + \frac{\delta_{\Ny, \Talpha}}{\Talpha} F_t^\top \Sigma_{F}^{-1} (\Sigma_\Lam^{(\gamma)})^{-1} 
\Gamma_{\LamY,i}^{(\gamma),\text{miss}}
(\Sigma_\Lam^{(\gamma)})^{-1}  \Sigma_{F}^{-1}  F_t \\
& \left.\left.-2 \cdot \frac{\delta_{\Ny, \Talpha}}{\Talpha} (\LamY)_i^\top (\Sigma_{\Lam,t}^{(\gamma)})^{-1}  \Gamma_{\LamY,F,i,t}^{(\gamma),\text{miss},\text{cov}}(\Sigma_\Lam^{(\gamma)})^{-1}  \Sigma_{F}^{-1}  F_t \right)\right) \quad \mathcal{G}^t - \text{stably},
\end{align*}
where $\Gamma_{\LamY,F,i,t}^{(\gamma),\text{miss},\text{cov}} =  g_{i,t}^{(\gamma),\text{cov}}((\LamY)_i, (\Sigma_\Lam^{(\gamma)})^{-1}\Sigma_F^{-1}F_t)$ with function $g_{i,t}^{(\gamma),\text{cov}}(\cdot)$ defined in Assumption \ref{assump: additional assumptions}.8. Equivalently, we have
\[
\sqrt{\delta_{\Ny, \Talpha}} (\Sigma_{C,ti}^{(\gamma)})^{-1/2} (\tilde{C}_{ti} - C_{ti} ) \overset{d}{\rightarrow} \Ncal(0,1),
\]
where 
\begin{align*}
\Sigma_{C,ti}^{(\gamma )} =&\frac{\delta_{\Ny, \Talpha}}{\Talpha} F_t^\top \Sigma_{F}^{-1}(\Sigma_\Lam^{(\gamma)})^{-1} \Gamma_{\LamY,i}^{(\gamma ),\text{obs}}  (\Sigma_\Lam^{(\gamma)})^{-1} \Sigma_{F}^{-1} F_t + \frac{\delta_{\Ny, \Talpha}}{\Ny } (\LamY)_i^\top (\Sigma_{\Lam,t}^{(\gamma)})^{-1}  \Gamma_{F,t}^{(\gamma),\text{obs}} (\Sigma_{\Lam,t}^{(\gamma)})^{-1} (\LamY)_i  \\
& + \frac{\delta_{\Ny, \Talpha}}{\Talpha} (\LamY)_i^\top (\Sigma_{\Lam,t}^{(\gamma)})^{-1} \Gamma_{F,t}^{(\gamma ),\text{miss}} (\Sigma_{\Lam,t}^{(\gamma)})^{-1}(\LamY)_i  + \frac{\delta_{\Ny, \Talpha}}{\Talpha} F_t^\top \Sigma_{F}^{-1} (\Sigma_\Lam^{(\gamma)})^{-1} 
\Gamma_{\LamY,i}^{(\gamma ),\text{miss}}
(\Sigma_\Lam^{(\gamma)})^{-1}  \Sigma_{F}^{-1}  F_t \\
& -2 \cdot \frac{\delta_{\Ny, \Talpha}}{\Talpha} (\LamY)_i^\top (\Sigma_{\Lam,t}^{(\gamma)})^{-1}  \Gamma_{\LamY,F,i,t}^{(\gamma ),\text{miss},\text{cov}}(\Sigma_\Lam^{(\gamma)})^{-1}  \Sigma_{F}^{-1}  F_t .
\end{align*}
We complete our proof.
\end{proof}

\vspace{8mm}
\subsection{Proof of Proposition \ref{prop: link of assumptions}}
We prove that Assumptions \ref{assump: obs pattern} and \ref{assump: simplified factor model} imply Assumption \ref{assump: factor model}, and Assumptions \ref{assump: obs pattern}, \ref{assump: simplified factor model}, and \ref{assump: simplified obs pattern} imply Assumption \ref{assump: additional assumptions} in the following.

\subsubsection{Assumptions \ref{assump: obs pattern} and \ref{assump: simplified factor model} imply Assumption \ref{assump: factor model}}
1. Assumptions \ref{assump: factor model}.1 and \ref{assump: factor model}.2 hold under Assumptions \ref{assump: obs pattern}, \ref{assump: simplified factor model}.1 and \ref{assump: simplified factor model}.2
\begin{proof}
Since $F_t \overset{i.i.d.}{\sim} (0,\Sigma_F)$, by LLN we have $\frac{1}{T}\sum_{t=1}^T F_tF_t^\top \overset{p}{\rightarrow} \Sigma_F$. Additionally, since $\E\|F_t\|^4$ is bounded, there is 
\begin{align*}
    \EE\Lv\sqrt{T}\Lp\frac{1}{T}\sum_{t=1}^T F_tF_t^\top-\Sigma_F\Rp\Rv^2 &= \sum_{p,q=1}^k T\cdot \EE\Ls\Lp\frac{1}{T}\sum_{t=1}^T F_{t,p}F_{t,q}-(\Sigma_{F})_{pq}\Rp^2 \Rs \\
    & = \sum_{p,q=1}^k \EE\Ls F_{t,p}^2F^2_{t,q}\Rs-(\Sigma_{F})_{pq}^2 \leq C,
\end{align*}
where $F_{t,p}$ denotes the $p$-th factor of $F_t$ and $(\Sigma_{F})_{pq}$ denotes the $(p,q)$-th entry of $\Sigma_F.$
Since the observation matrix $W^Y$ is independent of factors $F$, by similar arguments we can show that $\frac{1}{|Q^Z_{ij}|}\sum_{t\in Q^Z_{ij}} F_tF_t^\top \overset{p}{\rightarrow} \Sigma_F$ and $\EE\Big\|\sqrt{|Q^Z_{ij}|}\Big(\frac{1}{|Q^Z_{ij}|}\sum_{t\in Q^Z_{ij}}F_tF_t^\top-\Sigma_F\Big)\Big\|^2\leq C$. For factor loadings, since $(\LamX)_i \overset{i.i.d.}{\sim}(0,\Sigma_{\LamX})$, by LLN we have $\frac{1}{\Nx}\sum_{i=1}^{\Nx} (\LamX)_i(\LamX)_i^\top \overset{p}{\rightarrow} \Sigma_{\LamX}$. Other assumptions in Assumption \ref{assump: factor model}.2 automatically hold under Assumption \ref{assump: simplified factor model}.2.
\end{proof}

\vspace{3mm}
2. Assumption \ref{assump: factor model}.3 holds under Assumption \ref{assump: simplified factor model}.3
\begin{proof}
When Assumption \ref{assump: simplified factor model}.3 holds, Assumption \ref{assump: factor model}.3(a) automatically holds. Let $I(\cdot)$ be an indicator function where $I(A)=1$ if event $A$ happens and $I(A)=0$ otherwise. Under Assumption \ref{assump: simplified factor model}.3, there are $\gamma^{(\eX)}_{st,i} = \E[(\eX)_{ti}(\eX)_{si}] = \sigma^2_{\eX}\cdot I(t=s)$ and $\gamma^{(\eY)}_{st,i} = \E[(\eY)_{ti}(\eY)_{si}] = \sigma^2_{\eY}\cdot I(t=s)$. Let $\gamma_{st} = (\sigma^2_{\eX}+\sigma^2_{\eY})\cdot I(t=s)$.
It satisfies $|\gamma^{(\eX)}_{st,i}|\leq \gamma_{st}$, $|\gamma^{(\eY)}_{st,i}|\leq \gamma_{st}$ and $\sum_{s=1}^T \gamma_{st}\leq C$ for all $t$. By the same arguments, we can prove Assumption \ref{assump: factor model}.3(c). For Assumption \ref{assump: factor model}.3(d), there is $\tau_{ij,ts}^{(\eX)} = \sigma_{\eX}^2\cdot I(i=j,t=s)$. So $\sum_{j=1}^{\Nx}\sum_{s=1}^T |\tau_{ij,ts}^{(\eX)}|$ is bounded for all $i,t$. Similar arguments hold for $\tau_{ij,ts}^{(\eY)}$ and $\tau_{ij,ts}^{(\eX,\eY)}$. We denote $v_{t,ij}^{(y)} = (\eY)_{ti}(\eY)_{tj}-\EE[(\eY)_{ti}(\eY)_{tj}]$. We have $\EE[v_{t,ij}^{(y)}] = 0$, and since $\EE(\eY)^8_{ti}$ is bounded,
\begin{align*}
    \EE\Ls\frac{1}{|Q^Y_{ij}|^{1/2}}\sum_{t\in Q^Y_{ij}}v_{t,ij}^{(y)} \Rs^4 &= \frac{1}{|Q^Y_{ij}|^2} \sum_{t,s,u,w \in Q^Y_{ij}} \EE \Ls v_{t,ij}^{(y)}v_{s,ij}^{(y)}v_{u,ij}^{(y)}v_{w,ij}^{(y)} \Rs \\
    & = \frac{1}{|Q^Y_{ij}|^2} \Ls 3\sum_{t,s \in Q^Y_{ij}}\EE\Ls(v_{t,ij}^{(y)} )^2(v_{s,ij}^{(y)})^2 \Rs + \sum_{t \in Q^Y_{ij}} \EE\Ls(v_t^{(y)} )^4 \Rs\Rs \leq C.
\end{align*}
By similar arguments, we can prove that $\E \Big [ \frac{1}{T^{1/2}}\sum_{t =1}^T \big((\eX)_{ti}(\eX)_{tj}-\E[(\eX)_{ti}  \cdot(\eX)_{tj}]\big) \Big]^4$, and additionally, $\E \Big [ \frac{1}{|\QY_{jj}|^{1/2}}\sum_{t \in \QY_{jj}}\big((\eX)_{ti}({\eY})_{tj} - \E[(\eX)_{ti}({\eY})_{tj}]\big) \Big]^4$, are bounded.
\end{proof}

\vspace{3mm}
3. Assumption \ref{assump: factor model}.4 holds under Assumption \ref{assump: simplified factor model}.4
\begin{proof}
Since $F, \eY$, and $W^Y$ are independent, it is easy to see that for any $i,j=1,\cdots,\Ny$,
\begin{align*}
    \EE\Lv\frac{1}{\sqrt{|Q^Y_{ij}|}}\sum_{t\in Q^Y_{ij}}F_t (\eY)_{tj} \Rv^2 & = \sum_{p=1}^k \frac{1}{|Q^Y_{ij}|}\sum_{t\in Q^Y_{ij}} \EE[(\eY)^2_{tj}]\cdot \EE[F_{t,p}^2] \leq C.
\end{align*}
Similarly, $\EE\Lv\frac{1}{\sqrt{|Q^Y_{ii}|}}\sum_{t\in Q^Y_{ii}}F_t (\eX)_{tj'} \Rv^2 \leq C$ and $\EE\Lv\frac{1}{\sqrt{T}}\sum_{t=1}^T F_t (\eX)_{tj'} \Rv^2 \leq C$ for any $j'=1,\cdots,\Nx.$
\end{proof}

\vspace{10mm}
\subsubsection{Assumptions \ref{assump: obs pattern}, \ref{assump: simplified factor model} and \ref{assump: simplified obs pattern} imply Assumption \ref{assump: additional assumptions}}
1. Assumption \ref{assump: additional assumptions}.1 holds under Assumptions \ref{assump: obs pattern} and \ref{assump: simplified factor model}
\begin{proof}
Since factors $F$, loadings $\LamX,\LamY$, and errors $\eX,\eY$ are all i.i.d. with zero means, and $|Q^Y_{ij}|/T$ is bounded away from 0 for all $i,j$, we have
\begin{align*} 
   & \E\Lv\sqrt{\frac{T}{\Ny}}\sum_{j=1}^{\Ny} ({\LamY})_j \frac{1}{|\QY_{ij}|}\sum_{t\in \QY_{ij}}F_t^\top ({\eY})_{tj}\Rv^2 =  \sum_{p,q=1}^k \E\Ls\Bigg(\sqrt{\frac{T}{\Ny}}\sum_{j=1}^{\Ny} ({\LamY})_{j,p} \frac{1}{|\QY_{ij}|}\sum_{t\in \QY_{ij}}F_{t,q} ({\eY})_{tj}\Bigg)^2\Rs\\
    = & \sum_{p,q=1}^k \frac{T}{\Ny}\sum_{j,l=1}^{\Ny}\frac{1}{|Q^Y_{ij}|}\frac{1}{|Q^Y_{il}|}\sum_{t\in Q^Y_{ij}}\sum_{s\in Q^Y_{il}}\EE\Big[(\LamY)_{j,p}(\LamY)_{l,p}F_{t,q}F_{s,q}(\eY)_{tj}(\eY)_{sl}\Big] \\
    = & \sum_{p,q=1}^k \frac{T}{\Ny}\sum_{j=1}^{\Ny}\frac{1}{|Q^Y_{ij}|^2}\sum_{t\in Q^Y_{ij}}\EE\Big[(\LamY)_{j,p}^2F_{t,q}^2(\eY)_{tj}^2\Big] \leq   C.
\end{align*}  
Similarly we can prove that $\E\Big\|\sqrt{\frac{T}{\Nx}}\sum_{j=1}^{\Nx} ({\LamX})_j \frac{1}{|\QY_{ii}|}\sum_{t\in \QY_{ii}}F_t^\top (\eX)_{tj}\Big\|^2\leq C$.
\end{proof}

\vspace{3mm}
2. Assumption \ref{assump: additional assumptions}.2 holds under Assumptions \ref{assump: obs pattern} and \ref{assump: simplified factor model}
\begin{proof}
It holds that
\begin{align*}
    & \E \Lv\frac{\sqrt{\Ny T}}{\Ny\Nx }\sum_{i=1}^{\Ny}\sum_{j=1}^{\Nx }({\LamX})_j({\LamX})_j^\top \frac{1}{{|\QY_{ii}|}}\sum_{t\in \QY_{ii}} F_t({\LamY})_i^\top ({\eY})_{ti} \Rv^2 \\
    = & \sum_{p,q=1}^k \frac{T}{\Ny \Nx^2} \E\Ls\Lp\sum_{i=1}^{\Ny}\sum_{j=1}^{\Nx}\sum_{r=1}^k\frac{1}{|Q^Y_{ii}|} \sum_{t\in \QY_{ii}} (\LamX)_{j,p}(\LamX)_{j,r} F_{t,r}({\LamY})_{i,q} ({\eY})_{ti} \Rp^2  \Rs \\
    = & \sum_{p,q,r,m=1}^k \frac{T}{\Ny \Nx^2} \sum_{i=1}^{\Ny}\sum_{j,l=1}^{\Nx}\frac{1}{|Q^Y_{ii}|^2} \sum_{t\in \QY_{ii}}\E\big[(\LamX)_{j,p}(\LamX)_{j,r}(\LamX)_{l,p}(\LamX)_{l,m}\big]\cdot\E[F_{t,r}F_{t,m}]\cdot\E[({\LamY})_{i,q}^2]\cdot\E[(\eY)_{ti}^2]\\
    \leq &\sum_{p,q,r,m=1}^k \sum_{r,m=1}^k\frac{T}{|\QY_{ii}|}C \ \leq \ C,
\end{align*} 
where the first inequality holds since $\LamX$ has bounded fourth moments. By similar arguments, we can prove the other three bounds in Assumption \ref{assump: additional assumptions}.2.
\end{proof}

\vspace{3mm}
3. Assumption \ref{assump: additional assumptions}.3 holds under Assumptions \ref{assump: obs pattern} and \ref{assump: simplified factor model}
\begin{proof}
For simplicity, we let $v_{t,ij}^{(y)} = (\eY)_{ti}(\eY)_{tj}-\E\Ls(\eY)_{ti}(\eY)_{tj}\Rs$. We have $\E[v_{t,ij}^{(y)}] = 0$. According to Assumption \ref{assump: simplified factor model}, $v_{t,ij}^{(y)}$ is independent of $v_{s,hl}^{(y)}$ for any $s\neq t$ and is independent of loadings $\LamY$. As a result,
\begin{align*}
    \E\Lv \sqrt{\frac{T}{\Ny}}
 \sum_{j=1}^{\Ny} ({\LamY})_j \frac{1}{|\QY_{ij}|}\sum_{t\in \QY_{ij}}v_{t,ij}^{(y)} \Rv^2 &= \sum_{p=1}^k \frac{T}{\Ny} \E\Ls\Lp\sum_{j=1}^{\Ny} ({\LamY})_{j,p} \frac{1}{|\QY_{ij}|}\sum_{t\in \QY_{ij}}v_{t,ij}^{(y)}\Rp^2\Rs \\
 & = \sum_{p=1}^k \frac{T}{\Ny}\sum_{j,l=1}^{\Ny} \frac{1}{|\QY_{ij}|} \frac{1}{|\QY_{il}|}\sum_{t\in \QY_{ij}}\sum_{s\in \QY_{il}} \E\Big[ ({\LamY})_{j,p}({\LamY})_{l,p} v_{t,ij}^{(y)}v_{s,il}^{(y)}\Big] \\
 & = \sum_{p=1}^k \frac{T}{\Ny}\sum_{j=1}^{\Ny} \frac{1}{|\QY_{ij}|^2} \sum_{t\in \QY_{ij}} \E\Ls ({\LamY})_{j,p}^2\Rs\cdot \E\Ls ( v_{t,ij}^{(y)})^2\Rs\\
 & \leq C.
\end{align*}
Similarly, we can prove that $\E\Big\| \sqrt{\frac{T}{\Nx}}
 \sum_{j=1}^{\Nx} ({\LamX})_j \frac{1}{|\QY_{ii}|}\sum_{t\in \QY_{ii}}\big((\eX)_{tj}(\eY)_{ti}-\E[(\eX)_{tj}(\eY)_{ti}]\big)\Big\|^2\leq C$.
 
\end{proof}

\vspace{3mm}
4. Assumption \ref{assump: additional assumptions}.4 holds under Assumptions \ref{assump: obs pattern} and \ref{assump: simplified factor model}
\begin{proof}
In the following, we only prove that $\E \Big\|\sqrt{\frac{T}{\Ny^3}}\sum_{i,j=1}^{\Ny}\frac{1}{|\QY_{ij}|}\sum_{s\in \QY_{ij}}\WY_{ti}({\LamY})_j({\eY})_{ti} (({\eY})_{si} ({\eY})_{sj} - \E[({\eY})_{si}({\eY})_{sj}]) \Big\|^2 \leq C$ and other bounds can be proved similarly. As before, we define $v_{t,ij}^{(y)} = (\eY)_{ti}(\eY)_{tj}-\E[(\eY)_{ti}(\eY)_{tj}]$. It holds that
\begin{align*}
    &\E\Lv\sqrt{\frac{T}{\Ny^3}}\sum_{i,j=1}^{\Ny}\frac{1}{|\QY_{ij}|}\sum_{s\in \QY_{ij}}W^Y_{ti}({\LamY})_j({\eY})_{ti} v_{s,ij}^{(y)} \Rv^2 \\
    = & \sum_{p=1}^k \frac{T}{\Ny^3} \sum_{i,j,h,l=1}^{\Ny} \frac{1}{|\QY_{ij}|}\frac{1}{|\QY_{hl}|}\sum_{s\in \QY_{ij}}\sum_{u\in \QY_{hl}}\EE\Ls \WY_{ti}\WY_{th}({\LamY})_{j,p}({\LamY})_{l,p} ({\eY})_{ti}({\eY})_{th} v_{s,ij}^{(y)}v_{u,hl}^{(y)}\Rs \\
    = & \sum_{p=1}^k \frac{T}{\Ny^3} \sum_{i,j,h,l=1}^{\Ny} \frac{1}{|\QY_{ij}|}\frac{1}{|\QY_{hl}|}\sum_{s\in \QY_{ij}\cap\QY_{hl}}\EE\Ls\WY_{ti}\WY_{th}({\LamY})_{j,p}({\LamY})_{l,p} \Rs \cdot\E\Ls({\eY})_{ti}({\eY})_{th} v_{s,ij}^{(y)}v_{s,hl}^{(y)}\Rs.
\end{align*}
Note that if the indices $i,h,j,l$ take four different values, the RHS of the above equation will equal zero. Thus, the RHS of the above equation can be bounded by $C$. 
\end{proof}

\vspace{3mm}
5. Assumption \ref{assump: additional assumptions}.5 holds under Assumptions \ref{assump: obs pattern} and \ref{assump: simplified factor model}
\begin{proof}
For any $i,j=1,\cdots,\Ny,$ we let $\Delta_{F,ij} = \frac{1}{|\QY_{ij}|}\sum_{s\in \QY_{ij}}F_sF_s^\top - \frac{1}{T}\sum_{s=1}^TF_s F_s^\top$. It holds that
\begin{align*}
    \E\Lv\sqrt{\frac{{T}}{\Ny}} \sum_{i=1}^{\Ny} \Delta_{F,ij}\WY_{ti} ({\LamY})_i ({\eY})_{ti}\Rv^2 =& \frac{T}{\Ny} \sum_{i,l=1}^{\Ny} \E\Ls\WY_{ti}\WY_{tl} ({\LamY})_i^\top ({\LamY})_l \Delta_{F,ij} \Delta_{F,lj}({\eY})_{ti}({\eY})_{tl}   \Rs \\
    = & \frac{T}{\Ny} \sum_{i=1}^{\Ny} \E\Big[\WY_{ti} ({\LamY})_i^\top({\LamY})_i \Delta_{F,ij} \Delta_{F,lj}   \Big]\cdot \E\Ls({\eY})_{ti}^2\Rs .
\end{align*}
We will prove in part 8 that $\E\|\Delta_{F,ij}\|^2\leq C/T$. Once this holds, we can bound the RHS of the above equation by $C$. We can prove other bounds following similar arguments.
\end{proof}

\vspace{3mm}
6. Assumption \ref{assump: additional assumptions}.6 holds under Assumptions \ref{assump: obs pattern}, \ref{assump: simplified factor model} and \ref{assump: simplified obs pattern}
\begin{proof}
Since factors $F_t\overset{i.i.d.}{\sim} (0,\Sigma_F)$, idiosyncratic errors $(\eY)_{ti}\overset{i.i.d.}{\sim} (0,\sigma_{e_y}^2)$, and they are independent of the observation pattern, by CLT we have
$\frac{1}{|Q^Y_{ij}|^{1/2}}\sum_{t\in Q^Y_{ij}}F_t(\eY)_{ti} \overset{d}{\rightarrow} \Ncal(0, \sigma_{\eY}^2\Sigma_F)$.
Therefore, by Assumption \ref{assump: obs pattern}.2,
\[
\frac{\sqrt{T}}{|Q^Y_{ij}|}\sum_{t\in Q^Y_{ij}}F_t (\eY)_{ti} \overset{d}{\rightarrow} \Ncal\Lp 0,\frac{1}{q_{ij}}\sigma_{\eY}^2\Sigma_F\Rp, \qquad \forall \ i,j=1,\cdots,\Ny.
\]
Based on Slutsky's Theorem,
\begin{align*}
    \frac{1}{\Nx}\sum_{j=1}^{\Nx} (\LamX)_j (\LamX)_j^\top \frac{\sqrt{T}}{|Q^Y_{ii}|}\sum_{t\in Q^Y_{ii}}F_t (\eY)_{ti} 
    \overset{d}{\rightarrow} \Ncal\Lp 0,\frac{1}{q_{ii}} \Sigma_{\LamX} \sigma_{\eY}^2\Sigma_F \Sigma_{\LamX}\Rp.
\end{align*}
For any $i,j,l$, we have the asymptotic covariance matrix 
\begin{align*}
     \ACov\Bigg(\frac{\sqrt{T}}{|Q^Y_{ij}|}\sum_{t\in Q^Y_{ij}}F_t (\eY)_{ti},\ \frac{\sqrt{T}}{|Q^Y_{il}|}\sum_{t\in Q^Y_{il}}F_t (\eY)_{ti}\Bigg) &= \lim_{T\rightarrow\infty}  \frac{T}{|Q_{ij}^Y|\cdot |Q_{il}^Y|}\sum_{t\in Q^Y_{ij}}\sum_{s\in Q^Y_{il}}\EE\Ls F_t F_s^\top(\eY)_{ti}(\eY)_{si}\Rs \\
    &= \frac{q_{ij,il}}{q_{ij}q_{il}} \sigma^2_{\eY} \Sigma_F.
\end{align*}
This implies that
\begin{align*} 
    \frac{1}{\Ny}\sum_{j=1}^{\Ny} (\LamY)_j (\LamY)_j^\top \frac{\sqrt{T}}{|Q^Y_{ij}|}\sum_{t\in Q^Y_{ij}}F_t (\eY)_{ti} &=  \EE\Ls(\LamY)_j (\LamY)_j^\top\Rs\cdot \frac{1}{\Ny}\sum_{j=1}^{\Ny} 
    \frac{\sqrt{T}}{|Q^Y_{ij}|}\sum_{t\in Q^Y_{ij}}F_t (\eY)_{ti} + o_p(1) \\
    &\overset{d}{\rightarrow} \Ncal\bigg(0, \lim_{\Ny\rightarrow \infty} \frac{1}{\Ny^2}\sum_{j,l=1}^{\Ny}\frac{q_{ij,il}}{q_{ij}q_{il}} \Sigma_{\LamY} \sigma_{\eY}^2\Sigma_F \Sigma_{\LamY}\bigg) \\
    &= \Ncal\Lp 0, \omega_i^{(2,3)} \Sigma_{\LamY} \sigma_{\eY}^2\Sigma_F \Sigma_{\LamY}\Rp,
\end{align*}
where $\omega_i^{(2,3)}$ is defined in Assumption \ref{assump: simplified obs pattern}.2.
Furthermore, the two parts in Assumption \ref{assump: additional assumptions}.6 are jointly asymptotically normal with covariance matrix $\frac{1}{q_{ii}}\Sigma_{\LamX} \sigma_{\eY}^2\Sigma_F \Sigma_{\LamY}.$
Combining these terms, $\Gamma^{(\gamma),\text{obs}}_{{\LamY},i}$ equals to
\[
\Gamma^{(\gamma),\text{obs}}_{{\LamY},i} =  \lim \frac{\Nx^2}{(\Nx+\Ny)^2}  \sigma_{\eY}^2 \Ls \frac{1}{q_{ii}}\Lp\Sigma_{\LamX}+r \Sigma_{\LamY}\Rp \Sigma_F\Lp\Sigma_{\LamX}+r \Sigma_{\LamY}\Rp + \Lp\omega_i^{(2,3)}-\frac{1}{q_{ii}}\Rp r^2 \Sigma_{\LamY}\Sigma_F\Sigma_{\LamY} \Rs.
\]
\end{proof}

\vspace{3mm}
7. Assumption \ref{assump: additional assumptions}.7 holds under Assumptions \ref{assump: obs pattern}, \ref{assump: simplified factor model} and \ref{assump: simplified obs pattern}
\begin{proof}
By Assumptions \ref{assump: simplified factor model}, \ref{assump: simplified obs pattern}, and CLT, we have
\[
\frac{1}{\sqrt{\Nx}}\sum_{i=1}^{\Nx} (\LamX)_i (\eX)_{ti} \overset{d}{\rightarrow} \Ncal(0,\sigma_{\eX}^2\Sigma_{\LamX}),
\]
and  
\[
\frac{1}{\sqrt{\Ny}}\sum_{i=1}^{\Ny} \WY_{ti}(\LamY)_i (\eY)_{ti} \overset{d}{\rightarrow} \Ncal(0, \sigma_{\eY}^2\Sigma_{{\LamY},t}).
\]
Furthermore, $\frac{1}{\sqrt{\Nx}}\sum_{i=1}^{\Nx} (\LamX)_i (\eX)_{ti}$ and $\frac{1}{\sqrt{\Ny}}\sum_{i=1}^{\Ny} \WY_{ti}(\LamY)_i (\eY)_{ti}$ are asymptotically independent. So $\Gamma_{F,t}^{(\gamma),\text{obs}}$ defined in Assumption \ref{assump: additional assumptions}.7 simplifies to 
\begin{align*}
\Gamma_{F,t}^{(\gamma),\text{obs}}
= \lim\frac{\Nx^2}{(\Nx+\Ny)^2}\Lp\frac{\Ny}{\Nx}\sigma_{\eX}^2\Sigma_{\LamX}+ r^2\sigma_{\eY}^2\Sigma_{{\LamY},t}\Rp.
\end{align*}
Suppose there is some weak factor $F_w$ in $Y$ whose loading $\sum_{i=1}^{N_y}(\LamY)^2_{i,w}$ grows at the rate $g(\Ny) = p_w N_y$, where $p_w$ is defined in Assumption \ref{assump: simplified factor model}.2. For this weak factor $F_w$, $p_w$ decays to 0 but is nonzero as $N_y$ grows. We have $\frac{1}{\sqrt{g(\Ny)}} \sum_{i=1}^{\Ny} W^Y_{ti} (\Lam_y)_{i,w} (\eY)_{ti}  \overset{d}{\rightarrow}\Ncal(0,\sigma_{e_y}^2\Sigma_{{\Lam_{y}},t,w})$. Then, there is $\Gamma_{F_w,t}^{(\gamma),\text{obs}}=\lim\frac{\Nx^2}{(\Nx+\Ny)^2}\Lp \frac{ N_w}{ N_x}\sigma_{\eX}^2\Sigma_{\LamX} +r^2 \frac{{ p_wN_w}}{N_y} \sigma_{e_y}^2\Sigma_{{\Lam_{y}},t,w}\Rp$, where $N_w= \min(\Ny^2/g(\Ny),\Nx)$. 
\end{proof}

\vspace{3mm}
8. Assumption \ref{assump: additional assumptions}.8 holds under Assumptions \ref{assump: obs pattern}, \ref{assump: simplified factor model} and \ref{assump: simplified obs pattern}

\begin{proof}
It suffices to show that $\Lp\sqrt{T} \cdot \text{vec}(X_i^{(\gamma)}), \sqrt{T} \cdot \text{vec}(\mathbf{X}_t^{(\gamma)})\Rp$ is asymptotically normal. 

\vspace{3mm}
\emph{Step 1 -- $\sqrt{T} \cdot \text{vec}(X_i^{(\gamma)})$ is asymptotically normal}

\vspace{2mm}
Observe that $X_i^{(\gamma)} = \frac{1}{N_x+N_y}\sum_{j=1}^{\Nx+\Ny}\Lam_j^{(\gamma)}(\Lam_j^{(\gamma)})^\top \Big( \frac{1}{|Q^Z_{ij}|}\sum_{t\in Q^Z_{ij}}F_tF_t^\top - \frac{1}{T}\sum_{t=1}^T F_tF_t^\top \Big)$, where $\Lam^{(\gamma)} = [\LamX; \sqrt{\gamma} \LamY]\in \R^{(\Nx+\Ny)\times k}$ is the combined loadings. For any $k \times k$ matrices $A$ and $B$, there is $\text{vec}(AB) = (I_k \otimes A) \cdot \text{vec}(B).$ Therefore,  $\sqrt{T}\cdot\text{vec}(X_i^{(\gamma)})$ can be written as
\[
\sqrt{T}\cdot \text{vec}(X_i^{(\gamma)}) = \frac{1}{\Nx + \Ny}\sum_{j=1}^{\Nx+\Ny}\Lp I_k \otimes \Lam_j^{(\gamma)}(\Lam_j^{(\gamma)})^\top \Rp  \text{vec}\Lp\frac{\sqrt{T}}{|Q^Z_{ij}|}\sum_{t\in Q^Z_{ij}}F_tF_t^\top - \frac{1}{\sqrt{T}}\sum_{t=1}^TF_tF_t^\top \Rp.
\] 
To simplify notation, we let $v_{ij} = \text{vec}\Big(\frac{\sqrt{T}}{|Q^Z_{ij}|}\sum_{t\in Q^Z_{ij}}F_tF_t^\top - \frac{1}{\sqrt{T}}\sum_{t=1}^TF_tF_t^\top \Big)$. Furthermore, we define $q^Z_{ij}= \lim_{T\rightarrow\infty}|Q^Z_{ij}|/T$ and $q^Z_{ij,hl} =\lim_{T\rightarrow\infty}|Q^Z_{ij}\cap Q^Z_{hl}|/T$ for any $i,j,h,l$. According to CLT,
 \begin{align*}
  v_{ij} &= \sqrt{T}\Lp\frac{1}{|Q^Z_{ij}|} - \frac{1}{T}\Rp\sum_{t\in Q^Z_{ij}}\text{vec}(F_tF_t^\top) - \frac{1}{\sqrt{T}}\sum_{t\notin Q^Z_{ij}} \text{vec}(F_tF_t^\top)  \overset{d}{\rightarrow} \Ncal\left(0, \bigg(\frac{1}{q^Z_{ij}}-1\bigg) \Xi_F \right),  
\end{align*} 
where $\Xi_F = \Var(\text{vec}(F_tF_t^\top))$.
Additionally, the asymptotic covariance of $v_{ij}$ and $v_{hl}$ for any $i,j,h,l$ can be calculated as
\begin{align*}
    & \ACov(v_{ij},v_{hl}) \\
    = \ & \ACov\Bigg(\frac{\sqrt{T}}{|Q^Z_{ij}|}\sum_{t\in Q^Z_{ij}}\text{vec}(F_tF_t^\top) - \frac{1}{\sqrt{T}}\sum_{t=1}^T\text{vec}(F_tF_t^\top), \frac{\sqrt{T}}{|Q^Z_{hl}|}\sum_{t\in Q^Z_{hl}}\text{vec}(F_tF_t^\top) - \frac{1}{\sqrt{T}}\sum_{t=1}^T\text{vec}(F_tF_t^\top) \Bigg) \\
    =\ &\lim_{T\rightarrow\infty} \frac{T}{|Q^Z_{ij}|\cdot |Q^Z_{hl}|} \left|Q^Z_{ij}\cap Q^Z_{hl}\right|\cdot \Var(\text{vec}(F_tF_t^\top)) - 2\Var(\text{vec}(F_tF_t^\top))+ \Var(\text{vec}(F_tF_t^\top)) \\
    =\ & \bigg(\frac{q^Z_{ij,hl}}{q^Z_{ij}q^Z_{hl}}-1 \bigg)\Xi_F.
\end{align*} 
Particularly, we have $\ACov(v_{ij},v_{il}) = (q^Z_{ij,il}/(q^Z_{ij}q^Z_{il})-1)\Xi_F$.
When $i=\Nx+1,\cdots,\Nx+\Ny$, if $j$ or $l$ is chosen from $1,\cdots,\Nx$, then $q^Z_{ij,il}/(q^Z_{ij}q^Z_{il}) = 1/q_{i'i'}$ with $i'=i-\Nx$; otherwise, $q^Z_{ij,il}/(q^Z_{ij}q^Z_{il}) = q_{i'j',i'l'}/(q_{i'j'}q_{i'l'}),$ where $j'=j-\Nx$ and $l'=l-\Nx$. Let $u_{jl} =  (I_k \otimes \Lam_j^{(\gamma)}\Lam_j^{(\gamma)\top})\Xi_F (I_k \otimes \Lam_l^{(\gamma)}\Lam_l^{(\gamma)\top}).$ Observe that $u_{jl}$ is independent with $u_{mn}$ for distinct $j,l,m,n$. Thus, for any $p,q=1,\cdots,k^2,$ we have
 \begin{align*}
    \EE \Ls \frac{1}{(\Nx+\Ny)^2}\sum_{j,l=1}^{\Nx+\Ny} \bigg(\frac{q^Z_{ij,hl}}{q^Z_{ij}q^Z_{hl}}-1 \bigg)\Lp u_{jl,pq}-\E[u_{jl,pq}]\Rp \Rs^2 = O\Big(\frac{1}{\Ny}\Big).
\end{align*} 
When $q_{ij}$ and $q_{ij,hl}$ are independent of $(\LamX)_m(\LamX)_m^\top$ and $(\LamY)_m(\LamY)_m^\top$ for any $i,j,h,l,m$, 
\begin{align*}
    & \ACov\left( \sqrt{T}\cdot \text{vec}(X_i^{(\gamma)}),\sqrt{T}\cdot \text{vec}(X_i^{(\gamma)})\right) \\  =\ & \lim\frac{1}{(\Nx+\Ny)^2}\sum_{j,l=1}^{\Nx+\Ny} \left(I_k \otimes \Lam_j^{(\gamma)}\Lam_j^{(\gamma)\top}\right) \cdot \Cov(v_{ij},v_{il})\cdot \left(I_k \otimes \Lam_l^{(\gamma)}\Lam_l^{(\gamma)\top}\right)\\
    =\ & \lim \frac{1}{(\Nx+\Ny)^2}\sum_{j,l=1}^{\Nx+\Ny} \bigg(\frac{q^Z_{ij,hl}}{q^Z_{ij}q^Z_{hl}}-1 \bigg)\left(I_k \otimes \Lam_j^{(\gamma)}\Lam_j^{(\gamma)\top}\right) \Xi_F \left(I_k \otimes \Lam_l^{(\gamma)}\Lam_l^{(\gamma)\top}\right) \\
    =\ & \lim \frac{\Nx^2}{(\Nx+\Ny)^2}\bigg[\Big(\frac{1}{q_{i'i'}}-1\Big)\Lp I_k \otimes \Lp \Sigma_{\LamX} + r\Sigma_{\LamY}\Rp\Rp\Xi_F\Lp I_k \otimes \Lp \Sigma_{\LamX} + r\Sigma_{\LamY}\Rp\Rp\\
    &+ \Big(\omega_{i'}^{(2,3)}-\frac{1}{q_{i'i'}}\Big)\Lp I_k \otimes r \Sigma_{\LamY}\Rp\Xi_F\Lp I_k \otimes r \Sigma_{\LamY}\Rp\bigg],
\end{align*} 
where $r = \gamma \cdot \Ny/\Nx$ and $i'=i-\Nx$.
Therefore, for any $i=\Nx+1,\cdots,\Nx+\Ny$,
\begin{align*}
& \sqrt{T}\cdot \text{vec}(X_i^{(\gamma)})  \\
\overset{d}{\rightarrow} \ &  \Ncal\bigg( 0,\lim \frac{\Nx^2}{(\Nx+\Ny)^2}\bigg[\Big(\frac{1}{q_{i'i'}}-1\Big)\Lp I_k \otimes \Lp \Sigma_{\LamX} + r\Sigma_{\LamY}\Rp\Rp\Xi_F\Lp I_k \otimes \Lp \Sigma_{\LamX} + r\Sigma_{\LamY}\Rp\Rp\\
&\qquad\qquad\qquad\qquad\qquad+ \Big(\omega_{i'}^{(2,3)}-\frac{1}{q_{i'i'}}\Big)\Lp I_k \otimes r \Sigma_{\LamY}\Rp\Xi_F\Lp I_k \otimes r \Sigma_{\LamY}\Rp\bigg] \bigg).
\end{align*} \\

\emph{Step 2 -- $\sqrt{T} \cdot \text{vec}(\mathbf{X}_t^{(\gamma)})$ is asymptotically normal}

\vspace{2mm}
In Assumption \ref{assump: additional assumptions}.8, we have $\sqrt{T}\cdot\mathbf{X}_t^{(\gamma)} = \frac{1}{(\Nx+\Ny)^2}\sum_{i,j=1}^{\Nx+\Ny} \Lam_j^{(\gamma)}\Lam_j^{(\gamma)\top} \Big(\frac{\sqrt{T}}{|Q^Z_{ij}|}\sum_{t\in Q^Z_{ij}}F_tF_t^\top - \frac{1}{\sqrt{T}}\sum_{t=1}^TF_tF_t^\top \Big) W^Z_{ti} \Lam_i^{(\gamma)}\Lam_i^{(\gamma)\top}$, and its vectorized form can be written as \[\sqrt{T}\cdot \text{vec}(\mathbf{X}_t^{(\gamma)})= \frac{1}{(\Nx+\Ny)^2}\sum_{i,j=1}^{\Nx+\Ny}W^Z_{ti} \big(\Lam_i^{(\gamma)}\Lam_i^{(\gamma)\top} \otimes I_k\big)\big(I_k\otimes \Lam_j^{(\gamma)}\Lam_j^{(\gamma)\top}\big)v_{ij},\]
where $v_{ij}$ is defined in the first step. With similar arguments as in Step 1, we can show the existence of the following limit
\begin{align*}
\lim \frac{1}{(\Nx+\Ny)^4}\sum_{i,j,h,l=1}^{\Nx+\Ny} & \bigg(\frac{q^Z_{ij,hl}}{q^Z_{ij}q^Z_{hl}}-1 \bigg) W^Z_{ti}W^Z_{tl} \left(\Lam_i^{(\gamma)}\Lam_i^{(\gamma)\top} \otimes I_k\right)\left(I_k \otimes \Lam_j^{(\gamma)}\Lam_j^{(\gamma)\top}\right)  \\
& \Xi_F \left(I_k \otimes \Lam_h^{(\gamma)}\Lam_h^{(\gamma)\top}\right)\left(\Lam_l^{(\gamma)}\Lam_l^{(\gamma)\top} \otimes I_k\right).
\end{align*} 
Observe that when $i,j=1,\cdots,\Nx$ or $h,l=1,\cdots,\Nx$, $q^Z_{ij,hl}/(q^Z_{ij}q^Z_{hl}) = 1$; When $i,h=1,\cdots,\Nx$ and $j,l=\Nx+1,\cdots,\Nx+\Ny$, $q^Z_{ij,hl}/(q^Z_{ij}q^Z_{hl}) = q_{j'l'}/(q_{j'j'}q_{l'l'})$ with $j'=j-\Nx$ and $l'=l-\Nx$; When $h=1,\cdots,\Nx$ and $i,j,l=\Nx+1,\cdots,\Nx+\Ny$, $q^Z_{ij,hl}/(q^Z_{ij}q^Z_{hl}) = q_{i'j',l'l'}/(q_{i'j'}q_{l'l'})$ with $i'=i-\Nx$; When $i,j,h,l=\Nx+1,\cdots,\Nx+\Ny$, $q^Z_{ij,hl}/(q^Z_{ij}q^Z_{hl}) = q_{i'j',h'l'}/(q_{i'j'}q_{h'l'})$. By symmetry, other cases of $i,j,h,l$ can be considered similarly. Additionally, observe that for any $k\times k$ matrices $A$ and $B$, there are $(A\otimes I_k)(I_k\otimes B) = A\otimes B$ and $(I_k\otimes A)(B\otimes I_k) = B\otimes A$. As a result, we have
\begin{align*}
    &\ \ACov\left(\sqrt{T}\cdot \text{vec}(\mathbf{X}_t^{(\gamma)}) , \sqrt{T}\cdot \text{vec}(\mathbf{X}_t^{(\gamma)})\right)\\
    = & \lim\frac{1}{(\Nx+\Ny)^4}\sum_{i,j,l,h=1}^{\Nx+\Ny} W^Z_{ti}W^Z_{tl}\left(\Lam_i^{(\gamma)}\Lam_i^{(\gamma)\top} \otimes  \Lam_j^{(\gamma)}\Lam_j^{(\gamma)\top}\right) \Cov(v_{ij},v_{lh}) \left(\Lam_l^{(\gamma)}\Lam_l^{(\gamma)\top}\otimes \Lam_h^{(\gamma)}\Lam_h^{(\gamma)\top}\right) \\
    = & \lim \frac{1}{(\Nx+\Ny)^4}\sum_{i,j,h,l=1}^{\Nx+\Ny}  \bigg(\frac{q^Z_{ij,hl}}{q^Z_{ij}q^Z_{hl}}-1 \bigg)\left(W^Z_{ti}\Lam_i^{(\gamma)}\Lam_i^{(\gamma)\top} \otimes  \Lam_j^{(\gamma)}\Lam_j^{(\gamma)\top}\right) \Xi_F \left(W^Z_{tl}\Lam_l^{(\gamma)}\Lam_l^{(\gamma)\top}\otimes \Lam_h^{(\gamma)}\Lam_h^{(\gamma)\top}\right)\\
    = & \lim  \frac{\Nx^4}{(\Nx+\Ny)^4} \frac{1}{\Ny^2}\sum_{i,j=1}^{\Ny} \Big(\frac{q_{ij}}{q_{ii}q_{jj}}-1\Big) \Lp \Sigma_{\LamX} \otimes r\Sigma_{\LamY} + r\Sigma_{{\LamY},t} \otimes \Sigma_{\LamX} \Rp \Xi_F \Lp \Sigma_{\LamX} \otimes r\Sigma_{\LamY} + r\Sigma_{{\LamY},t} \otimes \Sigma_{\LamX} \Rp \\
    &+\lim  \frac{\Nx^4}{(\Nx+\Ny)^4} \frac{1}{\Ny^3}\sum_{i,j,l=1}^{\Ny} \Big(\frac{q_{jj,il}}{q_{jj}q_{il}}-1\Big) \big[\Lp \Sigma_{\LamX} \otimes r\Sigma_{\LamY} + r\Sigma_{{\LamY},t} \otimes \Sigma_{\LamX} \Rp \Xi_F \Lp  r\Sigma_{{\LamY},t} \otimes r\Sigma_{\LamY} \Rp +\\
    & \Lp  r\Sigma_{{\LamY},t} \otimes r\Sigma_{\LamY} \Rp\Xi_F\Lp \Sigma_{\LamX} \otimes r\Sigma_{\LamY} + r\Sigma_{{\LamY},t} \otimes \Sigma_{\LamX} \Rp\big] \\
    & + \lim  \frac{\Nx^4}{(\Nx+\Ny)^4} \frac{1}{\Ny^4}\sum_{i,j,l,h=1}^{\Ny} \Big(\frac{q_{il,jh}}{q_{jh}q_{il}}-1\Big) \Lp r\Sigma_{\LamY,t} \otimes r\Sigma_{\LamY}  \Rp \Xi_F \Lp r\Sigma_{\LamY,t} \otimes r\Sigma_{\LamY}  \Rp \\
    = & \lim  \frac{\Nx^4}{(\Nx+\Ny)^4} \bigg[  \Lp \Sigma_{\LamX} \otimes r\Sigma_{\LamY} + r\Sigma_{{\LamY},t} \otimes \Sigma_{\LamX} \Rp \Xi_F \Big((\omega^{(1)}-1)\Lp \Sigma_{\LamX} \otimes r\Sigma_{\LamY} + r\Sigma_{{\LamY},t} \otimes \Sigma_{\LamX} \Rp \\
    & + (\omega^{(2)}-1)\Lp  r\Sigma_{{\LamY},t} \otimes r\Sigma_{\LamY} \Rp\Big) +  \Lp r\Sigma_{\LamY,t} \otimes r\Sigma_{\LamY}\Rp \Xi_F \Big((\omega^{(2)}-1)\Lp \Sigma_{\LamX} \otimes r\Sigma_{\LamY} + r\Sigma_{{\LamY},t} \otimes \Sigma_{\LamX} \Rp \\
    &+ (\omega^{(3)}-1)  \Lp r\Sigma_{\LamY,t} \otimes r\Sigma_{\LamY}\Rp\Big)\bigg].
\end{align*} 
It follows that
\begin{align*}
    &\sqrt{T}\cdot \text{vec}(\mathbf{X}_t^{(\gamma)}) \overset{d}{\rightarrow} \\  & \Ncal\bigg(0,\lim \frac{\Nx^4}{(\Nx+\Ny)^4}\bigg[  \Lp \Sigma_{\LamX} \otimes r\Sigma_{\LamY} + r\Sigma_{{\LamY},t} \otimes \Sigma_{\LamX} \Rp \Xi_F \Big((\omega^{(1)}-1)\Lp \Sigma_{\LamX} \otimes r\Sigma_{\LamY} + r\Sigma_{{\LamY},t} \otimes \Sigma_{\LamX} \Rp  \\&+ (\omega^{(2)}-1)\Lp  r\Sigma_{{\LamY},t} \otimes r\Sigma_{\LamY} \Rp\Big) +  \Lp r\Sigma_{\LamY,t} \otimes r\Sigma_{\LamY}\Rp \Xi_F \Big((\omega^{(2)}-1)\Lp \Sigma_{\LamX} \otimes r\Sigma_{\LamY} + r\Sigma_{{\LamY},t} \otimes \Sigma_{\LamX} \Rp\\
    & + (\omega^{(3)}-1)  \Lp r\Sigma_{\LamY,t} \otimes r\Sigma_{\LamY}\Rp\Big)\bigg]\bigg).
\end{align*} \\

\emph{Step 3 -- $\sqrt{T} \cdot \text{vec}(X_i^{(\gamma)})$ and $\sqrt{T} \cdot \text{vec}(\mathbf{X}_t^{(\gamma)})$ are jointly asymptotically normal}

\vspace{2mm}
We only consider the case where $i=\Nx+1,\cdots,\Nx+\Ny$. It is easy to see that $\sqrt{T} \cdot \text{vec}(X_i^{(\gamma)})$ and $\sqrt{T} \cdot \text{vec}(\mathbf{X}_t^{(\gamma)})$ are jointly asymptotically normal since both of their randomnesses come from $v_{ij}$. By similar arguments, we can show the existence of the following limit
\begin{align*}
\lim \frac{1}{(\Nx+\Ny)^3}\sum_{j,l,h=1}^{\Nx+\Ny} \bigg(\frac{q^Z_{ij,hl}}{q^Z_{ij}q^Z_{hl}}-1 \bigg) W^Z_{tl} \left(\Lam_l^{(\gamma)}\Lam_l^{(\gamma)\top} \otimes I_k\right)\left(I_k \otimes \Lam_h^{(\gamma)}\Lam_h^{(\gamma)\top}\right)  \Xi_F \left(I_k \otimes \Lam_j^{(\gamma)}\Lam_j^{(\gamma)\top}\right).
\end{align*} 
Observe that when $h,l=1,\cdots,\Nx$, $q^Z_{ij,hl}/(q^Z_{ij}q^Z_{hl}) = 1$; When $j,h=1,\cdots,\Nx, l =\Nx+1,\cdots,\Nx+\Ny$, $q^Z_{ij,hl}/(q^Z_{ij}q^Z_{hl}) = q_{i'l'}/(q_{i'i'}q_{l'l'})$ with $i'=i-\Nx$ and $l'=l-\Nx$; When $h=1,\cdots,\Nx, j,l=\Nx+1,\cdots,\Nx+\Ny$, $q^Z_{ij,hl}/(q^Z_{ij}q^Z_{hl}) = q_{i'j',l'l'}/(q_{i'j'}q_{l'l'})$ with $j'=j-\Nx$; When $j,h,l=\Nx+1,\cdots,\Nx+\Ny$, $q^Z_{ij,hl}/(q^Z_{ij}q^Z_{hl}) = q_{i'j',h'l'}/(q_{i'j'}q_{h'l'})$, where $h'=h-\Nx$. Other cases can be similarly considered. As a result, we have
\begin{align*}
    & \ACov\left(\sqrt{T}\cdot \text{vec}(\mathbf{X}_t^{(\gamma)}) , \sqrt{T}\cdot \text{vec}(X_i^{(\gamma)})\right)\\
    = & \lim\frac{1}{(\Nx+\Ny)^3}\sum_{j,h,l=1}^{\Nx+\Ny} W^Z_{tl}\left(\Lam_l^{(\gamma)}\Lam_l^{(\gamma)\top} \otimes I_k\right)\left(I_k\otimes \Lam_h^{(\gamma)}\Lam_h^{(\gamma)\top}\right) \Cov(v_{lh},v_{ij}) \left(I_k\otimes \Lam_j^{(\gamma)}\Lam_j^{(\gamma)\top}\right) \\
    = & \lim \frac{1}{(\Nx+\Ny)^3}\sum_{j,l,h=1}^{\Nx+\Ny} \bigg(\frac{q^Z_{ij,hl}}{q^Z_{ij}q^Z_{hl}}-1 \bigg)  \left(W^Z_{tl}\Lam_l^{(\gamma)}\Lam_l^{(\gamma)\top} \otimes I_k\right)\left(I_k \otimes \Lam_h^{(\gamma)}\Lam_h^{(\gamma)\top}\right)  \Xi_F \left(I_k \otimes \Lam_j^{(\gamma)}\Lam_j^{(\gamma)\top}\right)\\
    = & \lim\frac{\Nx^3}{(\Nx+\Ny)^3}\bigg[ \Lp\Sigma_{\LamX}\otimes r\Sigma_{\LamY} + r\Sigma_{\LamY,t}\otimes \Sigma_{\LamX}\Rp \Xi_F\Lp (\omega_{i'}^{(1)}-1)(I_k\otimes \Sigma_{\LamX}) + (\omega_{i'}^{(2,2)}-1) \big(I_k \otimes r\Sigma_{\LamY}\big)\Rp \\
    & \qquad + \Lp r\Sigma_{\LamY,t}\otimes r\Sigma_{\LamY}\Rp \Xi_F \Lp(\omega_{i'}^{(2,1)}-1)(I_k\otimes \Sigma_{\LamX}) + (\omega_{i'}^{(3)}-1) \big(I_k\otimes r\Sigma_{\LamY}\big)\Rp \bigg],
\end{align*} 
where $\omega_i^{(1)}, \omega_i^{(2,1)}, \omega_i^{(2,2)}$ and $\omega_i^{(3)}$ are defined in Assumption \ref{assump: simplified obs pattern}.
For the special case where $\omega_i^{(1)}=\omega_i^{(2,1)}=\omega_i^{(2,2)}=\omega_i^{(3)} = \omega_i,$ the asymptotic covariance matrix is simplified to
\begin{align*}
    & \ACov\left(\sqrt{T}\cdot \text{vec}(\mathbf{X}_t^{(\gamma)}) , \sqrt{T}\cdot \text{vec}(X_i^{(\gamma)})\right)\\
    = & (\omega_{i'}-1) \cdot\lim\frac{\Nx^3}{(\Nx+\Ny)^3} \Lp\Lp\Sigma_{\LamX}+r\Sigma_{\LamY,t}\Rp\otimes r\Sigma_{\LamY} + r\Sigma_{\LamY,t}\otimes \Sigma_{\LamX}\Rp\Xi_F\Lp I_k\otimes \Lp\Sigma_{\LamX} +   r\Sigma_{\LamY}\Rp\Rp.
\end{align*} 
\end{proof}

\vspace{8mm}
\subsection{Proof of Corollary \ref{thm: simplified factor model}}
\begin{proof}
According to Proposition \ref{prop: link of assumptions}, Theorem \ref{thm: asymptotic distribution} holds under the simplified assumptions in Corollary \ref{thm: simplified factor model}. As a result, we just need to calculate the asymptotic variances in Theorem \ref{thm: asymptotic distribution} under the simplified model.

1. The asymptotic variance of loadings: 

We have $\Sigma_{\Lam}^{(\gamma)} = \lim \frac{\Nx}{\Nx+\Ny}\Lp\Sigma_{\LamX}+r\Sigma_{\LamY}\Rp$ and $\Sigma_{\Lam,t}^{(\gamma)} = \lim \frac{\Nx}{\Nx+\Ny}\Lp\Sigma_{\LamX}+r\Sigma_{\LamY,t}\Rp$, where $r = \gamma\cdot \Ny/\Nx$. According to the proof of Proposition \ref{prop: link of assumptions}, $\Gamma^{(\gamma),\text{obs}}_{{\LamY},i}$ defined in Assumption \ref{assump: additional assumptions}.6 equals to
\[
\Gamma^{(\gamma),\text{obs}}_{{\LamY},i} =  \lim \frac{\Nx^2}{(\Nx+\Ny)^2} \cdot \sigma_{\eY}^2 \Ls \frac{1}{q_{ii}}\Lp\Sigma_{\LamX}+r \Sigma_{\LamY}\Rp \Sigma_F\Lp\Sigma_{\LamX}+r \Sigma_{\LamY}\Rp + \Lp\omega_i^{(2,3)}-\frac{1}{q_{ii}}\Rp r^2 \Sigma_{\LamY}\Sigma_F\Sigma_{\LamY} \Rs.
\]
As a result, the first part of $\Sigma^{(\gamma)}_{\LamY,i}$ is
\begin{align*}
    &\Sigma_F^{-1}(\Sigma_{\Lam}^{(\gamma)})^{-1}\Gamma^{(\gamma),\text{obs}}_{{\LamY},i} (\Sigma_{\Lam}^{(\gamma)})^{-1}\Sigma_F^{-1} \\ =& \frac{1}{q_{ii}}\sigma_{\eY}^2\Sigma_F^{-1} + \Lp \omega_i^{(2,3)}-\frac{1}{q_{ii}}\Rp \sigma_{\eY}^2 r^2 \Sigma_F^{-1}\Lp\Sigma_{\LamX}+r\Sigma_{\LamY}\Rp^{-1}\Sigma_{\LamY}\Sigma_F\Sigma_{\LamY}\Lp\Sigma_{\LamX}+r\Sigma_{\LamY}\Rp^{-1}\Sigma_F^{-1}.
\end{align*}

Next, consider $\Gamma_{{\LamY},i}^{(\gamma),\text{miss}}=h^{(\gamma)}_{i+\Nx}((\LamY)_i)$ defined in Theorem \ref{thm: asymptotic distribution}.1. From the proof of Proposition \ref{prop: link of assumptions}, for any $i=1,\cdots,\Ny$,
\begin{align*}
& \sqrt{T}\cdot \text{vec}(X_{i+N_x}^{(\gamma)})  \\
\overset{d}{\rightarrow} \ &  \Ncal\bigg( 0,\lim \frac{\Nx^2}{(\Nx+\Ny)^2}\bigg[\Big(\frac{1}{q_{ii}}-1\Big)\Lp I_k \otimes \Lp \Sigma_{\LamX} + r\Sigma_{\LamY}\Rp\Rp\Xi_F\Lp I_k \otimes \Lp \Sigma_{\LamX} + r\Sigma_{\LamY}\Rp\Rp\\
&\qquad\qquad\qquad\qquad\qquad+ \Big(\omega_{i}^{(2,3)}-\frac{1}{q_{ii}}\Big)\Lp I_k \otimes r \Sigma_{\LamY}\Rp\Xi_F\Lp I_k \otimes r \Sigma_{\LamY}\Rp\bigg] \bigg),
\end{align*}
where $\Xi_F= \Var(\text{vec}(F_tF_t^\top)).$ Therefore, $\Gamma_{{\LamY},i}^{(\gamma),\text{miss}}$ can be calculated as
\begin{align*}
    \Gamma_{{\LamY},i}^{(\gamma),\text{miss}}=&\lim  ((\LamY)_i^\top \otimes I_k)\E \Ls T\cdot \text{vec}(X_{i+\Nx}) \text{vec}(X_{i+\Nx})^\top\Rs ((\LamY)_i \otimes I_k) \\
    = & \lim \frac{\Nx^2}{(\Nx+\Ny)^2}\bigg[\Big(\frac{1}{q_{ii}}-1\Big)\Lp (\LamY)_i^\top \otimes \Lp \Sigma_{\LamX} + r\Sigma_{\LamY}\Rp\Rp\Xi_F\Lp (\LamY)_i \otimes \Lp \Sigma_{\LamX} + r\Sigma_{\LamY}\Rp\Rp\\
    & \ \ \quad\qquad\qquad\qquad+ \Big(\omega_{i}^{(2,3)}-\frac{1}{q_{ii}}\Big)\Lp (\LamY)_i^\top \otimes r \Sigma_{\LamY}\Rp\Xi_F\Lp (\LamY)_i \otimes r \Sigma_{\LamY}\Rp\bigg].
\end{align*}
Thus, the second part of $\Sigma^{(\gamma)}_{\LamY,i}$ is
\begin{align*}
    &\Sigma_F^{-1}(\Sigma_{\Lam}^{(\gamma)})^{-1}\Gamma^{(\gamma),\text{miss}}_{{\LamY},i} (\Sigma_{\Lam}^{(\gamma)})^{-1}\Sigma_F^{-1} \\
    = & \Lp \frac{1}{q_{ii}}-1\Rp \Sigma_F^{-1}((\LamY)_i^\top \otimes I_k) \Xi_F \Lp(\LamY)_i \otimes I_k\Rp \Sigma_F^{-1} \\
    &+ \Lp\omega_i^{(2,3)}-\frac{1}{q_{ii}}\Rp\Sigma_F^{-1}\Lp \Sigma_{\LamX} + r\Sigma_{\LamY}\Rp^{-1} \Lp (\LamY)_i^\top \otimes r \Sigma_{\LamY}\Rp\Xi_F\Lp (\LamY)_i \otimes r \Sigma_{\LamY}\Rp \Lp \Sigma_{\LamX} + r\Sigma_{\LamY}\Rp^{-1}\Sigma_F^{-1}.
\end{align*}
Combining the two parts, we can deduce that
\[
\sqrt{T}(\Sigma_{\LamY,i}^{(\gamma)})^{-1/2}\left((H^{(\gamma)})^{-1}(\tilde{\Lam}_y)_i-(\LamY)_i\right)\overset{d}{\rightarrow}\Ncal \left(0, I_k \right),
\]
where
\begin{align*}
    \Sigma_{\LamY,i}^{(\gamma)} =& \frac{1}{q_{ii}}\sigma_{\eY}^2\Sigma_F^{-1} + \Lp\frac{1}{q_{ii}}-1\Rp \Sigma_F^{-1}((\LamY)_i^\top \otimes I_k)\Xi_F ((\LamY)_i \otimes I_k)\Sigma_F^{-1} + \Lp\omega_i^{(2,3)}-\frac{1}{q_{ii}}\Rp\Sigma_F^{-1} \\
    & \Lp \Sigma_{\LamX} + r\Sigma_{\LamY}\Rp^{-1}\Ls \sigma_{\eY}^2r^2\Sigma_{\LamY}\Sigma_F\Sigma_{\LamY} + \Lp (\LamY)_i^\top \otimes r \Sigma_{\LamY}\Rp\Xi_F\Lp (\LamY)_i \otimes r \Sigma_{\LamY}\Rp \Rs \Lp \Sigma_{\LamX} + r\Sigma_{\LamY}\Rp^{-1}\Sigma_F^{-1}.
\end{align*} \\

\noindent 2. The asymptotic variance of factors:

Based on the proof of Proposition \ref{prop: link of assumptions}, we have
\begin{align*}
\Gamma_{F,t}^{(\gamma),\text{obs}}
= \lim\frac{\Nx^2}{(\Nx+\Ny)^2}\Lp\frac{\Ny}{\Nx}\sigma_{\eX}^2\Sigma_{\LamX}+ r^2\sigma_{\eY}^2\Sigma_{{\LamY},t}\Rp.    
\end{align*}
As a result, the first part of $\Sigma_{F,t}^{(\gamma)}$ can be calculated as
\begin{align*}
&(\Sigma_{{\Lam},t}^{(\gamma)})^{-1}\Gamma_{F,t}^{(\gamma),\text{obs}}(\Sigma_{{\Lam},t}^{(\gamma)})^{-1} = \lim \Lp\Sigma_{\LamX}+r\Sigma_{{\LamY},t}\Rp^{-1} \Lp\frac{\Ny}{\Nx}\sigma_{\eX}^2\Sigma_{\LamX}+ r^2\sigma_{\eY}^2\Sigma_{{\LamY},t}\Rp\Lp\Sigma_{\LamX}+r\Sigma_{{\LamY},t}\Rp^{-1}.
\end{align*}
Consider the second part where $\Gamma^{(\gamma),\text{miss}}_{F,t}=g_t^{(\gamma)}\Lp(\Sigma_\Lam^{(\gamma)})^{-1}\Sigma_F^{-1}F_t\Rp$. We have proved in Proposition \ref{prop: link of assumptions}
\begin{align*}
    &\sqrt{T}\cdot \text{vec}(\mathbf{X}_t^{(\gamma)}) \overset{d}{\rightarrow} \\  & \Ncal\bigg(0,\lim \frac{\Nx^4}{(\Nx+\Ny)^4}\bigg[  \Lp \Sigma_{\LamX} \otimes r\Sigma_{\LamY} + r\Sigma_{{\LamY},t} \otimes \Sigma_{\LamX} \Rp \Xi_F \Big((\omega^{(1)}-1)\Lp \Sigma_{\LamX} \otimes r\Sigma_{\LamY} + r\Sigma_{{\LamY},t} \otimes \Sigma_{\LamX} \Rp  \\&+ (\omega^{(2)}-1)\Lp  r\Sigma_{{\LamY},t} \otimes r\Sigma_{\LamY} \Rp\Big) +  \Lp r\Sigma_{\LamY,t} \otimes r\Sigma_{\LamY}\Rp \Xi_F \Big((\omega^{(2)}-1)\Lp \Sigma_{\LamX} \otimes r\Sigma_{\LamY} + r\Sigma_{{\LamY},t} \otimes \Sigma_{\LamX} \Rp\\
    & + (\omega^{(3)}-1)  \Lp r\Sigma_{\LamY,t} \otimes r\Sigma_{\LamY}\Rp\Big)\bigg]\bigg).
\end{align*}
Therefore, we have
\begin{align*}
&\Gamma^{(\gamma),\text{miss}}_{F,t} \\
=& \lim  \Lp I_k \otimes F_t^\top\Sigma_{F}^{-1}(\Sigma_\Lam^{(\gamma)})^{-1}\Rp\E\Ls T\cdot\text{vec}(\mathbf{X}_t^{(\gamma)})\text{vec}(\mathbf{X}_t^{(\gamma)})^\top\Rs \Lp I_k \otimes (\Sigma_\Lam^{(\gamma)})^{-1}\Sigma_{F}^{-1}F_t\Rp\\
= & \lim  \frac{\Nx^2}{(\Nx+\Ny)^2}  \Lp I_k \otimes F_t^\top\Sigma_{F}^{-1}(\Sigma_{\LamX}+ r \Sigma_{\LamY})^{-1}\Rp \bigg[  \Lp \Sigma_{\LamX} \otimes r\Sigma_{\LamY} + r\Sigma_{{\LamY},t} \otimes \Sigma_{\LamX} \Rp \Xi_F \Big((\omega^{(1)}-1)\\
&\Lp \Sigma_{\LamX} \otimes r\Sigma_{\LamY} + r\Sigma_{{\LamY},t} \otimes \Sigma_{\LamX} \Rp + (\omega^{(2)}-1)\Lp  r\Sigma_{{\LamY},t} \otimes r\Sigma_{\LamY} \Rp\Big) +  \Lp r\Sigma_{\LamY,t} \otimes r\Sigma_{\LamY}\Rp \Xi_F \Big((\omega^{(2)}-1)\\
& \Lp \Sigma_{\LamX} \otimes r\Sigma_{\LamY} + r\Sigma_{{\LamY},t} \otimes \Sigma_{\LamX} \Rp+ (\omega^{(3)}-1)  \Lp r\Sigma_{\LamY,t} \otimes r\Sigma_{\LamY}\Rp\Big)\bigg]
\Lp I_k \otimes (\Sigma_{\LamX}+r\Sigma_{\LamY})^{-1}\Sigma_{F}^{-1}F_t\Rp.
\end{align*}
If all the factors in $F_y$ are strong factors in $Y$, combining the two parts, we have
\[
\sqrt{\delta_{\Ny,T}} (\Sigma_{F,t}^{(\gamma)})^{-1/2} \Lp H^{(\gamma)\top} \tilde{F}_t -F_t \Rp\overset{d}{\rightarrow} \Ncal(0,I_k),
\]
where 
\begin{align*}
    \Sigma_{F,t}^{(\gamma)} = & \Lp\Sigma_{\LamX}+r\Sigma_{{\LamY},t}\Rp^{-1}\Bigg[\frac{\delta_{\Ny T}}{\Ny}\Lp\frac{\Ny}{\Nx}\sigma_{\eX}^2\Sigma_{\LamX}+ r^2\sigma_{\eY}^2\Sigma_{{\LamY},t}\Rp + \frac{\delta_{\Ny T}}{T} \Lp I_k \otimes F_t^\top\Sigma_{F}^{-1}\big(\Sigma_{\LamX}+ r \Sigma_{\LamY}\big)^{-1}\Rp \\
    &\bigg[  \Lp \Sigma_{\LamX} \otimes r\Sigma_{\LamY} + r\Sigma_{{\LamY},t} \otimes \Sigma_{\LamX} \Rp \Xi_F \Big((\omega^{(1)}-1)\Lp \Sigma_{\LamX} \otimes r\Sigma_{\LamY} + r\Sigma_{{\LamY},t} \otimes \Sigma_{\LamX} \Rp + (\omega^{(2)}-1)\\
    &\Lp  r\Sigma_{{\LamY},t} \otimes r\Sigma_{\LamY} \Rp\Big) +  \Lp r\Sigma_{\LamY,t} \otimes r\Sigma_{\LamY}\Rp \Xi_F \Big((\omega^{(2)}-1) \Lp \Sigma_{\LamX} \otimes r\Sigma_{\LamY} + r\Sigma_{{\LamY},t} \otimes \Sigma_{\LamX} \Rp+ (\omega^{(3)}-1) \\
    &\Lp r\Sigma_{\LamY,t} \otimes r\Sigma_{\LamY}\Rp\Big)\bigg]
    \Lp I_k \otimes (\Sigma_{\LamX}+r\Sigma_{\LamY})^{-1}\Sigma_{F}^{-1}F_t\Rp\Bigg]\Lp\Sigma_{\LamX}+r\Sigma_{{\LamY},t}\Rp^{-1}.
\end{align*}
If some factor $F_w$ is weak in $Y$, then the asymptotic distribution of the estimation of this weak factor is
\[
\sqrt{\delta_{N_w,T}} (\Sigma_{F_w,t}^{(\gamma)})^{-1/2} \Lp (H^{(\gamma)\top} \tilde{F}_t)_w -F_{t,w} \Rp\overset{d}{\rightarrow} \Ncal(0,I_k).
\]
We have
\begin{align*}
    \Sigma_{F_w,t}^{(\gamma)}  = &  (\Sigma_{\LamX}+r\Sigma_{{\LamY},t})_w^{-1}\Ls\frac{\delta_{N_w,T}}{N_w}\Lp\frac{N_w}{\Nx}\sigma_{\eX}^2\Sigma_{\LamX,w}+ r^2\frac{p_w N_w}{\Ny}\sigma_{\eY}^2\Sigma_{{\LamY},t,w}\Rp \Rs (\Sigma_{\LamX}+r\Sigma_{{\LamY},t})^{-1}_w \\  & + \frac{\delta_{N_w,T}}{T}\cdot \Sigma_{F_w,t}^{(\gamma),miss},
\end{align*}
where $N_w = \min(\Ny/p_w, \Nx)$,
$\frac{1}{\Ny p_w}\sum_{i=1}^{\Ny}W^Y_{ti}(\LamY)_{i,w}(\LamY)_{i,w}^\top \overset{p}{\rightarrow} \Ncal(0,\Sigma_{{\LamY},t,w})$, $(\Sigma_{\LamX}+r\Sigma_{{\LamY},t})_w^{-1}$, $\Sigma_{\LamX,w}$ and $\Sigma^{(\gamma),\text{miss}}_{F_{w},t}$ are respectively the diagonal block in $(\Sigma_{\LamX}+r\Sigma_{{\LamY},t})^{-1}$, $\Sigma_{\LamX}$ and $\Sigma^{(\gamma),\text{miss}}_{F,t}$ corresponding to the weak factors. \\

\noindent 3. The asymptotic variance of common components:

For $\Gamma_{\LamY,F,i,t}^{(\gamma),\text{miss},\text{cov}} = g_{i,t}^{(\gamma),\text{cov}}\big((\LamY)_i, (\Sigma_\Lam^{(\gamma)})^{-1}\Sigma_F^{-1}F_t\big)$, define $\Psi_{i,t}^{(\gamma),\text{cov}} = \lim \EE\Ls T\cdot \text{vec}(X_{i+\Nx}^{(\gamma)})\text{vec}(\textbf{X}_t^{(\gamma)})^\top\Rs$ with $i=1,\cdots,\Ny$. It holds that $\Gamma_{\LamY,F,i,t}^{(\gamma),\text{miss},\text{cov}} = \big(I_k \otimes F_t^\top\Sigma_{F}^{-1}\big(\Sigma_\Lam^{(\gamma)}\big)^{-1}\big)\Psi_{i,t}^{(\gamma),{\rm cov}} ((\LamY)_i \otimes I_k)$. For general $\omega_i^{(1)}, \omega_i^{(2,1)}, \omega_i^{(2,2)}$ and $\omega_i^{(3)}$, we have
\begin{align*}
    \Psi_{i,t}^{(\gamma),{\rm cov}} = & \lim \frac{\Nx^3}{(\Nx+\Ny)^3}\bigg[ \Lp\Sigma_{\LamX}\otimes r\Sigma_{\LamY} + r\Sigma_{\LamY,t}\otimes \Sigma_{\LamX}\Rp\Xi_F\Big( (\omega_i^{(1)}-1)(I_k\otimes \Sigma_{\LamX}) + (\omega_i^{(2,2)}-1)  \\
    &(I_k \otimes r\Sigma_{\LamY})\Big)  + (r\Sigma_{\LamY,t}\otimes r\Sigma_{\LamY}) \Xi_F \Big((\omega_i^{(2,1)}-1)(I_k\otimes \Sigma_{\LamX}) + (\omega_i^{(3)}-1) (I_k\otimes r\Sigma_{\LamY})\Big) \bigg].
\end{align*}
As a result, there is
\begin{align*}
    &\Sigma_{\LamY,F,i,t}^{(\gamma),{\rm miss,cov}} \\= & (\Sigma_{{\Lam},t}^{(\gamma)})^{-1}\Gamma_{\LamY,F,i,t}^{(\gamma),\text{miss},\text{cov}}(\Sigma_{\Lam}^{(\gamma)})^{-1}\Sigma_F^{-1} \\
    = &\Lp\Sigma_{\LamX}+r\Sigma_{\LamY,t}\Rp^{-1} \Lp I_k \otimes F_t^\top\Sigma_{F}^{-1}\big(\Sigma_{\LamX}+ r \Sigma_{\LamY}\big)^{-1}\Rp\Big[ \Lp\Sigma_{\LamX}\otimes r\Sigma_{\LamY} + r\Sigma_{\LamY,t}\otimes  \Sigma_{\LamX}\Rp\Xi_F \\
    &\Big( (\omega_i^{(1)}-1)((\LamY)_i\otimes \Sigma_{\LamX}) + (\omega_i^{(2,2)}-1) ((\LamY)_i \otimes r\Sigma_{\LamY})\Big)+ (r\Sigma_{\LamY,t}\otimes r\Sigma_{\LamY})\Xi_F\\ &   \Lp(\omega_i^{(2,1)}-1)((\LamY)_i\otimes \Sigma_{\LamX}) + (\omega_i^{(3)}-1) \Lp(\LamY)_i\otimes r\Sigma_{\LamY}\Rp\Rp \Big]\Lp\Sigma_{\LamX}+r\Sigma_{\LamY}\Rp^{-1}\Sigma_F^{-1}.
\end{align*}
In the special case where $\omega_i^{(1)}=\omega_i^{(2,1)}=\omega_i^{(2,2)}=\omega_i^{(3)} = \omega_i,$ $\Sigma_{\LamY,F,i,t}^{(\gamma),{\rm miss,cov}}$ can be simplified as
\begin{align*}
    \Sigma_{\LamY,F,i,t}^{(\gamma),{\rm miss,cov}} = & (\omega_i-1)r \Lp\Sigma_{\LamX}+r\Sigma_{{\LamY},t}\Rp^{-1} \Big[\Lp I_k \otimes F_t^\top\Sigma_{F}^{-1}\big(\Sigma_{\LamX}+ r \Sigma_{\LamY}\big)^{-1}\Rp( \Sigma_{\LamX} \otimes \Sigma_{\LamY})+ \\
    &\Sigma_{{\LamY},t}\otimes (F_t^\top \Sigma_F^{-1})  \Big] \Xi_F \Lp (\LamY)_i \otimes I_k\Rp\Sigma_F^{-1}.
\end{align*}
This completes the proof.
\end{proof}

\vspace{8mm}
\subsection{Proof of Proposition \ref{prop: consistency example}}
\begin{proof}
We prove Proposition \ref{prop: consistency example} as a special case of the general Theorem \ref{thm: consistency for loadings}. In the following, we prove that the assumptions in Proposition \ref{prop: consistency example} imply the general model specified by Assumptions \ref{assump: factor shift}, \ref{assump: obs pattern}, \ref{assump: factor model} and \ref{assump: additional assumptions}.

Under the data generating process described in Section \ref{subsec: consistency effect}, the first factor is a strong factor in target $Y$ and is not contained in auxiliary panel $X$, so Assumption \ref{assump: factor shift} holds. Since there is no missing observation in $Y$, Assumption \ref{assump: obs pattern} automatically holds. The factors $F_t \overset{i.i.d.}{\sim} (0,\Sigma_F)$, where $\Sigma_F$ is a diagonal matrix with diagonal elements equal to $\sigma_F^2$. The loadings in $X$ follow
\[
\frac{1}{\Nx}\sum_{i=1}^{\Ny}(\LamX)_i (\LamX)_i^\top \overset{p}{\rightarrow} \Sigma_{\LamX} = \begin{bmatrix}
    0 & 0 \\
    0 & \sigma_{\LamX}^2
\end{bmatrix}.
\]
And for loadings in $Y$, we have
\[
\frac{1}{\Ny}\sum_{i=1}^{\Ny} (\LamY)_i (\LamY)_i^\top \overset{p}{\rightarrow} \Sigma_{\LamY} = \begin{bmatrix}
    \sigma_{\LamY}^2 & 0 \\
    0 & 0
\end{bmatrix}.
\]
Thus, $\Sigma_{\LamX} + \Sigma_{\LamY}$ is a positive definite matrix. Additionally, the factors and loadings have bounded fourth moments. These conditions imply Assumptions \ref{assump: factor model}.1 and \ref{assump: factor model}.2. In the data generating process in Section \ref{subsec: consistency effect}, the idiosyncratic errors have bounded eighth moments and are drawn i.i.d. from $(\eX)_{ti}\overset{i.i.d.}{\sim}(0,\sigma_{\eX}^2)$ and $(\eY)_{ti}\overset{i.i.d.}{\sim}(0,\sigma_{\eY}^2)$. As a result, Assumptions \ref{assump: factor model}.3 and \ref{assump: factor model}.4 hold. It is straightforward to show that the moment conditions in Assumption \ref{assump: additional assumptions} hold under this model.
\end{proof}

\vspace{8mm}
\subsection{Proof of Proposition \ref{prop: efficiency example}}
\begin{proof}
The data generating process and observation pattern in Section \ref{subsec: efficiency effect} are a special case of the simplified factor model in Assumptions \ref{assump: simplified factor model} and \ref{assump: simplified obs pattern}. 

Since factors, loadings, and idiosyncratic errors are i.i.d. distributed, we automatically obtain Assumption \ref{assump: simplified factor model}. Under the missing-at-random observation pattern with observed probability $p$, we can show that the quantities in Assumption \ref{assump: simplified obs pattern} satisfy $\omega_i^{(2,3)} = 1/p$, $\omega_i^{(1)} = \omega_i^{(2,1)} = \omega_i^{(2,2)} = \omega_i^{(3)} = 1,$ and $\omega^{(1)} = \omega^{(2)} = \omega^{(3)} = 1$. Plugging these parameters into the asymptotic variance of the estimated common components of $Y$ in Corollary \ref{thm: simplified factor model}.3, we have
\begin{align*}
\Sigma_{C,ti}^{(\gamma)} = & \frac{\delta_{{\Ny}T}}{T}\frac{\sigma_{\eY}^2}{p\sigma_F^2}F_t^2 + \frac{\delta_{{\Ny}T}}{T}\left(\frac{1}{p}-1\right) \sigma_F^{-4}\Var(F_t^2)({\LamY})_i^2F_t^2 \\
& + \frac{\delta_{{\Ny}T}}{{\Ny}}({\LamY})_i^2\left(\sigma_{\LamX}^2 + \gamma\frac{{\Ny}}{{\Nx}} p\sigma_{\LamY}^2\right)^{-2} \left(\frac{{\Ny}}{{\Nx}}\sigma_{\LamX}^2\sigma_{\eX}^2 + \gamma^2\frac{\Ny^2}{\Nx^2} p\sigma_{\LamY}^2\sigma_{\eY}^2 \right).
\end{align*}
The partial derivative of $\Sigma_{C,ti}^{(\gamma)}$ with respect to $\gamma$ equals to
\begin{align*}
\frac{\partial \Sigma_{C,ti}^{(\gamma)}}{\partial \gamma} = &\frac{\delta_{{\Ny}T}}{{\Ny}}({\LamY})_i^2\left(\sigma_{\LamX}^2 + \gamma\frac{{\Ny}}{{\Nx}} p\sigma_{\LamY}^2\right)^{-2} \Bigg[ -2\frac{\Ny}{\Nx}p\sigma_{\LamY}^2 \left(\sigma_{\LamX}^2 + \gamma\frac{{\Ny}}{{\Nx}} p\sigma_{\LamY}^2\right)^{-1} \\ &\cdot \left(\frac{{\Ny}}{{\Nx}}\sigma_{\LamX}^2\sigma_{\eX}^2 + \gamma^2\frac{\Ny^2}{\Nx^2} p\sigma_{\LamY}^2\sigma_{\eY}^2 \right) + 2\gamma \frac{\Ny^2}{\Nx^2} p\sigma_{\LamY}^2\sigma_{\eY}^2 \Bigg].
\end{align*}
We let $\partial \Sigma_{C,ti}^{(\gamma)}/ \partial \gamma = 0$ and obtain $\gamma^* = \sigma_{\eX}^2/\sigma_{\eY}^2$. Furthermore, $\partial^2 \Sigma_{C,ti}^{(\gamma^*)}/ \partial \gamma^2 > 0$. As a result, the optimized $\gamma$ that minimizes $\Sigma_{C,ti}^{(\gamma)}$ is $\gamma^* = \sigma_{\eX}^2/\sigma_{\eY}^2$ for any $i$ and $t$. This completes the proof.
\end{proof}

\vspace{8mm}
\subsection{Proof of Proposition \ref{prop: finite Y 1}}
\begin{proof}
    If we select $\gamma = r$ for some constant $r$, then target-PCA degenerates to the PCA estimator in \cite{xiong2019large} applied only to auxiliary data $X$. Therefore, when all the factors can be identified in $X$, Theorem \ref{thm: consistency for loadings} holds with convergence rate $\delta_{\Nx,T} = \min(\Nx,T)$, Theorem \ref{thm: asymptotic distribution}.2 holds with convergence rate $\sqrt{\delta_{\Nx,T}}$, and the asymptotic variance is independent of $Y$.
\end{proof}

\vspace{8mm}
\subsection{Proof of Proposition \ref{prop: finite Y 2}}
\begin{proof}
    The proof of Proposition \ref{prop: finite Y 2} is analogous to the proof of Theorems \ref{thm: consistency for loadings} and \ref{thm: asymptotic distribution} and follows the same arguments. Let $\delta_{\Nx,T,M} = \min(\delta_{\Nx,T},M)$. When $M\rightarrow \infty$ and $\gamma = r\cdot \Nx$ for some constant $r$, the upper bound for $(\Nx+\Ny)^{-2}\sum_{i,j=1}^{\Nx+\Ny} \gamma^2(i,j)$ in Lemma \ref{lemma: order of the four terms} changes to $C/\delta_{\Nx,T,M}$, while Lemmas \ref{lemma: D} and \ref{lemma: H,Q} remain the same. As a result, Theorem \ref{thm: consistency for loadings} holds with convergence rate $\delta_{\Nx,T,M}$. For Theorem \ref{thm: asymptotic distribution}, we replace the convergence rate $\delta_{\Ny,T}$ in Lemma \ref{lemma: small order terms} by $\delta_{\Nx,T,M}$. With modified assumptions, we can prove the statements of Theorem \ref{thm: asymptotic distribution}.2 with convergence rate $\sqrt{\delta_{\Nx,T,M}}$.
\end{proof}

\vspace{8mm}
\subsection{Proof of Proposition \ref{thm: asymptotic distribution weak signals}}
\begin{proof}
The proof of Proposition \ref{thm: asymptotic distribution weak signals} is analogous to the proof of Theorem \ref{thm: asymptotic distribution}. Specifically, when the number of time periods, for which any two and four units are observed, is proportional to $T^{\alpha}$, the rates based on $T$ in Lemmas \ref{lemma: order of the four terms}, \ref{lemma: order of 4 terms factor first} and \ref{lemma: small order terms} will simply be replaced by $T^\alpha$, while Lemmas \ref{lemma: D}, \ref{lemma: H,Q} and \ref{lemma: HDH} continue to hold. This will change the convergence rate in Theorem \ref{thm: consistency for loadings} to $\delta_{N_y, T^\alpha} = \min(\Ny, T^\alpha)$ and the convergence rate in Theorem \ref{thm: asymptotic distribution} to $\sqrt{T^\alpha}$ for the loadings of $Y$, and $\sqrt{\delta_{N_y, T^\alpha}}$ for the factors and common components of $Y$.
\end{proof}

\end{onehalfspacing}

\singlespacing
\bibliographystyle{econometrica}
{\small
\bibliography{reference}
}
\onehalfspacing